\documentclass[aos]{imsart}

\pdfoutput=1

\usepackage{footmisc} %in order for footnote hyperlinks to work, footmisc should be loaded before hyperref

\RequirePackage{amsthm,amsmath,amssymb}
\RequirePackage[sort&compress,numbers]{natbib}
\RequirePackage[colorlinks,citecolor=blue,urlcolor=blue,filecolor=cyan]{hyperref}
\RequirePackage{graphicx}% uncomment this for including figures

\usepackage{titlesec}

\titleclass{\part}{top} % make part like a chapter
\titleformat{\part}
[display]
{\centering}
{\title{Foundations of Structural Causal Models with Cycles and Latent Variables}}
{-3pt}
{\normalfont}
\titlespacing*{\part}{0pt}{0pt}{1cm}

\usepackage{xcolor}

\usepackage[OT1]{fontenc}
\usepackage[font=scriptsize,skip=2pt,textfont=it]{caption}
\usepackage{centernot}
\usepackage{float}
\usepackage{bm}
\usepackage{boldline}
\usepackage{array}
\usepackage{enumerate}
\usepackage{enumitem}
\usepackage{adjustbox}
\usepackage{thmtools}
\usepackage[outline]{contour}
\usepackage{url}
\usepackage{tikz}
\usepackage{multirow}
\usepackage{subfig}
\usepackage{nameref}
\startlocaldefs
\newtheorem{theorem}{Theorem}[section]
\newtheorem{proposition}[theorem]{Proposition}
\newtheorem{corollary}[theorem]{Corollary}

\newtheorem{definition}[theorem]{Definition}

\newtheorem{lemma}[theorem]{Lemma}

\newtheorem{example}[theorem]{Example}
\newtheorem*{remark}{Remark}
\let\oldFootnote\footnote
\newcommand\nextToken\relax

\renewcommand\footnote[1]{%
    \oldFootnote{#1}\futurelet\nextToken\isFootnote}

\newcommand\isFootnote{%
    \ifx\footnote\nextToken\textsuperscript{,}\fi}

\usetikzlibrary{cd}
\usetikzlibrary{arrows,fit}
\usetikzlibrary{patterns,calc}
\usetikzlibrary{matrix,decorations.pathreplacing,decorations.pathmorphing,calc,positioning}
\tikzstyle{var}=[circle,draw=black,fill=white,thick,minimum size=20pt,inner sep=0pt]
\tikzstyle{exvar}=[circle,draw=black,fill=gray,thick,minimum size=20pt,inner sep=2pt]
\tikzstyle{arr}=[->,>=stealth',draw=black,line width=1pt]
\tikzstyle{biarr}=[<->,>=stealth',draw=black,line width=1pt]
\newcommand\ellipsebyfoci[4]{% focus pt1, focus pt2, cste, ecc
  let \p1=(#1), \p2=(#2), \p3=($(\p1)!.5!(\p2)$)
  in \pgfextra{
    \pgfmathsetmacro{\angle}{atan2(\y2-\y1,\x2-\x1)}
    \pgfmathsetmacro{\focal}{veclen(\x2-\x1,\y2-\y1)/2/1cm}
    \pgfmathsetmacro{\lentotcm}{\focal*2*#3}
    \pgfmathsetmacro{\axeone}{(\lentotcm - 2 * \focal)/2+\focal}
    \pgfmathsetmacro{\axetwo}{sqrt((\lentotcm/2)*(\lentotcm/2)-\focal*\focal}
  }
  (\p3) ellipse (\axeone cm and #4*\axetwo cm);
}

\tikzset{
  symbol/.style={
    draw=none,
    every to/.append style={
      edge node={node [sloped, allow upside down, auto=false]{$#1$}}}
  }
}

\newcommand{\hyperrefbox}[2]{\begingroup\hypersetup{hidelinks,allbordercolors=blue,pdfborder={0 0 1}}\hyperref[#1]{#2}\endgroup}
\newcommand{\hyperrefsuppbox}[2]{\begingroup\hypersetup{hidelinks,allbordercolors=blue,pdfborder={0 0 1}}\hyperref[#1]{#2}\endgroup}
\newcommand{\hyperlinkbox}[2]{\begingroup\hypersetup{hidelinks,allbordercolors=blue,pdfborder={0 0 1}}\hyperlink{#1}{#2}\endgroup}

\newcommand\eref[1]{(\ref{#1})}

\newcommand\de[1]{\mathrm{de}({#1})}

\newcommand\intervene{\mathrm{do}}

\newcommand\marg{\mathrm{marg}}

\newcommand\pa{\mathrm{pa}}

\renewcommand\de{\mathrm{de}}
\newcommand\ch{\mathrm{ch}}
\newcommand\an{\mathrm{an}}

\newcommand\scc{\mathrm{sc}}

\newcommand\NN{{\mathbb N}}
\newcommand\RN{\mathbb{R}}

\newcommand\id{\mathbb I}
\newcommand\Prb{\mathbb{P}}
\newcommand\given{\,|\,}
\newcommand\indep{{\,\perp\mkern-12mu\perp\,}}

\newcommand\sigmablocked{{\,\perp\,}}

\newcommand\twin{\mathrm{twin}}
\newcommand{\equivobs}[1]{\equiv_{\text{obs}(#1)}}
\newcommand{\equivint}[1]{\equiv_{\text{int}(#1)}}
\newcommand{\equivcf}[1]{\equiv_{\text{cf}(#1)}}
\newcommand\acy{\mathrm{acy}}

\newcommand{\ot}{\leftarrow}
\newcommand{\oto}{\leftrightarrow}

\newcommand\B[1]{\bm{#1}}
\newcommand\C[1]{\mathcal{#1}}
\newcommand\BC[1]{\B{\C{#1}}}

\newcommand{\foralmostall}{{\scalebox{1}[1.156]{$\vee$}}\hspace*{-0.74em}\raisebox{0.5ex}{\scalebox{1}[0.8]{$\sim$}}}

\newcommand{\Stephan}[1]{{\color{orange}{#1}}}
\newcommand{\Jonas}[1]{{\color{blue}{Jonas: #1}}}
\newcommand{\Joris}[1]{{\color{purple}{Joris: #1}}}
\newcommand{\Patrick}[1]{{\color{orange}{Patrick: #1}}}

\renewcommand{\Stephan}[1]{}
\renewcommand{\Jonas}[1]{}
\renewcommand{\Joris}[1]{}
\renewcommand{\Patrick}[1]{}

\renewcommand\thmcontinues[1]{Continued}
\endlocaldefs

\begin{document}

\begin{frontmatter}

\title{Foundations of Structural Causal Models with Cycles and Latent Variables}
\runtitle{Foundations of Structural Causal Models}

%\if0
\begin{aug}
\author[A]{\fnms{Stephan} \snm{Bongers}\ead[label=e1]{s.r.bongers@uva.nl}},
\author[A]{\fnms{Patrick} \snm{Forr\'e}\ead[label=e4]{p.d.forre@uva.nl}},
\author[B]{\fnms{Jonas} \snm{Peters}\ead[label=e2]{jonas.peters@math.ku.dk}}
\and
\author[C]{\fnms{Joris M.} \snm{Mooij}\ead[label=e3]{j.m.mooij@uva.nl}}

%\thankstext{t1}{Supported in part by NWO, the Netherlands Organization for
%Scientific Research (VIDI grant 639.072.410 and VENI grant 639.031.036).}
%\thankstext{t2}{Supported in part by the European Research Council (ERC) under
%the European Union's Horizon 2020 research and innovation programme (grant
%agreement n$^{\mathrm{o}}$ 639466).}
%\thankstext{t3}{Supported by research grants from VILLUM FONDEN (18968) and the Carlsberg Foundation.}

\address[A]{Informatics Institute, 
University of Amsterdam,
\printead{e1,e4}}

\address[B]{Department of Mathematical Sciences,
University of Copenhagen,
\printead{e2}}

\address[C]{Korteweg-De Vries Institute,
University of Amsterdam,
\printead{e3}}
\end{aug}
%\fi

\begin{abstract}
Structural causal models (SCMs), also known as (nonparametric) structural
equation models (SEMs), are widely used for causal modeling purposes. In particular, acyclic 
SCMs, also known as recursive SEMs, form a well-studied subclass of SCMs that 
generalize causal Bayesian networks to allow for latent confounders. In this paper, we investigate SCMs 
in a more general setting, allowing for the presence of both latent confounders and cycles.
%Although SCMs with cycles have been used to model feedback systems in real world applications, in general they miss many of the desirable properties of its acyclic counterpart: 
We show that in the presence of cycles, many of the convenient properties of acyclic SCMs do not hold in general: they do not always have a solution; they do not always induce unique observational, interventional and counterfactual distributions; a marginalization does not always exist, and if it exists the marginal model does not always respect the latent projection; they do not always satisfy a Markov property; and their graphs are not always consistent with their causal semantics. 
%Only some of these properties have been extended to the cyclic setting for the discrete and linear cases, and more recently different Markov properties have been derived for graphs with cycles and hyperedges. 
%Except for the discrete and linear case, and their Markov properties, almost nothing general is known for SCMs with cycles and latent variables. 
%We show to which extent and under which conditions each of these properties hold for SCMs in general. 
We prove that for SCMs in general each of these properties does hold under certain solvability conditions. Our work generalizes results for SCMs with cycles that were only known for certain special cases so far. 
%We provide a necessary condition for the existence of a solution of an SCM, and 
We introduce the class of simple SCMs that extends the class of acyclic SCMs to the cyclic setting, while preserving many of the convenient properties of acyclic SCMs.
With this paper we aim to provide the foundations for a general theory of statistical causal modeling with SCMs.
%Structural causal models (SCMs), also known as (non-parametric) structural
%equation models (SEMs), are widely used for causal modeling purposes. Acyclic 
%SCMs, also known as recursive SEMs, form a special well-studied subclass of SCMs that 
%are closely related to causal Bayesian networks. This paper investigates SCMs in
%the presence of cycles and latent variables. We introduce a class of SCMs,
%called simple SCMs, that extends the class of acyclic SCMs to the cyclic
%setting, while preserving many of the convenient properties of
%acyclic SCMs: they induce unique observational, interventional and
%counterfactual distributions; they are closed under intervention and
%marginalization; the marginalization respects the latent projection; they 
%satisfy several Markov properties; and their graphs describe their causal
%relationships. Each of these properties does not hold in general for non-simple
%SCMs. We show to which extent and under which conditions each of these properties holds for non-simple SCMs.
%\Joris{Why model cycles at all? The abstract needs to be improved. Mention challenges
%that needed to be addressed in the cyclic case: solvability issues, measurability issues,
%finding appropriate generalizations of all basic concepts. Also, now it reads as if the
%main contribution is the introduction of simple SCMs. Turn it around: first mention what
%complications one has in the general cyclic case, and at the end mention simple SCMs as
%a nice class that are about as easy to work with as acyclic SCMs.}
\end{abstract}


%MSC2020 codes:
%62A09 - Graphical methods
%68T30 - Knowledge representation
%68T37 - Reasoning under uncertainty in the context of artificial intelligence


\end{frontmatter}

%%%%%%%%%%%%%%%%%%%%%%%%%%%%%%%%%%%%%%%%%%%%%%%%%%
\section{Introduction}
\label{sec:Introduction}
%%%%%%%%%%%%%%%%%%%%%%%%%%%%%%%%%%%%%%%%%%%%%%%%%%

\Joris{Wishlist of nice additions, will be postponed until after submission:
\begin{itemize}
  \item Introduce a DMG version of the graphical twin operator (the bidirected edges
between each node $i$ and its copy $i'$ that could be redundant in case $i$ has no exogenous parents are similar to these type of nondeterminism-only edges that are by default introduced into the graphical marginalization and acyclification of a DMG.)
  \item See if the old Proposition 5.2.8 can be added again, with a correct proof
  \item Add a proof that the graphical marginalization commutes rather than referring to Evans' proof
\end{itemize}}

Structural causal models (SCMs), also known as (nonparametric) structural equation models 
(SEMs), are widely used for causal modeling purposes
\citep{Bol89,SGS00,Pea09,PJS17}. They form the basis for many statistical
methods that aim at inferring knowledge of the underlying causal structure from
data \citep[see, e.g.,][]{MCKB09,MH13,PMJS14,BPE14,MPJ+16}. In these models, 
the causal relationships between the variables are expressed in the form of deterministic, functional relationships, 
and probabilities are introduced through the assumption that certain variables are exogenous 
latent random variables. SCMs arose out of certain causal models that were first introduced in 
genetics \citep{Wri21}, econometrics \citep{Haa43}, electrical engineering \citep{Mas53, Mas56} and the social 
sciences \citep{GD73, Dun75}. 

Acyclic SCMs, also known as recursive SEMs, form a special well-studied subclass of SCMs that
generalize causal Bayesian networks \citep{Pea09}.
%to allow for modeling latent confounders. 
They have many convenient properties~\citep[see, e.g.,][]{Pea85,LDLL90, Ver93, Lau96, Ric03, Eva16, Eva18}: (i) they induce a unique distribution over the variables; (ii) they are closed under perfect interventions; (iii) they are closed under marginalizations; (iv) their marginalization respects the latent projection; (v) they obey (various equivalent versions of) the Markov property and (vi) their graphs express the causal relationships encoded by the SCM in an intuitive manner.

One important limitation of acyclic SCMs is that they cannot 
model systems that involve causal cycles. 
%In many systems occurring in the real world, feedback loops between observed variables are present \citep[see e.g.,][]{Haa43,Mas53,Mas56,MH13,PBP19}. 
In many systems occurring in the real world, there are feedback loops between observed variables. For example, in economics the price of a product may be a function of the demanded or supplied quantities, and vice versa, the demanded and supplied quantities may be functions of the price. The underlying dynamic processes describing such systems have an acyclic causal structure over time. However, causal cycles may arise when one approximates such systems over time~\citep{Fis70,MMH18,MH20} or when one describes the equilibrium states of these systems~\citep{IS94,LSRH08,HEH12,MJS13,BM18,BBM19,PBP19}. In particular, in~\citep{BM18} it was shown that the equilibrium states of a system governed by (random) differential equations can be described by an SCM that represents their causal semantics, which gives rise to a plethora of SCMs that include cycles (we provide some examples of such feedback systems in Appendix~\ref{app:AppendixExamplesEquilibriumModels} of the Supplementary Material).
%For those systems, SCMs with cycles (or ``non-recursive SEMs'') may form an appropriate model class \citep{MJS13,BM18}. 
In contrast to their acyclic counterparts, SCMs with cycles have enjoyed less attention in the literature and are not as
well understood. In general, none of the above
properties (i)--(vi) hold in the class of SCMs.
%this class. 
However, some progress has been made in the case of discrete \citep{PD96,Nea00} and linear models \citep{Spi93, Spi94, Spi95, Ric96c, Kos96, HEH12}, and more recently, for more general cyclic models the Markov properties have been elucidated \citep{FM17}.
%in more generality \citep{FM17}. 

\paragraph*{Contributions}
The purpose of this paper is to provide the foundations for a general theory of statistical causal modeling with SCMs. We study properties of SCMs and allow for cycles, latent variables and nonlinear functional
relationships between the variables. We investigate to which extent and under which sufficient conditions each of the properties (i)--(vi) 
%still 
holds, in particular, in the presence of cycles. In the next paragraphs, we describe our contributions in more detail.

When there are cyclic functional relationships between variables, one encounters various technical complications, which
%, as we will see 
even arise in the linear setting. The structural equations of an acyclic SCM trivially have a unique 
solution. This unique solvability property ensures that the SCM gives rise to a 
unique, well-defined probability distribution on the variables. 
In the case of cycles, however, this property may be violated, and
consequently, the SCM may not have a solution at
all, or may allow for multiple different probability distributions~\citep{Hal98}.
%, which leads to ambiguity 
Even if one starts with a cyclic SCM that is uniquely solvable, 
performing an intervention on the SCM may lead to an intervened SCM that is not uniquely 
solvable. Hence, a cyclic SCM may not give rise to a unique, well-defined probability 
distribution corresponding to that intervention, and whether or not this happens may depend
on the intervention. We provide sufficient conditions for the existence and uniqueness of these probability distributions after intervention. In general, it is not clear whether the solutions of the structural equations of an SCM 
are measurable if cycles are present. In addition, we provide sufficient and necessary conditions for 
the measurability of solution functions of cyclic SCMs.

SCMs provide a detailed modeling description of a system. Not all information may be necessary
for a certain modeling task, which motivates to consider certain classes of SCMs to be equivalent.
In this paper, we formally introduce several of such equivalence relations. For example, we consider two SCMs observationally equivalent if they
cannot be distinguished based on observations alone. Observationally equivalent SCMs can often
still be distinguished by interventions. We consider two SCMs interventionally equivalent if
they cannot be distinguished based on observations and interventions. While
these concepts have been around in implicit form for acyclic SCMs, we formulate them in such a way that they also
apply to cyclic SCMs that have either no solution at all or have multiple different induced probability 
distributions on the variables. Finally, we consider two SCMs counterfactually equivalent if they 
cannot be distinguished based on observations and interventions and in addition encode the 
same counterfactual distributions, which are the distributions induced by the so-called twin SCM via the 
twin network method~\cite{BP94}. 
These different equivalence relations formalize the different levels of abstraction in the so-called
causal hierarchy~\citep{SP08,PM18}. In addition, we add another, strong version of equivalence, such that equivalent SCMs have the same solutions. This notion clarifies ambiguities when a function is constant in one of its arguments, for example.

Marginalization becomes useful if not all variables are observed: given a
joint probability distribution on some variables, we obtain a marginal distribution on a subset of
the variables by integrating out the remaining variables. Analogously, we
can marginalize an acyclic SCM by substituting the solutions of the structural equations of
a subset of the endogenous variables into the structural equations of the
remaining endogenous variables. For acyclic SCMs, the induced observational and
interventional distributions of the marginalized SCM coincide with the marginals of the distributions induced by the original SCM~\cite[see][a.o.]{Ver93,SRM+98,Eva16,Eva18}.
In other words, for acyclic SCMs the operation of marginalization preserves the  probabilistic and causal semantics (restricted to the remaining variables). 
We show that for cyclic SCMs a marginalization does not always exist without further assumptions.
%While this was already known for acyclic SCMs \cite[see][a.o.]{Ver93,SRM+98,Eva16,Eva18}, we show 
%that for cyclic SCMs this is not true without further assumptions.
%For example, the substitution method of acyclic SCMs is only feasible if there are no cyclic relationships in the model and breaks if cycles are present. 
In \citep{FM17} it is shown that for modular SCMs, which can be seen as an SCM together with an additional structure of a compatible system of solution functions, a marginalization can be defined that preserves the probabilistic and causal semantics. 
We prove that this additional structure is not necessary and use a local unique solvability condition instead.
%In \citep{FM17} it is shown that under the strong assumption of ``modularity'' a marginalization can also be defined for cyclic SCMs and that it preserves the probabilistic and causal semantics. We show here that marginalizations of SCMs can be defined under weaker assumptions. 
%Intuitively, the idea is to think of the subset of endogenous variables that are going to be marginalized out as a subsystem that interacts with the rest of the system. Under a certain unique solvability condition, one can ignore the internals of this subsystem and treat it effectively as a ``black box'', which has a unique output for every possible input, and this can be substituted into the structural equations of the remaining variables. 
Under this condition, we show that an SCM and its marginalization
%original and the marginal SCM 
are observationally, interventionally and counterfactually
equivalent on the remaining endogenous variables. Analogously, we define a
marginalization operation on the associated graph of an SCM, which
generalizes the latent projection \citep{Ver93,Tia02,Eva16}. 
In general, the marginalization of an SCM does not respect the latent projection of its associated graph, 
but we show that it does so under an additional local ancestral unique solvability condition.

In graphical models, Markov properties allow one to read off conditional independencies in a distribution
directly from a graph. 
Various equivalent formulations of Markov properties exist for acyclic SCMs \citep{Lau96},
one prominent example being the $d$-separation criterion, also known as 
the directed global Markov property, which was originally derived for Bayesian networks \citep{Pea85}.
Markov properties have been of key importance to derive various central results regarding causal reasoning and causal discovery. 
%was first shown for Bayesian networks \citep{Pea85}, but holds more generally. 
For cyclic SCMs, however, the usual Markov properties do not
hold in general, as was already pointed out by Spirtes~\cite{Spi94}. 
His solution in terms of collapsed graphs was recently generalized and reformulated 
for a general class of causal graphical models \citep{FM17} 
by adapting the notion of $d$-separation into what has been termed
$\sigma$-separation. This resulted in a general directed global Markov property 
expressed in terms of $\sigma$-separation instead of $d$-separation.
Here, we formulate these general Markov properties specifically within the framework of SCMs.
Again, they only hold under certain unique solvability conditions.
%that under certain 
%unique solvability assumptions the observational, interventional and counterfactual 
%distributions satisfy the (general) directed global Markov property relative to their corresponding graphs.

In addition to its interpretation in terms of conditional independencies, the graph of an acyclic
SCM also has a direct causal interpretation \citep{Pea09}. 
As was already observed in~\cite{Nea00}, the causal interpretation of SCMs with
cycles can be counterintuitive, as the causal semantics under interventions no longer
needs to be compatible with the structure imposed by the functional relations between the variables.
We resolve this issue by showing that under certain ancestral unique solvability conditions the causal interpretation of SCMs is consistent with its graph. 

Cycles lead to several technical complications 
related to solvability issues. 
We introduce a special subclass of (possibly cyclic) SCMs, the class of simple SCMs, for which most of these technical complications are absent and which preserves much of the simplicity of the theory for acyclic SCMs.
%setting,
%and generalizes causal Bayesian networks to admit cycles and confounders.
%A simple SCM is defined to be an SCM that is uniquely solvable after every perfect
%intervention. 
A simple SCM is an SCM that is uniquely solvable with respect to every subset of the variables. 
Because of this strong solvability assumption, simple SCMs have all the convenient properties
(i)--(vi): they always have uniquely defined
observational, interventional and counterfactual distributions; we can perform every perfect intervention and marginalization on them and the result is again a simple SCM; marginalization does respect the latent projection; they obey the general
directed global Markov property, and for special cases (including the acyclic, linear and discrete case)
they obey the (stronger) directed global Markov property; their graphs have a 
direct and intuitive causal interpretation. 
%We discuss how simple SCMs can be used for causal discovery and prediction. 

The scope of this paper is limited to establishing the foundations for statistical
causal modeling with cyclic SCMs (Figure~\ref{fig:OverviewCausalModels} in Appendix~\ref{app:AppendixOverviewCausalGraphicalModels} of the Supplementary Material shows an overview of how SCMs relate to other causal graphical models). For a detailed discussion of causal reasoning,
causal discovery and causal prediction with cyclic SCMs we refer the reader to other
literature \citep[e.g., ][]{Ric96a, Ric96b, RS99, EHS10, HEH12, HHEJ13, FDD12}.
Several recent results (generalizations of the do-calculus, adjustment criteria
and an identification algorithm) for modular SCMs \citep{FM18, FM19} directly apply
to the subclass of simple SCMs, as well. Finally, many causal discovery algorithms that have been designed for the acyclic case also apply to simple SCMs with no or only minor
changes \citep{MMC19,MooijClaassen_2005.00610}.

\paragraph*{Overview}

Figure~\ref{fig:OverviewDiagram} gives an overview of the different objects that can be constructed from an SCM and the different mappings between them. For pairs of mappings between the objects with the names in bold, we prove commutativity results which are summarized in Table~\ref{tab:CommutationRelations}.
\begin{figure}
\adjustbox{scale=0.78,center}{%
  \tikzcdset{arrow style=tikz,diagrams={>=latex}}
  \begin{tikzcd}[/tikz/row 1/.style={nodes={draw, rounded corners}}, /tikz/row 3/.style={nodes={draw, rounded corners}}, /tikz/row 4/.style={nodes={draw, rounded corners}}, /tikz/row 5/.style={nodes={draw, rounded corners}}, /tikz/row 7/.style={nodes={draw, rounded corners}}, every label/.append style = {font = \small}] 
    \hyperrefbox{def:intervenedSCM}{\textbf{\begin{tabular}{c} intervened \\ SCM \end{tabular}}} \arrow[ddd, dashed, every label/.append style={left,xshift=-1ex}, "{\text{\ref{thm:SolvabilityIffCondition}}}"] & \hyperrefbox{def:TwinSCM}{\textbf{\begin{tabular}{c} twin \\ SCM \end{tabular}}} &[5pt] \hyperrefbox{def:MarginalSCM}{\textbf{\begin{tabular}{c} marginal \\ SCM \end{tabular}}} &[-25pt] & \hyperrefbox{def:InterventionOnGraph}{\textbf{\begin{tabular}{c} intervened \\ graph \end{tabular}}} & \hyperrefbox{def:TwinGraph}{\textbf{\begin{tabular}{c} twin \\ graph \end{tabular}}} & \hyperrefbox{def:LatentProjection}{\textbf{\begin{tabular}{c} marginal \\ graph \end{tabular}}} \\
& & & & & & \\
& \hyperrefbox{def:SCM}{\textbf{SCM}} \arrow[uul, shift={(-3pt,0pt)}, line width=1pt, every label/.append style={xshift=1ex, yshift=2ex},"{\begin{tabular}{c} $\intervene$ \\ \ref{def:intervenedSCM} \end{tabular}}"] \arrow[uu, shift={(-3pt,0pt)}, line width=1pt, every label/.append style={xshift=1.1ex,yshift=0.6ex},"{\begin{tabular}{c} $\twin$ \\ \ref{def:TwinSCM} \end{tabular}}"] \arrow[uur, shift={(-5pt,0pt)}, dashed, line width=1pt, every label/.append style={right,xshift=1.4ex,yshift=-0.8ex},"{\begin{tabular}{c} $\marg$ \\ \ref{def:MarginalSCM} \end{tabular}}"] \arrow[rrrr, shift={(-8pt,0pt)}, line width=1pt, every label/.append style={xshift=3ex,yshift=-0.2ex},"{\C{G} \text{ (or } \C{G}^a \text{)}}","\ref{def:Graphs}"'] \arrow[ddl, shift={(-9pt,0pt)}, dashed, every label/.append style={xshift=0ex}, "{\text{\ref{thm:SolvabilityIffCondition}}}"] \arrow[dd, shift={(-9pt,0pt)}, dashed, every label/.append style={xshift=-0.8ex,yshift=0.5ex},"{\begin{tabular}{c} $\acy$ \\ \ref{def:Acyclification} \end{tabular}}"] & & & & \hyperrefbox{def:Graphs}{\textbf{\begin{tabular}{c} (augmented) \\ graph \end{tabular}}} \arrow[uul, shift={(-12pt,0pt)}, line width=1pt, every label/.append style={xshift=1ex, yshift=2ex},"{\begin{tabular}{c} $\intervene$ \\ \ref{def:InterventionOnGraph} \end{tabular}}"] \arrow[uu, shift={(-12pt,0pt)}, line width=1pt, every label/.append style={xshift=1ex,yshift=0.6ex},"{\begin{tabular}{c} $\twin$ \\ \ref{def:TwinGraph} \end{tabular}}"] \arrow[uur, shift={(-12pt,0pt)}, line width=1pt, every label/.append style={right,xshift=-0.5ex,yshift=-2ex},"{\begin{tabular}{c} $\marg$ \\ \ref{def:LatentProjection} \end{tabular}}"] \arrow[ddl, shift={(-18pt,0pt)}, every label/.append style={right,xshift=-2ex,yshift=-2ex},"{\begin{tabular}{c} $\acy$ \\ \ref{def:GraphicalAcyclification} \end{tabular}}"] \arrow[ddr, shift={(-18pt,0pt)}, dashed, every label/.append style={xshift=0.3ex,yshift=0.1ex},"\ref{def:DirectCausesConfounders}"] \arrow[dddd, shift={(-18pt,0pt)}, every label/.append style={xshift=-1.3ex},"{\begin{tabular}{c} $d/\sigma$-sep. \\ \ref{def:DSeparation} / \ref{def:SigmaSeparation} \end{tabular}}"] & \\[-28pt]
\hyperlinkbox{def:InterventionalDistribution}{\begin{tabular}{c} interventional \\ distribution(s) \end{tabular}} & & & & & & \\
  \hyperlinkbox{def:ObservationalDistribution}{\begin{tabular}{c} observational \\ distribution(s) \end{tabular}} \arrow[ddrr, shift={(-24pt,0pt)}] & \hyperrefsuppbox{def:Acyclification}{\begin{tabular}{c} acyclified \\ SCM \end{tabular}} \arrow[r, shift={(-24pt,0pt)}, every label/.append style={yshift=-0.2ex}, "\C{G}","\ref{def:Graphs}"'] & \hyperrefbox{def:Graphs}{\begin{tabular}{c} graph of the \\ acyclified \\ SCM \end{tabular}} \arrow[rr, shift={(-24pt,0pt)}, symbol=\subseteq, every label/.append style={xshift=-0.2ex,yshift=-4.1ex}, "{\ref{prop:CompatibilityAcyclification}}"] & & \hyperrefsuppbox{def:GraphicalAcyclification}{\begin{tabular}{c} acyclified \\ graph \end{tabular}} \arrow[rr, shift={(-24pt,0pt)}, phantom, every label/.append style={yshift=-1.8ex}, "{\begin{tabular}{c} \scalebox{2}{$\circlearrowleft$} \\ \scalebox{1.1}{\ref{prop:SigmaSeparationAsDSeparation}} \end{tabular}}", near start] \arrow[ddr, shift={(-24pt,0pt)}, every label/.append style={left,xshift=5ex,yshift=0.5ex},"{d\text{-sep. }\,\,\,\ref{def:DSeparation}}"] & & \hyperrefbox{def:DirectCausesConfounders}{{\renewcommand{\arraystretch}{1.25}\begin{tabular}{c} direct causes, \\ causes, \\ confounders \end{tabular}}} \\
& & & & & & \\
& & \hyperrefbox{thm:globalMarkovPropertiesSCMs}{\begin{tabular}{c} (conditional) \\ independencies \end{tabular}} \arrow[rrr, bend right=15, shift={(-30.5pt,0pt)}, dashed, rightarrow, every label/.append style={yshift=0.25ex}, "\text{faithfulness}"',"\text{\ref{def:dFaithfulness} / \ref{def:sigmaFaithfulness}}"] & & & \hyperrefbox{thm:globalMarkovPropertiesSCMs}{\begin{tabular}{c} $d/\sigma$- \\ separations \end{tabular}} \arrow[lll, bend right=15, shift={(-32.5pt,0pt)}, dashed, rightarrow, every label/.append style={yshift=-0.25ex}, "\text{Markov properties}"',"\text{\ref{def:Markov_property} / \ref{def:generalized_Markov_property}}"] \\[-16pt]
%& & \text{\ref{thm:dgMarkovPropertySCMThreeSpecialCases} / \ref{thm:gdgMarkovPropertySCM}}& & & \text{\ref{def:DSeparation} / \ref{def:SigmaSeparation}} & \\
& & & & & & \\
\end{tikzcd}
}
\caption{Overview of the objects constructed from an SCM and the mappings
between them. The numbers correspond to the definition, proposition or theorem of the corresponding object, mapping or result. When an arrow is dashed, the relation only holds under nontrivial assumptions that can be found in the corresponding definition or theorem. 
%The ``faithfulness'' dashed arrow corresponds to the faithfulness assumption that is often made in causal discovery, but we will not discuss it further in this paper. 
The symbol ``$\subseteq$'' stands for the subgraph of a directed mixed graph (see Definition~\ref{def:DirectedMixedGraph} in the Supplementary Material) and the symbol ``\protect\scalebox{1.2}{$\circlearrowleft$}'' denotes that the surrounding diagram commutes. Table~\ref{tab:CommutationRelations} gives an overview of the commutativity results for each pair of mappings between the objects with the names in bold.}
\label{fig:OverviewDiagram}
\end{figure}

\addtocounter{figure}{-1}
\begin{figure}
\captionsetup{name=Table}
\begin{tabular}{cccc}
  \scalebox{0.9}{{\renewcommand{\arraystretch}{1.2}
      \begin{tabular}[t]{c V{2.5} c|c|c}
    \textbf{SCMs} & $\intervene$ & $\twin$ & $\marg$ \\ \hlineB{2.5}
    $\C{G}, \C{G}^a$ & \ref{prop:InterventionOnGraph} & \ref{prop:CommuteTwinGraph} & (\ref{prop:LatentProjection}) \\ \hline
	$\intervene$ & \ref{prop:IntDisjSubsetCommute} & \ref{prop:CommuteInterveneTwin1} & \ref{prop:MarginalizationCommuteWithIntervention} \\ \hline
    $\twin$ & $\cdots$ & - & \ref{prop:MarginalizationCommuteWithTwinning} \\ \hline
	$\marg$ & $\cdots$ & $\cdots$ & \ref{prop:MarginalizationCommutes}
\end{tabular}}}
 & & &
\scalebox{0.9}{{\renewcommand{\arraystretch}{1.2}
    \begin{tabular}[t]{c V{2.5} c|c|c}
    \textbf{Graphs} & $\intervene$ & $\twin$ & $\marg$ \\ \hlineB{2}
  $\intervene$ & \ref{prop:IntDisjSubsetCommute} & \ref{prop:CommuteInterveneTwin2} & \ref{prop:LatentProjectionCommuteWithIntervention} \\ \hline
	$\twin$ & $\cdots$ & - & \ref{prop:LatentProjectionCommuteWithTwin} \\ \hline
	$\marg$ & $\cdots$ & $\cdots$ & \ref{prop:LatentProjectionCommutes}
\end{tabular}}}
\end{tabular}
  \vspace{0.5em}
  \caption{Overview of the commutativity results of different pairs of mappings, defined on SCMs (left table) and on graphs (right table). All results apply under the assumptions stated in the corresponding proposition. The entries denoted by dots are omitted due to symmetry. We do not consider the commutativity of the twin operation with itself in this paper. Proposition \ref{prop:LatentProjection} (in parentheses) is not a commutativity result but a weaker relation. The graphical twin operator is only defined for directed graphs.}
  \label{tab:CommutationRelations}
\end{figure}

\paragraph*{Outline}
This paper is structured as follows: In Section~\ref{sec:SCM}, we provide a
formal definition of SCMs and a natural notion of equivalence between SCMs, 
define the (augmented) graph corresponding to an SCM, and describe perfect interventions
and counterfactuals. In Section~\ref{sec:Solvability}, we discuss the concept of
(unique) solvability, its properties and how it relates to self-cycles. 
%and define the class of simple SCMs. 
In Section~\ref{sec:Equivalences}, we define and relate various equivalence relations between SCMs. In
Section~\ref{sec:Marginalization}, we define a marginalization operation that
is applicable to cyclic SCMs under certain conditions. We discuss several
properties of this marginalization operation and discuss the relation with 
a marginalization
operation defined on directed mixed graphs. In Section~\ref{sec:MarkovProperty}, we discuss Markov properties of SCMs. In Section~\ref{sec:CausalInterpretationSCMs}, we discuss the causal interpretation of the graphs of SCMs. Section~\ref{sec:SimpleSCMs} introduces and discusses the class of simple SCMs. 

The Supplementary Material introduces causal graphical models in Appendix~\ref{app:AppendixCausalGraphicalModels}. This section also contains details on Markov properties and modular SCMs. Appendix~\ref{app:AppendixSolvabilityResults} provides additional (unique) solvability properties, some results for linear
SCMs are discussed in Appendix~\ref{app:AppendixLinSCMs}, other examples in Appendix~\ref{app:AppendixExamples} and the proofs of all
the theoretical results are in Appendix~\ref{app:AppendixProofs}. Appendix~\ref{app:AppendixMST} contains
some lemmas and measurable selection theorems that are used in several proofs.

%%%%%%%%%%%%%%%%%%%%%%%%%%%%%%%%%%%%%%%%%%%%%%%%%%
\section{Structural causal models}
\label{sec:SCM}
%%%%%%%%%%%%%%%%%%%%%%%%%%%%%%%%%%%%%%%%%%%%%%%%%%

%In this section, we provide the basic definitions and properties of structural causal models (SCMs). We start in Section~\ref{sec:SCMdef} by formally defining SCMs and their solutions, and introducing a natural equivalence relation on SCMs. 
%In Section~\ref{sec:Graph}, we introduce graphical representations of an SCM.
%We then show in Section~\ref{sec:StructMinimalRepresentation} that each SCM has a representative with a sparsest graph, its structurally minimal representation. 
%In Section~\ref{sec:Interventions} we introduce perfect interventions, a key notion that grounds the causal semantics of SCMs. 
%We finish in Section~\ref{sec:Counterfactuals} with a definition of the twin SCM and how this gives rise to counterfactuals.
%
In this section, we provide the definition and properties of structural causal models (SCMs). Our definition of SCMs slightly deviates from existing definitions~\citep{Bol89,Pea09,SGS00}, because we make the definition of the SCM independent of the random variables that solve it. This enables us to deal with the various technical complications that arise in the presence of cycles.

%%%%%%%%%%%%%%%%%%%%%%%%%%%%%%%%%%%%%%%%%%%%%%%%%%
\subsection{Structural causal models and their solutions}
\label{sec:SCMdef}
%%%%%%%%%%%%%%%%%%%%%%%%%%%%%%%%%%%%%%%%%%%%%%%%%%

%In probabilistic models one usually considers random variables to be part of the model. When defining an SCM, we do not require that there are random variables satisfying the structural equations. As we will see in Section~\ref{sec:Interventions}, it may happen that a solvable SCM may have no solution anymore after performing an intervention on the SCM. Conversely, it may also happen that intervening on an SCM without any solution gives an SCM with a solution. By separating the SCMs from their solutions, interventions can always be defined. So, instead of defining causal models as probabilistic models, we take here the opposite approach and define probabilistic models in terms of causal models. From this viewpoint, every probabilistic model that we consider subsumes the existence of some underlying causal model \Joris{rephrase}. Conversely, a causal model does not necessarily subsume the existence of a probabilistic model, for example, if the causal model does not have a solution in terms of random variables.

\begin{definition}[Structural causal model]
\label{def:SCM}
A \emph{structural causal model (SCM)} is a tuple\footnote{We often use boldface for 
variables that have multiple components, for example, vectors
%or tuples
in a Cartesian product.}
$$\C{M} := \langle \C{I}, \C{J}, \BC{X}, \BC{E}, \B{f}, \Prb_{\BC{E}} \rangle \,,$$
where 
\begin{enumerate}
  \item 
    $\C{I}$ is a finite index set of \emph{endogenous variables}, 
  \item 
    $\C{J}$ is a disjoint finite index set of \emph{exogenous variables},
  \item 
    $\BC{X} = \prod_{i \in \C{I}} \C{X}_i$ is the product of the domains of the 
    endogenous variables, where each domain $\C{X}_i$ is a standard measurable space (see Definition~\ref{def:StandardMeasurableSpace}),
%    \footnote{A standard measurable space is a measurable space $(\BC{X},\B{\Sigma})$ that is isomorphic to a measurable space $(\BC{Y},\C{B}(\BC{Y}))$, where $\BC{Y}$ is a Polish space (i.e., a separable completely metrizable space) and $\C{B}(\BC{Y})$ are the Borel subsets of $\BC{Y}$ (i.e., the $\sigma$-algebra generated by the open sets in $\BC{Y}$). Throughout this paper, when we say that $\BC{X}$ is a standard measurable space, then we implicitly assume that there exists a $\sigma$-algebra $\B{\Sigma}$ such that $(\BC{X},\B{\Sigma})$ is a standard measurable space. In several proofs we assume without loss of generality that the standard measurable space \emph{is} a Polish space $\BC{Y}$ with $\sigma$-algebra $\C{B}(\BC{Y})$. Examples of standard measurable spaces are the open and closed subsets of $\RN^d$, and the finite sets with the usual complete metric. See Appendix~\ref{app:AppendixMST} of the Supplementary Material \citep{BPSM19s} for more details.}
  \item 
    $\BC{E} = \prod_{j \in \C{J}} \C{E}_j$ is the product of the domains of the 
    exogenous variables, where each domain $\C{E}_j$ is a standard measurable space, 
  \item 
    $\B{f} : \BC{X} \times \BC{E} \to \BC{X}$ is a measurable function that specifies 
    the \emph{causal mechanism}, 
  \item 
    $\Prb_{\BC{E}}=\prod_{j\in\C{J}}\Prb_{\C{E}_j}$ is a product measure, the \emph{exogenous distribution}, where 
$\Prb_{\C{E}_j}$ is a probability measure on $\C{E}_j$ for each $j\in\C{J}$.\footnote{For the case $\C{J}=\emptyset$, we have that $\BC{E}$ is the singleton $\B{1}$ and $\Prb_{\BC{E}}$ is the degenerate probability measure $\Prb_{\B{1}}$.}
\end{enumerate}
\end{definition}

%In structural causal models, the functional relationships between variables are expressed in terms of (deterministic) equations involving the causal mechanisms.
In SCMs, the functional relationships between variables are 
  expressed in terms of deterministic equations, where each equation expresses an endogenous variable (on the left-hand side) in terms of a causal mechanism depending on endogenous and exogenous variables (on the right-hand side). This allows us to model interventions in an unambiguous way by changing the causal mechanisms that target specific endogenous variables (see Section~\ref{sec:Interventions}).
\begin{definition}[Structural equations]
Let $\C{M}=\langle \C{I}, \C{J}, \BC{X}, \BC{E}, \B{f}, \Prb_{\BC{E}} \rangle$ be an SCM. We call the
set of equations
$$
  x_i=f_i(\B{x},\B{e}) \qquad\B{x}\in\BC{X}, \B{e}\in\BC{E}
$$
for $i \in \C{I}$ the \emph{structural equations} of the structural causal model $\C{M}$. 
\end{definition}

Although it is common to assume the absence of cyclic functional 
relations (see Definition~\ref{def:AcyclicSCM}), we make no such assumption here. In particular, we allow for self-cycles, which we will discuss in more detail in Sections~\ref{sec:Graph} and \ref{sec:SelfCycles}.

%Readers already familiar with structural causal models will note that we are still missing an important ingredient: the random variables that express solutions of SCMs.
The solutions of an SCM in terms of random variables are defined up to almost sure equality. Random variables that are almost surely equal are generally considered to be equivalent to each other for all practical purposes.
\begin{definition}[Solution]
\label{def:Solution}
A pair $(\B{X},\B{E})$ of random variables $\B{X} : \Omega \to \BC{X}, \B{E} : \Omega \to \BC{E}$, where $\Omega$ is a probability space, is a \emph{solution} of the SCM $\C{M} = \langle \C{I}, \C{J}, \BC{X}, \BC{E}, \B{f}, \Prb_{\BC{E}} \rangle$ if 
\begin{enumerate}
%  \item 
%    $\B{E}:\B{\Omega}\to\BC{E}$,
%  \item 
%    $\B{X}:\B{\Omega}\to\BC{X}$,
  \item  
    $\Prb^{\B{E}} = \Prb_{\BC{E}}$, that is, the distribution of $\B{E}$ is equal
    to $\Prb_{\BC{E}}$,\footnote{ 
    This implies that the components $E_j$ of $\B{E}$ are
    mutually independent, since $\Prb_{\BC{E}}=\prod_{j\in\C{J}}\C{E}_j$.} and
  \item 
    the \emph{structural equations} are satisfied, that is,
    $$
    \B{X} = \B{f}(\B{X},\B{E}) \text{ a.s..}
    $$
\end{enumerate}
For convenience, we call a random variable $\B{X}$ \emph{a solution of} $\C{M}$ if there exists a 
random variable $\B{E}$ such that $(\B{X},\B{E})$ forms a solution of $\C{M}$. 
\end{definition}
Often, the endogenous random variables $\B{X}$ can be observed, while the exogenous random 
variables $\B{E}$ are treated as latent. Latent exogenous variables are often referred to as 
``disturbance terms'' or ``noise variables.'' For a solution $\B{X}$, we 
call the distribution $\Prb^{\B{X}}$ the \emph{\hypertarget{def:ObservationalDistribution}{observational distribution} of 
$\C{M}$ associated to $\B{X}$}. In general, there may be multiple different observational 
distributions associated to an SCM due to the existence of different solutions of the structural 
equations. This is a consequence of the allowance of cycles in SCMs, as the following simple example illustrates. 

\begin{example}[Cyclic SCMs]
  \label{ex:SimpleCyclicExample}
  For brevity, we use throughout this paper the notation $\B{n} := \{1,2,\dots,n\}$ for $n\in\NN$.
  Let $\C{M} = \langle \B{2}, \B{1}, \RN^2, \RN, \B{f}, \Prb_{\RN} \rangle$ be an SCM\footnote{We will abuse notation by using nondisjoint subsets of the natural numbers to index both endogenous and exogenous variables; these should be understood to be disjoint copies of the natural numbers: if we write $\C{I} = \B{n}$ and $\C{J} = \B{m}$, we mean instead $\C{I} = \{1,2,\dots,n\}$ and $\C{J} = \{1',2',\dots,m'\}$ where $k'$ is a copy of~$k$. 
    %More formally, we mean $\C{I} = \{(k,0) : k \in \NN, 1 \le k \le n\}$ and $\C{J} = \{(k,1) : k \in \NN, 1 \le k \le m\}$.
  } with $f_1(\B{x},e) = x_2$ and $f_2(\B{x},e) = x_1$,
  and $\Prb_{\RN}$ an arbitrary probability measure on $\RN$. Then $(X,X)$ is a solution of $\C{M}$ for any arbitrary 
  random variable $X$ with values in $\RN$. Hence, any probability distribution on $\{(x,x) : x \in \RN\}$ is an observational distribution 
  associated to $\C{M}$. Now consider instead the same SCM but with $f_1(\B{x},e) = x_2 + 1$. This SCM has no solutions at all,
  and hence induces no observational distribution.
\end{example}
%Note that for two exogenous variables $\B{E}$ and $\tilde{\B{E}}$ we have that if $\B{E}=\tilde{\B{E}}$ a.s., then $\Prb^{\B{E}}=\Prb^{\tilde{\B{E}}}$.

Due to the fact that the structural equations only need to be satisfied almost surely, 
there may exist many different SCMs representing the same set of solutions (see Example~\ref{ex:StructuralEquationsUpToASEquality}). It therefore seems natural not to differentiate between 
%causal mechanisms 
structural equations that have different solutions on at most a $\Prb_{\BC{E}}$-null set of exogenous
variables. This leads to 
%the following 
an equivalence relation between SCMs. To be able to state the equivalence
relation concisely, we introduce the following notation: For subsets $\C{U}\subseteq\C{I}$ and $\C{V}\subseteq\C{J}$, we write
$\BC{X}_{\C{U}}:=\prod_{i\in\C{U}}\C{X}_i$ and
$\BC{E}_{\C{V}}:=\prod_{j\in\C{V}}\C{E}_j$. In particular, $\BC{X}_{\emptyset}$ and $\BC{E}_{\emptyset}$ are defined by the singleton $\B{1}$. Moreover, for a subset $\C{W} \subseteq \C{I} \cup \C{J}$, we use the convention that we write 
$\BC{X}_{\C{W}}$ and $\BC{E}_{\C{W}}$ instead of $\BC{X}_{\C{W} \cap \C{I}}$ and 
$\BC{E}_{\C{W} \cap \C{J}}$, respectively and we adopt a similar notation for the (random) 
variables in those spaces, that is, we write $\B{x}_{\C{W}}$ and $\B{e}_{\C{W}}$ instead of 
$\B{x}_{\C{W} \cap \C{I}}$ and $\B{e}_{\C{W} \cap \C{J}}$, respectively. This allows us to define the following natural equivalence relation for SCMs.\footnote{An attempt at coarsening this notion of equivalence by replacing the quantifier ``for all $\B{x}\in\BC{X}$'' by ``for almost every $\B{x}\in\BC{X}$ under the observational distribution $\Prb^{\B{X}}$'' will not lead to a well-defined equivalence relation, since in general the observational distribution $\Prb^{\B{X}}$ may be nonunique or even nonexistent. Refining it by replacing the quantifier ``for $\Prb_{\BC{E}}$-almost every $\B{e}\in\BC{E}$'' by ``for all $\B{e}\in\BC{E}$'' would make it too fine for our purposes, since we assume the exogenous distribution to be fixed and we assume as usual that random variables that are almost surely identical are indistinguishable in practice. Note that the ``for $\Prb_{\BC{E}}$-almost every $\B{e}\in\BC{E}$'' and ``for all $\B{x}\in\BC{X}$'' quantifiers do not commute in general (see Example~\ref{ex:ForAlmostEveryForAllQuantifier})}\footnote{We may extend this definition to allow $\tilde{\C{J}} \ne \C{J}$ and for a larger class of SCMs such that the exogenous distribution does not factorize. Then, for any $\C{M}$ that satisfies Definition~\ref{def:SCM}, except for that it may have a non-factorizing exogenous distribution, there exists an equivalent SCM with a factorizing exogenous distribution (and a different $\C{J}$); the latter can be obtained by partitioning the exogenous components into independent tuples. This motivates why we can restrict ourselves in Definition~\ref{def:SCM} to factorizing exogenous distributions only. For some more discussion on the representation of latent confounders, see also Example~\ref{ex:no_nice_reduction}.}
%\footnote{Another common definition of SCMs in the literature \citep[e.g.,][]{Pea09} assumes that there is a one-to-one correspondence between endogenous and exogenous variables in the sense that only exogenous variable $E_i$ may appear in the structural equation for endogenous variable $X_i$, but allows for an arbitrary (non-factorizing) joint distribution on the exogenous variables. In Example~\ref{ex:no_nice_reduction} we show that this class is less general, at least under certain smoothness assumptions on functions and distributions.}

\begin{definition}[Equivalence]
\label{def:EquivSCMs}
The two SCMs $\C{M}=\langle \C{I}, \C{J}, \BC{X}, \BC{E}, \B{f}, \Prb_{\BC{E}} \rangle$ and 
$\tilde{\C{M}}=\langle \C{I}, \C{J}, \BC{X}, \BC{E}, \tilde{\B{f}}, \Prb_{\BC{E}} \rangle$ are 
\emph{equivalent}, denoted by $\C{M} \equiv \tilde{\C{M}}$, if 
%$\B{f} \equiv \tilde{\B{f}}$.
for all $i \in I$, for $\Prb_{\BC{E}}$-almost every $\B{e} \in \BC{E}$ and
for all $\B{x} \in \BC{X}$,
  $$x_i = f_i(\B{x},\B{e}) \quad\iff\quad x_i = \tilde{f}_i(\B{x},\B{e}).$$
\end{definition}

Thus, two equivalent SCMs can only differ in terms of their causal mechanism. Importantly, equivalent SCMs have the same solutions and, as we will see in Sections~\ref{sec:Interventions} and \ref{sec:Counterfactuals}, they have the same causal and counterfactual semantics (see Definitions~\ref{def:intervenedSCM} and \ref{def:TwinSCM}, respectively). This equivalence relation on the set of all SCMs gives rise to the quotient set of equivalence classes of SCMs. %In this paper we prove properties and define operations on the equivalence classes of SCMs, by first proving the property and defining the operation for an SCM and then showing that this property and operation preserves the equivalence relation.
\subsection{The (augmented) graph}
\label{sec:Graph}
%%%%%%%%%%%%%%%%%%%%%%%%%%%%%%%%%%%%%%%%%%%%%%%%%%

We will now define two types of graphs that can be used for representing structural properties of the SCM.
These graphical representations are related to Wright's path diagrams \citep{Wri21}.
The structural properties of the functional relations between variables modeled by an SCM are specified by the causal mechanism of the SCM and can be encoded in an (augmented) graph. 
%In general, this (augmented) functional graph differs from the direct causal graph, as we will see in Section~\ref{sec:DirectCausalGraph}. 
For the graphical notation and standard terminology on directed (mixed) graphs
that is used throughout this paper, we refer the reader to
Appendix~\ref{app:AppendixDirectedMixedGraphs}.

We first define the parents of an endogenous variable.
\begin{definition}[Parent]
\label{def:Parents}
Let $\C{M} = \langle \C{I}, \C{J}, \BC{X}, \BC{E}, \B{f},
\Prb_{\BC{E}} \rangle$ be an SCM. We call $k \in \C{I}\cup\C{J}$ a \emph{parent of $i \in \C{I}$} if and only if there does not 
   exist a measurable function\footnote{For $\BC{X} = \prod_{i \in \C{I}}
     \C{X}_i$, $\C{I}$ some index set, $I\subseteq\C{I}$ and $k\in\C{I}$, we denote
     $\BC{X}_{\setminus I} =\prod_{i\in\C{I}\setminus I} \C{X}_i$ and $\BC{X}_{\setminus k} = \prod_{i \in \C{I} \setminus \{k\}} \C{X}_i$, 
   and similarly for their elements.} $\tilde f_i : \BC{X}_{\setminus k} \times
   \BC{E}_{\setminus k} \to \C{X}_i$ such that
%  $\tilde f_k \equiv f_k$ (see Definition~\ref{def:EquivMappings}).
   for $\Prb_{\BC{E}}$-almost every $\B{e} \in \BC{E}$ and for all $\B{x} \in \BC{X}$,
   $$x_i = f_i(\B{x},\B{e}) \quad\iff\quad x_i = \tilde f_i(\B{x}_{\setminus k},\B{e}_{\setminus k}).$$
\end{definition}
Exogenous variables have no parents by definition.
These parental relations are preserved under the equivalence relation $\equiv$
on SCMs. They can be represented by a directed graph or a directed mixed graph.\footnote{A \emph{directed mixed graph $\C{G}=(\C{V},\C{E},\C{B})$} consists of a set of nodes $\C{V}$, a set of directed edges $\C{E}$ and a set of bidirected edges $\C{B}$ (see Definition~\ref{def:DirectedMixedGraph} for a more precise definition).}
\begin{definition}[Graph and augmented graph]
\label{def:Graphs}
Let $\C{M} = \langle \C{I}, \C{J}, \BC{X}, \BC{E}, \B{f}, \Prb_{\BC{E}} \rangle$ be an SCM. We
define:
\begin{enumerate}
\item
\label{def:AugmentedGraph} 
the \emph{augmented graph} $\C{G}^a(\C{M})$ as the 
directed graph with nodes $\C{I} \cup \C{J}$ and directed edges $u \to v$ if and
only if $u\in\C{I}\cup\C{J}$ 
is a parent of $v\in\C{I}$;
%\footnote{By Definition~\ref{def:Parents}, $u\in\C{I}\cup\C{J}$ cannot be a parent of $v\in\C{J}$.}
\item
\label{def:Graph} 
the \emph{graph} $\C{G}(\C{M})$ as the directed mixed 
graph with nodes $\C{I}$, directed edges $u \to v$ if and only if
$u\in\C{I}$ is a parent 
of $v\in\C{I}$ and bidirected edges $u\leftrightarrow v$ if and only if there exists a $j\in\C{J}$ 
that is a parent of both $u \in \C{I}$ and $v \in \C{I}$.
\end{enumerate}
We call the mappings $\C{G}^a$ and $\C{G}$, that map $\C{M}$ to $\C{G}^a(\C{M})$ and $\C{G}(\C{M})$, the \emph{augmented graph mapping} and the \emph{graph mapping}, respectively.
\end{definition}
In particular, the augmented graph contains no directed edges 
pointing toward an exogenous variable, that is, $u\in\C{I}\cup\C{J}$ cannot be a parent of $v\in\C{J}$,
%between the exogenous variables, 
because they are not functionally related through the causal mechanism. 
We call a directed edge $i\to i$ in $\C{G}^a(\C{M})$ and $\C{G}(\C{M})$ 
(here, $i$ is a parent of itself) 
a \emph{self-cycle} at $i$. By definition, the mappings $\C{G}^a$ and $\C{G}$ are invariant under the equivalence
relation $\equiv$ on SCMs, and hence the equivalence class of an SCM $\C{M}$ is mapped to a unique augmented graph $\C{G}^a(\C{M})$ and a unique graph $\C{G}(\C{M})$. 

\begin{example}[Graphs of an SCM]
\label{ex:AugmentedGraphs}
Let $\C{M} = \langle \B{5}, \B{3}, \RN^5, \RN^3, \B{f}, \Prb_{\RN^3} \rangle$ be an SCM with causal mechanism given by
$$
  \begin{aligned}
    f_1(\B{x},\B{e}) &= x_1 - x_1^2 + \alpha e_1^2 \,, & f_3(\B{x},\B{e}) &=
    - x_4 + e_2 \,, & f_5(\B{x},\B{e}) &=  x_4 \cdot e_3 \,, \\
    f_2(\B{x},\B{e}) &= x_1 + x_3 + x_4 + e_1 \,, & f_4(\B{x},\B{e}) &= x_2 +
    e_2 \,, & & 
  \end{aligned}
$$
where $\alpha \neq 0$ and $\Prb_{\RN^3}$ is a product of three 
%non-degenerate
%THIS DEF IS WRONG\footnote{A probability measure $\Prb$ over a measurable space $\BC{X}$ is called \emph{degenerate} if it is the Dirac measure $\B{\delta}_{\B{x}}$ for some point $\B{x}\in\BC{X}$, which is defined on the measurable set $\BC{U}\subseteq\BC{X}$ by $\B{\delta}_{\B{x}}(\BC{U}):=\B{1}_{\BC{U}}(\B{x})$ with $\B{1}_{\BC{U}}$ the indicator function.} DEGENERACY MEANS MEASURE ONE ON LOWER-DIMENSIONAL SUBSPACE
probability measures $\Prb_{\RN}$ over $\RN$ that are non-degenerate. The augmented graph $\C{G}^a(\C{M})$ and the graph $\C{G}(\C{M})$ of $\C{M}$ are depicted\footnote{For visualizing an 
(augmented) graph, we adapt the common convention of using random variables, 
with the index set as a subscript, instead of using the index set itself. With a slight abuse of 
notation, we still use the random variables notation in the (augmented) graph in 
the case that the SCM has no solution at all.} in Figure~\ref{fig:AugmentedGraphs} (left and center). Observe that if $\alpha$ had been equal to zero, then the endogenous variable $1$ would not have any parents in $\C{G}^a(\C{M})$, that is, it would not have a self-cycle and directed edge from any exogenous variables in $\C{G}^a(\C{M})$, and it would not have a self-cycle and bidirected edge from any other variable in $\C{G}(\C{M})$. Moreover, if one of the probability measures $\Prb_{\RN}$ over $\RN$ were degenerate,
%the probability measure $\Prb_{\RN^3}$ had not been a product of three  measures over $\RN$, e.g., a product of three degenerate measures over $\RN$
then some of the 
%there would be no 
directed edges from the exogenous variables to the endogenous variables in the augmented graph $\C{G}^a(\C{M})$ and bidirected edges in the graph $\C{G}(\C{M})$ would be missing.
\begin{figure}
  \begin{center}
  \adjustbox{scale=0.8,center}{%
    \begin{tikzpicture}
      \begin{scope}
        \node[exvar] (E1) at (-0.5,2.25) {$E_1$};
        \node[exvar] (E2) at (1,2.25) {$E_2$};
        \node[exvar] (E3) at (2.5,2.25) {$E_3$};
        \node[var] (X1) at (-1,1) {$X_1$}; 
        \node[var] (X2) at (0,0) {$X_2$}; 
        \node[var] (X3) at (1,1) {$X_3$}; 
        \node[var] (X4) at (2,0) {$X_4$};
        \node[var] (X5) at (3,1) {$X_5$};
        \draw[arr] (X1) to (X2);
        \draw[arr] (X3) to (X2);
        \draw[arr] (X4) to (X3);
        \draw[arr] (X4) to (X5);
        \draw[arr] (X2) to [bend right=20] (X4);
        \draw[arr] (X4) to [bend right=20] (X2);
        \draw[arr] (X1) to [out=129.5,in=185.5,looseness=4.4] (X1);
        \draw[arr] (E1) to (X1);
        \draw[arr] (E1) to (X2);
        \draw[arr] (E2) to (X3);
        \draw[arr] (E2) to [bend left=20] (X4);
        \draw[arr] (E3) to (X5);
        \node at (-1.7,-0.2) {$\C{G}^a(\C{M})$};
      \end{scope}
%      \begin{scope}[shift={(6.5cm,0cm)}]
%        \node[exvar] (E1) at (-0.5,2.25) {$E_1$};
%        \node[exvar] (E2) at (1,2.25) {$E_2$};
%        \node[exvar] (E3) at (2.5,2.25) {$E_3$};
%        \node[var] (X1) at (-1,1) {$X_1$}; 
%        \node[var] (X2) at (0,0) {$X_2$}; 
%        \node[var] (X3) at (1,1) {$X_3$}; 
%        \node[var] (X4) at (2,0) {$X_4$};
%        \node[var] (X5) at (3,1) {$X_5$};
%        \draw[arr] (X1) to (X2);
%        \draw[arr] (X3) to (X2);
%        \draw[arr] (X4) to (X5);
%        \draw[arr] (X2) to [bend right=20] (X4);
%        \draw[arr] (X4) to [bend right=20] (X2);
%        \draw[arr] (X1) to [out=127.5,in=187.5,looseness=5] (X1);
%        \draw[arr] (E1) to (X1);
%        \draw[arr] (E1) to (X2);
%        \draw[arr] (E2) to [bend left=20] (X4);
%        \draw[arr] (E3) to (X5);
%        \node at (-1.7,-0.2) {$\C{G}^a(\C{M}_{\intervene(\{3\},1)})$};
%      \end{scope}
      \begin{scope}[shift={(6cm,0cm)}]
        \node[var] (X1) at (-1,1) {$X_1$}; 
        \node[var] (X2) at (0,0) {$X_2$}; 
        \node[var] (X3) at (1,1) {$X_3$}; 
        \node[var] (X4) at (2,0) {$X_4$};
        \node[var] (X5) at (3,1) {$X_5$};
        \draw[arr] (X1) to (X2);
        \draw[arr] (X3) to (X2);
        \draw[arr] (X4) to (X3);
        \draw[arr] (X4) to (X5);
        \draw[arr] (X2) to [bend right=20] (X4);
        \draw[arr] (X4) to [bend right=20] (X2);
        \draw[arr] (X1) to [out=129.5,in=185.5,looseness=4.4] (X1);
        \draw[biarr] (X2) to [bend right=40] (X1);
        \draw[biarr] (X4) to [bend right=40] (X3);
        \node at (-1.7,-0.2) {$\C{G}(\C{M})$};
      \end{scope}
      \begin{scope}[shift={(12cm,0cm)}]
        \node[var] (X1) at (-1,1) {$X_1$}; 
        \node[var] (X2) at (0,0) {$X_2$}; 
        \node[var] (X3) at (1,1) {$X_3$}; 
        \node[var] (X4) at (2,0) {$X_4$};
        \node[var] (X5) at (3,1) {$X_5$};
        \draw[arr] (X1) to (X2);
        \draw[arr] (X3) to (X2);
        \draw[arr] (X4) to (X5);
        \draw[arr] (X2) to [bend right=20] (X4);
        \draw[arr] (X4) to [bend right=20] (X2);
        \draw[arr] (X1) to [out=129.5,in=185.5,looseness=4.4] (X1);
        \draw[biarr] (X2) to [bend right=40] (X1);
        \node at (-1.7,-0.2) {$\C{G}(\C{M}_{\intervene(\{3\},1)})$};
      \end{scope}
  \end{tikzpicture}}
  \end{center}
  \caption{The augmented graph (left) and the graph (center) of the SCM
    $\C{M}$ of Example~\ref{ex:AugmentedGraphs} and the graph of the
  intervened SCM $\C{M}_{\intervene(\{3\},1)}$ of Example~\ref{ex:Interventions} (right).}
  \label{fig:AugmentedGraphs}
\end{figure}
\end{example}
As is illustrated in this example, the augmented graph provides a more detailed representation than the graph. Therefore, we use the augmented graph as the standard graphical representation for SCMs, unless stated otherwise. For an SCM $\C{M}$, we 
denote the sets $\pa_{\C{G}^a(\C{M})}(\C{U})$,
$\ch_{\C{G}^a(\C{M})}(\C{U})$, $\an_{\C{G}^a(\C{M})}(\C{U})$, etc., for some 
subset
$\C{U}\subseteq\C{I}\cup\C{J}$, by respectively $\pa(\C{U})$, $\ch(\C{U})$, $\an(\C{U})$, etc., when the notation is clear from the context.

\begin{definition}
\label{def:AcyclicSCM}
We call an SCM $\C{M}$ \emph{acyclic} if $\C{G}^a(\C{M})$ is a directed acyclic
graph (DAG). Otherwise, we call $\C{M}$ \emph{cyclic}.
\end{definition}
Equivalently, an SCM $\C{M}$ is acyclic if $\C{G}(\C{M})$ is an acyclic directed
mixed graph (ADMG) \citep{Ric03}.
%\begin{proposition}
%Every acyclic SCM $\C{M}$ has a solution.
%\end{proposition}
%\begin{proof}
%This follows directly from proposition~\ref{prop:AcyclicSCMUniquelySolvable}.
%\end{proof}
%Most of the existing literature focuses on acyclic SCMs \citep[see][a.o.]{Pea09,SGS00,Ver93,Tia02,Eva16,Eva18}. In the structural equation model (SEM) literature, acyclic SCMs are referred to as \emph{recursive} SEMs and cyclic SCMs as \emph{non-recursive} SEMs \citep{Bol89}.
Acyclic SCMs are also known as semi-Markovian SCMs \citep{Pea09,Tia02}. A commonly considered class of acyclic SCMs are the Markovian SCMs, which are acyclic SCMs for which each 
exogenous variable has at most one child.
%For these models it was first shown that they satisfy several Markov properties 
Several Markov properties were first shown for these models~\citep{Pea09,LDLL90,Tia02}.
% and the \emph{semi-Markovian} SCMs, which are acyclic SCMs for which each exogenous variable has at most two childs. 

%%%%%%%%%%%%%%%%%%%%%%%%%%%%%%%%%%%%%%%%%%%%%%%%%%
\subsection{Structurally minimal representations}
\label{sec:StructMinimalRepresentation}
%%%%%%%%%%%%%%%%%%%%%%%%%%%%%%%%%%%%%%%%%%%%%%%%%%

We have discussed an equivalence relation between SCMs in Section~\ref{sec:SCMdef}. In this subsection, we show that for each SCM there exists a representative of the equivalence class of that SCM for which each component of the causal mechanism does not depend on its nonparents \citep[see also][]{PJS17}. %Remark~6.6
\begin{definition}[Structurally minimal SCM]
\label{def:StructMinimalRepresentation}
Let $\C{M} = \langle \C{I}, \C{J}, \BC{X}, \BC{E}, \B{f}, \Prb_{\BC{E}} \rangle$ be an SCM. We call $\C{M}$ \emph{structurally minimal} if for all $i\in\C{I}$ there exists a mapping
$\tilde{f}_i: \BC{X}_{\pa(i)}\times\BC{E}_{\pa(i)} \to \C{X}_i$ such that
$f_i(\B{x},\B{e}) = \tilde{f}_i(\B{x}_{\pa(i)},\B{e}_{\pa(i)})$ for all $\B{e}\in\BC{E}$ and all $\B{x}\in\BC{X}$.
\end{definition}
We already encountered a structurally minimal SCM $\C{M}$ in Example~\ref{ex:AugmentedGraphs}. Taking instead $\alpha=0$ in that example gives an SCM $\C{M}$ that is not structurally minimal, since the endogenous variable $1$ is then not a parent of itself, while $f_1(\B{x},\B{e})$ depends on $x_1$. However, the equivalent SCM where we have replaced the causal mechanism of $1$ by $f_1(\B{x},\B{e})=0$ yields a structurally minimal SCM. In general, there always exists an equivalent structurally minimal SCM.

%The next example illustrates that not all SCMs are structurally minimal.
%\begin{example}[A non-structurally minimal SCM]
%\label{ex:NotStructMinimalGraphs}
%Consider an SCM $\C{M} = \langle \B{1}, \B{2}, \RN, \RN^2, f, \Prb_{\RN^2} \rangle$ with causal
%mechanism $f(x,\B{e}) = - x + e_1 + e_2$ and $\Prb_{\RN^2}$ a product of non-degenerate probability measures over $\RN$. Its augmented graph $\C{G}^a(\C{M})$ is depicted in 
%Figure~\ref{fig:NotStructMinimalGraphs}. 
%\begin{figure}
%  \begin{center}
%    \begin{tikzpicture}
%        \node[exvar] (E1) at (-3,1) {$E_1$};
%        \node[exvar] (E2) at (-1,1) {$E_2$};
%        \node[var] (X) at (-2,0) {$X$}; 
%        \draw[arr] (E1) -- (X);
%        \draw[arr] (E2) -- (X);
%    \end{tikzpicture}
%  \end{center}
%  \caption{The augmented graph of the acyclic SCM 
%  $\C{M}$ of Example~\ref{ex:NotStructMinimalGraphs}.}
%  \label{fig:NotStructMinimalGraphs}
%\end{figure}
%  The variable $x$ is not a parent of itself, while $f(x,\B{e})$ depends on $x$, and hence $\C{M}$ is not structurally minimal. However, the equivalent SCM $\tilde{\C{M}}$ defined by replacing the causal mechanism of $\C{M}$ by $\tilde{f}(x,\B{e})=\tfrac{1}{2}(e_1+e_2)$ is structurally minimal.
%\end{example}
%In general, there always exists an equivalent structurally minimal SCM like in the previous example.
\begin{proposition}[Existence of a structurally minimal SCM]
\label{prop:StructMinimalRepresentation}
For an SCM $\C{M} = \langle \C{I}, \C{J}, \BC{X}, \BC{E}, \B{f}, \Prb_{\BC{E}} \rangle$, 
there exists an equivalent SCM $\tilde{\C{M}}=\langle \C{I}, \C{J}, \BC{X}, \BC{E},
\tilde{\B{f}}, \Prb_{\BC{E}} \rangle$ that is structurally minimal.
\end{proposition}
%Obviously, such a structurally minimal representation of an SCM is unique up to equivalence.

%For a mapping $\B{f}_{\C{U}}:\BC{X}_{\C{I}}\times\BC{E}_{\C{J}}\to\BC{X}_{\C{U}}$, for some sets $\C{I},\C{J}$ and $\C{U}$, and some subset $\C{V}\subseteq\C{U}$ we write $(\B{f}_{\C{U}})_{\C{V}}$ for the $\C{V}$ components of $\B{f}_{\C{U}}$. 
%In particular, 
For a causal mechanism 
$\B{f}:\BC{X}\times\BC{E}\to\BC{X}$ and a subset $\C{U}\subseteq\C{I}$, we write
$\B{f}_{\C{U}}:\BC{X}\times\BC{E}\to\BC{X}_{\C{U}}$ for the $\C{U}$ components\footnote{For $\C{U}=\emptyset$, we always consider the trivial mapping $\B{f}_{\emptyset}:\BC{X}\times\BC{E}\to\BC{X}_{\emptyset}$ where $\BC{X}_{\emptyset}$ is the singleton $\B{1}$.} of $\B{f}$. A
structurally minimal representation is compatible with the (augmented) graph, in the
sense that for every $\C{U}\subseteq\C{I}$ there exists a unique measurable mapping $\tilde{\B{f}}_{\C{U}} :
\BC{X}_{\pa(\C{U})}\times\BC{E}_{\pa(\C{U})}\to\BC{X}_{\C{U}}$ such that $\B{f}_{\C{U}}(\B{x},\B{e}) = \tilde{\B{f}}_{\C{U}}(\B{x}_{\pa(\C{U})},\B{e}_{\pa(\C{U})})$ 
for all $\B{e}\in\BC{E}$ and all $\B{x}\in\BC{X}$. Moreover, for any $\C{U} \subseteq \C{I}$
there exists a unique measurable mapping $\tilde{\B{f}}_{\an(\C{U})} : 
\BC{X}_{\an(\C{U})} \times \BC{E}_{\an(\C{U})} \to \BC{X}_{\an(\C{U})}$
with
$\B{f}_{\an(\C{U})}(\B{x},\B{e}) = \tilde{\B{f}}_{\C{U}}(\B{x}_{\an(\C{U})},\B{e}_{\an(\C{U})})$ 
for all $\B{e}\in\BC{E}$ and all $\B{x}\in\BC{X}$.

%Note that there always exists a representation $\tilde{\B{f}}$ of the causal mechanism $\B{f}$ 
%such that each component may depend only on its parents, i.e., each component 
%$\tilde f_i$ can be considered without loss of generality to be a function $\tilde f_i : 
%\BC{X}_{\pa(i)} \times \BC{E}_{\pa(i)} \to \C{X}_i$. Similarly 
%\Joris{what does this mean? 
%do we get this already by applying the previous comment to each endogenous variable? I.e., 
%if we ensure that each component $f_i$ only depends on $\pa(i)$, do we already get that
%$f_{an(i)}$ only depends on $\an(i)$? If yes, it is more informative to replace the 'similarly'
%by something like 'by applying this to each endogenous variable'}
%, for $i \in \C{I}$, 
%one can always find a representation $\B{f}_{\an(i)} : \BC{X}_{\an(i)} \times 
%\BC{E}_{\an(i)} \to \BC{X}_{\an(i)}$ of $\B{f}$ such that the ancestral 
%components $\B{f}_{\an(i)}$ do not depend on non-ancestors of $i$.
%\Stephan{I have rewritten the paragraph. Is this more clear now?}

%%%%%%%%%%%%%%%%%%%%%%%%%%%%%%%%%%%%%%%%%%%%%%%%%%
\subsection{Interventions}
\label{sec:Interventions}
%%%%%%%%%%%%%%%%%%%%%%%%%%%%%%%%%%%%%%%%%%%%%%%%%%

To define the causal semantics of SCMs, we consider here an idealized class of interventions
introduced by Pearl~\cite{Pea09} that we refer to as perfect interventions. Other types of interventions, like mechanism changes \citep{TP01a}, fat-hand interventions \citep{EM07},
activity interventions \citep{MH13} and stochastic versions of
all these are at least as relevant, but we do not consider them here.

\begin{definition}[Perfect intervention on an SCM]
\label{def:intervenedSCM}
Let $\C{M} = \langle \C{I}, \C{J}, \BC{X}, \BC{E}, \B{f}, \Prb_{\BC{E}} \rangle$ be an SCM, $I \subseteq \C{I}$ a subset of endogenous variables and $\B{\xi}_I \in \BC{X}_I$ a value. The \emph{perfect intervention 
$\intervene(I, \B{\xi}_I)$} maps $\C{M}$ to the SCM 
$\C{M}_{\intervene(I, \B{\xi}_I)} := \langle \C{I}, \C{J}, \BC{X}, \BC{E}, 
\tilde{\B{f}}, \Prb_{\BC{E}} \rangle$, where the \emph{intervened causal mechanism} 
$\tilde{\B{f}}$ is given by
$$
  \begin{aligned}
  \tilde f_i(\B{x},\B{e}) = \begin{cases}
    \xi_i & i \in I \\
    f_i(\B{x},\B{e}) & i \in \C{I} \setminus I \,.
  \end{cases}
  \end{aligned}
$$
\end{definition}
This operation $\intervene(I, \B{\xi}_I)$ preserves the equivalence relation 
(see Definition~\ref{def:EquivSCMs}) on the set of all SCMs, and hence this mapping 
induces a well-defined mapping on the set of equivalence classes of SCMs. Previous work has considered interventions only on a specific subset of endogenous variables \citep{RWB+17, BH19, BBM19}. Instead, we 
%assume that it is only possible to intervene on a specific subset of endogenous variables, here we 
assume that we can intervene on any subset of endogenous variables in the model. 

We define an analogous operation $\intervene(I)$ on directed mixed graphs.
\begin{definition}[Perfect intervention on a directed mixed graph]
\label{def:InterventionOnGraph}
Let $\C{G}=(\C{V},\C{E},\C{B})$ be a directed mixed graph and $I \subseteq \C{V}$ a subset. The perfect intervention $\intervene(I)$ maps $\C{G}$ to 
the directed mixed graph $\intervene(I)(\C{G}) := (\C{V},\tilde{\C{E}},\tilde{\C{B}})$, where $\tilde{\C{E}} = \C{E} \setminus \{ v \to i : v\in\C{V}, i\in I\}$ and $\tilde{\C{B}} = \C{B}\setminus \{ v \oto i : v\in\C{V}, i\in\C{I} \}$. 
\end{definition}
%It 
This operation simply removes all incoming edges on the nodes in $I$. The two notions of 
intervention are compatible with the (augmented) graph mapping.

\begin{proposition}
\label{prop:InterventionOnGraph}
Let $\C{M} = \langle \C{I}, \C{J}, \BC{X}, \BC{E}, \B{f}, \Prb_{\BC{E}} \rangle$ be an SCM, $I \subseteq \C{I}$ a subset of endogenous variables and $\B{\xi}_I \in \BC{X}_I$ a value. Then $\big(\C{G}^a\circ\intervene(I, \B{\xi}_I)\big)(\C{M}) = 
\big(\intervene(I)\circ\C{G}^a\big)(\C{M})$ and $\big(\C{G}\circ\intervene(I, \B{\xi}_I)\big)(\C{M}) = 
\big(\intervene(I)\circ\C{G}\big)(\C{M})$.
\end{proposition}

The two notions of perfect intervention satisfy the following elementary properties.
\begin{proposition}
\label{prop:IntDisjSubsetAndAcyclicity}
For an SCM and a directed mixed graph, we have the following properties:
\begin{enumerate}[ref=\theproposition.(\arabic*)]
	\item \label{prop:IntDisjSubsetCommute} perfect interventions on disjoint subsets of variables commute;
	\item acyclicity is preserved under perfect intervention.
\end{enumerate}
\end{proposition}
%One reason for separating the random variables from the SCMs was that this allows us to define perfect interventions on SCMs independently of its solutions. 

The following example shows that an SCM with a solution may not have a solution anymore after performing a perfect intervention on the SCM, and vice versa that an SCM without a solution may yield an SCM with a solution after intervention.
\begin{example}[Intervened SCM and its graphs]
\label{ex:Interventions}
Consider the SCM $\C{M}$ of Example~\ref{ex:AugmentedGraphs} which has
a solution if and only if $\alpha \geq 0$. Applying the perfect intervention $\intervene(\{3\}, 1)$ to $\C{M}$ gives the intervened model $\C{M}_{\intervene(\{3\},1)}$ with the intervened causal mechanism
$$
  \begin{aligned}
    \tilde{f}_1(\B{x},\B{e}) &= x_1 - x_1^2 + \alpha e_1^2 \,, &\quad
    \tilde{f}_3(\B{x},\B{e}) &= 1 \,, &\quad \tilde{f}_5(\B{x},\B{e}) &= x_4
    \cdot e_3 \,, \\
    \tilde{f}_2(\B{x},\B{e}) &= x_1 + x_3 + x_4 + e_1 \,, &\quad \tilde{f}_4(\B{x},\B{e}) &= x_2 +
    e_2 \,, & &
  \end{aligned}
$$
for which 
%the augmented graph $\C{G}^a(\C{M}_{\intervene(\{3\},1)})$ and 
the graph $\C{G}(\C{M}_{\intervene(\{3\},1)})$ is depicted in
Figure~\ref{fig:AugmentedGraphs} (right). This is an example where a perfect
intervention leads to an intervened SCM
$\C{M}_{\intervene(\{3\}, 1)}$ that does not have a solution anymore. In addition, performing a perfect 
intervention $\intervene(\{4\}, 1)$ on $\C{M}_{\intervene(\{3\}, 1)}$ yields again an SCM with a solution for $\alpha \geq 0$. 
\end{example}

Recall that for each solution $\B{X}$ of an SCM $\C{M}$ we call the distribution $\Prb^{\B{X}}$ the observational distribution of $\C{M}$ associated to $\B{X}$. For cyclic SCMs, the observational distribution is in general not unique.\footnote{In order to assure the existence of a unique observational distribution it is common to consider only SCMs for which the structural equations have a unique solution (see, e.g., Definition~7.1.1 in~\cite{Pea09}). Although these SCMs induce a unique observational distribution, they generally do not induce a unique distribution after a perfect intervention.} For example, the SCM $\C{M}$ of Example~\ref{ex:AugmentedGraphs} has two different observational distributions if $\alpha > 0$. Similarly, an intervened SCM may induce a distribution that is not unique. Whenever the intervened SCM $\C{M}_{\intervene(I, \B{\xi}_I)}$ has a solution $\B{X}$ we therefore call the 
distribution $\Prb^{\B{X}}$ the \emph{\hypertarget{def:InterventionalDistribution}{interventional distribution} of $\C{M}$  under the perfect intervention $\intervene(I,\B{\xi}_I)$ associated to $\B{X}$}.\footnote{In the literature, one 
often finds the notation $p(\B{x})$ and $p(\B{x}\given \intervene(\B{X}_I=\B{x}_I))$ for the
densities of the observational and interventional distribution, respectively, in case these are
uniquely defined by the SCM \citep[e.g.,][]{Pea09}.}

%%%%%%%%%%%%%%%%%%%%%%%%%%%%%%%%%%%%%%%%%%%%%%%%%%
\subsection{Counterfactuals}
\label{sec:Counterfactuals}
%%%%%%%%%%%%%%%%%%%%%%%%%%%%%%%%%%%%%%%%%%%%%%%%%%

The causal semantics of an SCM are described by the interventions on the SCM. Adding another layer
of complexity, one can describe the counterfactual semantics of an SCM by the interventions on the so-called
twin SCM, an idea introduced in~\cite{BP94}.
 
\begin{definition}[Twin SCM]
\label{def:TwinSCM}
Let $\C{M}=\langle \C{I}, \C{J}, \BC{X}, \BC{E}, \B{f}, \Prb_{\BC{E}} \rangle$ be an SCM. The \emph{twin operation} maps $\C{M}$ to the \emph{twin structural causal model (twin SCM)}
$$
  \C{M}^{\twin} := \langle \C{I}\cup\C{I}', \C{J}, \BC{X}\times\BC{X}, \BC{E}, \tilde{\B{f}}, 
  \Prb_{\BC{E}} \rangle \,,
$$
where $\C{I}' = \{ i' : i \in \C{I}\}$ is a copy of $\C{I}$ and the causal mechanism $\tilde{\B{f}} : \BC{X} \times \BC{X} 
\times \BC{E} \to \BC{X} \times \BC{X}$ is the measurable function given by 
$\tilde{\B{f}}(\B{x},\B{x}',\B{e}) = \big(\B{f}(\B{x},\B{e}),\B{f}(\B{x}',\B{e})\big)$.
\end{definition}

The twin operation on SCMs preserves the equivalence relation $\equiv$ on the set of all SCMs. We define an analogous twin operation $\twin(\C{I})$ on directed graphs.
\begin{definition}[Twin graph]
\label{def:TwinGraph}
Let $\C{G}=(\C{V},\C{E})$ be a directed graph and $\C{I}\subseteq\C{V}$ a subset such that $\C{J}:=\C{V}\setminus \C{I}$ is \emph{exogenous}, that is, $\pa_{\C{G}}(\C{J}) = \emptyset$. The \emph{$\twin(\C{I})$ operation} maps $\C{G}$ to the \emph{twin graph w.r.t.\ $\C{I}$} defined by $\twin(\C{I})(\C{G}):=(\tilde{\C{V}}, \tilde{\C{E}})$, where:
\begin{enumerate}
  \item $\tilde{\C{V}}=\C{V}\cup \C{I}'$, where $\C{I}'$ is a copy of $\C{I}$,
  \item $\tilde{\C{E}}= \C{E}\cup\C{E}'$, where $\C{E}'$ is given by 
    $$\C{E}' = \{ j \to i' : j \in \C{J}, i \in \C{I}, j \to i \in \C{E} \} \cup \{ \tilde{i}' \to i' : \tilde{i},i \in \C{I}, \tilde{i} \to i \in \C{E} \}$$
    with $i',\tilde{i}' \in \C{I}'$ the respective copies of $i,\tilde{i} \in \C{I}$.
\end{enumerate}
\end{definition}

Twin operations are compatible with the augmented graph mapping and preserve acyclicity.
\begin{proposition}
\label{prop:CommuteTwinGraph}
Let $\C{M} = \langle \C{I}, \C{J}, \BC{X}, \BC{E}, \B{f}, \Prb_{\BC{E}} \rangle$ be an SCM. Then $(\C{G}^a \circ \twin)(\C{M}) = (\twin(\C{I})\circ \C{G}^a)(\C{M})$.
%and $(\C{G} \circ \twin)(\C{M}) = (\twin(\C{I}) \circ \C{G})(\C{M})$.
\end{proposition}

\begin{proposition}
\label{prop:TwinAcyclicityPreserved}
For SCMs and directed graphs, we have that acyclicity is preserved under the twin operation.
\end{proposition}

The perfect intervention and the twin operation for SCMs and directed graphs commute with each other in the following way.
\begin{proposition}
\label{prop:CommuteInterveneTwin}
Let $\C{M} = \langle \C{I}, \C{J}, \BC{X}, \BC{E}, \B{f}, \Prb_{\BC{E}} \rangle$ be an SCM and $\C{G}=(\C{V},\C{E})$ a directed graph. Then we have that perfect intervention commutes with the twin operation on both:
\begin{enumerate}[ref=\theproposition.(\arabic*)]
  \item \label{prop:CommuteInterveneTwin1} the SCM $\C{M}$: for a subset $I\subseteq\C{I}$ and value $\B{\xi}_I\in\BC{X}_I$, $(\intervene(I\cup I',\B{\xi}_{I\cup I'})) \circ \twin)(\C{M}) = (\twin\circ \intervene(I,\B{\xi}_I))(\C{M})$, and
  \item \label{prop:CommuteInterveneTwin2} the directed graph $\C{G}$: for subsets $I\subseteq \C{I}\subseteq\C{V}$ such that $\C{J}:=\C{V}\setminus\C{I}$ is exogenous, $(\intervene(I\cup I') \circ \twin(\C{I}))(\C{G}) = (\twin(\C{I})\circ \intervene(I))(\C{G})$,
\end{enumerate}
where $I'$ is the copy of $I$ in $\C{I}'$ and $\B{\xi}_{I'}=\B{\xi}_I$.
\end{proposition}

Whenever the intervened twin SCM $(\C{M}^{\twin})_{\intervene(\tilde{I}, \B{\xi}_{\tilde{I}})}$, 
where $\tilde{I}\subseteq\C{I}\cup\C{I}'$ and $\B{\xi}_{\tilde{I}} \in \BC{X}_{\tilde{I}}$, has 
a solution $(\B{X},\B{X}')$, we call the distribution 
$\Prb^{(\B{X},\B{X}')}$ the \emph{\hypertarget{def:CounterfactualDistribution}{counterfactual distribution} of $\C{M}$ under the perfect intervention $\intervene(\tilde{I},\B{\xi}_{\tilde{I}})$ associated to $(\B{X},\B{X}')$}. In Example~\ref{ex:PriceSupplyDemand}, we provide an example of how counterfactuals can be sensibly formulated for a well-known market equilibrium model described in terms of a cyclic SCM.

The interpretation of counterfactual statements has received a lot of attention
in the literature \citep{Lew79, Roe97, Byr07, BP94, Pea09}. For acyclic graphs, an alternative graphical approach to counterfactuals is the framework of 
Single World Intervention Graphs (SWIGs) \citep{RR13}. One topic of discussion is that there exist SCMs that induce the same observational and
interventional distributions, but differ in their counterfactual statements
\citep{Daw02} (see also Example~\ref{ex:Counterfactuals}). This raises the
question how one can estimate such SCMs from data. 
\section{Solvability}
\label{sec:Solvability}
%%%%%%%%%%%%%%%%%%%%%%%%%%%%%%%%%%%%%%%%%%%%%%%%%%

In this section, we introduce the notions of solvability and unique
solvability with respect to a subset of the endogenous variables of an SCM. 
They describe the existence and uniqueness of measurable solution 
functions for the subsystem of structural equations that correspond with a 
certain subset of the endogenous variables. 
These notions 
play a central role in 
formulating sufficient conditions under which several properties of acyclic 
SCMs may be extended to the cyclic setting.
For example, we show that solvability of an SCM is a sufficient and necessary
condition for the existence of a solution of an SCM. Further, unique solvability of an
SCM implies the uniqueness of the induced observational distribution. 
%\Stephan{Moreover, we discuss self-cycles and several other properties of (unique) solvability.} 
%Moreover, we discuss several properties of (unique) solvability and we introduce the class of simple SCMs, which have very convenient solvability properties.

%%%%%%%%%%%%%%%%%%%%%%%%%%%%%%%%%%%%%%%%%%%%%%%%%%
\subsection{Definition of solvability}
%%%%%%%%%%%%%%%%%%%%%%%%%%%%%%%%%%%%%%%%%%%%%%%%%%

Intuitively, one can think of the structural equations corresponding to a subset of endogenous variables $\C{O} \subseteq \C{I}$ as a description of how the subsystem formed by the variables $\C{O}$ interacts with the rest of the system $\C{I} \setminus \C{O}$ through the variables $\pa(\C{O})\setminus\C{O}$. A solution function w.r.t.\ $\C{O}$ assigns each input value $(\B{x}_{\pa(\C{O})\setminus\C{O}},\B{e}_{\pa(\C{O})})$ of this subsystem to a specific output value $\B{x}_{\C{O}}$ of the subsystem. This is formalized as follows.
\begin{definition}[Solvability]
\label{def:Solvability}
Let $\C{M} = \langle \C{I}, \C{J}, \BC{X}, \BC{E}, \B{f}, \Prb_{\BC{E}} \rangle$ be an SCM. We call $\C{M}$ \emph{solvable w.r.t.\ $\C{O} \subseteq \C{I}$} if there exists a measurable mapping 
$\B{g}_{\C{O}} : \BC{X}_{\pa(\C{O})\setminus\C{O}} \times \BC{E}_{\pa(\C{O})} \to \BC{X}_{\C{O}}$ 
such that for $\Prb_{\BC{E}}$-almost every $\B{e}\in\BC{E}$ and for all $\B{x} \in \BC{X}$,
$$
  \B{x}_{\C{O}} = \B{g}_{\C{O}}(\B{x}_{\pa(\C{O})\setminus\C{O}}, \B{e}_{\pa(\C{O})}) 
  \quad\implies\quad
  \B{x}_{\C{O}} = \B{f}_{\C{O}}(\B{x}, \B{e}) \,.
$$
We then call $\B{g}_{\C{O}}$ a \emph{measurable solution function w.r.t.\ $\C{O}$ for $\C{M}$}.
We call \emph{$\C{M}$ solvable} if it is solvable w.r.t.\ $\C{I}$.
\end{definition}
By definition, solvability w.r.t.\ a subset respects the equivalence relation
$\equiv$ on SCMs. The measurable solution functions w.r.t.\ a certain subset do not always exist, and if they exist, they are not always uniquely defined. For example, for the 
SCM $\C{M}$ in Example~\ref{ex:AugmentedGraphs}, the measurable solution functions w.r.t.\ $\{1\}$ are given by $g^{\pm}_1(e_1) = \pm \sqrt{\alpha e_1^2}$  if and only if $\alpha \geq 0$.
The following theorem states that various possible notions of ``solvability'' are equivalent.
\begin{theorem}[Sufficient and necessary conditions for solvability]
\label{thm:SolvabilityIffCondition}
For an SCM $\C{M} = \langle \C{I}, \C{J}, \BC{X}, \BC{E}, \B{f}, \Prb_{\BC{E}} \rangle$, the following are equivalent:
\begin{enumerate}
\item $\C{M}$ has a solution (see Definition~\ref{def:Solution});
\item for $\Prb_{\BC{E}}$-almost every $\B{e}\in\BC{E}$ the structural equations 
$
  \B{x} = \B{f}(\B{x},\B{e})
$
have a solution $\B{x}\in\BC{X}$;
\item $\C{M}$ is solvable (see Definition~\ref{def:Solvability}).
\end{enumerate}
\end{theorem}
While in the acyclic case, the above theorem is almost trivial, in the cyclic case the measure-theoretic aspects are not that obvious. In particular, to prove the existence of a \emph{measurable} solution function $\B{g} : \BC{E}_{\pa(\C{I})} \to \BC{X}$ in case the structural equations have a solution for almost every $\B{e} \in \BC{E}$, we make use of a strong measurable selection theorem (see Theorem~\ref{thm:MeasurableSelectionThm} or \citep{Kec95}). This theorem implies that if there exists a solution $\B{X}:\Omega\to\BC{X}$, then there necessarily exists a
random variable $\B{E}:\Omega\to\BC{E}$ and a mapping $\B{g}:\BC{E}_{\pa(\C{I})}\to\BC{X}$ such that
$\B{g}(\B{E}_{\pa(\C{I})})$ is a solution. However, it does not imply that there necessarily 
exists a random variable $\B{E}:\Omega\to\BC{E}$ and a mapping $\B{g}:\BC{E}_{\pa(\C{I})}\to\BC{X}$ such 
that $\B{X}=\B{g}(\B{E}_{\pa(\C{I})})$ holds a.s., for example, if $\B{X}$ is a nontrivial mixture of such solutions (see Example~\ref{ex:MixturesOfSolutions}).

Solvability w.r.t.\ a strict subset of $\C{I}$ is in general
neither sufficient nor necessary for the existence of a (global) solution of the
SCM. Consider, for example, the SCM $\C{M}$ in
Example~\ref{ex:AugmentedGraphs} with $\alpha < 0$. Even though this SCM is
solvable w.r.t.\ $\{2,3,4\}$, it is not (globally) solvable, and hence does not
have any solution. 
In Proposition~\ref{prop:SolvabilityIfCondition}, we provide a sufficient condition for solvability w.r.t.\ a strict subset of $\C{I}$ that is similar to condition (2) in Theorem~\ref{thm:SolvabilityIffCondition} in the sense that it is formulated in terms of the solutions of (a subset of) the structural equations without requiring measurability of the solutions.
%We prove in Proposition~\ref{prop:SolvabilityIfCondition} that there exists a sufficient condition for solvability w.r.t.\ a subset of $\C{I}$ that is similar to the sufficient and necessary condition (2) in Theorem~\ref{thm:SolvabilityIffCondition}.
For the class of linear SCMs, we provide in Proposition~\ref{prop:LinearSolvable} a sufficient and necessary condition for solvability w.r.t.\ a subset of $\C{I}$.

\subsection{Unique solvability}
\label{sec:UniqueSolvability}
%%%%%%%%%%%%%%%%%%%%%%%%%%%%%%%%%%%%%%%%%%%%%%%%%%

The notion of unique solvability w.r.t.\ a subset $\C{O}\subseteq\C{I}$ is similar to the notion of solvability, but with the additional requirement that the measurable solution function $\B{g}_{\C{O}}:\BC{X}_{\pa(\C{O})\setminus\C{O}}\times\BC{E}_{\pa(\C{O})}\to\BC{X}_{\C{O}}$ is unique up to a $\Prb_{\BC{E}}$-null set.
\begin{definition}[Unique solvability]
\label{def:UniqueSolvability}
Let $\C{M} = \langle \C{I}, \C{J}, \BC{X}, \BC{E}, \B{f}, \Prb_{\BC{E}} \rangle$ be an SCM. We call $\C{M}$ \emph{uniquely solvable w.r.t.\ $\C{O} \subseteq \C{I}$} if there exists a measurable mapping 
$\B{g}_{\C{O}} : \BC{X}_{\pa(\C{O})\setminus\C{O}} \times \BC{E}_{\pa(\C{O})} \to \BC{X}_{\C{O}}$ 
such that for $\Prb_{\BC{E}}$-almost every $\B{e}\in\BC{E}$ and for all $\B{x} \in \BC{X}$,
$$
    \B{x}_{\C{O}} = \B{g}_{\C{O}}(\B{x}_{\pa(\C{O})\setminus\C{O}}, \B{e}_{\pa(\C{O})}) 
    \quad\iff\quad
    \B{x}_{\C{O}} =
    \B{f}_{\C{O}}(\B{x},\B{e}) \,.
$$
We call \emph{$\C{M}$ uniquely solvable} if it is uniquely solvable w.r.t.\ $\C{I}$.
\end{definition}
If $\C{M}\equiv\tilde{\C{M}}$ and $\C{M}$ is uniquely solvable w.r.t.\ $\C{O}$, then
$\tilde{\C{M}}$ is uniquely solvable w.r.t.\ $\C{O}$, too, and the same mapping $\B{g}_{\C{O}}$ is a measurable solution
function w.r.t.\ $\C{O}$ for both $\C{M}$ and $\tilde{\C{M}}$.

The following result explains why the notions of (unique) solvability do not play an important role in the theory of acyclic SCMs.
\begin{proposition}
\label{prop:AcyclicSCMUniquelySolvable}
An acyclic SCM $\C{M} = \langle \C{I}, \C{J}, \BC{X}, \BC{E}, \B{f}, \Prb_{\BC{E}} \rangle$ is uniquely solvable w.r.t.\ every subset $\C{O} \subseteq \C{I}$. 
\end{proposition}
We now illustrate that also cyclic SCMs can be uniquely solvable w.r.t.\ 
%a
every subset.
\begin{example}[Cyclic SCM, uniquely solvable w.r.t.\ each subset]
  \label{ex:CyclicExample}
  \begin{figure}
  \adjustbox{scale=0.8,center}{%
  \begin{tikzpicture}
    \begin{scope}
    \node[var] (X1) at (0,1.25) {$X_1$};
    \node[var] (X2) at (1.5,1.25) {$X_2$};
    \node[var] (X3) at (0,0) {$X_3$};
    \node[var] (X4) at (1.5,0) {$X_4$};
    \draw[arr] (X1) -- (X3);
    \draw[arr] (X2) -- (X4);
    \draw[arr] (X3) to [bend right=20] (X4);
    \draw[arr] (X4) to [bend right=20] (X3);
    \node (M) at (-1.0,-0.2) {$\C{G}(\C{M})$};
    \end{scope}
    \begin{scope}[shift={(4,0)}]
    \node[var] (X1) at (0,1.25) {$X_1$};
    \node[var] (X2) at (1.5,1.25) {$X_2$};
    \node[var] (X3) at (0,0) {$X_3$};
    \node[var] (X4) at (1.5,0) {$X_4$};
    \draw[arr] (X1) -- (X3);
    \draw[arr] (X2) -- (X4);
    \draw[arr] (X2) -- (X3);
    \draw[arr] (X1) -- (X4);
    \draw[biarr] (X4) to [bend right=20] (X3);
    \node (Mt) at (-1.0,-0.2) {$\C{G}(\tilde{\C{M}})$};
    \end{scope}
    \begin{scope}[shift={(10cm,0cm)}]
      \node[var] (X1) at (0,0) {$X_1$};
      \node[var] (X2) at (1.5,0) {$X_2$}; 
      %\node[exvar] (E1) at (0,1.25) {$E_1$};
      %\node[exvar] (E2) at (1.75,1.25) {$E_2$}; 
      \draw[arr] (X1) -- (X2);
      %\draw[arr] (E1) -- (X1);
      %\draw[arr] (E2) -- (X2);
      \node (Mb) at (-1.0,-0.2) {$\C{G}(\bar{\C{M}})$};
    \end{scope}
    \begin{scope}[shift={(14cm,0cm)}]
      \node[var] (X1) at (0,0) {$X_1$};
      \node[var] (X2) at (1.5,0) {$X_2$}; 
      %\node[exvar] (E1) at (0,1.25) {$E_1$};
      %\node[exvar] (E2) at (1.75,1.25) {$E_2$}; 
      %\draw[arr] (E1) -- (X1);
      %\draw[arr] (E2) -- (X2);
      \node (Mh) at (-1.0,-0.2) {$\C{G}(\hat{\C{M}})$};
    \end{scope}
\end{tikzpicture}}
\caption{Left: The graphs of the observationally equivalent SCMs $\C{M}$ and
$\tilde{\C{M}}$ of Examples~\ref{ex:CyclicExample} and \ref{ex:ObservationalEquivalenceDoesNotImplyEquivalence}, respectively. Right: The graphs of the interventionally equivalent SCMs $\bar{\C{M}}$ and $\hat{\C{M}}$ of Example~\ref{ex:InterventionalEquivalenceDoesNotImplySameAugmentedGraph}.}
  \label{fig:CyclicAndSimpleExample}
\end{figure}
Consider the SCM $\C{M} = \langle \B{4}, \B{4}, \RN^4, \RN^4, \B{f},
\Prb_{\RN^4} \rangle$ with causal mechanism given by
$$
f_1(\B{x},\B{e}) = e_1 \,,\,\quad f_2(\B{x},\B{e})=e_2 \,,\,\quad f_3(\B{x},\B{e})=x_1 x_4
+ e_3 \,,\,\quad f_4(\B{x},\B{e})=x_2 x_3 + e_4 
$$
and $\Prb_{\RN^4}$ the standard-normal distribution on $\RN^4$. This SCM $\C{M}$ is uniquely solvable w.r.t.\ every subset and its (augmented) graph includes a cycle (see Figure~\ref{fig:CyclicAndSimpleExample}).
\end{example}

Theorem~\ref{thm:SolvabilityIffCondition} provides sufficient and necessary conditions for (global) solvability. The next theorem states that under the additional uniqueness requirement there exists a sufficient and necessary condition for unique solvability w.r.t.\ any subset (for solvability w.r.t.\ a subset we only have the sufficient condition provided in Proposition~\ref{prop:SolvabilityIfCondition}), and moreover, that all solutions of a uniquely solvable SCM induce the same observational distribution.
\begin{theorem}[Sufficient and necessary conditions for unique solvability]
\label{thm:UniqueSolvabilityIffCondition}
Let $\C{M} = \langle \C{I}, \C{J}, \BC{X}, \BC{E}, \B{f}, \Prb_{\BC{E}} \rangle$ be an SCM and $\C{O}\subseteq\C{I}$ a subset.
%of endogenous variables. 
The following are equivalent:
\begin{enumerate}
\item for $\Prb_{\BC{E}}$-almost every $\B{e}\in\BC{E}$ and for all $\B{x}_{\setminus\C{O}} \in \BC{X}_{\setminus\C{O}}$ the structural equations
$$
    \B{x}_{\C{O}} = \B{f}_{\C{O}}(\B{x}, \B{e})
$$
have a unique solution $\B{x}_{\C{O}}\in\BC{X}_{\C{O}}$;
\item $\C{M}$ is uniquely solvable w.r.t.\ $\C{O}$.
\end{enumerate}
Furthermore, if $\C{M}$ is uniquely solvable, then there exists a solution, and all solutions have the same observational distribution.
\end{theorem}
It is well known that under acyclicity the observational distribution is unique. Theorem~\ref{thm:UniqueSolvabilityIffCondition} generalizes this result to settings with cycles.
%This result generalizes the known sufficient acyclicity condition for the unicity of the observational distribution of an SCM to the cyclic setting.
For linear SCMs, the unique solvability condition w.r.t.\ a subset is equivalent to a matrix invertibility condition (see Proposition~\ref{prop:LinearUniquelySolvable}).

In general, (unique) solvability w.r.t.\ $\C{O} \subseteq \C{I}$ does not imply 
(unique) solvability w.r.t.\ a strict superset $\C{O} \subsetneq \C{V} \subseteq \C{I}$ nor w.r.t.\ a strict subset $\C{W} \subsetneq \C{O}$ (see Example~\ref{ex:SolvabilityProperties}). Moreover, (unique) solvability is in general not preserved under unions and intersections (see Appendix~\ref{app:AppendixSolvabilityUnionAndIntersections}).

\subsection{Self-cycles}
\label{sec:SelfCycles}
%%%%%%%%%%%%%%%%%%%%%%%%%%%%%%%%%%%%%%%%%%%%%%%%%%

%Intuitively, a self-cycle at an endogenous variable denotes that that variable cannot be uniquely determined by its parents up to a $\Prb_{\BC{E}}$-null set.
%, since it needs to satisfy the structural equation of that variable (in their structural minimal representation) up to a $\Prb_{\BC{E}}$-null set. For example, for the self-cycle at the endogenous variable $1$ in the SCM $\C{M}$ of Example~\ref{ex:AugmentedGraphs} not all parents, up to $\Prb_{\BC{E}}$-null set, satisfy structural equation $x_1=x_1-x_1^2+\alpha e_1^2$ if $\alpha\neq 0$. In contrast, the lack of a self-cycle at an endogenous variable denotes that that variable can be uniquely determined by its parents up to a $\Prb_{\BC{E}}$-null set. For example, if we take $\alpha=0$, then the variable $1$ in $\C{M}$ does not have a self-cycle at $1$ and in that case $x_1$ is equal to zero up to $\Prb_{\BC{E}}$-null set. 
%In other words, a self-cycle at an endogenous variable in the (augmented) graph of an SCM indicates that the SCM is not uniquely solvable w.r.t.\ that variable, and vice versa. 

%The notion of a self-cycle is closely related to the notion of unique solvability. 
One can think of a structural equation of a single endogenous variable $i\in\C{I}$ as describing a small subsystem that interacts with the rest of the system. If the output $x_i$ of this subsystem is uniquely determined by the input $(\B{x}_{\setminus i},\B{e})$ from the rest of the system (up to a $\Prb_{\BC{E}}$-null set), then $i$ is not a parent of itself (see Definition~\ref{def:Parents}).
\begin{proposition}[Self-cycles]
\label{prop:ObstructionToSelfCycles}
The SCM $\C{M} = \langle \C{I}, \C{J}, \BC{X}, \BC{E}, \B{f}, \Prb_{\BC{E}} \rangle$ is uniquely
solvable w.r.t.\ $\{i\}$ for $i\in\C{I}$ if and only if $\C{G}^a(\C{M})$ (or
$\C{G}(\C{M})$) has no self-cycle $i \to i$ at $i\in\C{I}$.
\end{proposition}
A self-cycle at an endogenous variable denotes that that variable is not uniquely determined by its parents, up to a $\Prb_{\BC{E}}$-null set. This implies that an SCM with a self-cycle at an endogenous variable in its graph can be either solvable, or not solvable, w.r.t.\ that variable. For the SCM $\C{M}$ of Example~\ref{ex:AugmentedGraphs}, we have indeed that it is solvable w.r.t.\ $\{1\}$ for $\alpha > 0$, while for $\alpha <0$ it is not. For linear SCMs with structural equations $X_i=\sum_{j\in\C{I}} B_{ij}X_j+\sum_{k\in\C{J}}\Gamma_{ik} E_k$, the endogenous variable $i\in\C{I}$ has a self-cycle if and only if $B_{ii}=1$ (see also Appendix~\ref{app:AppendixLinSCMs}).

%%%%%%%%%%%%%%%%%%%%%%%%%%%%%%%%%%%%%%%%%%%%%%%%%%
\subsection{Interventions}
\label{sec:SolvabilityInterventions}
%%%%%%%%%%%%%%%%%%%%%%%%%%%%%%%%%%%%%%%%%%%%%%%%%%

%We know that a solvable SCM has a solution (see Theorem~\ref{thm:SolvabilityIffCondition}), and moreover that all solutions of a uniquely solvable SCM induce the same observational distribution (see Theorem~\ref{thm:UniqueSolvabilityIffCondition}).
The property of (unique) solvability is in general not preserved under perfect intervention. For
example, a (uniquely) solvable SCM can lead to a nonuniquely 
solvable SCM after intervention, which either has no solution or has solutions
with multiple induced distributions (see, e.g., Examples~\ref{ex:Interventions} and \ref{ex:InterventionUniqueSolvability}). A sufficient condition for the intervened SCM to be (uniquely) solvable is that the original SCM has to be (uniquely) solvable w.r.t.\ the subset of nonintervened endogenous variables. 
\begin{proposition}
\label{prop:SolvabilityPreservedIntervenedModel}
Let $\C{M} = \langle \C{I}, \C{J}, \BC{X}, \BC{E}, \B{f}, \Prb_{\BC{E}} \rangle$ be an SCM that is (uniquely) solvable w.r.t.\ $\C{O} \subseteq \C{I}$. Then, for any set $I$ such that $\pa(\C{O})\setminus\C{O} \subseteq I \subseteq \C{I}\setminus\C{O}$ and value $\B{\xi}_I\in\BC{X}_I$ the intervened SCM $\C{M}_{\intervene(I,\B{\xi}_I)}$ is (uniquely) solvable w.r.t.\ $\C{O}\cup I$.
\end{proposition}
%\begin{proposition}
%\label{prop:SolvabilityPreservedIntervenedModel}
%Let $\C{M} = \langle \C{I}, \C{J}, \BC{X}, \BC{E}, \B{f}, \Prb_{\BC{E}} \rangle$ be an SCM and let 
%$\C{O} \subseteq \C{I}$. If $\C{M}$ is (uniquely) solvable w.r.t.\ $\C{O}$, then for every 
%$\B{\xi}_{\pa(\C{O})\setminus\C{O}} \in \BC{X}_{\pa(\C{O})\setminus\C{O}}$ the intervened SCM 
%$\C{M}_{\intervene(\pa(\C{O})\setminus\C{O},\B{\xi}_{\pa(\C{O})\setminus\C{O}})}$ is 
%(uniquely) solvable w.r.t.\ $\pa(\C{O})\cup\C{O}$. In particular, for every
%$\B{\xi}_{\C{I}\setminus\C{O}}\in\BC{X}_{\C{I}\setminus\C{O}}$ the intervened SCM
%$\C{M}_{\intervene(\C{I}\setminus\C{O},\B{\xi}_{\C{I}\setminus\C{O}})}$ is
%(uniquely) solvable.
%\end{proposition}

Proposition~\ref{prop:AcyclicSCMUniquelySolvable} shows that acyclic SCMs are uniquely solvable w.r.t.\ every subset and hence are uniquely solvable after every perfect intervention. This also directly follows from the fact that acyclicity is preserved under perfect intervention (see Proposition~\ref{prop:IntDisjSubsetAndAcyclicity}). Moreover, since acyclicity is preserved under the twin operation (see Proposition~\ref{prop:TwinAcyclicityPreserved}), an acyclic SCM induces unique observational, interventional and counterfactual distributions.

\subsection{Ancestral (unique) solvability}
\label{sec:AncestralSolvabilityAndUniqueSolvability}
%%%%%%%%%%%%%%%%%%%%%%%%%%%%%%%%%%%%%%%%%%%%%%%%%%

We saw that, in general, solvability w.r.t.\ $\C{O}\subseteq\C{I}$ does not
imply solvability w.r.t.\ a strict subset of $\C{O}$. Here we show that
it does imply solvability w.r.t.\ the ancestral subsets in
$\C{G}(\C{M})_{\C{O}}$, that is, in the induced subgraph of the graph
$\C{G}(\C{M})$ on $\C{O}$. A subset $\C{A}\subseteq\C{O}$ is called an \emph{ancestral subset} in $\C{G}(\C{M})_{\C{O}}$ if $\C{A}=\an_{\C{G}(\C{M})_{\C{O}}}(\C{A})$, where $\an_{\C{G}(\C{M})_{\C{O}}}(\C{A})$ are the ancestors of $\C{A}$ according to the induced subgraph\footnote{%
Here, one can also use the augmented graph $\C{G}^a(\C{M})$ on $\C{O}$ since
$\an_{\C{G}(\C{M})_{\C{O}}}(\C{A}) = \an_{\C{G}^a(\C{M})_{\C{O}}}(\C{A})$ for
every subset $\C{A}\subseteq\C{O}\subseteq\C{I}$.} $\C{G}(\C{M})_{\C{O}}$.
\begin{definition}[Ancestral (unique) solvability]
\label{def:AncestralUniqueSolvability}
Let $\C{M}=\langle \C{I}, \C{J}, \BC{X}, \BC{E}, \B{f},
\Prb_{\BC{E}} \rangle$ be an SCM. We call $\C{M}$ \emph{ancestrally (uniquely) solvable w.r.t.\
$\C{O}\subseteq\C{I}$} if $\C{M}$ is (uniquely) solvable w.r.t.\ every ancestral subset in $\C{G}(\C{M})_{\C{O}}$. We call $\C{M}$ \emph{ancestrally (uniquely)
solvable} if it is ancestrally (uniquely) solvable w.r.t.\ $\C{I}$.
\end{definition}

\begin{proposition}[Solvability is equivalent to ancestral solvability]
\label{prop:SolvabilityAncestralSolvability}
The SCM $\C{M} = \langle \C{I}, \C{J}, \BC{X}, \BC{E}, \B{f}, \Prb_{\BC{E}} \rangle$ is solvable w.r.t.\ the subset $\C{O}\subseteq\C{I}$ if and only if $\C{M}$ is ancestrally solvable w.r.t.\ $\C{O}$. 
\end{proposition}
A similar result does not hold for unique solvability. Although
ancestral unique solvability w.r.t.\ $\C{O}\subseteq\C{I}$ implies unique solvability
w.r.t.\ $\C{O}$, the converse does not hold in general, as the following example illustrates.

\begin{example}[Unique solvability w.r.t.\ $\C{O}$ does not imply ancestral unique solvability w.r.t.\ $\C{O}$]
\label{ex:UniqueSolvabilityNotUniqAncestralSubset}
Consider the SCM $\C{M}=\langle \B{4}, \B{1}, \RN^4, \RN, \B{f}, \Prb_{\RN} \rangle$ with causal
mechanism given by
$$
  f_1(\B{x},e) = e \,, \,\, f_2(\B{x},e) = x_2 \cdot (1-\B{1}_{\{0\}}(x_1-x_3)) + 1 \,, \,\, f_3(\B{x},e)
  = x_3 \,, \,\, f_4(\B{x},e) = x_3
$$
and $\Prb_{\RN}$ the standard-normal measure on $\RN$. This SCM is uniquely solvable w.r.t.\ the set $\{2,3\}$, and thus solvable w.r.t.\ this set. Although it is solvable w.r.t., the ancestral subset $\{3\}$ in $\C{G}(\C{M})_{\{2,3\}}$, depicted in Figure~\ref{fig:NoLatentProjection} (left), it is not uniquely
solvable w.r.t.\ this subset, because the structural equation $x_3=x_3$ holds for any $x_3\in\RN$. Hence, it is not ancestrally uniquely solvable w.r.t.\ $\{2,3\}$.
\begin{figure}
  \begin{center}
  \adjustbox{scale=0.8,center}{%
    \begin{tikzpicture}
%      \begin{scope}
%      \node[var] (X1) at (1,-1.25) {$X_1$};
%      \node[var] (X2) at (0,0) {$X_2$};
%      \node[var] (X3) at (2,0) {$X_3$};
%      \node[var] (X4) at (3,-1.25) {$X_4$};
%%      \node[exvar] (E) at (-1,-1.25) {$E$};
%      \draw[arr] (X1) -- (X2);
%      \draw[arr] (X3) -- (X2);
%      \draw[arr] (X2) to [out=60,in=120,looseness=5] (X2);
%      \draw[arr] (X3) to [out=60,in=120,looseness=5] (X3);
%      \draw[arr] (X3) -- (X4);
%%      \draw[arr] (E) -- (X3);
%%      \node at (1,-2.3) {$\C{G}^a(\C{M})$};
%      \end{scope}
      \begin{scope}
      \node[var] (X1) at (-2,0) {$X_1$};
      \node[var] (X2) at (0,0) {$X_2$};
      \node[var] (X3) at (2,0) {$X_3$};
      \node[var] (X4) at (4,0) {$X_4$};
%      \node[exvar] (E) at (-1,-1.25) {$E$};
      \draw[arr] (X1) -- (X2);
      \draw[arr] (X3) -- (X2);
      \draw[arr] (X2) to [out=62,in=118,looseness=4.4] (X2);
      \draw[arr] (X3) to [out=62,in=118,looseness=4.4] (X3);
      \draw[arr] (X3) -- (X4);
%      \draw[arr] (E) -- (X3);
%      \node at (1,-2.3) {$\C{G}^a(\C{M})$};
      \node (BoxL) [draw=black, fit=(X2) (X3), inner xsep=0.3cm, inner ysep=0.5cm, dashed, 
      ultra thick, fill=black!50, fill opacity=0.2] {};
      \node at (1,0.45) {$\C{L}$};
%      \node at (0,-1) {}; %spacing
      \coordinate (A) at (6,0);
      \coordinate (B) at (7,0);
      \end{scope}
      \begin{scope}[shift={(8,0)}]
      \node[var] (X1) at (1,0) {$X_1$};
      \node[var] (X4) at (3,0) {$X_4$};
%      \node[exvar] (E) at (-1,-1.25) {$E$};
      \draw[arr] (X1) -- (X4);
%      \draw[arr] (E) -- (X3);
%      \node at (1,-2.3) {$\C{G}^a(\C{M}_{\marg(\C{L})})$};
      \end{scope}
    \draw[->,double,thick, draw=black, shorten >=0cm,shorten <=0cm] (A) to node[above=4pt] 
    {$\marg(\C{L})$} (B);
  \end{tikzpicture}}
  \end{center}
  \caption{The graphs of the SCM $\C{M}$ (left) of Example~\ref{ex:UniqueSolvabilityNotUniqAncestralSubset} and the marginal SCM $\C{M}_{\marg(\{2,3\})}$ 
  (right) of Example~\ref{ex:NoLatentProjection}.}
  \label{fig:NoLatentProjection}
\end{figure}
\end{example}

%\begin{example}[Unique solvability w.r.t.\ $\C{O}$ does not imply ancestral unique solvability w.r.t.\ $\C{O}$]
%\label{ex:UniqueSolvabilityNotUniqAncestralSubset}
%Consider the SCM $\C{M}=\langle \B{3}, \B{1}, \RN^3, \RN, \B{f}, \Prb_{\RN} \rangle$ where the causal
%mechanism is given by
%$$
%  f_1(\B{x},e) = x_1 \cdot (1-\B{1}_{\{0\}}(x_2-x_3)) + 1 \,, \,\, f_2(\B{x},e)
%  = x_2 \,, \,\, f_3(\B{x},e) = e 
%$$
%and $\Prb_{\RN}$ is the standard-normal measure on $\RN$. This SCM is uniquely solvable w.r.t.\ the set $\{1,2\}$, and thus solvable w.r.t.\ this set. Although it is solvable w.r.t.\ the ancestral subset $\{2\}$ in $\C{G}(\C{M})_{\{1,2\}}$ it is not uniquely
%solvable w.r.t.\ this subset. Hence, it is not ancestrally uniquely solvable w.r.t.\ $\{1,2\}$. 
%\end{example}
However, for the class of linear SCMs we have that unique solvability w.r.t.\ $\C{O}$ always implies ancestral unique solvability w.r.t.\ $\C{O}$ (see Proposition~\ref{prop:LinearEquivalentUniqueSolvability}).

Although in general unique solvability is not preserved under unions, in Proposition~\ref{prop:UniqueSolvabilityAncestralUnion} we show that if an SCM is uniquely solvable w.r.t.\ two ancestral subsets and w.r.t.\ their intersection, then it is uniquely solvable w.r.t.\ their union. In general, the property of ancestral unique solvability is not preserved under perfect
intervention, as can be seen in Example~\ref{ex:InterventionUniqueSolvability}.
%as is illustrated in the following example.
%\begin{example}[Ancestral unique solvability is not preserved under perfect intervention]
%\label{ex:AncestralUniqueSolvabilityNotPreservedUnderIntervention}
%Consider the SCM $\C{M}$ of
%Example~\ref{ex:UniqueSolvabilityNotUniqAncestralSubset}, but with the causal
%mechanism component $f_2$ replaced by $f_2(\B{x},e)=x_3$. This SCM $\C{M}$ is
%ancestrally uniquely solvable, however $\C{M}_{\intervene(\{2\},\xi_2)}$ for some $\xi_2\in\RN$ is not ancestrally uniquely solvable.
%\end{example}
%
The notion of ancestral unique solvability will appear in various results in Sections~\ref{sec:Marginalization}
and \ref{sec:MarkovProperty}.

%%%%%%%%%%%%%%%%%%%%%%%%%%%%%%%%%%%%%%%%%%%%%%%%%%
\section{Equivalences}
\label{sec:Equivalences}
%%%%%%%%%%%%%%%%%%%%%%%%%%%%%%%%%%%%%%%%%%%%%%%%%%

In Section~\ref{sec:SCM}, we already encountered an equivalence relation on 
the class of SCMs (see Definition~\ref{def:EquivSCMs}). The (augmented) graph of an SCM, its solutions and its
induced observational, 
interventional and counterfactual distributions are preserved under this equivalence relation. In this section, we give several coarser
equivalence relations on the class of SCMs: observational, interventional and
counterfactual equivalence.  

%%%%%%%%%%%%%%%%%%%%%%%%%%%%%%%%%%%%%%%%%%%%%%%%%%
\subsection{Observational equivalence}
%%%%%%%%%%%%%%%%%%%%%%%%%%%%%%%%%%%%%%%%%%%%%%%%%%

Observational equivalence is the property that two SCMs are indistinguishable on the 
basis of their observational distributions.
\begin{definition}[Observational equivalence]
\label{def:ObservationallyEquivalent}
Two SCMs $\C{M} = \langle \C{I}, \C{J}, \BC{X}, \BC{E}, \B{f}, \Prb_{\BC{E}} \rangle$ and 
$\tilde{\C{M}} = \langle \tilde{\C{I}}, \tilde{\C{J}}, \tilde{\BC{X}}, \tilde{\BC{E}}, \tilde{\B{f}}, \Prb_{\tilde{\BC{E}}} \rangle$ are \emph{observationally equivalent w.r.t.\ $\C{O} \subseteq \C{I}\cap\tilde{\C{I}}$}, denoted by $\C{M}\equivobs{\C{O}}\tilde{\C{M}}$, 
if $\BC{X}_{\C{O}} = \tilde{\BC{X}}_{\C{O}}$ and for 
all solutions $\B{X}$ of $\C{M}$ there exists a solution $\tilde{\B{X}}$ of $\tilde{\C{M}}$ such that 
$\Prb^{\B{X}_{\C{O}}} = \Prb^{\tilde{\B{X}}_{\C{O}}}$ and for all solutions 
$\tilde{\B{X}}$ of $\tilde{\C{M}}$ there exists a solution $\B{X}$ of $\C{M}$ such that 
$\Prb^{\B{X}_{\C{O}}} = \Prb^{\tilde{\B{X}}_{\C{O}}}$. $\C{M}$ and $\tilde{\C{M}}$ are called 
\emph{observationally equivalent} if they are observationally equivalent w.r.t.\ $\C{I} = \tilde{\C{I}}$.
%with respect to 
\end{definition}
%That is, observationally equivalent SCMs induce the same observational distributions. 
Equivalent SCMs have the same solutions, and hence they are observationally equivalent w.r.t.\ every subset $\C{O}\subseteq\C{I}$. However, observational equivalence does not imply equivalence.
\begin{example}[Observational equivalence does not imply equivalence]
\label{ex:ObservationalEquivalenceDoesNotImplyEquivalence}
Consider the SCM $\tilde{\C{M}}$ that is the same as $\C{M}$ of Example~\ref{ex:CyclicExample} but with the causal mechanism $\tilde{\B{f}}$ given by
$$
\tilde{f}_1(\B{x},\B{e}):=e_1 ,\, \quad 
\tilde{f}_2(\B{x},\B{e}):=e_2 ,\, \quad 
\tilde{f}_3(\B{x},\B{e}):=\tfrac{x_1 e_4 + e_3}{1-x_1 x_2} ,\, \quad
\tilde{f}_4(\B{x},\B{e}):=\tfrac{x_2 e_3 + e_4}{1-x_1 x_2} \,.
$$
This SCM $\tilde{\C{M}}$ is observationally equivalent to the SCM $\C{M}$. Because both SCMs have a different (augmented) graph they are not equivalent to each other (see Figure~\ref{fig:CyclicAndSimpleExample}). 
\end{example}

This example shows that if two SCMs $\C{M}$ and $\tilde{\C{M}}$ are observationally equivalent, then their associated 
augmented graphs $\C{G}^a(\C{M})$ and $\C{G}^a(\tilde{\C{M}})$ are not necessarily equal to each other.
%Although the SCMs of this example are observationally equivalent, they are not interventionally equivalent, as we will see in the next subsection.

%For two observationally equivalent SCMs $\C{M}$ and $\tilde{\C{M}}$ it is generally not true 
%that unique solvability of $\C{M}$ implies unique solvability of 
%$\tilde{\C{M}}$, as the following example illustrates:

%This next example is not applicable: it is not a standard measurable space.
%\begin{example}
%Consider the uniquely solvable SCM $\C{M} = \langle \B{1}, \B{1}, \C{X}, \C{E}, f, 
%\Prb_{\C{E}}\rangle$, where we take $\C{X}=\C{E} = \{-1,1\}$ with the trivial $\sigma$-algebra 
%$\{\emptyset,\{-1,1\}\}$ and the causal mechanism $f(e) = e$. Then 
%all solutions of $\C{M}$ induce the same probability distribution $\Prb_{\C{X}}$ on $\C{X}$. 
%Consider another SCM $\tilde{\C{M}}$ that differs from $\C{M}$ only by its causal mechanism 
%  $\tilde{f}$, which we take to be $\tilde{f}(e) = x^2+x-e^2$. In 
%example~\ref{ex:UniquelySolvableNotNecessary} we saw that $\tilde{\C{M}}$ is a 
%non-uniquely solvable SCM, that induce the same distribution $\Prb_{\C{X}}$ on 
%$\C{X}$ as the one of $\C{M}$ and hence $\C{M}$ and $\tilde{\C{M}}$ are observationally 
%equivalent.
%\end{example}

%%%%%%%%%%%%%%%%%%%%%%%%%%%%%%%%%%%%%%%%%%%%%%%%%%
\subsection{Interventional equivalence}
%%%%%%%%%%%%%%%%%%%%%%%%%%%%%%%%%%%%%%%%%%%%%%%%%%

We consider two SCMs to be interventionally equivalent if they induce the same interventional distributions under all perfect interventions.
\begin{definition}[Interventional equivalence]
\label{def:InterventionallyEquivalent}
Two SCMs $\C{M} = \langle \C{I}, \C{J}, \BC{X}, \BC{E}, \B{f}, \Prb_{\BC{E}} \rangle$ and 
$\tilde{\C{M}} = \langle \tilde{\C{I}}, \tilde{\C{J}}, \tilde{\BC{X}}, \tilde{\BC{E}}, \tilde{\B{f}}, \Prb_{\tilde{\BC{E}}} \rangle$ are \emph{interventionally equivalent w.r.t.\ $\C{O} \subseteq \C{I}\cap\tilde{\C{I}}$}, denoted by $\C{M}\equivint{\C{O}}\tilde{\C{M}}$, if $\BC{X}_{\C{O}}=\tilde{\BC{X}}_{\C{O}}$ and for every $I \subseteq \C{O}$ and every value 
$\B{\xi}_I \in \BC{X}_I$ their intervened models $\C{M}_{\intervene(I,\B{\xi}_I)}$ and 
$\tilde{\C{M}}_{\intervene(I,\B{\xi}_I)}$ are observationally equivalent with respect to 
$\C{O}$. $\C{M}$ and $\tilde{\C{M}}$ are called \emph{interventionally equivalent} if they are interventionally 
equivalent w.r.t.\ $\C{I} = \tilde{\C{I}}$.
%with respect to 
\end{definition}
%That is, interventionally equivalent SCMs induce the same interventional distributions.
Equivalent SCMs have the same solutions under every perfect intervention, and hence they are interventionally equivalent w.r.t.\ every subset $\C{O}\subseteq\C{I}$. SCMs that are interventionally equivalent w.r.t.\ a subset $\C{O}\subseteq\C{I}$ are interventionally equivalent w.r.t.\ every strict subset $\C{W}\subsetneq\C{O}$. But in general, they are not interventionally equivalent w.r.t.\ a strict superset $\C{O}\subsetneq\C{V}\subseteq\C{I}$, as can be seen in Example~\ref{ex:ObservationalEquivalenceDoesNotImplyEquivalence}, where the SCMs $\C{M}$ and $\tilde{\C{M}}$ are interventionally equivalent w.r.t.\ $\{1,2\}$ but are not interventionally equivalent.
Interventional equivalence w.r.t.\ $\C{O}\subseteq\C{I}$ implies observational equivalence w.r.t.\ $\C{O}$, since the empty perfect intervention ($I = \emptyset$) is a special case of a perfect intervention. However, observational equivalence w.r.t.\ $\C{O}\subseteq\C{I}$ does not imply interventional equivalence w.r.t.\ $\C{O}$ in general, as can be seen in Example~\ref{ex:ObservationalEquivalenceDoesNotImplyEquivalence}, where the SCMs $\C{M}$ and $\tilde{\C{M}}$ are observationally equivalent but not interventionally equivalent. 
%(see Figure~\ref{fig:SimpleCyclicExample}).
%We saw already in Examples~\ref{ex:LinearGaussianANM2} that interventional equivalence is a strictly finer notion than observational equivalence. 
%Even for this finer notion of 

Although interventional equivalence is a finer notion than observational equivalence, we have that if two SCMs $\C{M}$ and $\tilde{\C{M}}$ are 
interventionally equivalent, then their associated augmented graphs $\C{G}^a(\C{M})$ and 
$\C{G}^a(\tilde{\C{M}})$ are not necessarily equal to each other. 
%as is shown in the  following example.
\begin{example}[Interventionally equivalent SCMs with different graphs]
\label{ex:InterventionalEquivalenceDoesNotImplySameAugmentedGraph}
Consider the SCM $\bar{\C{M}} = 
\langle \B{2},\B{2},\{-1,1\}^2,\{-1,1\}^2,\bar{\B{f}},\Prb_{\BC{E}}\rangle$ and the SCM $\hat{\C{M}}$ that is the same as $\bar{\C{M}}$ except for its causal mechanism $\hat{\B{f}}$, where the causal mechanisms are given by
$$
\bar{f}_1(\B{x},\B{e}) = e_1 \,, \quad \bar{f}_2(\B{x},\B{e}) = x_1e_2  \,,\quad\quad
  \hat{f}_1(\B{x},\B{e}) =e_1 \,, \quad \hat{f}_2(\B{x},\B{e}) =e_2 \,,
$$
and $\Prb_{\BC{E}} = \Prb^{\B{E}}$ with $E_1,E_2 \sim \C{U}(\{-1,1\})$ uniformly
distributed and $E_1 \indep E_2$. Then $\bar{\C{M}}$ and $\hat{\C{M}}$ are interventionally equivalent although $\C{G}(\bar{\C{M}})$ is 
not equal to $\C{G}(\hat{\C{M}})$ (see Figure~\ref{fig:CyclicAndSimpleExample}).
\end{example}

Example~\ref{ex:no_nice_reduction} showcases an SCM with two endogenous and three exogenous variables, for which there is no interventionally equivalent SCM (satisfying smoothness constraints) with one exogenous variable taking values in $\mathbb{R}^2$ whose first and second components enter in the first and second structural equation, respectively. In this sense, representing confounders with dependent exogenous variables can be nontrivial in nonlinear models.

%%%%%%%%%%%%%%%%%%%%%%%%%%%%%%%%%%%%%%%%%%%%%%%%%%
\subsection{Counterfactual equivalence}
%%%%%%%%%%%%%%%%%%%%%%%%%%%%%%%%%%%%%%%%%%%%%%%%%%

We consider two SCMs to be counterfactually equivalent if their twin SCMs induce the same
counterfactual distributions under every perfect intervention.
\begin{definition}[Counterfactual equivalence]
\label{def:CounterfactuallyEquivalent}
Two SCMs $\C{M} = \langle \C{I}, \C{J}, \BC{X}, \BC{E}, \B{f}, \Prb_{\BC{E}} \rangle$ and 
$\tilde{\C{M}} = \langle \tilde{\C{I}}, \tilde{\C{J}}, \tilde{\BC{X}}, \tilde{\BC{E}}, \tilde{\B{f}}, \Prb_{\tilde{\BC{E}}} \rangle$ are \emph{counterfactually equivalent with respect to 
$\C{O} \subseteq \C{I} \cap \tilde{\C{I}}$}, denoted by $\C{M}\equivcf{\C{O}}\tilde{\C{M}}$, 
if $\C{M}^{\twin}$ and $\tilde{\C{M}}^{\twin}$ are 
interventionally equivalent with respect to $\C{O}\cup\C{O}'$, where $\C{O}'$ corresponds to 
the copy of $\C{O}$ in $\C{I}'\cap \tilde{\C{I}}'$. $\C{M}$ and $\tilde{\C{M}}$ are called \emph{counterfactually equivalent} if 
they are counterfactually equivalent with respect to $\C{I} = \tilde{\C{I}}$.
\end{definition}
%That is, counterfactually equivalent SCMs induce the same counterfactual distributions.

The notion of counterfactual equivalence is coarser than equivalence and finer than interventional equivalence.

\begin{proposition}
\label{prop:CounterfactualEqRelatedToOtherEqs}
For SCMs, we have that equivalence implies counterfactual equivalence w.r.t.\ $\C{O}$, which in turn implies interventional equivalence w.r.t.\ $\C{O}$, for any $\C{O}\subseteq\C{I}$. 
\end{proposition}

%The notion of counterfactual equivalence is finer than that of equivalence.
%\begin{proposition}
%\label{prop:EquivalentSCMsAreCounterfactualEq}
%Equivalent SCMs are counterfactually equivalent w.r.t.\ every subset $\C{O}\subseteq\C{I}$.
%\end{proposition}
%
%The counterfactual equivalence relation is related to the interventional equivalence relation in the following way. 
%\begin{proposition}
%\label{prop:CounterfactuallyEquivalentSCMsAreInterventionallyEq}
%If two SCMs $\C{M}$ and $\tilde{\C{M}}$ are counterfactually equivalent w.r.t.\ 
%$\C{O} \subseteq \C{I}\cap\tilde{\C{I}}$, then $\C{M}$ and $\tilde{\C{M}}$ are 
%interventionally equivalent w.r.t.\ $\C{O}$.
%\end{proposition}
%However, the converse is not true in general, as the next simple example shows.

% Double example -> See Dawid's example in the appendix
%\begin{example}[Interventional equivalence does not imply counterfactual equivalence]
%\label{ex:InterventionNotImplyCounterfact}
%Consider the same SCMs as in 
%Example~\ref{ex:InterventionalEquivalenceDoesNotImplySameAugmentedGraph}. We have seen that they are interventionally equivalent. However, they are not 
%counterfactually equivalent, as $\C{M}^{\twin}_{\intervene(\{1',1\},(1,-1))}$ is not 
%observationally equivalent to $\bar{\C{M}}^{\twin}_{\intervene(\{1',1\},(1,-1))}$. To see this, 
%consider the counterfactual query $p(X_{2'}=1 \mid \intervene(X_{1'}=1, X_1=-1), X_2=1)$. 
%Both SCMs give a different answer and hence $\C{M}$ and $\bar{\C{M}}$ are not counterfactually equivalent.
%\end{example}

%Even 
Interventionally equivalent SCMs that have the same causal mechanism (that differ only 
in their exogenous distribution) may not be counterfactually equivalent (see, e.g., Example~\ref{ex:Counterfactuals}).
Although the notion of counterfactual equivalence is finer than the notion
of observational and interventional equivalence, the (augmented) graphs
for counterfactually equivalent SCMs are in general not equal to each
other (see Example~\ref{ex:CounterfactualEquivalenceDoesNotImplySameAugmentedGraph}).
\subsection{Relations between equivalences}
%%%%%%%%%%%%%%%%%%%%%%%%%%%%%%%%%%%%%%%%%%%%%%%%%%

The definitions of observational, interventional and counterfactual equivalence provide
equivalence relations on the set of all SCMs. 
For two SCMs to be observationally, 
interventionally or counterfactually equivalent w.r.t.\ $\C{O}\subseteq\C{I}\cap\tilde{\C{I}}$, the domains of their endogenous variables $\C{O}$ have to be equal, that is, $\BC{X}_{\C{O}}=\tilde{\BC{X}}_{\C{O}}$. Apart from that, the index sets of the endogenous and the exogenous 
variables, the spaces of the other endogenous and exogenous variables, the causal mechanism and the exogenous probability measure may all differ. The observational, interventional and 
counterfactual equivalence classes w.r.t.\ $\C{O}\subseteq\C{I}\cap\tilde{\C{I}}$ are related in the following way (see Proposition~\ref{prop:CounterfactualEqRelatedToOtherEqs}):
$$
\begin{aligned}
    \C{M}\text{ and }\tilde{\C{M}}\text{ are equivalent} 
    &\quad \implies & \C{M}\text{ and }\tilde{\C{M}}\text{ are counterfactually equivalent w.r.t.\ } \C{O} \\
    &\quad \implies & \C{M}\text{ and }\tilde{\C{M}}\text{ are interventionally equivalent w.r.t.\ } \C{O} \\
    &\quad \implies & \C{M}\text{ and }\tilde{\C{M}}\text{ are observationally equivalent w.r.t.\ } \C{O} \,. \\
\end{aligned}
$$
%This hierarchy formally establishes the ``ladder of causation'' \citep{SP08,PM18,Pea09} and allows us to compare SCMs at different levels of abstractions. 

This hierarchy allows us to compare SCMs at different levels of abstraction and formally establishes the ``ladder'' of causation (last two implications)~\citep{SP08,PM18,Pea09}.

%%%%%%%%%%%%%%%%%%%%%%%%%%%%%%%%%%%%%%%%%%%%%%%%%%
\section{Marginalizations}
\label{sec:Marginalization}
%%%%%%%%%%%%%%%%%%%%%%%%%%%%%%%%%%%%%%%%%%%%%%%%%%

In this section, we show how, and under which condition, one can marginalize an SCM over a subset
$\C{L}\subseteq\C{I}$ of endogenous variables (thereby ``hiding'' the variables $\C{L}$), to another SCM on the margin
$\C{I}\setminus\C{L}$ that is observationally, interventionally and even
counterfactually equivalent with respect to $\C{I}\setminus\C{L}$. 
In other words, we provide a formal notion of marginalization and show that this preserves the
probabilistic, causal and counterfactual semantics on the margin.  
%In
%particular, this can be seen as marginalizing all observational, interventional
%and counterfactual distributions associated to the SCM down to their corresponding distributions on the margin.

The problem of marginalization of directed graphical models has been addressed for acyclic graph structures, for example, ADMGs and mDAGs \citep[see][a.o.]{Ver93,Ric03,RS02,Eva16,Eva18}, and more recently in \citep{FM17} for certain graph structures (``HEDGes'') that may include cycles. 
%Although in the acyclic setting it has been shown that the marginalization for some of these graph structures preserves the probabilistic and causal semantics, in the cyclic setting this has only been shown for modular SCMs~\citep{FM17}, which can be seen as an SCM together with an additional structure of a compatible system of solution functions (see Appendix~\ref{app:ModularSCMs}). 
Although in the acyclic setting it has been shown that the marginalization for some of these graph structures preserves the probabilistic and causal semantics, in the cyclic setting this has only been shown for modular SCMs~\citep{FM17}. 
%The underlying SCM of a modular SCM always has an additional structure of a compatible system of solution functions (see Appendix~\ref{app:ModularSCMs}). 
We show that without the additional structure of a compatible system of solution functions (see Appendix~\ref{app:ModularSCMs}) one can still define a marginalization for SCMs under certain local unique solvability conditions.
%under the \emph{modularity} assumption on the underlying model~\citep{FM17}. We show that without this modularity assumption one can still define a marginalization for SCMs.
Intuitively, the idea is that if the state of a subsystem of endogenous variables is uniquely determined by the parents outside of this subsystem, then one can ignore the internals of this subsystem by treating it as a ``black box'' that can be described by certain measurable solution functions (see Figure~\ref{fig:NoLatentProjection}). One can marginalize over this subsystem by substituting these measurable solution functions into the rest of the model, thereby removing the functional dependencies on the variables of the subsystem from the rest of the system, while preserving the probabilistic, causal and the counterfactual semantics of the rest of the system. 
We show that in general this marginalization operation defined on SCMs does not respect the latent projection on its associated (augmented) graph, where the latent projection is a similar marginalization operation defined on directed mixed graphs~\citep{Ver93,Tia02,Eva16}. We show that under certain stronger local ancestral unique solvability conditions the marginalization does respect the latent projection.

\subsection{Marginalization of a structural causal model}
\label{sec:MarginalSCMs}
%%%%%%%%%%%%%%%%%%%%%%%%%%%%%%%%%%%%%%%%%%%%%%%%%%

Before we show how one can marginalize an SCM w.r.t.\ a subset of endogenous variables, we first point out that in general it is not always possible to find an SCM on the margin that
preserves the causal semantics,
as the following example illustrates.
\begin{example}[No SCM on the margin preserves the causal semantics]
Consider the SCM $\C{M} = \langle \B{3}, \emptyset, \RN^3, \B{1}, \B{f}, \Prb_{\B{1}} \rangle$ with causal
mechanism
$
  f_1(\B{x}) = x_1 + x_2 + x_3 \,, \quad f_2(\B{x}) = x_2 \,, \quad f_3(\B{x}) = 0 \,.
$ 
Then there exists no SCM $\tilde{\C{M}}$ on the endogenous variables $\{2,3\}$ that is interventionally
equivalent to $\C{M}$ w.r.t.\ $\{2,3\}$. To see this, suppose there exists such an SCM $\tilde{\C{M}}$, then for every $(\xi_2, \xi_3)\in\BC{X}_{\{2,3\}}$ such that $\xi_2 + \xi_3 \neq 0$ the
intervened model $\tilde{\C{M}}_{\intervene(\{2,3\},(\xi_2,\xi_3))}$ has a solution but
$\C{M}_{\intervene(\{2,3\},(\xi_2,\xi_3))}$ does not.
\end{example}
More generally, for an SCM $\C{M}$ that is not solvable w.r.t.\ a subset $\C{L}\subseteq\C{I}$ there is no SCM
$\tilde{\C{M}}$ on the endogenous variables $\C{I}\setminus\C{L}$ that is interventionally equivalent w.r.t.\
$\C{I}\setminus\C{L}$.

The following example illustrates that for an SCM that is uniquely solvable w.r.t.\ a subset there exists an SCM on the margin that preserves the causal semantics.
\begin{example}[SCM on the margin that preserves the causal semantics]
\label{ex:Marginalization}
Consider the SCM $\C{M}$ of Example~\ref{ex:UniqueSolvabilityNotUniqAncestralSubset} that is uniquely solvable w.r.t.\ the subset $\C{L}=\{2,3\}$ (depicted by the gray box in Figure~\ref{fig:NoLatentProjection}). Substituting the measurable solution functions $\B{g}_{\C{L}}$ into the causal mechanism components $f_1$ and $f_4$ for the remaining endogenous variables $\{1,4\}$ gives a ``marginal'' causal mechanism $\tilde{f}_1(\B{x},e):=e$ and $\tilde{f}_4(\B{x},e):=x_1$. This defines an SCM $\tilde{\C{M}}$ on the margin $\C{I}\setminus\C{L}=\{1,4\}$ that is interventionally equivalent w.r.t.\ $\C{I}\setminus\C{L}$ to $\C{M}$.
\end{example}

In general, for an SCM $\C{M}$ and a given subset $\C{L}\subseteq\C{I}$ of endogenous variables and
its complement $\C{O}=\C{I}\setminus\C{L}$, we can consider the ``subsystem'' of structural equations
$\B{x}_{\C{L}}=\B{f}_{\C{L}}(\B{x}_{\C{L}},\B{x}_{\C{O}},\B{e})$. If $\C{M}$ is uniquely
solvable w.r.t.\ $\C{L}$ with measurable solution function
$\B{g}_{\C{L}}:\BC{X}_{\pa(\C{L})\setminus\C{L}}\times\BC{E}_{\pa(\C{L})}\to\BC{X}_{\C{L}}$, then for each input 
$(\B{x}_{\pa(\C{L})\setminus\C{L}},\B{e}_{\pa(\C{L})}) \in 
\BC{X}_{\pa(\C{L})\setminus\C{L}}\times\BC{E}_{\pa(\C{L})}$ of the subsystem, there exists an
output $\B{x}_{\C{L}}\in\BC{X}_{\C{L}}$, which is unique for $\Prb_{\BC{E}_{\pa(\C{L})}}$-almost
every $\B{e}_{\pa(\C{L})}\in\BC{E}_{\pa(\C{L})}$ and for all $\B{x}_{\pa(\C{L})\setminus\C{L}}\in\BC{X}_{\pa(\C{L})\setminus\C{L}}$.
%$\Prb_{\BC{E}_{\pa(\C{L})}}$-unique\footnote{A mapping is $\Prb$-unique if it is unique up to a $\Prb$-null set.} 
We can remove this subsystem
of endogenous variables from the model by substitution. This leads to a marginal SCM that is
observationally, interventionally and counterfactually equivalent to the original SCM w.r.t.\ the margin, as we prove in Theorem~\ref{thm:MarginalizationEquivalences}.

\begin{definition}[Marginalization of an SCM]
\label{def:MarginalSCM}
Let $\C{M} = \langle \C{I}, \C{J}, \BC{X}, \BC{E}, \B{f}, \Prb_{\BC{E}} \rangle$ be an SCM that is uniquely solvable w.r.t.\ a subset $\C{L}\subseteq\C{I}$ and let $\C{O}=\C{I}\setminus\C{L}$. For
$\B{g}_{\C{L}}:\BC{X}_{\pa(\C{L})\setminus\C{L}}\times\BC{E}_{\pa(\C{L})}\to\BC{L}$, any measurable
solution function of $\C{M}$ w.r.t.\ $\C{L}$, we call
the SCM $\C{M}_{\marg(\C{L})} := \langle \C{O}, \C{J}, 
\BC{X}_{\C{O}}, \BC{E}, \tilde{\B{f}}, \Prb_{\BC{E}} \rangle$ with the \emph{marginal causal 
mechanism} $\tilde{\B{f}} : \BC{X}_{\C{O}} \times \BC{E} \to \BC{X}_{\C{O}}$ given by 
$$ 
  \tilde{\B{f}}(\B{x}_{\C{O}},\B{e}) = 
  \B{f}_{\C{O}}(\B{g}_{\C{L}}(\B{x}_{\pa(\C{L})\setminus\C{L}},
  \B{e}_{\pa(\C{L})}),\B{x}_{\C{O}},\B{e}) \,,
$$
a \emph{marginalization of $\C{M}$ w.r.t.\ $\C{L}$}. 
We denote by $\marg(\C{L})(\C{M})$ the equivalence class of the marginalizations of $\C{M}$ w.r.t.\ $\C{L}$.
\end{definition}
The marginalization of $\C{M}$ w.r.t.\ $\C{L}$ is defined up to the equivalence $\equiv$ on SCMs, since the measurable solution functions $\B{g}_{\C{L}}$ are uniquely defined up to $\Prb_{\BC{E}}$-null sets.
%Note that for a specific $\C{L}\subsetneq\C{I}$ there may exist more than one marginalization 
%w.r.t.\ $\C{L}$, depending on the choice of the measurable solution function $\B{g}_{\C{L}}$. 
%However, all marginalizations w.r.t.\ $\C{L}$ map $\C{M}$ to a 
%representative of the same equivalence class $\marg(\C{L})(\C{M})$ of SCMs. Moreover, marginalizing two equivalent 
%SCMs w.r.t.\ $\C{L}$ yield two equivalent marginal SCMs. Thus, the mapping 
%$\marg(\C{L})$ induces a well-defined mapping between the equivalence classes of SCMs.
With this definition at hand, we can always construct a marginal SCM over a subset of the endogenous variables of an acyclic SCM by mere substitution (see also Proposition~\ref{prop:AcyclicSCMUniquelySolvable}). Moreover, this definition extends that notion to SCMs that are uniquely solvable w.r.t.\ a certain subset. For linear
SCMs this condition translates into a matrix invertibility
condition, and since substitution preserves linearity, marginalization yields a 
%which can be substituted into the causal mechanism to yield another 
linear marginal SCM (see Proposition~\ref{prop:MarginalLinearModel}). 

In general, marginalization is not always defined for all subsets. For instance, the SCM of
Example~\ref{ex:UniqueSolvabilityNotUniqAncestralSubset} cannot be marginalized over
the variable $3$ (due to the self-cycle at $3$), but can be marginalized over the variables $2$ and $3$ together. It follows
from Proposition~\ref{prop:ObstructionToSelfCycles} that we can only marginalize over a single variable if
that variable has no self-cycle. Note that we may introduce new self-cycles if we marginalize over a
subset of variables, as can be seen, for example, from the SCM $\C{M}$ in
Example~\ref{ex:AugmentedGraphs}. This SCM has only one self-cycle; however,
marginalizing w.r.t.\  
$\{2\}$ gives a marginal SCM with another self-cycle at variable $4$.

The definition of marginalization satisfies an intuitive property: if we can marginalize over two
disjoint subsets after each other, then we can also marginalize over the union of those
subsets at once, and the respective results agree.
\begin{proposition}
\label{prop:MarginalizationCommutes}
%Given an SCM $\C{M} = \langle \C{I}, \C{J}, \BC{X}, \BC{E}, \B{f}, \Prb_{\BC{E}}
%\rangle$ and two disjoint subsets $\C{L}_1,\C{L}_2\subsetneq\C{I}$. If $\C{M}$ is
%uniquely solvable w.r.t.\ $\C{L}_1$ and $\C{M}_{\marg(\C{L}_1)}$ is uniquely solvable w.r.t.\
%$\C{L}_2$, then $\C{M}$ is uniquely solvable w.r.t.\ $\C{L}_1\cup\C{L}_2$, and
%$\marg(\C{L}_2)\circ\marg(\C{L}_1)(\C{M})\equiv\marg(\C{L}_1\cup\C{L}_2)(\C{M})$.
\Joris{Added ``if'' which replaces the old Proposition 5.2.8 which assumed
ancestral solvability w.r.t.\ $\C{L}_1\cup\C{L}_2$ and claimed ancestral solvability
of $\C{M}_{\marg(\C{L}_1)}$ w.r.t.\ $\C{L}_2$.}
Let $\C{M} = \langle \C{I}, \C{J}, \BC{X}, \BC{E}, \B{f}, \Prb_{\BC{E}}\rangle$ be an SCM
that is uniquely solvable w.r.t.\ a subset $\C{L}_1 \subseteq\C{I}$ and let $\C{L}_2\subseteq\C{I}$ be a subset disjoint from $\C{L}_1$. Then $\C{M}_{\marg(\C{L}_1)}$ is uniquely solvable w.r.t.\ $\C{L}_2$ if and only if
$\C{M}$ is uniquely solvable w.r.t.\ $\C{L}_1 \cup \C{L}_2$,
Moreover, $\marg(\C{L}_2)\circ\marg(\C{L}_1)(\C{M})=\marg(\C{L}_1\cup\C{L}_2)(\C{M})$.
%If in addition $\C{M}$ is uniquely solvable w.r.t.\ $\C{L}_2$ and
%$\C{M}_{\marg(\C{L}_2)}$ is uniquely solvable w.r.t.\ $\C{L}_1$, then $(\marg(\C{L}_2)\circ\marg(\C{L}_1))(\C{M}) \equiv 
%(\marg(\C{L}_1)\circ\marg(\C{L}_2)(\C{M})) \equiv \marg(\C{L}_1 \cup
%\C{L}_2)(\C{M})$. \Joris{I would omit the last sentence---it is more obvious than it
%takes effort to parse it and trying to see whether it states something non-obvious
%(``did I confuse a 1 for a 2 somewhere while reading, or does it really just restate
%  the same statement?'')}
\end{proposition}
In this proposition, $\C{L}_1$ and $\C{L}_2$ have to be disjoint, since marginalizing 
first over $\C{L}_1$ gives a marginal SCM $\C{M}_{\marg(\C{L}_1)}$ with endogenous 
variables $\C{I}\setminus\C{L}_1$.

Next, we show that the distributions of a marginal SCM are identical to the
marginal distributions induced by the original SCM. A simple proof of this result proceeds by showing that 
both the intervention and the twin operation commute with marginalization.
%the operations of intervention and marginalization commute.
\begin{proposition}
\label{prop:MarginalizationCommuteWithInterventionAndTwin}
Let $\C{M}$ be an SCM that is uniquely solvable w.r.t.\ a subset $\C{L}\subseteq\C{I}$. Then the marginalization $\marg(\C{L})$ commutes with both:
\begin{enumerate}[ref=\theproposition.(\arabic*)]
\item \label{prop:MarginalizationCommuteWithIntervention} the perfect intervention $\intervene(I,\B{\xi}_I)$ for a subset $I\subseteq\C{I}\setminus\C{L}$ and a value $\B{\xi}_I\in\BC{X}_I$, that is, $(\marg(\C{L})\circ\intervene(I,\B{\xi}_I))(\C{M}) = (\intervene(I,\B{\xi})\circ\marg(\C{L}))(\C{M})$, and
\item \label{prop:MarginalizationCommuteWithTwinning} the twin operation $\twin$, that is, $(\marg(\C{L}\cup\C{L}')\circ\twin)(\C{M}) = (\twin\circ\marg(\C{L}))(\C{M})$,
\end{enumerate}
where $\C{L}'$ is the copy of $\C{L}$ in $\C{I}'$.
\end{proposition}

%\begin{proposition}
%\label{prop:MarginalizationCommuteWithIntervention}
%Given an SCM $\C{M}$, a subset $\C{L}\subsetneq\C{I}$ such that $\C{M}$ is uniquely solvable 
%w.r.t.\ $\C{L}$, a subset $I\subseteq\C{I}\setminus\C{L}$ and a value $\B{\xi}_I\in\BC{X}_I$. 
%Marginalization $\marg(\C{L})$ commutes with perfect intervention $\intervene(I,\B{\xi}_I)$, 
%i.e., $(\marg(\C{L})\circ\intervene(I,\B{\xi}_I))(\C{M}) \equiv 
%(\intervene(I,\B{\xi})\circ\marg(\C{L}))(\C{M})$.
%\end{proposition}
%
%Similarly, we show that the twin operation commutes with the marginalization
%operation.
%\begin{proposition}
%\label{prop:MarginalizationCommuteWithTwinning}
%Let $\C{M}$ be an SCM and $\C{L}\subsetneq\C{I}$ a subset such that $\C{M}$ is
%uniquely solvable w.r.t.\ $\C{L}$. Marginalization $\marg(\C{L})$ commutes with
%the twin operation, i.e., $(\marg(\C{L}\cup\C{L}')\circ\twin)(\C{M}) \equiv
%(\twin\circ\marg(\C{L}))(\C{M})$, where $\C{L}'$ is the copy of $\C{L}$ in
%$\C{I}'$.
%\end{proposition}

With Proposition~\ref{prop:MarginalizationCommuteWithInterventionAndTwin} at hand, we can prove the main result of this subsection.
\begin{theorem}[Marginalization of an SCM preserves the observational, causal and counterfactual semantics]
\label{thm:MarginalizationEquivalences}
Let $\C{M}$ be an SCM that is uniquely solvable w.r.t.\ a subset $\C{L}\subseteq\C{I}$. Then $\C{M}$ and $\marg(\C{L})(\C{M})$ are observationally,
interventionally and counterfactually equivalent w.r.t.\ 
$\C{I}\setminus\C{L}$.
\end{theorem}

This shows that our definition of marginalization (Definition~\ref{def:MarginalSCM}) preserves the probabilistic, causal and counterfactual semantics, under a certain local unique solvability condition. Moreover, this allows us to marginalize SCMs w.r.t.\ a certain subset that do not satisfy the additional assumptions imposed by modular SCMs, for example, the SCM $\C{M}$ of Example~\ref{ex:UniqueSolvabilityNotUniqAncestralSubset} does not have any additional structure of a compatible system of solution functions, but $\C{M}$ can be marginalized w.r.t.\ the subset $\{2,3\}$ (see Appendix~\ref{app:ModularSCMs}).
%without \Stephan{the additional assumptions imposed by modular SCMs (see Appendix~\ref{app:ModularSCMs}).}

In general, interventional equivalence does not imply counterfactual equivalence (see, e.g., Example~\ref{ex:Counterfactuals}).
%As we saw in Example~\ref{ex:InterventionNotImplyCounterfact} it is
%generally not true that interventional equivalence implies counterfactual
%equivalence. 
However, for our definition of marginalization we arrive at a
marginal SCM that is not only interventionally equivalent, but also counterfactually
equivalent w.r.t.\ the margin. 

For an SCM $\C{M}$, unique solvability w.r.t.\ a certain subset
$\C{L}\subseteq\C{I}$ is a sufficient, but not a necessary condition for the existence of an SCM $\tilde{\C{M}}$ on the margin $\C{I}\setminus\C{L}$ such that $\C{M}$ and $\tilde{\C{M}}$ are counterfactually equivalent w.r.t.\
$\C{I}\setminus\C{L}$ (see, e.g., Example~\ref{ex:MarginalizationSufficientCondition}). Hence, in certain cases it may be possible to relax the uniqueness condition.

%%%%%%%%%%%%%%%%%%%%%%%%%%%%%%%%%%%%%%%%%%%%%%%%%%
\subsection{Marginalization of a graph}
%%%%%%%%%%%%%%%%%%%%%%%%%%%%%%%%%%%%%%%%%%%%%%%%%%

We now turn to a marginalization operation for directed mixed graphs, 
which we call the latent projection. This name is inspired from a similar construction on
directed mixed graphs in \citep{Ver93}. In \citep{Ver93}, the authors concentrate on a mapping between directed mixed graphs and show that it preserves conditional independence properties \citep[see also][]{Tia02}. In this subsection, we provide a sufficient
condition for the marginalization of an SCM to respect the latent projection, that is, that the
augmented graph of the marginal SCM is a subgraph of the latent projection of the
augmented graph of the original SCM.

\begin{definition}[Marginalization/latent projection of a directed mixed graph]
\label{def:LatentProjection}
Let $\C{G} = (\C{V},\C{E},\C{B})$ be a directed mixed graph and $\C{L}
\subseteq \C{V}$ a subset. The \emph{marginalization of $\C{G}$ w.r.t.\
$\C{L}$} or the \emph{latent projection of $\C{G}$ onto $\C{V}\setminus\C{L}$} maps $\C{G}$ to the \emph{marginal graph} $\marg(\C{L})(\C{G}) := (\tilde{\C{V}}, \tilde{\C{E}}, \tilde{\C{B}})$, where:
\begin{enumerate}
  \item $\tilde{\C{V}} = \C{V} \setminus \C{L}$,
  \item for $i,j \in \tilde{\C{V}}$: $i\to j \in \tilde{\C{E}}$ if and only if there exists a directed path $i \to \ell_1 \to \dots \to \ell_n \to j$ in $\C{G}$ with $n \ge 0$ and $\ell_1,\dots,\ell_n \in \C{L}$,
  \item for $i\ne j \in \tilde{\C{V}}$: $i\oto j \in \tilde{\C{B}}$ if and only if \Joris{the conditions were incomplete!}
    \begin{enumerate}
%      \item $i \oto j \in \C{B}$, or
%      \item there exists an $n > k \geq 1$, $\ell_1,\dots,\ell_n \in \C{L}$ and two directed paths $i \ot \ell_1 \ot \cdots \ot \ell_k$ and $\ell_{k+1}\to \cdots \to \ell_n \to j$ in $\C{G}$ such that $\ell_k=\ell_{k+1}$ or $\ell_k\oto\ell_{k+1}\in\C{B}$. 
      \item there exist $n,m \ge 0$, $\ell_1,\dots,\ell_n \in \C{L}$, $\tilde{\ell}_1,\dots,\tilde{\ell}_m \in \C{L}$ such that
        $i \ot l_1 \ot l_2 \ot \cdots \ot \ell_n \oto \tilde{\ell}_m \to \tilde{\ell}_{m-1} \to \dots \to \tilde{\ell}_1 \to j$ in $\C{G}$, or
      \item there exist $n,m \ge 1$, $\ell_1,\dots,\ell_n \in \C{L}$, $\tilde{\ell}_1,\dots,\tilde{\ell}_m \in \C{L}$ such that 
        $i \ot l_1 \ot l_2 \ot \cdots \ot \ell_n$ and $\tilde{\ell}_m \to \tilde{\ell}_{m-1} \to \dots \to \tilde{\ell}_1 \to j$ in $\C{G}$ and $\ell_n = \tilde{\ell}_m$.
    \end{enumerate}
    \Joris{Note here that the graphical marginalization defined on a mixed graph also may add too many bidirected edges, assuming implicitly that each endogenous variable has at least one exogenous parent. This is similar to the graphical twin operation for DMGs, so we could actually define that one too.}
\end{enumerate}
\end{definition}
Note that this gives $\C{G}(\C{M}) = \marg(\C{J})(\C{G}^a(\C{M}))$ for any SCM $\C{M}$.
Further, for a subgraph $\C{H} \subseteq \C{G}$ we have $\marg(\C{L})(\C{H}) \subseteq \marg(\C{L})(\C{G})$
for any subset of nodes $\C{L}$. It does not matter in which order we project out the nodes or if we 
%do
perform several projections at once.
\begin{proposition}
\label{prop:LatentProjectionCommutes}
Let $\C{G}=(\C{V},\C{E},\C{B})$ be a directed mixed graph and $\C{L}_1,\C{L}_2\subseteq\C{V}$ two disjoint subsets. Then $(\marg(\C{L}_1)\circ\marg(\C{L}_2))(\C{G})=(\marg(\C{L}_2)\circ\marg(\C{L}_1))(\C{G})=\marg(\C{L}_1\cup\C{L}_2)(\C{G})$.
\end{proposition}
Similar to the definition of marginalization for SCMs, this definition of the
latent projection commutes with both the (graphical) perfect intervention and the twin operation.

\begin{proposition}
\label{prop:LatentProjectionCommuteWithInterventionAndTwin}
Let $\C{G}=(\C{V},\C{E},\C{B})$ be a directed mixed graph and $\C{L},\C{I},I\subseteq\C{V}$ subsets. Then the marginalization $\marg(\C{L})$ commutes with both:
\begin{enumerate}[ref=\theproposition.(\arabic*)]
  \item \label{prop:LatentProjectionCommuteWithIntervention} perfect intervention $\intervene(I)$ if $I$ is disjoint from $\C{L}$, that is, $(\marg(\C{L})\circ\intervene(I))(\C{G}) = (\intervene(I)\circ\marg(\C{L}))(\C{G})$, and
  \item \label{prop:LatentProjectionCommuteWithTwin} the twin operation $\twin(\C{I})$ if $\C{B}=\emptyset$, $\C{J}:=\C{V}\setminus\C{I}$ is exogenous (i.e., $\pa_{\C{G}}(\C{J})=\emptyset$) and $\C{L}\subseteq\C{I}$, that is, $(\marg(\C{L}\cup\C{L}')\circ\twin(\C{I}))(\C{G})=(\twin(\C{I}\setminus\C{L})\circ\marg(\C{L}))(\C{G})$,
\end{enumerate}
where $\C{L}'$ is the copy of $\C{L}$ in $\C{I}'$.
\end{proposition}

%\begin{proposition}
%\label{prop:LatentProjectionCommuteWithIntervention}
%Given a directed mixed graph $\C{G}=(\C{V},\C{E},\C{B})$ and subsets $\C{L},I\subseteq\C{V}$. If
%$\C{L}$ and $I$ are disjoint, then $(\marg(\C{L})\circ\intervene(I))(\C{G}) =
%(\intervene(I)\circ\marg(\C{L}))(\C{G})$.
%\end{proposition}
%\begin{proposition}
%\label{prop:LatentProjectionCommuteWithTwinning}
%Given a directed graph $\C{G}=(\C{V},\C{E})$ and subsets $\C{L},I\subseteq\C{V}$ such
%that $\pa_{\C{G}}(\C{V}\setminus I) = \emptyset$. If
%$\C{L}\subseteq I$, then 
%$(\marg(\C{L}\cup\C{L}')\circ\twin(I))(\C{G})=(\twin(I\setminus\C{L})\circ\marg(\C{L}))(\C{G})$,
%where $\C{L}'$ is the copy of $\C{L}$ in $I'$.
%\end{proposition}

An example of an SCM for which a marginalization respects the latent projection is the SCM $\C{M}$ of Example~\ref{ex:AugmentedGraphs}. Marginalizing $\C{M}$ w.r.t.\ $\C{L}=\{2\}$ gives a marginal SCM $\C{M}_{\marg(\C{L})}$ with a graph that is a subgraph of the latent projection of the graph of the SCM $\C{M}$ onto $\C{I}\setminus\C{L}$.
%Marginalizing, for example, the SCM $\C{M}$ of Example~\ref{ex:Marginalization} w.r.t.\ $\{2\}$ gives a graph of the marginal SCM $\C{M}_{\marg(\{2\})}$ that is a subgraph of the latent projection of the graph onto the other variables.
%for which the graph 
%In Example~\ref{ex:Marginalization}
%we already saw an example of a marginalization that respects the latent projection. 
In general, not all marginalizations respect the latent projection, as is illustrated in the following
example.
\begin{example}[Marginalization does not respect the latent projection]
\label{ex:NoLatentProjection}
Consider the SCM $\C{M}$ of Example~\ref{ex:UniqueSolvabilityNotUniqAncestralSubset}. Although $\C{M}$ and its marginalization $\C{M}_{\marg(\C{L})}$ with $\C{L}=\{2,3\}$ are interventionally equivalent w.r.t.\ $\C{I}\setminus\C{L}=\{1,4\}$, the graph $\C{G}(\C{M}_{\marg(\C{L})})$ is not a subgraph of the latent projection of $\C{G}(\C{M})$ onto $\C{I}\setminus\C{L}$, as
can be verified from the graphs depicted in Figure~\ref{fig:NoLatentProjection}.
\end{example}

%\begin{example}[Marginalization does not respect the latent projection]
%\label{ex:NoLatentProjection}
%Consider the SCM $\C{M}=\langle \B{4}, \B{1}, \RN^4, \RN, \B{f}, \Prb_{\RN} \rangle$ with causal
%mechanism given by
%$$
%  \begin{aligned}
%    f_1(\B{x},e) &= x_1 (1 - \B{1}_{\{0\}}(x_2 - x_3)) + 1 \,, &f_2(\B{x},e) &= x_2 \,, \\
%    f_3(\B{x},e) &= e \,, &f_4(\B{x},e) &= x_2 
%  \end{aligned}
%$$
%and $\Prb_{\RN}$ the standard-normal measure on $\RN$. Although $\C{M}$ and its marginalization $\C{M}_{\marg(\C{L})}$ with $\C{L}=\{1,2\}$ are
%interventionally equivalent w.r.t.\ $\C{I}\setminus\C{L}=\{3,4\}$, the augmented graph
%$\C{G}^a(\C{M}_{\marg(\C{L})})$ is not a subgraph of the latent projection of
%$\C{G}^a(\C{M})$ onto $\C{I}\setminus\C{L}$, as
%can be verified from the augmented graphs depicted in Figure~\ref{fig:NoLatentProjection}.
%\end{example}

Under the local ancestral unique solvability condition, which is a stronger condition than the local unique solvability condition (i.e., ancestral unique solvability w.r.t.\ a subset implies unique solvability w.r.t.\ that subset), one can prove that the marginalization of an SCM respects the latent projection.
%Under the stronger local ancestral unique solvability condition 
%one can prove that
%the marginalization of an SCM respects the latent projection. 
\begin{proposition}
\label{prop:LatentProjection}
Let $\C{M}$ be an SCM that is ancestrally uniquely solvable w.r.t.\ a subset $\C{L}\subseteq\C{I}$. Then $\big(\C{G}^a\circ\marg(\C{L})\big)(\C{M}) \subseteq 
\big(\marg(\C{L})\circ\C{G}^a\big)(\C{M})$ and $\big(\C{G}\circ\marg(\C{L})\big)(\C{M}) \subseteq 
\big(\marg(\C{L})\circ\C{G}\big)(\C{M})$.
\end{proposition}
The (augmented) graph of a marginalized SCM can be a strict subgraph of the corresponding latent projection if, for example, certain paths cancel each other out after the substitution of the measurable solution function(s) into the causal mechanism(s) on the margin (see Example~\ref{ex:LatentProjection}). For acyclic SCMs, we recover with Proposition~\ref{prop:LatentProjection} the known result that this class is closed under marginalization (see Proposition~\ref{prop:AcyclicSCMUniquelySolvable})~\citep{Eva16}. For linear SCMs, we have that unique solvability w.r.t.\ a subset $\C{L}$ holds if and only if ancestral unique solvability w.r.t.\ $\C{L}$ holds (see Proposition~\ref{prop:LinearEquivalentUniqueSolvability}), and hence, a marginalization of a linear SCM always respects the latent projection.

\Joris{\textbf{The following turned out to have several problems in its proof. We couldn't fix
it, but still believe that it should be true. However, the proposal is to remove it for now.}

Under the stronger conditions of Proposition~\ref{prop:LatentProjection} we can also prove:
\begin{proposition}
\label{prop:MarinalizingTwoSubsetsEquiv2conjecture}
Given an SCM $\C{M}$ and two disjoint subsets $\C{L}_1,\C{L}_2\subseteq\C{I}$
such that $\C{M}$ is ancestrally uniquely solvable w.r.t.\ $\C{L}_1$ and w.r.t.\ $\C{L}_1\cup\C{L}_2$, then $\C{M}_{\marg(\C{L}_1)}$ is ancestrally uniquely solvable w.r.t.\ $\C{L}_2$. Moreover $\marg(\C{L}_2)\circ\marg(\C{L}_1)(\C{M})\equiv\marg(\C{L}_1\cup\C{L}_2)(\C{M})$.
\end{proposition}
Interestingly, the following (weaker) older version also had a bug in its proof:
\begin{proposition}
\label{prop:MarginalizingTwoSubsetsEquiv2older}
Given an SCM $\C{M}$ and two disjoint subsets $\C{L}_1,\C{L}_2\subseteq\C{I}$ such that $\C{M}$ is
uniquely solvable w.r.t.\ $\C{L}_1$ and 
ancestrally uniquely solvable w.r.t.\ $\C{L}_1\cup\C{L}_2$,
then $\C{M}_{\marg(\C{L}_1)}$ is uniquely solvable w.r.t.\ $\C{L}_2$. Moreover $\marg(\C{L}_2)\circ\marg(\C{L}_1)(\C{M})\equiv\marg(\C{L}_1\cup\C{L}_2)(\C{M})$.
\end{proposition}
In fixing the bug I strengthened to the following 
even older version (which is easiest to prove!), which now became part of
Proposition~\ref{prop:MarginalizationCommutes}:
\begin{proposition}
\label{prop:MarginalizingTwoSubsetsEquiv}
Let $\C{M} = \langle \C{I}, \C{J}, \BC{X}, \BC{E}, \B{f}, \Prb_{\BC{E}}\rangle$ be an SCM
that is uniquely solvable w.r.t.\ $\C{L}_1 \subseteq\C{I}$. Let $\C{L}_2\subseteq\C{I}$ be
disjoint from $\C{L}_1$.
Then $\C{M}_{\marg(\C{L}_1)}$ is uniquely solvable w.r.t.\ $\C{L}_2$ if and only if
$\C{M}$ is uniquely solvable w.r.t.\ $\C{L}_1 \cup \C{L}_2$,
Moreover $\marg(\C{L}_2)\circ\marg(\C{L}_1)(\C{M})\equiv\marg(\C{L}_1\cup\C{L}_2)(\C{M})$.
\end{proposition}
}

%%%%%%%%%%%%%%%%%%%%%%%%%%%%%%%%%%%%%%%%%%%%%%%%%%
\section{Markov properties}
\label{sec:MarkovProperty}
%%%%%%%%%%%%%%%%%%%%%%%%%%%%%%%%%%%%%%%%%%%%%%%%%%

In this section, we give a short overview of Markov properties for
SCMs with cycles. 
%These Markov properties were recently derived in \citep{FM17} for HEDGes, a graphical representation that is similar to the augmented graph of SCMs. Here, we present some of the highlights of the work of \citep{FM17} and reformulate these results in terms of SCMs. 
We make use of the Markov properties that were recently developed by Forr{\'e} and Mooij~\cite{FM17} for HEDGes, a graphical representation that is similar to the augmented graph of SCMs. We briefly summarize some of their main results and apply them to the class of SCMs. In Appendix~\ref{app:MarkovProperty}, we provide a more thorough introduction and give an intuitive derivation, which can act as an entry point for the reader into the more extensive discussion of Markov properties provided in~\cite{FM17}.

Markov properties associate a set of conditional independence relations to a
graph. The directed global Markov property for directed acyclic graphs (see Definitions~\ref{def:DSeparation} and \ref{def:Markov_property}), also
known as the $d$-separation criterion~\citep{Pea85}, is one of the most widely used.
It directly extends to a similar property for acyclic directed mixed graphs (ADMGs) \citep{Ric03}. It does not hold in general for cyclic SCMs, however, as
was already observed earlier \citep{Spi94, Spi95}.
\begin{example}[Directed global Markov property does not hold for cyclic SCM]
\label{ex:SpirtesExample}
One can check that for every solution $\B{X}$ of the SCM $\C{M}$ of Example~\ref{ex:CyclicExample}, $X_1$ is not independent of $X_2$ given $\{X_3,X_4\}$. However, the
variables $X_1$ and $X_2$ are $d$-separated given $\{X_3,X_4\}$ in $\C{G}(\C{M})$ (see Figure~\ref{fig:CyclicAndSimpleExample}). Hence the
global directed Markov property does not hold here. 
\end{example}
Although some progress has been made in the case of discrete \citep{PD96,Nea00,FM17}
and linear models \citep{Spi93, Spi94, Spi95, Ric96c, Kos96, HEH12,FM17}, only recently a general directed global Markov property has been introduced for more general cyclic models~\citep{FM17}, that is based on $\sigma$-separation (see Definition~\ref{def:SigmaSeparation} and \ref{def:generalized_Markov_property}), an extension of $d$-separation. This notion of $\sigma$-separation was derived from the notion of $d$-separation in the acyclification of the graph~\citep{FM17} (see Definition~\ref{def:GraphicalAcyclification}). The acyclification of a graph generalizes the idea of the collapsed graph developed by Spirtes~\citep{Spi94} and can, in particular, be applied to the graphs of SCMs.
%was inspired by the work of Spirtes~\citep{Spi94}.
The main idea of the acyclification is that under the condition that the SCM is uniquely solvable w.r.t.\ each strongly connected component, we can replace the causal mechanisms of these strongly connected components by their measurable solution functions, which results in an acyclic SCM. This acyclified SCM (see Definition~\ref{def:Acyclification}) is observationally equivalent to the original SCM (see Proposition~\ref{prop:Acyclification}).

\begin{example}[Construction of an observationally equivalent acyclic SCM]
\label{ex:AcyclificationExample}
The SCM $\C{M}$ of Example~\ref{ex:CyclicExample} is uniquely solvable w.r.t.\ all its strongly connected components, that is, the subsets $\{1\}$, $\{2\}$ and $\{3,4\}$. Replacing the causal mechanisms of these strongly connected components by their measurable solution functions gives the observationally equivalent SCM $\tilde{\C{M}}$ of Example~\ref{ex:ObservationalEquivalenceDoesNotImplyEquivalence}. Because $\tilde{\C{M}}$ is acyclic (see
Figure~\ref{fig:CyclicAndSimpleExample}) we can apply the directed global Markov property 
to $\tilde{\C{M}}$. The fact that $X_1$ and $X_2$ are not $d$-separated given $\{X_3,X_4\}$ in 
$\C{G}(\tilde{\C{M}})$ is in line with $X_1$ being dependent of $X_2$ given $\{X_3,X_4\}$
for every solution $\B{X}$ of $\tilde{\C{M}}$ (and hence of $\C{M}$).
\end{example}
This acyclification preserves solutions, and $d$-separation in the
acyclification can directly be translated into $\sigma$-separation on
the original graph (see Proposition~\ref{prop:SigmaSeparationAsDSeparation}). This leads to the general directed global Markov property. The following theorem summarizes the main results of \citep{FM17} applied to SCMs.
\begin{theorem}[Global Markov properties for SCMs~\citep{FM17}]
\label{thm:globalMarkovPropertiesSCMs}
Let $\C{M}$ be a uniquely solvable SCM. Then its observational distribution $\Prb^{\B{X}}$ exists, is unique and the following two statements hold:
\begin{enumerate}[ref=\theproposition.(\arabic*)]
\item \label{thm:DirectedGlobalMarkovPropertSCMs}
  $\Prb^{\B{X}}$ satisfies the \emph{directed global Markov property} (``$d$-separation criterion'') relative to $\C{G}(\C{M})$ (see Definition~\ref{def:Markov_property}) if $\C{M}$ satisfies at least one of the following conditions:
  \begin{enumerate}[ref=\theproposition.(1.\alph*)]
      \item %$\C{M}$ has a graph $\C{G}(\C{M})$ with a bidirected edge $i\leftrightarrow j$ for all $i \ne j\in\C{I}$ such that $j\in\scc_{\C{G(\C{M})}}(i)$. 
        $\C{M}$ is acyclic;
      \item \label{thm:DirectedGlobalMarkovPropertSCMsb} all endogenous spaces $\C{X}_i$ are discrete and $\C{M}$ is ancestrally
        uniquely solvable;
      \item $\C{M}$ is linear (see Definition~\ref{def:LinearSCM}), 
        each of its causal mechanisms $\{f_i\}_{i\in\C{I}}$ has a non-trivial dependence on at least one exogenous variable, and
    %    $\Prb^{\B{X}}$ has a density w.r.t.\ the Lebesgue measure,
        $\Prb_{\BC{E}}$ has a density w.r.t.\ the Lebesgue measure on $\RN^{\C{J}}$.
    \end{enumerate}
  \item \label{thm:GeneralDirectedGlobalMarkovPropertSCMs} $\Prb^{\B{X}}$ satisfies the \emph{general directed global Markov property} (``$\sigma$-separation criterion'') relative to $\C{G}(\C{M})$ (see Definition~\ref{def:generalized_Markov_property}) if $\C{M}$ is uniquely solvable w.r.t.\ each strongly connected component of $\C{G}(\C{M})$.\footnote{Since \citep{FM17} also provides results under the weaker condition that an SCM is solvable
(not necessarily uniquely) w.r.t.\ each strongly connected component of $\C{G}(\C{M})$, one might
believe that Theorem~\ref{thm:GeneralDirectedGlobalMarkovPropertSCMs} could be generalized to stating that in that case,
any of its observational distributions satisfies the general directed global
Markov property. However, that is not true:
consider, for example, the SCM $\C{M} = \langle \B{2}, \emptyset, \RN^2, \B{1}, \B{f},
\Prb_{\B{1}} \rangle$ with $f_1(\B{x})=x_1$ and $f_2(\B{x})=x_2$. Then $\C{M}$ is solvable
w.r.t.\ each of its strongly connected components $\{1\}$ and $\{2\}$. The solution with $X_1=X_2$, where $X_2$ has a nondegenerate distribution, shows a dependence between $X_1$ and $X_2$, and thus $X_1 \indep X_2$ does not hold. In general, all strongly connected components that admit multiple solutions may be dependent on any other variable(s) in the model.
}
\end{enumerate}
\end{theorem}
The general directed global Markov property is generally weaker than the directed global Markov property, since $\sigma$-separation implies $d$-separation. The acyclic case is well known and was first shown in the context of linear-Gaussian structural
equation models \citep{SRM+98,Kos99}.
The discrete case fixes the erroneous theorem by Pearl and Dechter~\cite{PD96}, for which a
counterexample was found by Neal~\cite{Nea00}, by adding the ancestral unique
solvability condition, and extends it to allow for bidirected edges in the graph. 
The linear case is an extension of existing results for the linear-Gaussian setting 
without bidirected edges \citep{Spi94, Spi95, Kos96} to a linear (possibly
non-Gaussian) setting with bidirected edges in the graph. 

%The results in Theorem~\ref{thm:DirectedGlobalMarkovPropertSCMsb} and \ref{thm:GeneralDirectedGlobalMarkovPropertSCMs} are in general not preserved under perfect intervention (see Example~\ref{ex:InterventionUniqueSolvability}). As the class of simple SCMs is preserved under perfect intervention and the twin operation (see Proposition~\ref{prop:SimplenessClosedUnderResults}) we have that all the results of Theorem~\ref{thm:globalMarkovPropertiesSCMs} apply equally well to all the observational, interventional and counterfactual distributions of a simple SCM (see Corollary~\ref{coro:gdgMarkovPropertyInterventionalSCM} in the Supplementary Material~\citep{BPSM19s}).

In constraint-based approaches to causal discovery, one usually assumes the converse of the (general) directed global Markov property to hold~\citep{SGS00,Pea09}, which is called $\sigma$-faithfulness respectively $d$-faithfulness (see Definition~\ref{def:dFaithfulness} and \ref{def:sigmaFaithfulness}). Meek~\citep{Mee95} showed that for multinomial and linear-Gaussian DAG (i.e., acyclic and causally sufficient SCMs) models, $d$-faithfulness holds for all parameter values up to a measure zero set. 
%In general, $d$- or $\sigma$-faithfulness does not hold for SCMs. For example, consider solution variables $X_1$ and $X_4$ in the SCM of Example~\ref{ex:UniqueSolvabilityNotUniqAncestralSubset} which are $d$- and $\sigma$-separated but are not independent for every . 
Up to our knowledge no such results have been shown in more general parametric or nonparametric settings (neither for $d$-faitfhulness in acyclic or cyclic settings, nor for $\sigma$-faithfulness).

%%%%%%%%%%%%%%%%%%%%%%%%%%%%%%%%%%%%%%%%%%%%%%%%%%
\section{Causal interpretation of the graph of SCMs}
\label{sec:CausalInterpretationSCMs}
%%%%%%%%%%%%%%%%%%%%%%%%%%%%%%%%%%%%%%%%%%%%%%%%%%

In Example~\ref{ex:InterventionalEquivalenceDoesNotImplySameAugmentedGraph}, we 
already saw that sometimes no information in the observational, interventional and even the counterfactual distributions suffices to decide whether a directed path or bidirected edge is present in the graph, or not. Here, we do not attempt to provide a complete characterization of 
the conditions under which the presence or absence of a directed path or bidirected edge in
the graph can be identified from the observational and interventional
distributions. Instead, we give sufficient conditions to
detect a directed path and bidirected edge in the graph.

In general, cyclic SCMs may have none, one or multiple induced observational distributions, and this may change after intervening in the system. Here, we restrict ourselves to
%take here a conservative approach and give only a causal interpretation of the 
graphs of SCMs where the induced (marginal) observational and interventional distributions are uniquely defined.

%%%%%%%%%%%%%%%%%%%%%%%%%%%%%%%%%%%%%%%%%%%%%%%%%%
\subsection{Directed paths and edges}
%%%%%%%%%%%%%%%%%%%%%%%%%%%%%%%%%%%%%%%%%%%%%%%%%%

For cyclic SCMs, the causal interpretation of the SCM is not always consistent with its graph. This can be illustrated with the SCM $\C{M}$ of Example~\ref{ex:NoLatentProjection}. Here, one sees a difference
in the marginal distribution $\Prb_{\C{M}_{\intervene(\{1\},\xi_1)}}$ on $\C{X}_4$ for different values of $\xi_1$, although variable $1$ is not an ancestor of variable $4$ and each marginal distribution $\Prb_{\C{M}_{\intervene(\{1\},\xi_1)}}$ on $\C{X}_4$ is uniquely defined. This counterintuitive behavior that an intervention on a nonancestor of a variable can change the distribution of that variable was already observed by Neal~\cite{Nea00}. However, under a specific unique solvability condition, we obtain a direct causal interpretation for the absence of a directed edge or directed path in the graph of an SCM.
\begin{proposition}[Sufficient condition for detecting a directed edge in the latent projection of the graph of an SCM]
\label{prop:DirectedPathEdges}
Consider an SCM $\C{M} = \langle \C{I}, \C{J}, \BC{X}, \BC{E}, \B{f},
\Prb_{\BC{E}} \rangle$, a subset $\C{O}\subseteq\C{I}$ and $i,j\in\C{O}$ such that $i\ne j$. Let $\B{\xi}_{I}\in\BC{X}_I$, where $I:=\C{O}\setminus\{i,j\}$, such that $\C{M}_{\intervene(I,\B{\xi}_I)}$ is uniquely solvable w.r.t.\ $\an_{\C{G}(\C{M}_{\intervene(I,\B{\xi}_I)})_{\setminus i}}(j)$. If
there exist values $\xi_i \ne \tilde{\xi}_i \in \C{X}_i$ such that both 
$(\C{M}_{\intervene(I,\B{\xi}_I)})_{\intervene(\{i\},\xi_i)}$ and
$(\C{M}_{\intervene(I,\B{\xi}_I)})_{\intervene(\{i\},\tilde{\xi}_i)}$
induce unique marginal distributions on $\C{X}_j$, and these two induced distributions do not coincide, that is, there exists a measurable set $\C{B}_j\subseteq\C{X}_j$ such that
$$
\Prb_{(\C{M}_{\intervene(I,\B{\xi}_I)})_{\intervene(\{i\},\xi_i)}}(X_j\in\C{B}_j)
  \ne 
\Prb_{(\C{M}_{\intervene(I,\B{\xi}_I)})_{\intervene(\{i\},\tilde{\xi}_i)}}(X_j\in\C{B}_j) \,,
$$
the directed edge $i\to j$ is present in the
latent projection $\marg(\C{I}\setminus \C{O})(\C{G}(\C{M}))$ of $\C{G}(\C{M})$ on
$\C{O}$.
\end{proposition}
Two cases are of special interest: $\C{O}=\C{I}$, which corresponds with a
directed edge $i\to j$ in $\C{G}(\C{M})$, and $\C{O}=\{i,j\}$, which corresponds with a
directed path $i\to\cdots\to j$ in $\C{G}(\C{M})$. 
%We carefully avoided making statements when an induced marginal distribution is not defined or not unique.
%\footnote{E.g.\ it would be
%tempting to propose the condition that there exists a value
%$\B{\xi}_I \in \BC{X}_I$ and there
%exist values
%$\xi_i \ne \tilde{\xi}_i \in \C{X}_i$ such that 
%$\big(\C{M}_{\intervene(I,\B{\xi}_I)}\big)_{\intervene(i,\xi_i)}$ and
%$\big(\C{M}_{\intervene(I,\B{\xi}_I)}\big)_{\intervene(i,\tilde{\xi}_i)}$
%are not observationally equivalent with respect to $\{j\}$. However, it is not clear what this would mean in practice
%if the marginal observational distributions do not exist or are not unique. The assumptions of existence and uniqueness
%enable one to test in practice whether the two interventional distributions coincide by using two finite samples from 
%both distributions.}

The condition in Proposition~\ref{prop:DirectedPathEdges} is a sufficient condition for
determining whether a directed edge or path is present in the graph. In general, not all directed edges and paths can be identified from the interventional distributions with this sufficient condition. For example, no interventional distribution satisfies the condition of Proposition~\ref{prop:DirectedPathEdges} for the SCM $\bar{\C{M}}$ in Example~\ref{ex:InterventionalEquivalenceDoesNotImplySameAugmentedGraph}, although there is a directed edge $1\to 2$ in the graph $\C{G}(\bar{\C{M}})$.

%%%%%%%%%%%%%%%%%%%%%%%%%%%%%%%%%%%%%%%%%%%%%%%%%%
\subsection{Bidirected edges}
%%%%%%%%%%%%%%%%%%%%%%%%%%%%%%%%%%%%%%%%%%%%%%%%%%

It is well known that there exists a similar sufficient condition for detecting
bidirected edges in the graph of an acyclic SCM also known as the common-cause principle~\citep[see, e.g.,][]{Pea09}. In the two variables case, this criterion informally states that there exists a bidirected edge between the variables $i$ and $j$ in the graph of the SCM, if the marginal interventional distribution of $X_j$ under the intervention $\intervene(\{i\},x_i)$ differs from the conditional distribution of $X_j$ given $X_i=x_i$ (see Example~\ref{ex:DetectingBiDirectedEdges}).
%\begin{example}[Detecting a bidirected edge in the graph of an SCM]
%\label{ex:DetectingBiDirectedEdges}
%Consider the acyclic SCM $\C{M}$ of Example~\ref{ex:InterventionalEquivalenceDoesNotImplySameAugmentedGraph} and the SCM $\hat{\C{M}}$ that is the same as $\C{M}$ except for its causal mechanism, which is given by
%$$
%\hat{f}_1(\B{x},\B{e}) = e_1 \,, \quad \hat{f}_2(\B{x},\B{e}) = x_1 e_1 \,.
%$$
%For the SCM $\hat{\C{M}}$ we observe that the marginal interventional probability $\Prb_{\hat{\C{M}}_{\intervene(\{1\},\xi_1)}}(X_2=-1)$ is not equal to the conditional probability $\Prb_{\hat{\C{M}}}(X_2=-1 \given X_1=\xi_1)$ for both $\xi_1=-1$ and $\xi_1=1$. This observation suffices to identify the presence of the bidirected edge $1\oto 2$ in the graph $\C{G}(\hat{\C{M}})$. For the SCM $\C{M}$, whose graph does not contain the bidirected edge $1\oto 2$, the marginal interventional distribution and conditional distribution coincide.
%\end{example}
The following proposition provides a generalization of this 
%well-known result 
sufficient condition for detecting bidirected edges in graphs of SCMs that may include cycles.
%to SCMs that may include cycles.
\begin{proposition}[Sufficient condition for detecting a bidirected edge in the latent projection of the graph of an SCM]
\label{prop:BidirectedEdges}
Consider an SCM $\C{M} = \langle \C{I}, \C{J}, \BC{X}, \BC{E}, \B{f}, \Prb_{\BC{E}} \rangle$, a subset $\C{O} \subseteq \C{I}$ and $i,j\in \C{O}$ such that $i\ne j$. Let $\B{\xi}_I\in\BC{X}_I$, where $I:=\C{O}\setminus \{i,j\}$, such that $\C{M}_{\intervene(I,\B{\xi}_I)}$ is uniquely solvable w.r.t.\ both $\an_{\C{G}(\C{M}_{\intervene(I,\B{\xi}_I)})}(i)$ and $\an_{\C{G}(\C{M}_{\intervene(I,\B{\xi}_I)})_{\setminus i}}(j)$. Assume that for every $\xi_i\in\C{X}_i$ both $\C{M}_{\intervene(I,\B{\xi}_I)}$ and $(\C{M}_{\intervene(I,\B{\xi}_I)})_{\intervene(\{i\},\xi_i)}$ induce a unique marginal distribution on $\C{X}_j\times\C{X}_i$ and $\C{X}_j$, respectively. If $j\notin\an_{\C{G}(\C{M}_{\intervene(I,\B{\xi}_I)})}(i)$ and there exists a measurable set $\C{B}_j\subseteq\C{X}_j$ such that for every version of
the regular conditional probability
$\Prb_{\C{M}_{\intervene(I,\B{\xi}_I)}}(X_j \in \C{B}_j \given
X_i = \xi_i)$, there exists a value $\xi_i\in\C{X}_i$ such that
$$
\Prb_{(\C{M}_{\intervene(I,\B{\xi}_I)})_{\intervene(\{i\},\xi_i)}}(X_j
  \in \C{B}_j)
 \ne 
  \Prb_{\C{M}_{\intervene(I,\B{\xi}_I)}}(X_j \in \C{B}_j \given
  X_i = \xi_i) \,,
$$ 
then there exists a bidirected edge $i\oto j$ in the
latent projection $\marg(\C{I}\setminus \C{O})(\C{G}(\C{M}))$ of $\C{G}(\C{M})$ on
$\C{O}$.
\end{proposition}
This proposition gives a sufficient condition for determining that a bidirected edge is present in the graph. In general, not all bidirected edges in the graph can be identified from the observational, interventional and even the counterfactual distributions, as we saw in 
Example~\ref{ex:CounterfactualEquivalenceDoesNotImplySameAugmentedGraph}. In this example, there exists a bidirected edge $1\oto 2 \in\C{G}(\C{M})$ while the density $p(x_2 \given \intervene(X_1=x_1)) = p(x_2 \given X_1 = x_1)$ for all $x_1\in\C{X}_1$.
For the acyclic setting, the above criterion is generally considered as a universal way to detect a confounder (note that then one can also deal with the case $j\in\an_{\C{G}(\C{M}_{\intervene(I,\B{\xi}_I)})}(i)$ by swapping the roles of $i$ and $j$). If $i$ and $j$ are part of a cycle, the above sufficient condition cannot be applied, and in that case, to the best of our knowledge, no simple sufficient conditions for detecting the presence of a bidirected edge are known.

%%%%%%%%%%%%%%%%%%%%%%%%%%%%%%%%%%%%%%%%%%%%%%%%%%
\section{Simple SCMs}
\label{sec:SimpleSCMs}
%%%%%%%%%%%%%%%%%%%%%%%%%%%%%%%%%%%%%%%%%%%%%%%%%%

In this section, we introduce the well-behaved class of simple SCMs. Simple SCMs satisfy all the local unique solvability conditions to ensure that this class is closed under both perfect intervention and marginalization. They extend the subclass of acyclic SCMs to the cyclic setting, while preserving many of their convenient properties.

\begin{definition}[Simple SCM]
\label{def:SimpleSCM}
Let $\C{M}=\langle \C{I}, \C{J}, \BC{X}, \BC{E}, \B{f}, \Prb_{\BC{E}} \rangle$ be an SCM. We call $\C{M}$ \emph{simple} if it is uniquely solvable w.r.t.\ every subset $\C{O}\subseteq\C{I}$.
\end{definition}
Loosely speaking, an SCM is simple if any subset of its structural equations can be solved uniquely for its associated variables in terms of the other variables that appear in these equations. An example of a simple SCM is given in Example~\ref{ex:HarmonicOscillator}.

%This allows us to marginalize over any subsystem by mere substitution of their measurable solution functions (see Definition~\ref{def:MarginalSCM}). 
On simple SCMs one can perform any number of marginalizations (see Definition~\ref{def:MarginalSCM}) in any order (see Proposition~\ref{prop:MarginalizationCommutes}). All these marginalizations respect the latent projection (see Proposition~\ref{prop:LatentProjection}) and each resulting marginal SCM is again simple. Moreover, we show that this class is closed under intervention and the twin operation.
\begin{proposition}
\label{prop:SimplenessClosedUnderResults}
The class of simple SCMs is closed under marginalization, perfect intervention and the twin operation.
\end{proposition}
The class of simple SCMs contains the acyclic SCMs as a subclass (see Proposition~\ref{prop:AcyclicSCMUniquelySolvable}). In particular, a simple SCM has no self-cycles (see Proposition~\ref{prop:ObstructionToSelfCycles}), since a self-cycle denotes that that variable cannot be uniquely (up to a $\Prb_{\BC{E}}$-null set) determined by its parents.

From Proposition~\ref{prop:SimplenessClosedUnderResults}, it follows that the results summarized in Theorem~\ref{thm:globalMarkovPropertiesSCMs} also apply to all the observational, interventional and counterfactual distributions of simple SCMs.
\begin{corollary}[Global Markov properties for simple SCMs]
\label{coro:gdgMarkovPropertySimpleSCM}
Let $\C{M}$ be a simple SCM. Then the: 
\begin{enumerate}
  \item observational distribution,
  \item interventional distribution after perfect intervention on
    $I\subset\C{I}$, 
  \item counterfactual distribution after perfect intervention on
   $\tilde{I}\subseteq\C{I}\cup\C{I}'$,
 \end{enumerate}
all exist, are unique and satisfy the 
general directed global Markov property relative to $\C{G}(\C{M})$,
$\intervene(I)(\C{G}(\C{M}))$ and $\intervene(\tilde{I})(\twin(\C{G}(\C{M})))$, respectively. Moreover, if $\C{M}$ satisfies at least one of the three
conditions (1a), (1b), (1c) of
Theorem~\ref{thm:globalMarkovPropertiesSCMs}, then they also obey the directed global Markov property relative to $\C{G}(\C{M})$,
$\intervene(I)(\C{G}(\C{M}))$ and $\intervene(\tilde{I})(\twin(\C{G}(\C{M})))$, respectively.
\end{corollary}
%We will see in the next sections that simple SCMs have the following additional convenient properties: they obey several Markov properties (see Section~\ref{sec:MarkovProperty}); and their graphs express their causal semantics (see Section~\ref{sec:CausalInterpretationSCMs}). 
Many of these properties are also shown to hold for the class of \emph{modular SCMs}~\citep{FM17}, which contains, in particular, the class of simple SCMs (see Appendix~\ref{app:ModularSCMs} for more details). 
%In the upcoming sections we show that without this additional structure, these properties do not hold in general, and we will investigate to which extent and under which conditions each of these properties still hold. 
%The advantage of simple SCMs is that we can always construct a ``system of (measurable) solution functions'' that is ``compatible'' (see Proposition~\ref{prop:SimpleSCMsCompatibleSysOfSolFunctions}).

Moreover, simple SCMs satisfy the unique solvability conditions of Proposition~\ref{prop:DirectedPathEdges} and \ref{prop:BidirectedEdges},
%and moreover, the (marginal) observational and interventional distributions always exist and are uniquely defined
which allows us to define the causal relationships for simple SCMs in terms of its graph.
\begin{definition}[Causal relationships for simple SCMs]
\label{def:DirectCausesConfounders}
Let $\C{M}$ be a simple SCM. 
\begin{enumerate}
\item
If there exists a directed edge $i\to j \in\C{G}(\C{M})$, that is, $i\in\pa(j)$, then we call \emph{$i$ a direct cause of $j$ according to $\C{M}$};
\item 
If there exists a directed path $i\to\cdots\to j$ in $\C{G}(\C{M})$, that is, $i\in\an(j)$, then we call \emph{$i$ a cause of $j$ according to $\C{M}$};
\item
  If there exists a bidirected edge $i\oto j \in \C{G}(\C{M})$, then we call \emph{$i$ and $j$ (latently) confounded according to $\C{M}$}.
\end{enumerate}
\end{definition}

In summary, we have the following sufficient conditions for determining the different causal and confoundedness relationships according to a specific simple SCM $\C{M}$.
\begin{corollary}[Sufficient conditions for the presence of causal and confoundedness relationships for simple SCMs]
\label{coro:SufficientConditionCausesConfounder}
Let $\C{M}$ be a simple SCM and $i,j\in\C{I}$ such that $i\ne j$ and $I:=\C{I}\setminus\{i,j\}$. Then:
\begin{enumerate}
\item 
If there exist values $\B{\xi}_I\in\BC{X}_I$ and $\xi_i \ne \tilde{\xi}_i \in \C{X}_i$ and a measurable set $\C{B}_j\subseteq\C{X}_j$ such that
$$
\Prb_{(\C{M}_{\intervene(I,\B{\xi}_I)})_{\intervene(\{i\},\xi_i)}}(X_j\in\C{B}_j)
  \ne 
\Prb_{(\C{M}_{\intervene(I,\B{\xi}_I)})_{\intervene(\{i\},\tilde{\xi}_i)}}(X_j\in\C{B}_j) \,,
$$
then $i$ is a direct cause of $j$ according to $\C{M}$, that is, $i\to j \in \C{G}(\C{M})$;
\item
If there exist values $\xi_i \ne \tilde{\xi}_i \in \C{X}_i$ and a measurable set $\C{B}_j\subseteq\C{X}_j$ such that
$$
\Prb_{\C{M}_{\intervene(\{i\},\xi_i)}}(X_j\in\C{B}_j)
  \ne 
\Prb_{\C{M}_{\intervene(\{i\},\tilde{\xi}_i)}}(X_j\in\C{B}_j) \,,
$$
then $i$ is a cause of $j$ according to $\C{M}$, that is, $i\to\cdots\to j$ in $\C{G}(\C{M})$;
\item
If $j\notin\an_{\C{G}(\C{M}_{\intervene(I,\B{\xi}_I)})}(i)$ and there exist a value $\B{\xi}_I\in\BC{X}_I$ and a measurable set $\C{B}_j\subseteq\C{X}_j$ such that for every version of the regular conditional probability
$\Prb_{\C{M}_{\intervene(I,\B{\xi}_I)}}(X_j \in \C{B}_j \given
X_i = \xi_i)$ there exists a value $\xi_i\in\C{X}_i$ such that
$$
\Prb_{(\C{M}_{\intervene(I,\B{\xi}_I)})_{\intervene(\{i\},\xi_i)}}(X_j
  \in \C{B}_j)
  \ne 
  \Prb_{\C{M}_{\intervene(I,\B{\xi}_I)}}(X_j \in \C{B}_j \given
  X_i = \xi_i) \,,
$$ 
then $i$ and $j$ are confounded according to $\C{M}$, that is, $i\oto j \in \C{G}(\C{M})$.
\end{enumerate}
\end{corollary}

%Corollary~\ref{coro:SufficientConditionCausesConfounder} summarizes sufficient conditions for determining the different causal relationships according to a simple SCM $\C{M}$. 
For simple SCMs, it is in general not possible to identify all the causal and confoundedness relationships in the graph from the observational, interventional or even the counterfactual distributions. Examples~\ref{ex:InterventionalEquivalenceDoesNotImplySameAugmentedGraph} and
\ref{ex:CounterfactualEquivalenceDoesNotImplySameAugmentedGraph} show that this is already impossible for acyclic SCMs without further assumptions.

%\begin{remark}
Finally, there is a connection between SCMs and potential outcomes \citep{Rub74} that generalizes
to the cyclic setting. One of the consequences of Proposition~\ref{prop:SimplenessClosedUnderResults}
is that all counterfactuals are defined for a simple SCM (even if it is cyclic). This allows us to
define potential outcomes in terms of a simple SCM in the following way.
\begin{definition}[Potential outcome]
Let $\C{M}=\langle \C{I},\C{J}, \BC{X}, \BC{E}, \B{f}, \Prb_{\BC{E}} \rangle$ be a simple SCM, $I\subseteq\C{I}$ a subset, $\B{\xi}_I\in\BC{X}_I$ a value and $\B{E}$ a random variable such that $\Prb^{\B{E}}=\Prb_{\BC{E}}$. The \emph{potential outcome under the perfect intervention $\intervene(I,\B{\xi}_I)$} is defined as $\B{X}_{\B{\xi}_I} := \B{g}_{\C{M}_{\intervene(I,\B{\xi}_I)}}(\B{E}_{\pa(\C{I})})$, where $\B{g}_{\C{M}_{\intervene(I,\B{\xi}_I)}}:\BC{E}_{\pa(\C{I})}\to\BC{X}$ is a measurable solution function for $\C{M}_{\intervene(I,\B{\xi}_I)}$.
\end{definition}
\section{Discussion}
\label{sec:Discussion}
%%%%%%%%%%%%%%%%%%%%%%%%%%%%%%%%%%%%%%%%%%%%%%%%%%

In this paper, we studied the basic properties of SCMs in the presence of
cycles and latent variables without restricting to linear functional
relationships between the variables. We saw that cyclic SCMs behave differently in many aspects
than acyclic SCMs. Indeed, in the presence of cycles, many of
the convenient properties of acyclic SCMs do not hold in general: SCMs do not
always have a solution; they do not always induce unique observational,
interventional and counterfactual distributions; a marginalization does not
always exist, and if it exists the marginal model does not always respect the
latent projection; they do not always satisfy a Markov property and their
graphs are not always consistent with their causal semantics. 

We introduced various notions of (unique) solvability and showed that
under appropriate (unique) solvability conditions, many of the operations and
results for the acyclic setting can be extended to SCMs with cycles.
For example, we introduced several equivalence relations between SCMs to 
compare SCMs at different levels of abstraction, we showed how to define
marginal SCMs on a subset of the variables that are (in various ways)
equivalent to the original SCM, we discussed under which conditions the distributions 
satisfy the (general) directed global Markov property relative to their graphs 
and we showed under which conditions the graph of an SCM can be interpreted 
causally. Most of these results are shown under sufficient
conditions that are not necessary (e.g., for the marginalization operation this was shown in Example~\ref{ex:MarginalizationSufficientCondition}). It may therefore be possible to further relax some of the conditions.

These insights led us to introduce the more well-behaved class of simple SCMs, which forms an extension of the class of acyclic SCMs to the cyclic setting that preserves many of its convenient properties: 
simple SCMs induce unique observational, interventional and
counterfactual distributions; 
%(see Proposition~\ref{prop:SimplenessClosedUnderResults}); 
the class of simple SCMs is closed under both perfect intervention and marginalization;
%(see Proposition~\ref{prop:SimplenessClosedUnderResults}); 
the marginalization respects the latent projection; 
%(see Proposition~\ref{prop:LatentProjection});
the induced distributions obey the general directed global Markov property and
obey the directed global Markov property in the acyclic, discrete and linear
case. 
%(see Corollary~\ref{coro:gdgMarkovPropertyInterventionalSCM}). 
This class does not contain SCMs that have
self-cycles 
%(see Proposition~\ref{prop:ObstructionToSelfCycles}) 
and graphs of simple SCMs have a direct and intuitive causal interpretation.
%(see Definition~\ref{def:DirectCausesConfounders} and Corollary~\ref{coro:SufficientConditionCausesConfounder}).

One key property of simple SCMs is that the solutions always satisfy
the conditional independencies implied by $\sigma$-separation.
%(see Definition~\ref{def:SigmaSeparation} and Corollary~\ref{coro:gdgMarkovPropertySimpleSCM}). 
By simply replacing $d$-separation with $\sigma$-separation it turns out that one can
directly extend results and algorithms for acyclic SCMs to the more general
class of simple SCMs. For example, adjustment criteria (including the back-door criterion),
Pearl's $\intervene$-calculus and Tian's ID algorithm for the identification of causal effects
have been extended recently to the class of modular SCMs, which contains the class of simple SCMs \citep{FM19}. 
%Also, 
Several causal discovery algorithms have already been proposed that work with simple SCMs, for example, the first constraint-based causal discovery algorithm that can deal with cycles and nonlinear functional relationships~\citep{FM18}. Also, Local Causal Discovery (LCD) \citep{Cooper1997}, Y-structures \citep{ManiPhD2006} and  the Joint Causal Inference framework (JCI) all apply to simple SCMs~\cite{MMC19} even though they were originally developed for acyclic SCMs only. Recently, it has been shown that even the well-known
Fast Causal Inference (FCI) algorithm \citep{SMR1999,Zhang2008_AI} is directly applicable to simple SCMs 
\citep{MooijClaassen_2005.00610} and provides a consistent estimate of the Markov equivalence class (under 
the faithfulness assumption). Moreover, a method for constructing nonlinear simple SCMs using neural networks and sampling from them has been proposed~\citep{FM18}. This illustrates that the class of simple SCMs forms a convenient and practical extension of the class of acyclic SCMs that can be used for the purposes of causal modeling, reasoning, discovery and prediction.
%Summarizing, when working with SCMs that include cyclic and non-linear relationships, we suggest to work with simple SCMs. The operations we discuss in the paper are well-behaved on this class of SCMs and many technical difficulties can be avoided.

We hope that this work will provide the foundations for a general
theory of statistical causal modeling with SCMs. 
Future work might consist of reparametrizing and reducing the space of
the exogenous variables of an SCM while preserving the causal and counterfactual
semantics; extending and generalizing the identifiability results for (direct) causes and confounders; extending the graphs of SCMs to represent selection bias; proving completeness results for some Markov properties for a subclass of SCMs that contains cycles.

%%%%%%%%%%%%%%%%%%%%%%%%%%%%%%%%%%%%%%%%%%%%%%%%%%
\section*{Acknowledgments}
%%%%%%%%%%%%%%%%%%%%%%%%%%%%%%%%%%%%%%%%%%%%%%%%%%

%\thankstext{t1}{Supported in part by NWO, the Netherlands Organization for
%Scientific Research (VIDI grant 639.072.410 and VENI grant 639.031.036).}
%\thankstext{t2}{Supported in part by the European Research Council (ERC) under
%the European Union's Horizon 2020 research and innovation programme (grant
%agreement n$^{\mathrm{o}}$ 639466).}
%\thankstext{t3}{Supported by research grants from VILLUM FONDEN (18968) and the Carlsberg Foundation.}
S.\ Bongers and J.M.\ Mooij are supported in part by NWO, the Netherlands Organization for Scientific Research (VIDI grant 639.072.410 and VENI grant 639.031.036). P.\ Forr{\'e} and J.M.\ Mooij are supported in part by the European Research Council (ERC) under the European Union's Horizon 2020 research and innovation programme (grant agreement n$^{\mathrm{o}}$ 639466). J.\ Peters is supported by research grants from VILLUM FONDEN (18968) and the Carlsberg Foundation.

The authors are grateful to Bernhard Sch{\"o}lkopf and Robin Evans for stimulating discussions, and to Noud de
Kroon, Tineke Blom
and Alexander Ly for providing helpful comments on earlier
drafts. We thank two anonymous reviewers and the associate editor for helpful comments.

%\begin{supplement}
%%\sname{}
%\label{supplement}
%\stitle{Supplement to ``Foundations of Structural Causal Models with Cycles and Latent Variables''}
%\slink[doi]{COMPLETED BY THE TYPESETTER}
%\sdatatype{.pdf}
%\sdescription{This Supplementary Material contains a summary of the basic terminology and results for causal graphical models, additional (unique) solvability properties, some results for linear SCMs, other examples, the proofs of all
%theoretical results and the measurable selection theorems that are used in
%several proofs.}
%\end{supplement}

%\pagebreak

\appendix

\part{Supplementary Material} % Start the Supplementary Material part

This Supplementary Material contains a summary of the basic terminology and results for causal graphical models (Appendix~\ref{app:AppendixCausalGraphicalModels}), additional (unique) solvability properties (Appendix~\ref{app:AppendixSolvabilityResults}), some results for linear SCMs (Appendix~\ref{app:AppendixLinSCMs}), other examples (Appendix~\ref{app:AppendixExamples}), the proofs of all the
theoretical results (Appendix~\ref{app:AppendixProofs}) and the measurable selection theorems (Appendix~\ref{app:AppendixMST}) that are used in several proofs.

%%%%%%%%%%%%%%%%%%%%%%%%%%%%%%%%%%%%%%%%%%%%%%%%%%
\section{Causal graphical models}
\label{app:AppendixCausalGraphicalModels}
%%%%%%%%%%%%%%%%%%%%%%%%%%%%%%%%%%%%%%%%%%%%%%%%%%

In this appendix, we provide a summary of the basic terminology and results for causal graphical models. In Appendix~\ref{app:AppendixDirectedMixedGraphs} we provide the terminology for directed (mixed) graphs. In Appendix~\ref{app:MarkovProperty} we give an introduction and an intuitive derivation of Markov properties for SCMs with cycles. In Appendix~\ref{app:ModularSCMs} we provide a definition of modular SCMs and show how they relate to SCMs. In Appendix~\ref{app:AppendixOverviewCausalGraphicalModels} we provide an overview of the causal graphical models related to SCMs. The proofs of the theoretical results in this appendix are given in Appendix~\ref{app:AppendixProofs}.

%%%%%%%%%%%%%%%%%%%%%%%%%%%%%%%%%%%%%%%%%%%%%%%%%%
\subsection{Directed (mixed) graphs}
\label{app:AppendixDirectedMixedGraphs}
%%%%%%%%%%%%%%%%%%%%%%%%%%%%%%%%%%%%%%%%%%%%%%%%%%

In this subsection, we introduce the terminology for
directed (mixed) graphs, where we do allow for cycles
\citep{Lau96,Ric03,Pea09,FM17}.

\begin{definition}[Directed (mixed) graph]
\label{def:DirectedMixedGraph}
\leavevmode
\begin{enumerate}
\item 
A \emph{directed graph} is a pair $\C{G} = (\C{V},\C{E})$, where $\C{V}$ is a set of nodes and 
$\C{E}$ is a set of directed edges, which is a subset $\C{E} \subseteq \C{V} \times \C{V}$ of ordered 
pairs of nodes. Each element $(i,j) \in \C{E}$ can be represented by the directed edge $i \to j$
or equivalently $j \ot i$. In particular, $(i,i) \in \C{E}$ represents a \emph{self-cycle} $i\to i$. 
\item 
A \emph{directed mixed graph} is a triple $\C{G} = (\C{V},\C{E},\C{B})$, where the pair
$(\C{V},\C{E})$ forms a directed graph and $\C{B}$ is a set of bidirected edges,
which is a subset
$\C{B}\subseteq\{ \{i,j\} : i,j\in \C{V}, i\neq j \}$ of unordered (distinct)
pairs of nodes. Each element
$\{i,j\}\in\C{B}$ can be represented by the bidirected edge $i\leftrightarrow j$ or equivalently $j \leftrightarrow i$. Note that a directed graph can be considered as a directed mixed graph without bidirected edges. 
\item
Let $\C{G}=(\C{V},\C{E},\C{B})$ be a directed mixed graph. A directed mixed graph
$\tilde{\C{G}}=(\tilde{\C{V}},\tilde{\C{E}},\tilde{\C{B}})$ is a \emph{subgraph of $\C{G}$} if $\tilde{\C{V}}\subseteq\C{V}$, $\tilde{\C{E}}\subseteq\C{E}$ and
$\tilde{\C{B}}\subseteq\C{B}$, in which case we write $\tilde{\C{G}} \subseteq \C{G}$. For a subset $\C{W}\subseteq\C{V}$, we define the \emph{induced
subgraph of $\C{G}$ on $\C{W}$} by
$\C{G}_{\C{W}} := (\C{W},\tilde{\C{E}},\tilde{\C{B}})$, where $\tilde{\C{E}}$ and $\tilde{\C{B}}$ are
the set of directed and bidirected edges in $\C{E}$ and $\C{B}$, respectively, that lie in
$\C{W}\times\C{W}$ and $\{ \{i,j\} : i,j\in\C{W}, i\neq j \}$, respectively.

\item
\sloppy 
%A \emph{walk} between $i$ and $j$ in a directed (mixed) graph $\C{G}$ is a sequence
%of nodes $(i=i_0,i_1,\dots,i_{n-1},i_n=j)$ and a sequence of edges (directed or bidirected) 
%$(\epsilon_1,\dots,\epsilon_n)$ for some $n\geq 0$ such that $\epsilon_k \in \{i_{k-1}\to i_k, i_{k-1}\leftarrow i_k, i_{k-1}\leftrightarrow i_k \}$ 
%for $k=1,2,\dots,n$ 
A \emph{walk} between $i, j \in \C{V}$ in a directed mixed graph $\C{G}$
is a tuple $(i_0, \epsilon_1, i_1, \epsilon_2, i_2, \dots, \epsilon_n, i_n)$ of alternating nodes
and edges in $\C{G}$ for some $n \ge 0$, where all $i_0, \dots, i_n \in \C{V}$, all $\epsilon_1, \dots, \epsilon_n \in \C{E} \cup \C{B}$ such that $\epsilon_k \in \{i_{k-1}\to i_k, i_{k-1}\leftarrow i_k, i_{k-1}\leftrightarrow i_k \}$ for all $k = 1, \dots, n$, and it starts with node $i_0 = i$ and ends with node $i_n = j$. Note that $n=0$ corresponds with a trivial walk consisting of a single node. If all nodes $i_0,\dots,i_n$ 
are distinct, it is called a \emph{path}. A walk (path) of the form $i\to \dots\to j$, that is, $\epsilon_k$ is $i_{k-1}\to i_k$ for all $k=1,2,\dots,n$, is called a \emph{directed walk
(path)} from $i$ to $j$. 
\item
A \emph{cycle} through $i \in \C{V}$ in a directed mixed graph $\C{G}$ is a directed path from $i$ to some node $j$ extended with the edge $j \to i \in \C{E}$. 
%A \emph{cycle} is a sequence of edges $(\epsilon_1,\dots,\epsilon_{n+1})$ for some $n \ge 0$
%such that $(\epsilon_1,\dots,\epsilon_n)$ forms a directed path from $i$ to $j$ and 
%$\epsilon_{n+1}$ is $j\to i$. 
In particular, a self-cycle $i \to i \in \C{E}$ is a cycle. Note that a path cannot contain any cycles. A directed graph and a directed mixed graph are said to be \emph{acyclic} if they contain no cycles, and are then referred to as a \emph{directed acyclic graph (DAG)} and an 
\emph{acyclic directed mixed graph (ADMG)}, respectively.
\item
For a directed mixed graph $\C{G}$ and a node $i \in \C{V}$ we define the set of 
\emph{parents} of $i$ by 
$\pa_{\C{G}}(i) := \{j \in \C{V} : j \to i \in \C{E}\}$, the set of children of $i$ by 
$\ch_{\C{G}}(i) := \{j \in \C{V} : i \to j \in \C{E}\}$, the set of \emph{ancestors} of $i$ by 
$$
  \an_{\C{G}}(i) := \{ j \in \C{V}: 
  \text{there is a directed path from $j$ to $i$ in $\C{G}$} \}
$$ 
and the set of \emph{descendants} of $i$ by 
$$
  \de_{\C{G}}(i) := \{ j \in \C{V} : 
  \text{there is a directed path from $i$ to $j$ in $\C{G}$} \} \,.
$$
Note that we have $\{i\}\cup\pa_{\C{G}}(i)\subseteq\an_{\C{G}}(i)$ and 
$\{i\}\cup\ch_{\C{G}}(i)\subseteq\de_{\C{G}}(i)$. We can apply all these definitions to subsets 
$\C{U} \subseteq \C{V}$ by taking unions, for example $\pa_{\C{G}}(\C{U}) := \cup_{i \in \C{U}} \pa_{\C{G}}(i)$. A subset $\C{A}\subseteq\C{V}$ is called an \emph{ancestral subset} in $\C{G}$ if $\C{A}=\an_{\C{G}}(\C{A})$, that is, $\C{A}$ is closed under taking ancestors of $\C{A}$ in $\C{G}$. 
\item
Let $\C{G}=(\C{V},\C{E},\C{B})$ be a directed mixed graph. We call $\C{G}$ \emph{strongly connected} if for every pair of distinct nodes $i,j\in\C{V}$, the graph contains a cycle that passes through both $i$ and $j$. The \emph{strongly connected component of $i\in\C{V}$}, denoted by
$\scc_{\C{G}}(i)$, is the
maximal subset $\C{S}\subseteq\C{V}$ such that $i\in\C{S}$ and the induced subgraph $\C{G}_{\C{S}}$ is 
strongly connected. Equivalently, $\scc_{\C{G}}(i)=\an_{\C{G}}(i)\cap\de_{\C{G}}(i)$.
\item A \emph{loop} in a directed mixed graph $\C{G}=(\C{V},\C{E},\C{B})$ is a subset $\C{O}\subseteq\C{V}$ that is strongly connected in the induced subgraph $\C{G}_{\C{O}}$ of $\C{G}$ on $\C{O}$.
\item
For a directed graph $\C{G}=(\C{V},\C{E})$, we define the \emph{graph of strongly connected components of $\C{G}$} as the directed graph $\C{G}^{\scc}:=(\C{V}^{\scc},\C{E}^{\scc})$, where $\C{V}^{\scc}$ are the strongly connected components of $\C{G}$, that is, $\C{V}^{\scc}$ are the equivalence classes in $\C{V}/{\sim}$ with the equivalence relation $i\sim j$ if and only if $i\in\scc_{\C{G}}(j)$, and $\C{E}^{\scc} = (\C{E}\setminus\{i\to i : i\in \C{V}\})/{\sim}$ with the equivalence relation $(i\to j) \sim (i'\to j')$ if and only if $i \sim i'$ and $j \sim j'$.
\end{enumerate}
\end{definition}
We omit the subscript $\C{G}$ whenever it is clear which directed (mixed) graph $\C{G}$ we are referring to.

\begin{lemma}[DAG of strongly connected components]
\label{lemm:GraphOfStronglyConnectedComponentsDAG}
Let $\C{G}=(\C{V},\C{E})$ be a directed graph. Then $\C{G}^{\scc}$, the graph of strongly connected 
components of $\C{G}$, is a DAG.
\end{lemma}

%%%%%%%%%%%%%%%%%%%%%%%%%%%%%%%%%%%%%%%%%%%%%%%%%%
\subsection{Markov properties}
\label{app:MarkovProperty}
%%%%%%%%%%%%%%%%%%%%%%%%%%%%%%%%%%%%%%%%%%%%%%%%%%

In this subsection, we give a short overview of Markov properties for
SCMs with cycles. 
%These Markov properties were recently derived in \citep{FM17} for HEDGes, a graphical representation that is similar to the augmented graph of SCMs. Here, we present some of the highlights of the work of \citep{FM17} and reformulate these results in terms of SCMs. 
We will make use of the Markov properties that were recently developed by Forr{\'e} and Mooij~\cite{FM17} for HEDGes, a graphical representation that is similar to the augmented graph of SCMs. We briefly summarize some of their main results and apply them to the class of SCMs. We also provide a shorter 
and more intuitive derivation so that this subsection can act as an entry point
for the reader into the more extensive discussion of Markov 
properties provided in~\cite{FM17}.

Markov properties associate a set of conditional independence relations to a
graph. The directed global Markov property for directed acyclic graphs, also
known as the $d$-separation criterion \citep{Pea85}, is one of the most widely used.
It directly extends to a similar property for acyclic directed mixed graphs (ADMGs) \citep{Ric03}. It does not hold in general for cyclic SCMs, however, as
was already observed earlier \citep{Spi94, Spi95}.
Under some conditions (roughly speaking, linearity or discrete variables) 
the directed global Markov property can be shown to hold also in the presence of cycles \citep{FM17}.

Inspired by work of Spirtes~\cite{Spi94}, Forr{\'e} and Mooij~\cite{FM17} recognized that in the general
cyclic case a different extension of $d$-separation, termed
$\sigma$-separation, is needed, leading to the general directed global
Markov property. One key result in \citep{FM17}
implies that under the assumption of unique solvability w.r.t.\ each strongly connected 
component of its graph, the observational distribution of an SCM satisfies the 
general directed global Markov property w.r.t.\ its graph. The solvability 
assumptions are in general not preserved under interventions. Under the stronger 
assumption of simplicity, however, they are, and one obtains the corollary that
also all interventional and counterfactual distributions of a simple SCM satisfy 
the general directed global Markov property w.r.t.\ to their corresponding graphs. 

For a more extensive study of different Markov properties that can be associated to SCMs we refer the reader to \citep{FM17}. 
%In Appendix~\ref{app:AppendixCausalGraphicalModels} of the Supplementary Material~\citep{BPSM19s} one can find all the relevant graphical notation and terminology for the remainder of this section.

%%%%%%%%%%%%%%%%%%%%%%%%%%%%%%%%%%})%%%%%%%%%%%%%%
\subsubsection{The directed global Markov property}
%%%%%%%%%%%%%%%%%%%%%%%%%%%%%%%%%%%%%%%%%%%%%%%%%%

%In this section we introduce the notation and terminology for the directed global Markov
%property. This directed global Markov
%property associates a conditional independence relation in the observational
%distribution of the SCM to each $d$-separation entailed by the graph. 
%Here, we use a formulation
%of $d$-separation that generalizes $d$-separation
%for DAGs \citep{Pea85} and $m$-separation for ADMGs \citep{Ric03} and
%mDAGs \citep{Eva16}, as discussed in more detail in \citep{FM17}.

Conditional independencies in the observational distribution of an acyclic SCM can be read off from its graph by using the graphical criterion called
$d$-separation \citep{Pea09}. The directed global Markov
property associates a conditional independence relation in the observational
distribution of the SCM to each $d$-separation entailed by the graph. 
Here, we use a formulation
of $d$-separation that generalizes $d$-separation
for DAGs \citep{Pea85} and $m$-separation for ADMGs \citep{Ric03} and
mDAGs \citep{Eva16}.
%We refer the reader to Appendix~\ref{app:DirectedGlobalMarkovProperty} of the Supplementary Material \citep{BPSM19s} for the definitions of $d$-separation (Definition~\ref{def:DSeparation}) and the directed global Markov property (Definition~\ref{def:Markov_property}).

\begin{definition}[Collider]
\label{def:Colliders}
Let $\pi = (i_0,\epsilon_1,i_1,\epsilon_2,i_2,\dots,\epsilon_n,i_n)$ be a walk (path) in a directed
mixed graph $\C{G}=(\C{V},\C{E},\C{B})$.
A node $i_k$ on $\pi$ is called a \emph{collider on $\pi$} if it is a non-endpoint node
($1 \le k < n$) and the two edges $\epsilon_k,\epsilon_{k+1}$ meet head-to-head on $i_k$ (i.e., if 
the subwalk $(i_{k-1}, \epsilon_k, i_k, \epsilon_{k+1}, i_{k+1})$ is of the form
$i_{k-1} \to i_k \ot i_{k+1}$, $i_{k-1} \oto i_k \ot i_{k+1}$, 
$i_{k-1} \to i_k \oto i_{k+1}$ or $i_{k-1} \oto i_k \oto i_{k+1}$).
The node $i_k$ is called a \emph{non-collider on $\pi$} otherwise,
that is, if it is an endpoint node ($k=0$ or $k=n$) or if
the subwalk $(i_{k-1}, \epsilon_k, i_k, \epsilon_{k+1}, i_{k+1})$ is of the form
$i_{k-1} \to i_k \to i_{k+1}$, $i_{k-1} \ot i_k \ot i_{k+1}$,
$i_{k-1} \ot i_k \to i_{k+1}$, $i_{k-1} \oto i_k \to i_{k+1}$ or $i_{k-1} \ot i_k \oto i_{k+1}$.
\end{definition}
Note in particular that the end points of a walk are non-colliders on the walk.
\begin{definition}[$d$-separation]
\label{def:DSeparation}
Let $\C{G}=(\C{V},\C{E},\C{B})$ be a directed mixed graph and let $C\subseteq\C{V}$ be a subset of nodes. 
A walk (path) $\pi = (i_0, \epsilon_1, i_1, \dots, i_n)$ in $\C{G}$ is said to be \emph{$C$-$d$-blocked or $d$-blocked by $C$} if
\begin{enumerate}
%  \item its first node $i_0 \in C$ or its last node $i_n \in C$, or
  \item it contains a collider $i_k \notin \an_{\C{G}}(C)$, or
  \item it contains a non-collider $i_k \in C$.
\end{enumerate}
The walk (path) $\pi$ is said to be \emph{$C$-$d$-open} if it is not $d$-blocked by $C$.
For two subsets of nodes $A, B \subseteq \C{V}$, we say that \emph{$A$ is $d$-separated from $B$ given $C$ in $\C{G}$} 
if all paths between any node in $A$ and any node in $B$ are $d$-blocked by $C$, and write
$$
A \mathrel{\mathop{\sigmablocked}^{d}_{\C{G}}} B \given C \,.
$$
\end{definition}
The next lemma is a straightforward generalization of Lemma~3.3 in \citep{Gei90} to the cyclic
setting. It implies that it suffices to formulate $d$-separation in terms of paths rather than walks. 
\begin{lemma}
\label{lemm:DWalksPaths}
Let $\C{G}=(\C{V},\C{E},\C{B})$ be a directed mixed graph, $C\subseteq\C{V}$
and $i,j\in\C{V}$. There exists a $C$-$d$-open walk between $i$ and $j$ in
$\C{G}$ if and only if there exists a $C$-$d$-open path between $i$ and $j$ in $\C{G}$.
\end{lemma}

\begin{definition}[Directed global Markov property]
\label{def:Markov_property}
Let $\C{G}=(\C{V},\C{E},\C{B})$ be a directed mixed graph and $\Prb_{\C{V}}$ a probability distribution 
on $\BC{X}_{\C{V}}=\prod_{i\in\C{V}}\C{X}_i$, where each $\C{X}_i$ is a standard probability space. The probability distribution $\Prb_{\C{V}}$ satisfies 
the \emph{directed global Markov property} relative to $\C{G}$ if for all subsets 
$A,B,C\subseteq\C{V}$ we have
$$
  A \mathrel{\mathop{\sigmablocked}^{d}_{\C{G}}} B \given C \quad \Longrightarrow \quad
  \B{X}_A \mathrel{\mathop{\indep}_{\Prb_{\C{V}}}} \B{X}_B \given \B{X}_C \,,
$$
that is, $(X_i)_{i\in A}$ and $(X_i)_{i\in B}$ are conditionally
independent given $(X_i)_{i\in C}$ under $\Prb_{\C{V}}$, 
where we take the canonical projections $X_i:\BC{X}_{\C{V}}\to\C{X}_i$ as random variables.
\end{definition}

From the results in \citep{FM17} it directly follows that
for the observational distribution of an SCM, the directed global Markov property w.r.t.\ the graph of the SCM (also known as the $d$-separation criterion), holds under one of the following assumptions.
\begin{theorem}[Directed global Markov property for SCMs~\citep{FM17}]
\label{thm:dgMarkovPropertySCMThreeSpecialCases}
Let $\C{M}$ be a uniquely solvable SCM that satisfies at least one of the following three conditions:
\begin{enumerate}
  \item %$\C{M}$ has a graph $\C{G}(\C{M})$ with a bidirected edge $i\leftrightarrow j$ for all $i \ne j\in\C{I}$ such that $j\in\scc_{\C{G(\C{M})}}(i)$. 
    $\C{M}$ is acyclic;
  \item all endogenous spaces $\C{X}_i$ are discrete and $\C{M}$ is ancestrally
    uniquely solvable;
  \item $\C{M}$ is linear (see Definition~\ref{def:LinearSCM}), 
    each of its causal mechanisms $\{f_i\}_{i\in\C{I}}$ has a nontrivial dependence on at least one exogenous variable, and
%    $\Prb^{\B{X}}$ has a density w.r.t.\ the Lebesgue measure,
    $\Prb_{\BC{E}}$ has a density w.r.t.\ the Lebesgue measure on $\RN^{\C{J}}$.
\end{enumerate}
Then its observational distribution $\Prb^{\B{X}}$ exists, is unique and satisfies the directed global Markov property relative to $\C{G}(\C{M})$ (see Definition~\ref{def:Markov_property}).
\end{theorem}
The acyclic case is well known and was first shown in the context of linear-Gaussian structural
equation models \citep{SRM+98,Kos99}.
The discrete case fixes the erroneous theorem by Pearl and Dechter~\cite{PD96}, for which a
counterexample was found by Neal~\cite{Nea00}, by adding the ancestral unique
solvability condition, and extends it to allow for bidirected edges in the graph. 
The linear case is an extension of existing results for the linear-Gaussian setting 
without bidirected edges \citep{Spi94, Spi95, Kos96} to a linear (possibly
non-Gaussian) setting with bidirected edges in the graph. 

The following counterexample of
an SCM for which the directed global Markov property does not hold was already given in~\citep{Spi94,Spi95}.
\begin{example}[Directed global Markov property does not hold for cyclic SCM]
  \label{ex:SpirtesExampleAppendix}
  \begin{figure}
  \adjustbox{scale=0.8,center}{%
  \begin{tikzpicture}
    \begin{scope}
    \node[var] (X1) at (0,0) {$X_1$};
    \node[var] (X2) at (1.5,0) {$X_2$};
    \node[var] (X3) at (0,-1.25) {$X_3$};
    \node[var] (X4) at (1.5,-1.25) {$X_4$};
    \draw[arr] (X1) -- (X3);
    \draw[arr] (X2) -- (X4);
    \draw[arr] (X3) to [bend right=20] (X4);
    \draw[arr] (X4) to [bend right=20] (X3);
    \end{scope}
    \begin{scope}[shift={(4,0)}]
    \node[var] (X1) at (0,0) {$X_1$};
    \node[var] (X2) at (1.5,0) {$X_2$};
    \node[var] (X3) at (0,-1.25) {$X_3$};
    \node[var] (X4) at (1.5,-1.25) {$X_4$};
    \draw[arr] (X1) -- (X3);
    \draw[arr] (X2) -- (X4);
    \draw[arr] (X2) -- (X3);
    \draw[arr] (X1) -- (X4);
    \draw[biarr] (X4) to [bend right=20] (X3);
    \end{scope}
\end{tikzpicture}}
  \caption{The graphs of the observationally equivalent SCMs $\C{M}$ (left) and
  $\tilde{\C{M}}$ (right) of Example~\ref{ex:SpirtesExampleAppendix} and \ref{ex:AcyclificationExampleAppendix}.}
  \label{fig:SpirtesExampleAppendix}
\end{figure}
Consider the SCM $\C{M} = \langle \B{4}, \B{4}, \RN^4, \RN^4, \B{f},
\Prb_{\RN^4} \rangle$ with causal mechanism given by
$$
f_1(\B{x},\B{e}) = e_1 \,,\,\quad f_2(\B{x},\B{e})=e_2 \,,\,\quad f_3(\B{x},\B{e})=x_1 x_4
+ e_3 \,,\,\quad f_4(\B{x},\B{e})=x_2 x_3 + e_4 
$$
and $\Prb_{\RN^4}$ is the standard-normal distribution on $\RN^4$. The graph of $\C{M}$ is depicted in
Figure~\ref{fig:SpirtesExampleAppendix} on the left. The model is uniquely solvable (it is even simple). One can check
that for every solution $\B{X}$ of $\C{M}$, $X_1$ is not independent of $X_2$ given $\{X_3,X_4\}$. However, the
variables $X_1$ and $X_2$ are $d$-separated given $\{X_3,X_4\}$ in $\C{G}(\C{M})$. Hence the
global directed Markov property does not hold here. 
\end{example}

In constraint-based approaches to causal discovery, one usually assumes the converse of the directed global Markov property to hold~\citep{SGS00,Pea09}.
  \begin{definition}[$d$-Faithfulness]
\label{def:dFaithfulness}
Let $\C{G}=(\C{V},\C{E},\C{B})$ be a directed mixed graph and $\Prb_{\C{V}}$ a probability distribution 
on $\BC{X}_{\C{V}}=\prod_{i\in\C{V}}\C{X}_i$, where each $\C{X}_i$ is a standard probability space. The probability distribution $\Prb_{\C{V}}$ is \emph{$d$-faithful} to $\C{G}$ if for all subsets 
$A,B,C\subseteq\C{V}$ we have
$$
  A \mathrel{\mathop{\sigmablocked}^{d}_{\C{G}}} B \given C \quad \Longleftarrow \quad
  \B{X}_A \mathrel{\mathop{\indep}_{\Prb_{\C{V}}}} \B{X}_B \given \B{X}_C \,,
$$
where we take the canonical projections $X_i:\BC{X}_{\C{V}}\to\C{X}_i$ as random variables.
\end{definition}
In other words, the $d$-faithfulness assumption states that the graph explains, via $d$-separation, all the conditional independencies that are present in the observational distribution. Meek~\citep{Mee95} showed that for multinomial and linear-Gaussian DAG (i.e., acyclic and causally sufficient SCMs) models, $d$-faithfulness holds for all parameter values up to a measure zero set (in a natural parameterization). 
%In general, $d$- or $\sigma$-faithfulness does not hold for SCMs. For example, consider solution variables $X_1$ and $X_4$ in the SCM of Example~\ref{ex:UniqueSolvabilityNotUniqAncestralSubset} which are $d$- and $\sigma$-separated but are not independent for every . 
%Up to our knowledge no such result has been shown for any subclass of SCMs that contains cycles, nor in more general acyclic settings. 
Up to our knowledge no such results have been shown in more general parametric or nonparametric settings (neither in the acyclic case, nor in the cyclic one).

%%%%%%%%%%%%%%%%%%%%%%%%%%%%%%%%%%%%%%%%%%%%%%%%%%
\subsubsection{The general directed global Markov property}
%%%%%%%%%%%%%%%%%%%%%%%%%%%%%%%%%%%%%%%%%%%%%%%%%%

In \citep{FM17} the general directed global Markov property is introduced, that is based on $\sigma$-separation, an extension of $d$-separation. This notion of $\sigma$-separation was derived from the notion of $d$-separation in the acyclification of the graph. The acyclification of a graph generalizes the idea of the collapsed graph for directed graphs, developed by Spirtes~\citep{Spi94}, to HEDGes. In particular, this notion can be applied to directed mixed graphs, and thus to the graphs of SCMs. The main idea of the acyclification is that under the condition that the SCM is uniquely solvable w.r.t.\ each strongly connected component, we can replace the causal mechanisms of these strongly connected components by their measurable solution functions, which results in an acyclic SCM. 
%%In \citep{FM17}, the notion of ``acyclification'' is introduced. This
%is an alternative to the ``collapsed graph'' for directed graphs of \citep{Spi94}\Stephan{, and extends to HEDGes. In particular, this notion can be applied to directed mixed graphs, and thus to the graphs of SCMs.} 
%%that is applicable to graphs of SCMs. 
%The main idea of the acyclification is that under the condition that the SCM is uniquely solvable w.r.t.\ each strongly connected 
%component, the causal mechanisms of these strongly connected components can be 
%replaced by their measurable solution functions, which results in an acyclic SCM.
This acyclification preserves the solutions, and $d$-separation in the
acyclification can directly be translated into $\sigma$-separation in
the original graph. This then leads to the general directed
global Markov property. We will discuss this now in more detail.

\begin{example}[Construction of an observationally equivalent acyclic SCM]
\label{ex:AcyclificationExampleAppendix}
Consider the SCM $\C{M}$ of Example~\ref{ex:SpirtesExampleAppendix} which is uniquely solvable w.r.t.\ all its strongly connected components, i.e., the subsets $\{1\}$, $\{2\}$ and $\{3,4\}$. Replacing the causal mechanisms of these strongly connected components by their measurable solution functions gives the SCM $\tilde{\C{M}}$ that is the same as $\C{M}$ except that its causal mechanism $\tilde{\B{f}}$ is given by
$$
\tilde{f}_1(\B{x},\B{e}):=e_1 ,\, \quad 
\tilde{f}_2(\B{x},\B{e}):=e_2 ,\, \quad 
\tilde{f}_3(\B{x},\B{e}):=\tfrac{x_1 e_4 + e_3}{1-x_1 x_2} ,\, \quad
\tilde{f}_4(\B{x},\B{e}):=\tfrac{x_2 e_3 + e_4}{1-x_1 x_2} \,.
$$
By construction, $\C{M}$ and $\tilde{\C{M}}$ are observationally equivalent. Because $\tilde{\C{M}}$ is acyclic (see
Figure~\ref{fig:SpirtesExampleAppendix} on the right) we can apply the directed global Markov property 
to $\tilde{\C{M}}$. The fact that $X_1$ and $X_2$ are not $d$-separated given $\{X_3,X_4\}$ in 
$\C{G}(\tilde{\C{M}})$ is in line with $X_1$ being dependent of $X_2$ given $\{X_3,X_4\}$
for every solution $\B{X}$ of $\tilde{\C{M}}$ (and hence of $\C{M}$).
\end{example}

One of the key insights in \citep{FM17} is that this example can easily be generalized
as follows.
\begin{definition}[Acyclification of an SCM]
\label{def:Acyclification}
Let $\C{M}=\langle \C{I}, \C{J}, \BC{X}, \BC{E}, \B{f},
\Prb_{\BC{E}} \rangle$ be an SCM  that is uniquely solvable w.r.t.\ each strongly connected
component of $\C{G}(\C{M})$. For each $i\in\C{I}$, let $g_i$ be the $i^{\text{th}}$ component of a measurable solution function $\B{g}_{\scc(i)}:\BC{X}_{\pa(\scc(i))\setminus\scc(i)}\times\BC{E}_{\pa(\scc(i))}\to\BC{X}_{\scc(i)}$
of $\C{M}$ w.r.t.\ $\scc(i)$, where $\pa$ and $\scc$ denote the parents and strongly connected components according
to $\C{G}^a(\C{M})$, respectively. We call the SCM
$\C{M}^{\acy} := \langle \C{I}, \C{J}, \BC{X}, \BC{E}, \hat{\B{f}},
\Prb_{\BC{E}} \rangle$ with the \emph{acyclified causal mechanism}
$\hat{\B{f}}:\BC{X}\times\BC{E}\to\BC{X}$ given by
$$
\hat{f}_i(\B{x},\B{e}) =
  g_i(\B{x}_{\pa(\scc(i))\setminus\scc(i)},\B{e}_{\pa(\scc(i))})\,,\quad i\in\C{I} \,,
$$
an \emph{acyclification of $\C{M}$}. We denote by $\acy(\C{M})$ the equivalence class of the acyclifications of $\C{M}$.
\end{definition}
%Note that the strongly connected components of $\C{G}(\C{M})$ form a partition
%of the set of endogenous variables $\C{I}$. 
Note that $\acy(\C{M})$ is well-defined: 
all acyclifications of an SCM $\C{M}$ belong to the
same equivalence class of SCMs. 
\begin{proposition}
\label{prop:Acyclification}
Let $\C{M}$ be an SCM that is uniquely solvable w.r.t.\ each strongly connected
component of $\C{G}(\C{M})$. Then an acyclification $\C{M}^{\acy}$ of $\C{M}$ is acyclic and
observationally equivalent to $\C{M}$. 
\end{proposition}

We can also define a graphical acyclification for directed mixed graphs, which is
a special case of the operation defined in \citep{FM17} for HEDGes.
\begin{definition}[Acyclification of a directed mixed graph]
\label{def:GraphicalAcyclification}
Let $\C{G} = (\C{V},\C{E},\C{B})$ be a directed mixed graph. The \emph{acyclification
of $\C{G}$} maps $\C{G}$ to the \emph{acyclified graph} $\C{G}^{\acy} := (\C{V},\hat{\C{E}},\hat{\C{B}})$ with
directed edges $j \to i \in \hat{\C{E}}$ if and only if $j \in \pa_{\C{G}}(\scc_{\C{G}}(i))\setminus\scc_{\C{G}}(i)$
and bidirected edges $i \oto j \in \hat{\C{B}}$ if and only if there exist $i' \in \scc_{\C{G}}(i)$ and $j' \in \scc_{\C{G}}(j)$ with $i'=j'$ or $i' \oto j' \in \C{B}$.
\end{definition}
The following compatibility result is immediate from the definitions.
\begin{proposition}
\label{prop:CompatibilityAcyclification}
Let $\C{M}$ be an SCM that is uniquely solvable w.r.t.\ each strongly connected
component of $\C{G}(\C{M})$. Then $\C{G}^a(\acy(\C{M})) \subseteq \acy(\C{G}^a(\C{M}))$ and
$\C{G}(\acy(\C{M})) \subseteq \acy(\C{G}(\C{M}))$.
\end{proposition}

The following example illustrates that the graph of the acyclification of an SCM can be a strict subgraph of the acyclification of the graph of the SCM.
\begin{example}[Graph of the acyclification of the SCM is a strict subgraph of the acyclification of its graph]
\label{ex:AcyclificationSubgraph}
Consider the SCM $\C{M}=\langle \B{2}, \B{1}, \RN^2, \RN, \B{f}, \Prb_{\RN} \rangle$ with the causal mechanism defined by
$$
f_1(\B{x},e) = x_2 - e \,,\quad f_2(\B{x},e) = \tfrac{1}{2} x_1 + e
$$
and $\Prb_{\RN}$ the standard Gaussian measure on $\RN$. The SCM $\C{M}$ is uniquely solvable w.r.t.\ the (only) strongly connected component $\{1,2\}$. An acyclification of $\C{M}$ is the acyclified SCM $\C{M}^{\acy}$ with the acyclified causal mechanism $\hat{\B{f}}$ defined by
$$
\hat{f}_1(\B{x},e) = 0 \,,\quad \hat{f}_2(\B{x},e) =  e \,.
$$
The graph $\C{G}(\acy(\C{M}))$ is a strict subgraph of $\acy(\C{G}(\C{M}))$ as can be seen in Figure~\ref{fig:AcyclificationSubgraph}.
\begin{figure}
\adjustbox{scale=0.8,center}{%
\begin{tikzpicture}
  \begin{scope}
  \node[var] (X1) at (0,0) {$X_1$};
  \node[var] (X2) at (1.75,0) {$X_2$};
  \node at (0.75,-.75) {$\C{G}(\C{M})$};
  \draw[arr,bend left=15] (X1) to (X2);
  \draw[arr,bend left=15] (X2) to (X1);
  \draw[biarr] (X1) to [bend left=45] (X2);
  \end{scope}
  \begin{scope}[shift={(4,0)}]
  \node[var] (X1) at (0,0) {$X_1$};
  \node[var] (X2) at (1.75,0) {$X_2$};
  \node at (0.75,-.75) {$\C{G}(\acy(\C{M}))$};
  \end{scope}
  \begin{scope}[shift={(8,0)}]
  \node[var] (X1) at (0,0) {$X_1$};
  \node[var] (X2) at (1.75,0) {$X_2$};
  \draw[biarr] (X1) to [bend left=45] (X2);
  \node at (0.75,-.75) {$\acy(\C{G}(\C{M}))$};
  \end{scope}
\end{tikzpicture}}
\caption{The graphs of the original SCM $\C{M}$ (left), of the acyclified SCM (center), and of the acyclification of the graph of $\C{M}$ (right) corresponding to Example~\ref{ex:AcyclificationSubgraph}.}
\label{fig:AcyclificationSubgraph}
\end{figure}
\end{example}

Translating the notion of $d$-separation from the acyclified graph 
back to the original graph led to the notion of $\sigma$-separation.
\begin{definition}[$\sigma$-separation~\citep{FM17}]
\label{def:SigmaSeparation}
Let $\C{G}=(\C{V},\C{E},\C{B})$ be a directed mixed graph and let
$C\subseteq\C{V}$ be a subset of nodes. 
A walk (path) $\pi = (i_0, \epsilon_1, i_1, \dots, i_n)$ in $\C{G}$ is said to be \emph{$C$-$\sigma$-blocked or $\sigma$-blocked by $C$} if
\begin{enumerate}
  \item its first node $i_0 \in C$ or its last node $i_n \in C$, or
  \item it contains a collider $i_k \notin \an_{\C{G}}(C)$, or
  \item it contains a non-endpoint non-collider $i_k \in C$ that points towards a neighboring node on $\pi$
  that lies in a different strongly connected component of $\C{G}$, that is, such that $i_{k-1} \ot i_k$ in $\pi$
  and $i_{k-1} \notin \scc_{\C{G}}(i_k)$, or $i_{k} \to i_{k+1}$ in $\pi$ and
  $i_{k+1} \notin \scc_{\C{G}}(i_k)$.
\end{enumerate}
The walk (path) $\pi$ is said to be \emph{$C$-$\sigma$-open} if it is not $\sigma$-blocked by $C$.
For two subsets of nodes $A, B \subseteq \C{V}$, we say that \emph{$A$ is $\sigma$-separated from $B$ given $C$ in $\C{G}$} 
if all paths between any node in $A$ and any node in $B$ are $\sigma$-blocked by $C$, and write
$$
A \mathrel{\mathop{\sigmablocked}^{\sigma}_{\C{G}}} B \given C \,.
$$
\end{definition}
The only difference between $\sigma$-separation and $d$-separation is that $d$-separation does not have the
extra condition on the non-collider that it has to point to a node in a different
strongly connected component. It is therefore obvious that $\sigma$-separation reduces to
$d$-separation for acyclic graphs, since $\scc_{\C{G}}(i) = \{i\}$ for each $i\in\C{V}$ in that case.

Although for proofs it is often easier to make use of walks, it suffices to formulate $\sigma$-separation in
term of paths rather than walks because of the following result, which is analogous to a similar result
for $d$-separation (see Lemma~\ref{lemm:DWalksPaths}).
\begin{lemma}
\label{lemm:SigmaWalksPaths}
Let $\C{G}=(\C{V},\C{E},\C{B})$ be a directed mixed graph, $C\subseteq\C{V}$
and $i, j\in\C{V}$. There exists a $C$-$\sigma$-open walk between $i$ and
$j$ in $\C{G}$ if and only if there exists a $C$-$\sigma$-open path between $i$ and $j$ in $\C{G}$.
\end{lemma}

It is clear from the definitions that $\sigma$-separation implies $d$-separation. The other way around does not hold in general, as can be seen in the following example.
\begin{example}[$d$-separation does not imply $\sigma$-separation]
\label{ex:SpirtesExampleSigmaWeakerD}
Consider the directed graph $\C{G}$ as depicted in Figure~\ref{fig:SpirtesExampleAppendix} (left). Here 
$X_1$ is $d$-separated from $X_2$ given $\{X_3,X_4\}$, but $X_1$ is not $\sigma$-separated from 
$X_2$ given $\{X_3,X_4\}$.
\end{example}

The following result in \citep{FM17} relates $\sigma$-separation to $d$-separation.
\begin{proposition}
\label{prop:SigmaSeparationAsDSeparation}
Let $\C{G} = (\C{V},\C{E},\C{B})$ be a directed mixed graph. Then for $A,B,C \subseteq \C{V}$,
$$
  A \mathrel{\mathop{\sigmablocked}^{\sigma}_{\C{G}}} B \given C \iff
  A \mathrel{\mathop{\sigmablocked}^{d}_{\acy(\C{G})}} B \given C \,.
$$
\end{proposition}
%\begin{example}
%\label{ex:DSepDoesNotImplySigmaSep}
%Consider the directed graph $\C{G}$ as depicted in Figure~\ref{fig:DSepDoesNotImplySigmaSep}. Here 
%$X_1$ is $d$-separated from $X_3$ given $\{X_2,X_4\}$, but $X_1$ is not $\sigma$-separated from 
%$X_3$ given $\{X_2,X_4\}$.
%\Joris{Maybe better to instead point out the difference for the graph in Figure~\ref{fig:SpirtesExample} instead!
%That saves space and makes the storyline more coherent.}
%\begin{figure}
%  \begin{center}
%    \begin{tikzpicture}
%      \node[var] (X1) at (0,0) {$X_1$};
%      \node[var] (X2) at (1.25,0) {$X_2$};
%      \node[var] (X3) at (1.25,-1.25) {$X_3$};
%      \node[var] (X4) at (0,-1.25) {$X_4$};
%      \draw[arr] (X1) -- (X2);
%      \draw[arr] (X2) -- (X3);
%      \draw[arr] (X3) -- (X4);
%      \draw[arr] (X4) -- (X1);
%    \end{tikzpicture}
%  \end{center}
%  \caption{A directed graph $\C{G}$ for which $d$-separation does not imply $\sigma$-separation (see
%Example~\ref{ex:DSepDoesNotImplySigmaSep}).}
%  \label{fig:DSepDoesNotImplySigmaSep}
%\end{figure}
%\end{example}

%In Example~\ref{ex:AcyclificationExample} one can easily verify that $X_1$ and
%$X_2$ are not $\sigma$-separated given $X_3, X_4$ in the graph
%$\C{G}(\C{M})$. 
By replacing in Definition~\ref{def:Markov_property} ``$d$-separation'' by ``$\sigma$-separation'', one obtains the formulation of
what Forr{\'e} and Mooij~\cite{FM17} termed the general directed global Markov property.
\begin{definition}[General directed global Markov property~\citep{FM17}]
\label{def:generalized_Markov_property}
Let $\C{G}=(\C{V},\C{E},\C{B})$ be a directed mixed graph and $\Prb_{\C{V}}$ a probability distribution 
on $\BC{X}_{\C{V}}=\prod_{i\in\C{V}}\C{X}_i$, where each $\C{X}_i$ is a standard probability space. The probability distribution $\Prb_{\C{V}}$ satisfies 
the \emph{general directed global Markov property} relative to $\C{G}$ if for all subsets 
$A,B,C\subseteq\C{V}$ we have
$$
  A \mathrel{\mathop{\sigmablocked}^{\sigma}_{\C{G}}} B \given C \quad \Longrightarrow \quad
  \B{X}_A \mathrel{\mathop{\indep}_{\Prb_{\C{V}}}} \B{X}_B \given \B{X}_C \,,
$$
that is, $(X_i)_{i\in A}$ and $(X_i)_{i\in B}$ are conditionally
independent given $(X_i)_{i\in C}$ under $\Prb_{\C{V}}$, 
where we take the canonical projections $X_i:\BC{X}_{\C{V}}\to\C{X}_i$ as random variables.
\end{definition}
The fact that $\sigma$-separation implies $d$-separation means that the directed global Markov
property implies the general directed global Markov property. In other words, the general
directed global Markov property is weaker than the directed global Markov property.
It is actually strictly weaker, as we saw in Example~\ref{ex:SpirtesExampleSigmaWeakerD}.

The following fundamental result, also known as the $\sigma$-separation criterion, follows directly from the theory in \citep{FM17}.
\begin{theorem}[General directed global Markov property for SCMs]
\label{thm:gdgMarkovPropertySCM}
Let $\C{M}$ be an SCM that is uniquely solvable w.r.t.\ each strongly connected component of
$\C{G}(\C{M})$. Then its observational distribution $\Prb^{\B{X}}$ exists, is
unique and it satisfies the general directed global Markov property relative to
$\C{G}(\C{M})$.\footnote{Since \citep{FM17} also provides results under the weaker condition that an SCM is solvable
(not necessarily uniquely) w.r.t.\ each strongly connected component of $\C{G}(\C{M})$, one might
believe that Theorem~\ref{thm:gdgMarkovPropertySCM} could be generalized to stating that in that case,
any of its observational distributions satisfies the general directed global
Markov property. However, that is not true:
consider for example the SCM $\C{M} = \langle \B{2}, \emptyset, \RN^2, \B{1}, \B{f},
\Prb_{\B{1}} \rangle$ with $f_1(\B{x})=x_1$ and $f_2(\B{x})=x_2$. Then $\C{M}$ is solvable
w.r.t.\ each of its strongly connected components $\{1\}$ and $\{2\}$. The solution with $X_1=X_2$ shows a dependence between $X_1$ and $X_2$ and thus $X_1 \indep X_2$ does not hold. In general, all strongly connected components that admit multiple solutions may be dependent on any other variable(s) in the model.
}
%\footnote{Since \citep{FM17} also provides results under the weaker condition that an SCM is solvable
%(not necessarily uniquely) w.r.t.\ each strongly connected component of $\C{G}(\C{M})$, one might
%believe that Theorem~\ref{thm:gdgMarkovPropertySCM} could be generalized to stating that in that case,
%any of its observational distributions satisfies the general directed global
%Markov property. However, that is not true:
%consider for example the SCM $\C{M} = \langle \B{2}, \emptyset, \RN^2, \B{1}, \B{f},
%\Prb_{\B{1}} \rangle$ with $f_1(\B{x})=x_1$ and $f_2(\B{x})=x_2$. Then $\C{M}$ is solvable
%w.r.t.\ each of its strongly connected components $\{1\}$ and $\{2\}$. Since the graph contains no edges, one might na\"ively
%expect that for every solution $\B{X}$ we have $X_1 \indep X_2$. However, that is
%clearly not true: the solution with $X_1=X_2$ shows
%a dependence between $X_1$ and $X_2$. In general, all strongly connected components
%that admit multiple
%solutions may be dependent on any other variable(s) in the model.}
\end{theorem}
The proof is based on the reasoning that, for $A, B, C \subseteq \C{I}$, if $A$
is $\sigma$-separated from $B$ given $C$ in $\C{G}(\C{M})$, then $A$
is $d$-separated from $B$ by $C$ in $\acy(\C{G}(\C{M}))$ and hence in $\C{G}(\acy(\C{M}))$, and since
$\acy(\C{M})$ is acyclic and observationally equivalent to $\C{M}$, it follows from the directed global
Markov property applied to $\acy(\C{M})$ that 
$\B{X}_A \mathrel{\mathop{\indep}_{\Prb^{\B{X}}}} \B{X}_B \given \B{X}_C$ for every solution $\B{X}$ of $\C{M}$. 
Note that the ancestral unique solvability condition for the discrete case is
strictly weaker than the condition of unique solvability w.r.t.\ each strongly
connected component in Theorem~\ref{thm:gdgMarkovPropertySCM}. For the linear
case, the condition of unique solvability is equivalent to the condition of
unique solvability w.r.t.\ each strongly connected component (see Proposition~\ref{prop:LinearEquivalentUniqueSolvability}).

The results in Theorems~\ref{thm:dgMarkovPropertySCMThreeSpecialCases} and \ref{thm:gdgMarkovPropertySCM} are not preserved under perfect intervention, 
because intervening on a strongly connected component could split 
it into several strongly connected components with different solvability properties. As the class of simple SCMs is preserved under perfect intervention and the twin operation (Proposition~\ref{prop:SimplenessClosedUnderResults}), 
%and a simple SCM
%is uniquely solvable w.r.t.\ every subset of its endogenous variables
%(Proposition~\ref{prop:SimpleSCMiffUniquelySolvableAfterAnyIntervention}), 
we
obtain the following corollary.
\begin{corollary}[Global Markov properties for simple SCMs]
\label{coro:gdgMarkovPropertyInterventionalSCM}
Let $\C{M}$ be a simple SCM. Then the:
\begin{enumerate}
  \item observational distribution,
  \item interventional distribution after perfect intervention on
    $I\subset\C{I}$, 
  \item counterfactual distribution after perfect intervention on
   $\tilde{I}\subseteq\C{I}\cup\C{I}'$,
 \end{enumerate}
all exist, are unique and satisfy the 
general directed global Markov property relative to $\C{G}(\C{M})$,
$\intervene(I)(\C{G}(\C{M}))$ and $\intervene(\tilde{I})(\twin(\C{G}(\C{M})))$, respectively. Moreover, if $\C{M}$ satisfies at least one of the three
conditions (1), (2), (3) of
Theorem~\ref{thm:dgMarkovPropertySCMThreeSpecialCases}, then they also satisfies the directed global Markov property relative to $\C{G}(\C{M})$,
$\intervene(I)(\C{G}(\C{M}))$ and $\intervene(\tilde{I})(\twin(\C{G}(\C{M})))$, respectively.
\end{corollary}

Similar to $d$-faithfulness, $\sigma$-faithfulness\footnote{In~\citep{Ric96c} it is called ``collapsed graph faithfulness''.} is defined as follows.
\begin{definition}[$\sigma$-Faithfulness]
\label{def:sigmaFaithfulness}
Let $\C{G}=(\C{V},\C{E},\C{B})$ be a directed mixed graph and $\Prb_{\C{V}}$ a probability distribution 
on $\BC{X}_{\C{V}}=\prod_{i\in\C{V}}\C{X}_i$, where each $\C{X}_i$ is a standard probability space. The probability distribution $\Prb_{\C{V}}$ is \emph{$\sigma$-faithful} to $\C{G}$ if for all subsets 
$A,B,C\subseteq\C{V}$ we have
$$
  A \mathrel{\mathop{\sigmablocked}^{\sigma}_{\C{G}}} B \given C \quad \Longleftarrow \quad
  \B{X}_A \mathrel{\mathop{\indep}_{\Prb_{\C{V}}}} \B{X}_B \given \B{X}_C \,,
$$
where we take the canonical projections $X_i:\BC{X}_{\C{V}}\to\C{X}_i$ as random variables.
\end{definition}
In other words, the graph explains, via $\sigma$-separation, all the conditional independencies that are present in the observational distribution. 
%Similar to $d$-faithfulness, $\sigma$-faithfulness does not always hold for SCMs in general. One can easily check that the SCM of Example~\ref{ex:UniqueSolvabilityNotUniqAncestralSubset} is not $\sigma$-faithful. 
Although it has been conjectured~\citep{Spi95} that under certain conditions $\sigma$-faithfulness should hold, formulating and proving such completeness results is an open problem to the best of our knowledge.

%%%%%%%%%%%%%%%%%%%%%%%%%%%%%%%%%%%%%%%%%%%%%%%%%%
\subsection{Modular SCMs}
\label{app:ModularSCMs}
%%%%%%%%%%%%%%%%%%%%%%%%%%%%%%%%%%%%%%%%%%%%%%%%%%

In this subsection, we relate the class of (simple) SCMs to that of modular SCMs. Modular SCMs introduced by Forr{\'e} and Mooij~\citep{FM17} are causal graphical models on which marginalizations and interventions are defined and they satisfy the general directed global Markov property. 
%We show that the solutions of a modular SCM can always be described by an SCM. Intuitively, a modular SCM can be seen as an SCM together with an additional structure of a compatible system of solution functions. In particular, for simple SCMs one can always construct such a compatible system of solution functions. 
For a comprehensive account on modular SCMs we refer the reader to \citep{FM17}.

%We show that the solutions of a modular SCM can always be described by its underlying SCM. We show that the underlying SCM of a modular SCM always has a compatible system of solution functions. Intuitively, a modular SCM can be seen as an SCM together with this additional structure of a compatible system of solution functions. In particular, for simple SCMs one can always construct such a compatible system of solution functions. For a comprehensive account of modular SCMs we refer the reader to \citep{FM17}.

%We show that an SCM with an additional structure of ``compatible system of solution functions'' defines a modular SCM. In particular, we show that a simple SCM defines a unique modular SCM (up to an exogenous representation and a $\Prb_{\BC{E}}$-null set). For a comprehensive account of modular SCMs we refer the reader to \citep{FM17}.

%%%%%%%%%%%%%%%%%%%%%%%%%%%%%%%%%%%%%%%%%%%%%%%%%%
\subsubsection{Definition of a modular SCM}
%%%%%%%%%%%%%%%%%%%%%%%%%%%%%%%%%%%%%%%%%%%%%%%%%%

In contrast to an SCM from which a graph can be derived, a modular SCM is defined in terms of a graphical object, which Forr{\'e} and Mooij~\citep{FM17} call a directed graph with hyperedges (HEDG). The hyperedges of a HEDG are described in terms of a simplicial complex.
\begin{definition}[Simplicial complex]
Let $\C{V}$ be a finite set. A \emph{simplicial complex} $\C{H}$ over $\C{V}$ is a set of subsets of $\C{V}$ such that
\begin{enumerate}
\item all single element sets $\{v\}$ are in $\C{H}$ for $v\in\C{V}$, and
\item if $\C{F}\in\C{H}$, then also all subsets $\tilde{\C{F}}\subseteq F$ are elements of $\C{H}$.
\end{enumerate}
\end{definition}

\begin{definition}[Directed graph with hyperedges (HEDGes)~\citep{FM17}]
A \emph{directed graph with hyperedges (HEDG)} is a triple $\C{G} = (\C{V}, \C{E}, \C{H})$, where $(\C{V},\C{E})$ is a directed graph and $\C{H}$ a simplicial complex over the set of nodes $\C{V}$. The elements $\C{F}$ of $\C{H}$ are called \emph{hyperedges} of $\C{G}$. The elements $\C{F}$ of $\C{H}$ that are inclusion-maximal elements of $\C{H}$ are called \emph{maximal hyperedges} and are denoted by $\hat{\C{H}}$.
\end{definition}
A HEDG $\C{G} = (\C{V},\C{E}, \C{H})$ can be represented as a directed graph $\bar{\C{G}}:=(\C{V},\C{E})$ consisting of nodes $\C{V}$ and directed edges $\C{E}$, with additional maximal hyperedges $\C{F}\in\hat{\C{H}}$ with $|\C{F}| \geq 2$ (i.e., not corresponding to single element sets $\{v\}\in\hat{\C{H}}$), that point to their target nodes $v\in\C{F}$. For a HEDG $\C{G}$, we define $\pa_{\C{G}}$, $\ch_{\C{G}}$, etc., in terms of the underlying directed graph $\bar{\C{G}}$, that is, $\pa_{\bar{\C{G}}}$, $\ch_{\bar{\C{G}}}$, etc., respectively. 

A \emph{loop} in a HEDG $\C{G}=(\C{V}, \C{E}, \C{H})$ is a subset $\C{O}\subseteq\C{V}$ that is a loop in the underlying directed graph $\bar{\C{G}}=(\C{V},\C{E})$. In other words, a loop of $\C{G}$ is a set of nodes $\C{O}\subseteq\C{V}$ such that for every two nodes $v,w\in\C{O}$ there are directed paths $v\to\cdots\to w$ and $w \to\cdots\to v$ in $\C{G}$ for which all the intermediate nodes lie in $\C{O}$ (if any exist). In particular, a loop may consist of a single element $\{v\}$ for $v\in\C{V}$. The set of loops in $\C{G}$ is denoted by $\C{L}(\C{G})$.

In order to define a modular SCM one needs the notion of a compatible system of solution functions, which assigns to each loop a separate solution function such that all these solution functions are ``compatible'' with each other.
\begin{definition}[Compatible system of solution functions\protect\footnote{We deviate from the terminology in \citep{FM17} where this is called a ``compatible system of structural equations''.}]
Let $\C{G} = (\C{V},\C{E},\C{H})$ be a HEDG. For every $v\in\C{V}$ and maximal hyperedge $\C{F}$ in $\hat{\C{H}}$, let $\C{X}_v$ and $\C{E}_{\C{F}}$ be standard measurable spaces. For a subset $\C{O}\subseteq\C{V}$ we define\footnote{We use the ``hat'' notation $\widehat{\BC{E}}_{\C{O}}$ to distinguish it from the ordinary subscript convention that $\BC{E}_{\C{O}}=\prod_{\C{F}\in\C{O}}\C{E}_{\C{F}}$ for some subset $\C{O}\subseteq\hat{\C{H}}$.}
$$
  \BC{X}_{\C{O}} := \prod_{v\in\C{O}}\C{X}_v \quad\text{and}\quad 
  \widehat{\BC{E}}_{\C{O}} := \prod_{\substack{ \C{F}\in\hat{\C{H}} \\ \C{F}\cap\C{O} \neq\emptyset}}\C{E}_{\C{F}} \,.
$$
Consider a family of measurable mappings $(\B{g}_{\C{O}})_{\C{O}\in\C{L}(\C{G})}$ indexed by $\C{L}(\C{G})$ which are of the form
$$
\B{g}_{\C{O}} : \BC{X}_{\pa_{\C{G}}(\C{O})\setminus\C{O}}\times\widehat{\BC{E}}_{\C{O}} \to \BC{X}_{\C{O}} \,.
$$
We call the family of measurable mappings $(\B{g}_{\C{O}})_{\C{O}\in\C{L}(\C{G})}$ a \emph{compatible system of solution functions}, if for all $\C{O},\tilde{\C{O}}\in\C{L}(\C{G})$ with $\tilde{\C{O}}\subseteq\C{O}$ and for all $\widehat{\B{e}}_{\C{O}}\in\widehat{\BC{E}}_{\C{O}}$ and $\B{x}_{\pa_{\C{G}}(\C{O})\cup\C{O}}\in\BC{X}_{\pa_{\C{G}}(\C{O})\cup\C{O}}$ we have
$$
\B{x}_{\C{O}} = \B{g}_{\C{O}}(\B{x}_{\pa_{\C{G}}(\C{O})\setminus\C{O}},\widehat{\B{e}}_{\C{O}}) 
  \quad\implies\quad
  \B{x}_{\tilde{\C{O}}} = \B{g}_{\tilde{\C{O}}}(\B{x}_{\pa_{\C{G}}(\tilde{\C{O}})\setminus\tilde{\C{O}}},\widehat{\B{e}}_{\tilde{\C{O}}}) \,.
$$
%where $\widehat{\B{e}}_{\tilde{\C{O}}}$, $\B{x}_{\tilde{\C{O}}}$, etc.\ are the corresponding components of $\widehat{\B{e}}_{\C{O}}$, $\B{x}_{\C{O}}$, etc.\ respectively.
\end{definition}

This structure of a compatible system of solution functions is at the heart of the defnition of a modular SCM.
\begin{definition}[Modular structural causal model (mSCM)~\citep{FM17}]
A \emph{modular structural causal model (mSCM)} is a tuple
$$
  \widehat{\C{M}} := \langle \C{G}, \BC{X}, \BC{E}, (\B{g}_{\C{O}})_{\C{O}\in\C{L}(\C{G})}, \Prb_{\BC{E}} \rangle \,,
$$
where
\begin{enumerate}
\item $\C{G}=(\C{V},\C{E},\C{H})$ is a HEDG,
\item $\BC{X} = \prod_{v\in\C{V}} \C{X}_v$ is the product of standard measurable spaces $\C{X}_v$,
\item $\BC{E} = \prod_{\C{F}\in\hat{\C{H}}} \C{E}_{\C{F}}$ is the product of standard measurable spaces $\C{E}_{\C{F}}$,
\item $(\B{g}_{\C{O}})_{\C{O}\in\C{L}(\C{G})}$ is a compatible system of solution functions,
\item $\Prb_{\BC{E}}=\prod_{\C{F}\in\hat{\C{H}}}\Prb_{\C{E}_{\C{F}}}$ is a product measure, where $\Prb_{\C{E}_{\C{F}}}$ is a probability measure on $\C{E}_{\C{F}}$ for each $\C{F}\in\hat{\C{H}}$.
\end{enumerate}
\end{definition}

Let $\widehat{\C{M}} = \langle \C{G}, \BC{X}, \BC{E}, (\B{g}_{\C{O}})_{\C{O}\in\C{L}(\C{G})}, \Prb_{\BC{E}} \rangle$ be a modular SCM and $\C{O}_1,\dots,\C{O}_r\in\C{L}(\C{G})$ the strongly connected components of $\C{G}$ ordered according to a topological order of the DAG of strongly connected components of $\C{G}$. Then for any random variable $\B{E}:\Omega\to\BC{E}$ such that $\Prb^{\B{E}}=\Prb_{\BC{E}}$ one can inductively define the random variables $X_v:=(\B{g}_{\C{O}_i})_v(\B{X}_{\pa_{\C{G}}(\C{O}_i)\setminus\C{O}_i},\widehat{\B{E}}_{\C{O}_i})$ for all $v\in\C{O}_i$ for all $i\geq 1$, starting at $X_v := (\B{g}_{\C{O}_1})_v(\widehat{\B{E}}_{\C{O}_1})$ for all $v\in\C{O}_1$. Because $(\B{g}_{\C{O}})_{\C{O}\in\C{L}(\C{G})}$ is a compatible system of solution functions, we have for every $\C{O}\in\C{L}(\C{G})$
$$
\B{X}_{\C{O}}=\B{g}_{\C{O}}(\B{X}_{\pa_{\C{G}}(\C{O})\setminus\C{O}},\widehat{\B{E}}_{\C{O}}) \,.
$$
We call the random variable $\B{X}$ a \emph{solution} of the modular SCM $\widehat{\C{M}}$. Note that the solution $\B{X}$ depends on the choice of the random variable $\B{E}:\Omega\to\BC{E}$.

The causal semantics of modular SCMs can be defined in terms of perfect interventions, which is defined as follows.
\begin{definition}[Perfect intervention on an mSCM]
\label{def:intervenedModularSCM}
Consider a modular SCM $\widehat{\C{M}} = \langle \C{G}, \BC{X}, \BC{E}, (\B{g}_{\C{O}})_{\C{O}\in\C{L}(\C{G})}, \Prb_{\BC{E}} \rangle$, a subset $I \subseteq \C{V}$ of endogenous variables and a value $\B{\xi}_I \in \BC{X}_I$. The \emph{perfect intervention 
$\intervene(I, \B{\xi}_I)$} maps $\widehat{\C{M}}$ to the modular SCM 
$$
  \widehat{\C{M}}_{\intervene(I, \B{\xi}_I)} := \langle \C{G}^{\intervene}, \BC{X}, \BC{E}^{\intervene}, (\B{g}_{\C{O}}^{\intervene})_{\C{O}\in\C{L}(\C{G}^{\intervene})}, \Prb_{\BC{E}^{\intervene}} \rangle \,,
$$
where 
\begin{enumerate}
  \item $\C{G}^{\intervene} = (\C{V}, \C{E}^{\intervene}, \C{H}^{\intervene})$, where 
    $$
    \C{E}^{\intervene}=\C{E}\setminus\{v\to w \,:\, v\in\C{V}, w\in I\}
    $$
    $$
    \C{H}^{\intervene}=\{\C{F}\setminus I \,:\, \C{F}\in\C{H}\}\cup \{\{v\} \,:\, v\in I\} \,,
    $$
\item $\phi: \{\C{F}\in\hat{\C{H}} \,:\, \C{F}\setminus I \neq \emptyset\} \to \hat{\C{H}}^{\intervene}\setminus\{\{v\} \,:\, v\in I\}$ is a mapping such that $\phi(\C{F})\supseteq \C{F}\setminus I$ for all $\C{F}\in\hat{\C{H}}$ for which $\C{F}\setminus I \neq \emptyset$,
  \item $\BC{E}^{\intervene}=\prod_{\tilde{\C{F}}\in\hat{\C{H}}^{\intervene}}\C{E}^{\intervene}_{\tilde{\C{F}}}$, where
    $$
    \C{E}^{\intervene}_{\tilde{\C{F}}} = \begin{cases}
      \C{X}_v & \text{if $\tilde{\C{F}}=\{v\}$ for $v\in I$} \\
      \prod_{\C{F}=\phi^{-1}(\tilde{\C{F}})} \C{E}_{\C{F}} &\text{if }\tilde{\C{F}}\in\hat{\C{H}}^{\intervene}\setminus\{\{v\} \,:\, v\in I\} \,,
    \end{cases}
    $$
  \item for every $\C{O}\in\C{L}(\C{G}^{\intervene})$ 
    $$
    \B{g}^{\intervene}_{\C{O}}= \begin{cases}
      \mathbb{I}_{\{v\}} &\text{if $\C{O}=\{v\}$ for $v\in I$} \\
      \B{g}_{\C{O}} &\text{otherwise,}
      \end{cases}
    $$
    (note that if $\C{O}$ is a loop in $\C{G}^{\intervene}$, then it is a loop in $\C{G}$),
  \item $\Prb_{\BC{E}^{\intervene}}=\prod_{\tilde{\C{F}}\in\hat{\C{H}}^{\intervene}}\Prb_{\C{E}^{\intervene}_{\tilde{\C{F}}}}$, where
      $$
      \Prb_{\C{E}^{\intervene}_{\tilde{\C{F}}}} = \begin{cases}
        \delta_{\xi_v} & \text{if $\tilde{\C{F}}=\{v\}$ for $v\in I$} \\
        \prod_{\C{F}=\phi^{-1}(\tilde{\C{F}})} \Prb_{\C{E}_{\C{F}}} &\text{if }\tilde{\C{F}}\in\hat{\C{H}}^{\intervene}\setminus\{\{v\} \,:\, v\in I\} \,.
      \end{cases}
      $$
\end{enumerate}
\end{definition}
In contrast to SCMs, these perfect interventions on modular SCMs are directly defined on the underlying HEDG and depend on the choice of the mapping $\phi$.

%%%%%%%%%%%%%%%%%%%%%%%%%%%%%%%%%%%%%%%%%%%%%%%%%%
\subsubsection{Relation between SCMs and modular SCMs}
%%%%%%%%%%%%%%%%%%%%%%%%%%%%%%%%%%%%%%%%%%%%%%%%%%

The solutions of a modular SCM can be described by an SCM that is loop-wisely solvable.
\begin{definition}[Underlying SCM]
  Let $\widehat{\C{M}} = \langle \C{G}, \BC{X}, \BC{E}, (\B{g}_{\C{O}})_{\C{O}\in\C{L}(\C{G})}, \Prb_{\BC{E}} \rangle$ be a modular SCM. Then the mapping $\iota$ maps $\widehat{\C{M}}$ to the \emph{underlying SCM} $\tilde{\C{M}} := \langle \tilde{\C{I}}, \tilde{\C{J}}, \tilde{\BC{X}}, \tilde{\BC{E}}, \tilde{\B{f}}, \Prb_{\tilde{\BC{E}}} \rangle$, where
\begin{enumerate}
  \item $\tilde{\C{I}} = \C{V}$,
  \item $\tilde{\C{J}} = \hat{\C{H}}$,
  \item $\tilde{\BC{X}} = \BC{X}$,
  \item $\tilde{\BC{E}} = \BC{E}$,
  \item $\tilde{\B{f}}$ is given by $\tilde{f}_v=(\B{g}_{\{v\}})_v$ for all $v\in\C{V}$,
  \item $\Prb_{\tilde{\BC{E}}}=\Prb_{\BC{E}}$.
\end{enumerate}
\end{definition}
Every solution $\B{X}$ of a modular SCM $\widehat{\C{M}}$ is also a solution of the underlying SCM $\iota(\widehat{\C{M}})$.

Observe that for the modular SCM $\widehat{\C{M}}$ we have that the induced subgraph $\C{G}^a(\iota(\widehat{\C{M}}))_{\tilde{\C{I}}}$, of the augmented graph of the underlying SCM $\C{G}^a(\iota(\widehat{\C{M}}))$ on $\tilde{\C{I}}$, is a subgraph of the underlying HEDG $\C{G}$, that is, $\C{G}^a(\iota(\widehat{\C{M}}))_{\tilde{\C{I}}}\subseteq\C{G}$. This implies that, in general, the underlying HEDG $\C{G}$ of $\widehat{\C{M}}$ may have more loops than the loops in $\C{G}(\iota(\widehat{\C{M}}))$. For a subset $\C{O}\subseteq\tilde{\C{I}}$, we have for the exogenous parents of the underlying SCM $\iota(\widehat{\C{M}})$
$$
\pa(\C{O})\cap\tilde{\C{J}} \subseteq \{ \C{F}\in\tilde{\C{J}} \,:\, \C{F}\cap\C{O} \neq\emptyset \} \,,
$$
where $\pa(\C{O})$ denotes the set of parents of $\C{O}$ in $\C{G}^a(\iota(\widehat{\C{M}}))$. Hence, in general, not all the hyperedges $\C{F}\in\C{H}$ such that $|\C{F}|=2$ (i.e., bidirected edges) are in the set of bidirected edges $\C{B}$ of the graph of the underlying SCM $\C{G}(\iota(\widehat{\C{M}}))=(\C{V},\C{E},\C{B})$. We conclude that the graph of the underlying SCM is, in general, a sparser graph than the HEDG of the modular SCM.

Next, we show that the compatible system of solution functions of a modular SCM induces a compatible system of solution functions on the underlying SCM. For this we need the notion of loop-wise solvability for SCMs.
\begin{definition}[Loop-wise (unique) solvability for SCMs]
We call an SCM $\C{M}$
\begin{enumerate}
  \item \emph{loop-wisely solvable}, if $\C{M}$ is solvable w.r.t.\ every loop $\C{O}\in\C{L}(\C{G}(\C{M}))$, and
    \item \emph{loop-wisely uniquely solvable}, if $\C{M}$ is uniquely solvable w.r.t.\ every loop $\C{O}\in\C{L}(\C{G}(\C{M}))$.
  \end{enumerate}
\end{definition}

\begin{definition}[Compatible system of solution functions for SCMs]
For a loop-wisely solvable SCM $\C{M}$, we call a family of measurable solution functions $(\B{g}_{\C{O}})_{\C{O}\in\C{L}(\C{G}(\C{M}))}$, where $\B{g}_{\C{O}}$ is a measurable solution function of $\C{M}$ w.r.t.\ $\C{O}$, a \emph{compatible system of solution functions}, if for all $\C{O},\tilde{\C{O}}\in\C{L}(\C{G}(\C{M}))$ with $\tilde{\C{O}}\subseteq\C{O}$ and for $\Prb_{\BC{E}}$-almost every $\B{e}\in\BC{E}$ and for all $\B{x}\in\BC{X}$ we have
$$
  \B{x}_{\C{O}} = \B{g}_{\C{O}}(\B{x}_{\pa(\C{O})\setminus\C{O}},\B{e}_{\pa(\C{O})}) \quad\implies\quad
  \B{x}_{\tilde{\C{O}}} = \B{g}_{\tilde{\C{O}}}(\B{x}_{\pa(\tilde{\C{O}})\setminus\tilde{\C{O}}},\B{e}_{\pa(\tilde{\C{O}})})  \,.
$$
\end{definition}

The underlying SCM of a modular SCM always has a compatible system of solution functions, by construction.
\begin{proposition}
\label{prop:ExistenceOfCompatibleSystemOfMeasurableSolutionFunctions}
Let $\widehat{\C{M}} = \langle \C{G}, \BC{X}, \BC{E}, (\B{g}_{\C{O}})_{\C{O}\in\C{L}(\C{G})}, \Prb_{\BC{E}} \rangle$ be a modular SCM. Then the underlying SCM $\tilde{\C{M}}:=\iota(\widehat{\C{M}})$ is loop-wisely solvable. Moreover, it has a compatible system of solution functions $(\B{g}_{\C{O}})_{\C{O}\in\C{L}(\C{G}(\tilde{\C{M}}))}$, where $\B{g}_{\C{O}}$ is a measurable solution function of $\tilde{\C{M}}$ w.r.t.\ $\C{O}$.
\end{proposition}
This shows that a modular SCM can be seen as an SCM together with an additional structure of a compatible system of solution functions, and is, in particular, loop-wisely solvable.

Moreover, the class of simple SCMs corresponds exactly with those SCMs that are loop-wisely uniquely solvable.
\begin{lemma}
\label{lemm:LoopwiseUniqueSolvabilityAndSimple}
An SCM $\C{M}$ is simple if and only if it is loop-wisely uniquely solvable.
\end{lemma}

In particular, for simple SCMs, or loop-wisely uniquely solvable SCMs, there always exists a compatible system of solution functions.
\begin{proposition}
\label{prop:SimpleSCMsCompatibleSysOfSolFunctions}
Let $\C{M} = \langle \C{I}, \C{J}, \BC{X}, \BC{E}, \B{f}, \Prb_{\BC{E}} \rangle$ be a simple SCM. Then every family of measurable solution functions $(\B{g}_{\C{O}})_{\C{O}\in\C{L}(\C{G}(\C{M}))}$, where $\B{g}_{\C{O}}$ is a measurable solution function of $\C{M}$ w.r.t.\ $\C{O}$, is a compatible system of solution functions.
\end{proposition}

%%%%%%%%%%%%%%%%%%%%%%%%%%%%%%%%%%%%%%%%%%%%%%%%%%
\subsection{Overview of causal graphical models}
\label{app:AppendixOverviewCausalGraphicalModels}
%%%%%%%%%%%%%%%%%%%%%%%%%%%%%%%%%%%%%%%%%%%%%%%%%%

\begin{figure}[t]
\begin{center}
\adjustbox{scale=0.8,center}{%
\begin{tikzpicture}
  \coordinate (a) at (0,0);
  \coordinate (b) at (1,0);
  \coordinate (c) at (2,0);
  \coordinate (d) at (3,0);
  \coordinate (e) at (4,0);
  \coordinate (f) at (5,0);
  \coordinate (g) at (6,0);
  \def\ellipseb{\ellipsebyfoci{a}{b}{2}{0.52}};
  \def\ellipsec{\ellipsebyfoci{a}{c}{2}{0.52}};
  \def\ellipsed{\ellipsebyfoci{a}{d}{2}{0.52}};
  \def\ellipsee{\ellipsebyfoci{a}{e}{2}{0.52}};
  \def\ellipsef{\ellipsebyfoci{a}{f}{2}{0.52}};
  \def\ellipseg{\ellipsebyfoci{a}{g}{2}{0.52}};
  \fill[gray!20] \ellipsef;
  \fill[gray!40] \ellipsec;
  \draw \ellipseb;
  \draw \ellipsec;
  \draw \ellipsed;
  \draw \ellipsee;
  \draw \ellipsef;
  \draw \ellipseg;
%  \ellipsebyfoci{draw}{a}{b}{2.5}{0.5};
%  \ellipsebyfoci{draw}{a}{c}{2.5}{0.5};
%  \ellipsebyfoci{draw}{a}{d}{2.5}{0.5};
%  \fill[red] (a) circle(2pt);
%  \fill[red] (b) circle(2pt);
%  \fill[red] (c) circle(2pt);
%  \fill[red] (d) circle(2pt);
%  \fill[red] (e) circle(2pt);
%  \fill[red] (f) circle(2pt);
  \node at (0.52,0) (A) {\begin{tabular}{c} causal \\ BNs \end{tabular}};
  \node at (2.18,0) (B) {\begin{tabular}{c} acyclic \\ SCMs \end{tabular}};
  \node at (3.68,0) (C) {\begin{tabular}{c} simple \\ SCMs \end{tabular}};
  \node at (5.2,0) (D) {\begin{tabular}{c} modular \\ SCMs \end{tabular}};
  \node at (6.71,0) (E) {SCMs};
  \node at (8.25,0) (F) {CCMs};
\end{tikzpicture}}
\end{center}
\caption{Overview of causal graphical models. The ``gray'' and ``dark gray'' areas contain all the causal graphical models that can be modeled by an SCM and an acyclic SCM, respectively.}
\label{fig:OverviewCausalModels}
\end{figure}

Figure~\ref{fig:OverviewCausalModels} gives an overview of the causal graphical models related to SCMs. The ``gray'' area contains all the causal graphical models that can be modeled by an SCM, by which we mean, that there exists an SCM that can describe all its observational and interventional distributions. The ``dark gray'' area contains all the causal graphical models which can be modeled by an acyclic SCM. Acyclic SCMs generalize causal Bayesian networks (causal BNs)~\citep{Pea09} to allow for latent confounders and to derive counterfactuals. Simple SCMs form a subclass of SCMs that extends acyclic SCMs to the cyclic setting, while preserving many of their convenient properties. Modular SCMs~\citep{FM17} can be seen as SCMs that have an additional structure of compatible system of solution functions and contain, in particular, the class of simple SCMs. Forr{\'e} and Mooij~\citep{FM17} showed that modular SCMs satisfy various convenient properties, like marginalization and the general directed global Markov property. We show that for SCMs in general various of those properties still hold under certain solvability conditions. A generalization of SCMs, known as \emph{causal constraints models (CCMs)}, has been proposed~\citep{BBM19} in order to completely model the causal semantics of the equilibrium solutions of a dynamical system given the initial conditions. This class of CCMs is rich enough to model the causal semantics of SCMs, but does not come with a single graphical representation that provides both a Markov property and a causal interpretation~\citep{BDM20}.

%%%%%%%%%%%%%%%%%%%%%%%%%%%%%%%%%%%%%%%%%%%%%%%%%%
\section{(Unique) solvability properties}
\label{app:AppendixSolvabilityResults}
\setcounter{section}{2}
\setcounter{theorem}{0}
%%%%%%%%%%%%%%%%%%%%%%%%%%%%%%%%%%%%%%%%%%%%%%%%%%

In this appendix, we provide additional (unique) solvability properties for SCMs. 
In Appendix~\ref{app:AppendixSolvabilitySufficientStrictSubset} we provide a sufficient condition of solvability w.r.t.\ (strict) subsets. In Appendix~\ref{app:AppendixSolvabilitySuperAndSubsets} we discuss how (unique) solvability is preserved under strict super- and subsets. In Appendix~\ref{app:AppendixSolvabilityUnionAndIntersections} we discuss how (unique) solvability is preserved under unions and intersections. The proofs of the theoretical results in this appendix are given in Appendix~\ref{app:AppendixProofs}.

%%%%%%%%%%%%%%%%%%%%%%%%%%%%%%%%%%%%%%%%%%%%%%%%%%
\subsection{Sufficient condition for solvability w.r.t.\ subsets}
\label{app:AppendixSolvabilitySufficientStrictSubset}
%%%%%%%%%%%%%%%%%%%%%%%%%%%%%%%%%%%%%%%%%%%%%%%%%%

For solvability w.r.t.\ a (strict) subset of $\C{I}$ there exists a sufficient condition that is similar to the sufficient (and necessary) condition (2) in Theorem~\ref{thm:SolvabilityIffCondition} in the sense that it is formulated in terms of the solutions of (a subset of) the structural equations, but no measurability is required.
\begin{proposition}[Sufficient condition for solvability w.r.t.\ a subset]
\label{prop:SolvabilityIfCondition}
Let $\C{M} = \langle \C{I}, \C{J}, \BC{X}, \BC{E}, \B{f}, \Prb_{\BC{E}} \rangle$ be an SCM and $\C{O}\subseteq\C{I}$ a subset.
%If for $\Prb_{\BC{E}}$-almost every $\B{e}$
%and for all $\B{x}_{\setminus\C{O}}\in\BC{X}_{\setminus\C{O}}$ the structural equations
%$$
%  \B{x}_{\C{O}} = \B{f}_{\C{O}}(\B{x}, \B{e}) \,,
%$$
%have a non-empty and $\sigma$-compact space of solutions $\B{x}_{\C{O}}\in\BC{X}_{\C{O}}$, then $\C{M}$ is solvable w.r.t.\ $\C{O}$.
If for $\Prb_{\BC{E}}$-almost every $\B{e}\in\BC{E}$ and for all $\B{x}_{\setminus \C{O}}\in\BC{X}_{\setminus\C{O}}$ the topological space
$$
\BC{S}_{(\B{e},\B{x}_{\setminus\C{O}})} := 
\{ \B{x}_{\C{O}}\in\BC{X}_{\C{O}} :
\B{x}_{\C{O}} = \B{f}_{\C{O}}(\B{x},\B{e}) \} \,,
$$
with the subspace topology induced by $\BC{X}_{\C{O}}$ is nonempty and $\sigma$-compact,\footnote{%A topological space $\BC{X}$ is called \emph{compact} if every open cover of $\BC{X}$ has a finite subcover, i.e., for every set $\{\BC{U}_i\subseteq\BC{X}\}_{i\in I}$ of open subsets of $\BC{X}$ such that $\BC{X} = \cup_{i\in I} \BC{U}_i$ there exists a finite subset $J\subseteq I$ such that $\BC{X}=\cup_{j\in J} \BC{U}_j$. 
A topological space $\BC{X}$ is called \emph{$\sigma$-compact} if it is the union of a countable set of compact topological spaces.} 
then $\C{M}$ is solvable w.r.t.\ $\C{O}$.
\end{proposition}
%Again, in the cyclic case the existence of a measurable solution function seems much less straightforward than in the acyclic case.

%In other words, if for $\Prb_{\BC{E}}$-almost every $\B{e}$ and for all $\B{x}_{\setminus \C{O}}\in\BC{X}_{\setminus\C{O}}$ the space
%$$
%\BC{S}_{(\B{e},\B{x}_{\setminus\C{O}})} := 
%\{ \B{x}_{\C{O}}\in\BC{X}_{\C{O}} :
%\B{x}_{\C{O}} = \B{f}_{\C{O}}(\B{x},\B{e}) \} \,,
%$$
%which inherits the topology from $\BC{X}_{\C{O}}$, is non-empty and $\sigma$-compact, then $\C{M}$ is solvable w.r.t.\ $\C{O}$. 
For many purposes, this condition of $\sigma$-compactness suffices since it contains for example all countable discrete spaces, every interval of the real line, and moreover all the 
Euclidean spaces. In particular, it suffices to prove a sufficient and necessary condition for unique solvability w.r.t.\ a subset, in terms of the solutions of a subset of the structural equations (see Theorem~\ref{thm:UniqueSolvabilityIffCondition}). For larger solution spaces, we refer the reader to \citep{Kec95}. For the class of linear SCMs (see Definition~\ref{def:LinearSCM}), we provide in Proposition~\ref{prop:LinearSolvable} a sufficient and necessary condition for solvability w.r.t.\ a (strict) subset of $\C{I}$. 

%which suffices as a sufficient condition for the marginalization of an SCM (see Theorem~\ref{thm:UniqueSolvabilityIffCondition} and Definition~\ref{def:MarginalSCM}).

%Note that for the case $\C{O}=\C{I}$ the condition in
%Proposition~\ref{prop:SolvabilityIfCondition} implies in particular condition (2) in Theorem~\ref{thm:SolvabilityIffCondition} from which solvability also follows.

%\begin{example}
%Consider an SCM $\C{M} = \langle \B{1}, \B{1}, \CN, \RN, \B{f}, \Prb_{\C{E}}
%\rangle$ with the causal mechanism $f_1(\B{x},\B{e}) = x_1 - x_1^n + e_1$, where
%$n\in\NN$ and $\Prb_{\C{E}} = \Prb^{E}$ with $E \sim \C{N}(0,1)$. For $n \ge 1$,
%this SCM is solvable w.r.t.\ the subset $\{1\}$, since for $\Prb_{\C{E}}$-almost
%all $e$ the set of solutions of the structural equation:
%$$
%    x_1 = x_1 - x_1^n + e \,,
%$$
%consists of at most $n$ elements $x_1$. This gives $n$ mappings $g_1 : \C{E} \to \C{X}$, 
%which each map the value $e$ to one of the $n$ distinct $n^{\text{th}}$ 
%roots of $e$.
%\end{example}

%%%%%%%%%%%%%%%%%%%%%%%%%%%%%%%%%%%%%%%%%%%%%%%%%%
\subsection{(Unique) solvability w.r.t.\ strict super- and subsets}
\label{app:AppendixSolvabilitySuperAndSubsets}
%%%%%%%%%%%%%%%%%%%%%%%%%%%%%%%%%%%%%%%%%%%%%%%%%%

In general, (unique) solvability w.r.t.\ $\C{O} \subseteq \C{I}$ does not imply 
(unique) solvability w.r.t.\ a strict superset $\C{O} \subsetneq \C{V} \subseteq \C{I}$ nor w.r.t.\ a strict subset $\C{W} \subsetneq \C{O}$, as can be seen in the following example.
\begin{example}[Solvability is not preserved under strict sub- or supersets]
\label{ex:SolvabilityProperties}
Consider the SCM $\C{M}=\langle \B{3}, \emptyset, \RN^3, \B{1}, \B{f}, \Prb_{\B{1}} \rangle$ where the causal mechanism is given by
$$ 
    f_1(\B{x}) = x_1 \cdot (1-\B{1}_{\{1\}}(x_2)) + 1 \,,\,\, f_2(\B{x}) =
    x_2 \,, \,\, f_3(\B{x}) = x_3 \cdot (1-\B{1}_{\{-1\}}(x_2)) + 1 \,.
$$
This SCM is (uniquely) solvable w.r.t.\ the subsets $\{1,2\}$, $\{2,3\}$, however it is not (uniquely) solvable w.r.t.\ the subsets $\{1\}$, $\{3\}$ and $\{1,2,3\}$, and not uniquely solvable w.r.t.\ $\{2\}$.
\end{example}
However, in Proposition~\ref{prop:SolvabilityAncestralSolvability} we show that solvability w.r.t.\ $\C{O}$ implies solvability w.r.t.\ every ancestral subset in $\C{G}(\C{M})_{\C{O}}$.

%%%%%%%%%%%%%%%%%%%%%%%%%%%%%%%%%%%%%%%%%%%%%%%%%%
\subsection{(Unique) solvability w.r.t.\ unions and intersections}
\label{app:AppendixSolvabilityUnionAndIntersections}
%%%%%%%%%%%%%%%%%%%%%%%%%%%%%%%%%%%%%%%%%%%%%%%%%%

In general, (unique) solvability is not preserved under unions and intersections. The following example illustrates that 
(unique) solvability is in general not preserved under intersections.
\begin{example}[Solvability is not preserved under intersections]
\label{ex:SolvabilityProperties2}
Consider the SCM $\C{M}=\langle \B{3}, \emptyset, \RN^3, \B{1}, \B{f}, \Prb_{\B{1}} \rangle$ where the causal mechanism is given by
$$ 
  f_1(\B{x}) = 0 \,, \,\, f_2(\B{x}) = x_2 \cdot (1 - \B{1}_{\{0\}}(x_1\cdot
  x_3)) + 1 \,, \,\, f_3(\B{x}) = 0 \,.
$$
Then $\C{M}$ is (uniquely) solvable w.r.t.\ $\{1,2\}$ and $\{2,3\}$, however it
is not (uniquely) solvable w.r.t.\ 
their intersection.
\end{example}

Example~\ref{ex:SolvabilityProperties} gives an example where (unique) solvability is not preserved 
under unions. Even, if we take the union of disjoint subsets, (unique) solvability is not preserved 
(see Example~\ref{ex:SimpleCyclicExample}). 
%We have seen that unique solvability is not preserved under taking unions.
Although, in general, unique solvability is not preserved under unions, we show next that unique solvability is preserved under the union of ancestral subsets, under the following assumptions.
\begin{proposition}[Combining measurable solution functions on different sets]
\label{prop:UniqueSolvabilityAncestralUnion}
Let $\C{M}=\langle \C{I}, \C{J}, \BC{X}, \BC{E}, \B{f},
\Prb_{\BC{E}} \rangle$ be an SCM, $\C{O}\subseteq\C{I}$ a subset and $\C{A},\tilde{\C{A}}\subseteq\C{O}$ two ancestral subsets in $\C{G}(\C{M})_{\C{O}}$. If $\C{M}$ is uniquely solvable w.r.t.\ $\C{A}$, $\tilde{\C{A}}$ and $\C{A} \cap \tilde{\C{A}}$, then $\C{M}$ is uniquely solvable w.r.t.\ $\C{A}\cup \tilde{\C{A}}$. 
\end{proposition}

%\Joris{Fixed this statement, as it was wrong (the condition that the SCM is uniquely solvable w.r.t.\ the intersection was missing). Example~\ref{ex:SolvabilityProperties} is a counterexample
%when taking $\C{O} = \C{I} = \{1,2,3\}$ and $\C{A} = \{1,2\}$ (ancestral subset) and $\tilde{\C{A}} = \{2,3\}$ (ancestral subset).
%It is uniquely solvable w.r.t.\ both $\C{A}$ and $\tilde{\C{A}}$, but not w.r.t.\ $\C{A} \cup \tilde{\C{A}}$ (nor w.r.t.\
%$\C{A} \cap \tilde{\C{A}}$).}

A consequence of this property is that in order to check whether an SCM is ancestrally uniquely solvable w.r.t.\ $\C{O}$, 
it suffices to check that it is uniquely solvable w.r.t.\ the ancestral subsets for each node in $\C{O}$.
\begin{corollary}
\label{coro:UniqueSolvabilityAncestralUniqueSolvability}
Let $\C{M}=\langle \C{I}, \C{J}, \BC{X}, \BC{E}, \B{f},
\Prb_{\BC{E}} \rangle$ be an SCM and $\C{O}\subseteq\C{I}$ a subset. Then $\C{M}$ is
ancestrally uniquely solvable w.r.t.\ $\C{O}$ if and only if $\C{M}$ is uniquely solvable w.r.t.\ $\an_{\C{G}(\C{M})_{\C{O}}}(i)$ for every
$i \in\C{O}$.
\end{corollary}

%%%%%%%%%%%%%%%%%%%%%%%%%%%%%%%%%%%%%%%%%%%%%%%%%%
\section{Linear SCMs}
\label{app:AppendixLinSCMs}
\setcounter{section}{3}
\setcounter{theorem}{0}
%%%%%%%%%%%%%%%%%%%%%%%%%%%%%%%%%%%%%%%%%%%%%%%%%%

In this appendix, we provide some results about (unique) solvability and
marginalization for linear SCMs. Linear SCMs form a special class of SCMs that has seen much attention in the literature \citep[see, e.g.,][]{Bol89,HEH12}. The proofs of the theoretical results in this appendix are given in Appendix~\ref{app:AppendixProofs}.
\begin{definition}[Linear SCM]
\label{def:LinearSCM}
We call an SCM $\C{M}=\langle \C{I}, \C{J}, \RN^{\C{I}}, \RN^{\C{J}}, \B{f}, \Prb_{\RN^{\C{J}}} \rangle$
\emph{linear} if each component of the causal mechanism is a linear combination of the endogenous and exogenous variables, that is
$$
  f_i(\B{x},\B{e}) = \sum_{j\in\C{I}} B_{ij} x_j + \sum_{k\in\C{J}}\Gamma_{ik} e_k \,,
$$
where $i\in\C{I}$, $B\in\RN^{\C{I}\times\C{I}}$ and $\Gamma\in\RN^{\C{I}\times\C{J}}$ 
are matrices, and $\Prb_{\RN^{\C{J}}}$ is a product probability 
measure\footnote{Note that we do not assume that the probability measure $\Prb_{\RN^{\C{J}}}$ 
is Gaussian.} on $\RN^{\C{J}}$.
\end{definition}
For a subset $\C{O}\subseteq\C{I}$ we also use the shorthand vector-notation
$$
  \B{f}_{\C{O}}(\B{x},\B{e})=B_{\C{OI}}\B{x}+\Gamma_{\C{OJ}}\B{e} \,.
$$
A nonzero coefficient $B_{ij}$ for $i,j\in\C{I}$ such that $i\neq j$ corresponds with a directed edge $j\to i$ in the (augmented) graph, and a coefficient $B_{ii}=1$ for $i\in\C{I}$ corresponds with a self-cycle $i\to i$ in the (augmented) graph of the SCM. A nonzero coefficient $\Gamma_{ij}$ for $i\in\C{I}$, $j\in\C{J}$ with $\Prb_{\C{E}_j}$ a nondegenerate probability distribution over $\RN$ corresponds with a directed edge $j \to i$ in the augmented graph. A nonzero entry $(\Gamma \Gamma^T)_{ij}$ for $i,j\in\C{I}$ with $i\neq j$ such that there exists a $k\in\C{J}$ for which $\Gamma_{ik},\Gamma_{jk}\neq 0$ and $\Prb_{\C{E}_k}$ a nondegenerate probability distribution over $\RN$ corresponds with a bidirected edge $i \oto j$ in the graph of the SCM.

For linear SCMs, the solvability condition w.r.t.\ a subset, Definition~\ref{def:Solvability}, translates into a matrix condition. In order to state this condition we need to define the pseudoinverse (or the Moore-Penrose inverse) $A^+$ of a real matrix $A$ \citep{Pen55,GK65}. 
%For every $n\times m$ real or complex matrix $A$, there exists a \emph{singular value decomposition}, which is given by the factorization $A=U\Sigma V^*$, where $U$ is an $m\times m$ unitary matrix, $\Sigma$ is an $m\times n$ diagonal matrix with nonnegative real entries, $V$ is an $n\times n$ unitary matrix and $V^*$ is its conjugate transpose. 
The \emph{pseudoinverse of the matrix $A$} is defined by $A^+:=V\Sigma^+ U^*$, where $A=U\Sigma V^*$ is the singular value decomposition of $A$ and $\Sigma^+$ is obtained by replacing each nonzero entry on the diagonal of $\Sigma$ by its reciprocal \citep{GK65}. One of its useful properties is that $A A^+ A = A$.
%The \emph{pseudoinverse of the matrix $A$} is now defined by $A^+:=V\Sigma^+ U^*$, where $\Sigma^+$ is obtained by replacing each non-zero entry on the diagonal of $\Sigma$ by its reciprocal \citep{GK65}.

\begin{proposition}[Sufficient and necessary condition for solvability w.r.t.\ a subset for linear SCMs]
\label{prop:LinearSolvable}
Let $\C{M}$ be a linear SCM and $\C{L}\subseteq\C{I}$ and $\C{O}=\C{I}\setminus\C{L}$. Then $\C{M}$ is solvable w.r.t.\ 
$\C{L}$ if and only if for the matrix $A_{\C{L}\C{L}}=\mathbb{I}_{\C{L}}-B_{\C{LL}}$, for $\Prb_{\BC{E}}$-almost every $\B{e}\in\BC{E}$ and for all $\B{x}_{\C{O}}\in\BC{X}_{\C{O}}$ the identity
$$
A_{\C{L}\C{L}} A_{\C{L}\C{L}}^+ (B_{\C{LO}}\B{x}_{\C{O}}+\Gamma_{\C{LJ}}\B{e}) = B_{\C{LO}}\B{x}_{\C{O}}+\Gamma_{\C{LJ}}\B{e}
$$
is satisfied, where $A_{\C{L}\C{L}}^+$ is the pseudoinverse of $A_{\C{L}\C{L}}$. Moreover, if $\C{M}$ is solvable w.r.t.\ $\C{L}$, then for every vector $\B{v}\in \RN^{\C{L}}$ the mapping 
$\B{g}^{\B{v}}_{\C{L}}:\RN^{\C{O}}\times\RN^{\C{J}}\to\RN^{\C{L}}$ given by
$$
\B{g}^{\B{v}}_{\C{L}}(\B{x}_{\C{O}},\B{e}) = A_{\C{L}\C{L}}^+(B_{\C{LO}}\B{x}_{\C{O}}+\Gamma_{\C{LJ}}\B{e}) + [\mathbb{I}_{\C{L}} - A_{\C{L}\C{L}}^+ A_{\C{L}\C{L}}] \B{v} \,,
$$
is a measurable solution function for $\C{M}$ w.r.t.\ $\C{L}$.
\end{proposition}

For linear SCMs, the unique solvability condition w.r.t.\ a subset translates into a matrix invertibility condition, as was already shown in \citep{HEH12}.
\begin{proposition}[Sufficient and necessary condition for unique solvability w.r.t.\ a subset for linear SCMs]
\label{prop:LinearUniquelySolvable}
Let $\C{M}$ be a linear SCM, $\C{L}\subseteq\C{I}$ and $\C{O}=\C{I}\setminus\C{L}$. 
Then $\C{M}$ is uniquely solvable w.r.t.\ 
$\C{L}$ if and only if the matrix $A_{\C{L}\C{L}}=\mathbb{I}_{\C{L}}-B_{\C{LL}}$ is invertible. Moreover, if $\C{M}$ is uniquely solvable w.r.t.\ $\C{L}$, then the mapping 
$\B{g}_{\C{L}}:\RN^{\C{O}}\times\RN^{\C{J}}\to\RN^{\C{L}}$ given by
$$
  \B{g}_{\C{L}}(\B{x}_{\C{O}},\B{e}) = 
  A_{\C{L}\C{L}}^{-1}(B_{\C{LO}}\B{x}_{\C{O}}+\Gamma_{\C{LJ}}\B{e}) \,,
$$
is a measurable solution function for $\C{M}$ w.r.t.\ $\C{L}$.
\end{proposition}
Note that if $A_{\C{LL}}$ is invertible, then $A_{\C{LL}}^+=A_{\C{LL}}^{-1}$ (see Lemma~1.3 in \citep{Pen55}), and the matrix condition of Proposition~\ref{prop:LinearSolvable} is always satisfied and all the measurable solution functions $\B{g}^{\B{v}}_{\C{L}}$ of Proposition~\ref{prop:LinearSolvable} are (up to a $\Prb_{\BC{E}}$-null set) equal to the solution function $\B{g}_{\C{L}}$ of Proposition~\ref{prop:LinearUniquelySolvable}.

\begin{remark}
A sufficient condition for $A_{\C{LL}}$ to be invertible is that the spectral radius of $B_{\C{LL}}$ is less than one. If that is the case, then $A_{\C{LL}}^{-1}=\sum_{n=0}^{\infty}(B_{\C{LL}})^n$. Note that the nonzero nondiagonal entries of the matrix $B_{\C{LL}}$ represent the directed edges in the induced subgraph $\C{G}(\C{M})_{\C{L}}$. In particular, if the diagonal entries of the matrix $B_{\C{LL}}$ are zero, then for $n\in\NN$, the coefficients of the matrix $(B_{\C{LL}})^n$ in the sum represent the sum of the product of the edge weights $B_{ij}$ over directed paths of length $n$ in the induced subgraph $\C{G}(\C{M})_{\C{L}}$.
\end{remark}

From Proposition~\ref{prop:SolvabilityAncestralSolvability} we know that an SCM is solvable w.r.t.\ $\C{L}$ if and only if it is ancestrally solvable w.r.t.\ $\C{L}$. In particular, this result also holds for linear SCMs. We saw in Example~\ref{ex:UniqueSolvabilityNotUniqAncestralSubset} that a similar result for unique solvability does not hold, that is, in general, it does not hold that unique solvability w.r.t.\ $\C{L}$ implies ancestral unique solvability w.r.t.\ $\C{L}$. For the class of linear SCMs we do have the following positive result.
\begin{proposition}[Equivalent unique solvability conditions for linear SCMs]
\label{prop:LinearEquivalentUniqueSolvability}
For a linear SCM $\C{M}$ and a subset $\C{L}\subseteq\C{I}$ the following are equivalent:
\begin{enumerate}
  \item $\C{M}$ is uniquely solvable w.r.t.\ $\C{L}$;
  \item $\C{M}$ is ancestrally uniquely solvable w.r.t.\ $\C{L}$;
  \item $\C{M}$ is uniquely solvable w.r.t.\ each strongly connected component in $\C{G}(\C{M})_{\C{L}}$.
\end{enumerate}
\end{proposition}

%\begin{example}
%\label{ex:LinearSCM}
%Consider the linear SCM $\C{M} = \langle \B{3}, \emptyset, \RN^3, 1, \B{f}, \Prb_1 \rangle$ 
%with the causal mechanism:
%$$
%  \begin{aligned}
%    f_1(\B{x}) &= x_2 + x_3 \\
%    f_2(\B{x}) &= x_2 \\
%    f_3(\B{x}) &= x_3 \,,
%  \end{aligned}
%$$
%and let $\C{L}=\{1,2\}$. Then the mapping $\B{g}_{\C{L}}:\C{X}_3 \to \BC{X}_{\C{L}}$ given by  
%$$
%  \begin{aligned}
%    g_1(x_3) &= x_3 \\
%    g_2(x_3) &= 0
%  \end{aligned}
%$$
%makes $\C{M}$ solvable w.r.t.\ $\C{L}$.
%\end{example}
%Note that in example~\ref{ex:LinearSCM} the matrix $I-B_{\C{LL}}$ is not invertible.

Under the condition of unique solvability w.r.t.\ a subset $\C{L}$ we can define the marginalization w.r.t.\ $\C{L}$ of a linear SCM by mere substitution.
\begin{proposition}[Marginalization of a linear SCM]
\label{prop:MarginalLinearModel}
Let $\C{M}$ be a linear SCM and $\C{L}\subseteq\C{I}$ a subset of endogenous variables such that 
$\mathbb{I}_{\C{L}}-B_{\C{LL}}$ is invertible. Then there exists a marginalization $\C{M}_{\marg(\C{L})}$ that is 
linear and with marginal causal mechanism 
$\tilde{\B{f}}:\RN^{\C{O}}\times\RN^{\C{J}}\to\RN^{\C{O}}$ given by
$$ 
  \begin{aligned}
    \tilde{\B{f}}(\B{x}_{\C{O}},\B{e}) &=
  [B_{\C{OO}}+B_{\C{OL}}A_{\C{LL}}^{-1}B_{\C{LO}}]\B{x}_{\C{O}} + [B_{\C{OL}}A_{\C{LL}}^{-1}\Gamma_{\C{LJ}}+\Gamma_{\C{OJ}}]\B{e} \,,
  \end{aligned}
$$
where $A_{\C{LL}}=\mathbb{I}_{\C{L}} - B_{\C{LL}}$. Moreover, this marginalization respects the latent projection, that is, $\big(\C{G}^a \circ \marg(\C{L})\big)(\C{M}) \subseteq \big(\marg(\C{L})\circ\C{G}^a\big)(\C{M})$.
\end{proposition}
From Theorem~\ref{thm:MarginalizationEquivalences} we know that $\C{M}$ and its
marginalization $\C{M}_{\marg(\C{L})}$ over $\C{L}$ are observationally,
interventionally and counterfactually equivalent w.r.t.\ $\C{O}$. A similar result can also be found in \citep{HEH12}. In contrast to nonlinear SCMs, this class of linear SCMs has the convenient property that every marginalization of a model of this class respects the latent projection. Moreover, the subclass of simple linear SCMs is even closed under marginalization.

%%%%%%%%%%%%%%%%%%%%%%%%%%%%%%%%%%%%%%%%%%%%%%%%%%
\section{Examples}
\label{app:AppendixExamples}
\setcounter{section}{4}
\setcounter{theorem}{0}
%%%%%%%%%%%%%%%%%%%%%%%%%%%%%%%%%%%%%%%%%%%%%%%%%%

In this appendix, we provide additional examples. In Appendix~\ref{app:AppendixExamplesEquilibriumModels} we provide some examples of SCMs that describe the equilibrium states of certain feedback systems governed by (random) differential equations~\citep{BM18} that motivated our study of cyclic SCMs. In Appendix~\ref{app:AppendixAdditionalExamples} we provide additional examples that support the main text.

%%%%%%%%%%%%%%%%%%%%%%%%%%%%%%%%%%%%%%%%%%%%%%%%%%
\subsection{SCMs as equilibrium models}
\label{app:AppendixExamplesEquilibriumModels}
%%%%%%%%%%%%%%%%%%%%%%%%%%%%%%%%%%%%%%%%%%%%%%%%%%

In many systems occurring in the real world feedback loops between observed variables are present. For example, in economics, the price of a product may be a function of the demanded or supplied quantities, and vice versa; or in physics, two masses that are connected by a spring may exert forces on each other. Such systems are often described by a system of (random) differential equations. In~\citep{BM18} it was shown that SCMs are capable of modeling the causal semantics of the equilibrium states of such systems. For illustration purposes we provide the following toy example of interacting masses that are attached to springs.
\begin{example}[Damped coupled harmonic oscillator]
\label{ex:HarmonicOscillator}
Consider a one-dimensional system of $d$ point masses $m_i\in\RN$
($i=1,\dots,d$) with positions $Q_i$, which are
coupled by springs, with spring constants $k_i>0$ and equilibrium lengths $\ell_i>0$
($i=0,\dots,d$), under influence of friction with friction coefficients $b_i\in\RN$ ($i=1,\dots,d$) and with fixed
endpoints $Q_0=0$ and $Q_{d+1}=L>0$ (see Figure~\ref{fig:MassSpringSystem} (top)).
\begin{figure}
  \centerline{\adjustbox{scale=0.8,center}{%
      \begin{tikzpicture}[scale=0.6]
    \fill[pattern=north east lines,draw=none] (-.75,-1) rectangle (0,1);
    \draw (0,-1) -- (0,1);
    \node (Q1) at (1.75,0.6) {$m_1$};
    \node (Q2) at (3.5,0.6) {$m_2$};
    \node (Q3) at (5.25,0.6) {$m_3$};
    \node (Q4) at (7,0.6) {$m_4$};
    \node (Q5) at (8.75,0.6) {$m_5$};
    \begin{scope}[xshift=-1.5mm]
      \node (l0) at (0.875,-0.6) {$\ell_0$};
      \node (l1) at (2.625,-0.6) {$\ell_1$};
      \node (l2) at (4.375,-0.6) {$\ell_2$};
      \node (l3) at (6.125,-0.6) {$\ell_3$};
      \node (l4) at (7.875,-0.6) {$\ell_4$};
      \node (l5) at (9.625,-0.6) {$\ell_5$};
    \end{scope}
    \fill[black] (1.75,0) circle (.2);
    \fill[black] (3.5,0) circle (.2);
    \fill[black] (5.25,0) circle (.2);
    \fill[black] (7,0) circle (.2);
    \fill[black] (8.75,0) circle (.2);
    \draw (10.5,-1) -- (10.5,1);
    \fill[pattern=north east lines,draw=none] (10.5,-1) rectangle (11.25,1);
    \draw[thick,decorate,decoration={coil,aspect=0.7,amplitude=4,segment length=3mm}] (0,0) -- (1.75,0);
    \draw[thick,decorate,decoration={coil,aspect=0.7,amplitude=4,segment length=3mm}] (1.75,0) -- (3.5,0);
    \draw[thick,decorate,decoration={coil,aspect=0.7,amplitude=4,segment length=3mm}] (3.5,0) -- (5.25,0);
    \draw[thick,decorate,decoration={coil,aspect=0.7,amplitude=4,segment length=3mm}] (5.25,0) -- (7,0);
    \draw[thick,decorate,decoration={coil,aspect=0.7,amplitude=4,segment length=3mm}] (7,0) -- (8.75,0);
    \draw[thick,decorate,decoration={coil,aspect=0.7,amplitude=4,segment length=3mm}] (8.75,0) -- (10.5,0);
    \node at (0,1.3) {\tiny $Q_0 = 0$};
    \node at (10.5,1.3) {\tiny $Q_6 = L$};
  \end{tikzpicture}}}
  \bigskip
  \adjustbox{scale=0.8,center}{%
  \begin{tikzpicture}[scale=0.8]
    \node[var] (Q1) at (1.75,0.6) {$Q_1$};
    \node[var] (Q2) at (3.5,0.6) {$Q_2$};
    \node[var] (Q3) at (5.25,0.6) {$Q_3$};
    \node[var] (Q4) at (7,0.6) {$Q_4$};
    \node[var] (Q5) at (8.75,0.6) {$Q_5$};
    \draw[arr,bend left=20] (Q1) to (Q2);
    \draw[arr,bend left=20] (Q2) to (Q3);
    \draw[arr,bend left=20] (Q3) to (Q4);
    \draw[arr,bend left=20] (Q4) to (Q5);
    \draw[arr,bend left=20] (Q5) to (Q4);
    \draw[arr,bend left=20] (Q4) to (Q3);
    \draw[arr,bend left=20] (Q3) to (Q2);
    \draw[arr,bend left=20] (Q2) to (Q1);
  \end{tikzpicture}}
  \caption{Damped coupled harmonic oscillator (top) and the graph of the SCM $\C{M}$ that describes the positions of the masses at equilibrium (bottom) of Example~\ref{ex:HarmonicOscillator} for $d=5$.}
  \label{fig:MassSpringSystem}
\end{figure}
The equations of motion of this system are provided by the following differential equations
$$
  \frac{d^2Q_i}{dt^2} = \frac{k_i}{m_i} (Q_{i+1} - Q_i - \ell_i) + \frac{k_{i-1}}{m_i} (Q_{i-1} - Q_i + \ell_{i-1}) -\frac{b_i}{m_i} \frac{dQ_i}{dt} \quad\quad(i=1,\dots,d) \,.
$$
The dynamics of the masses, in terms of the position, velocity and acceleration, is described by a single and separate equation of motion for each mass. Under friction, that is, $b_i>0$ ($i=1,\dots,d$), there is a unique equilibrium position, where the sum of forces vanishes for each mass. If one starts out of equilibrium, for example, by moving one or several masses out of equilibrium, then the masses will start to oscillate and converge to their unique equilibrium position. At equilibrium (i.e., for $t\to\infty$) the velocity $\frac{dQ_i}{dt}$ and acceleration $\frac{d^2Q_i}{dt^2}$ of the masses vanish (i.e., $\frac{dQ_i}{dt},\frac{d^2Q_i}{dt^2}\to 0$), and thus the following equation holds at equilibrium
$$
  0 = \frac{k_i}{m_i} (Q_{i+1} - Q_i - \ell_i) + \frac{k_{i-1}}{m_i} (Q_{i-1} - Q_i + \ell_{i-1}) \,,
$$
for each mass ($i=1,\dots,d$). Hence, for each mass $i=1,\dots,d$ its equilibrium position $Q_i$ is given by
$$
  Q_i = \frac{k_i(Q_{i+1} - \ell_i) + k_{i-1}(Q_{i-1}+\ell_{i-1})}{k_i + k_{i-1}}\,.
$$
By considering the $\ell_i$ and $k_i$ and $L$ as fixed parameters, we arrive at a linear SCM (see \cite{BM18} for more details about constructing an SCM from a dynamical system)
$$
\C{M}=\langle \{1,\dots,d\},\emptyset,\RN^d,\B{1},\B{f},\Prb_{\B{1}} \rangle \,,
$$
where the causal mechanism $\B{f}$ is given by
$$
f_i(\B{q}) = \frac{k_i(q_{i+1} - \ell_i) + k_{i-1}(q_{i-1}+\ell_{i-1})}{k_i + k_{i-1}} \,.
$$
Alternatively, (some of) the parameters could be treated as exogenous variables instead. Its graph is depicted in Figure~\ref{fig:MassSpringSystem} (bottom). This SCM allows us to describe the equilibrium behavior of the system under perfect intervention. For example, when forcing the mass $j$ to a fixed position $Q_j=\xi_j$ with $0\leq \xi_j\leq L$, the equilibrium positions of the masses correspond to the solutions of the intervened model $\C{M}_{\intervene(\{j\},\xi_j)}$. It is an easy exercise to show that $\C{M}$ is a simple SCM by using Proposition~\ref{prop:LinearUniquelySolvable}.
\end{example}

Next, we show that the well known market equilibrium model from economics, which has been thoroughly discussed in the literature \citep[see, e.g.,][]{RichardsonRobins2014}, can be described by a (non-simple) SCM. This example illustrates how self-cycles enrich the class of SCMs.
\begin{example}[Price, supply and demand]
  \label{ex:PriceSupplyDemandDynamical}
Let $X_D$ denote the demand and $X_S$ the supply of a quantity of a product. The
price of the product is denoted by $X_P$. The following system of differential equations describes how the
demanded and supplied quantities are determined by the price, and how price
adjustments occur in the market:
$$
  \begin{aligned}
    X_D &= \beta_D X_P + E_D \\
    X_S &= \beta_S X_P + E_S \\
    \frac{dX_P}{dt} &= X_D - X_S \,,
  \end{aligned}
$$
where $E_D$ and $E_S$ are exogenous random influences on the demand and supply, respectively, $\beta_D<0$ is the reciprocal of the slope of the demand curve, and $\beta_S>0$ is the reciprocal of the slope of the supply curve. At the situation known as a ``market equilibrium'', the price is determined implicitly by the condition that demanded and supplied quantities should be equal, since $\frac{dX_P}{dt}=0$ at equilibrium. Applying the results in \citep{BM18} gives rise to a linear SCM $\C{M} = \langle \{P,S,D\}, \{S,D\}, \RN^3, \RN^2, \B{f}, \Prb_{\BC{E}} \rangle$
at equilibrium with the causal mechanism defined by
$$
  \begin{aligned}
    f_D(\B{x},\B{e}) &:= \beta_D x_P + e_D \\
    f_S(\B{x},\B{e}) &:= \beta_S x_P + e_S \\
    f_P(\B{x},\B{e}) &:= x_P + (x_D - x_S) \,.
  \end{aligned}
$$
Note how we use a self-cycle for $P$ in order to implement the equilibrium
equation $X_D = X_S$ as the causal mechanism for the price $P$.\footnote{Richardson and Robins~\cite{RichardsonRobins2014}
argue that this market equilibrium model cannot be modeled as an SCM. We observe that it can,
as long as one allows for self-cycles.} Moreover, $\C{M}$ is uniquely solvable. Its augmented graph is depicted in Figure~\ref{fig:PriceSupplyDemand} (left).
\begin{figure}[t]
\begin{center}
\adjustbox{scale=0.8,center}{%
    \begin{tikzpicture}
    \begin{scope}
%      \node at (-2.5,2) {(c)};
      \node[exvar] (ED) at (1.5,0) {$E_D$};
      \node[var] (D) at (1.5,-1.5) {$X_D$};
      \node[exvar] (ES) at (-1.5,0) {$E_S$};
      \node[var] (S) at (-1.5,-1.5) {$X_S$};
      \node[var] (P) at (0,-1.5) {$X_P$};
      \draw[arr, bend left=20] (P) edge (S);
      \draw[arr, bend left=20] (P) edge (D);
      \draw[arr, bend left=20] (S) edge (P);
      \draw[arr, bend left=20] (D) edge (P);
      \draw[arr] (ED) to (D);
      \draw[arr] (ES) to (S);
      \draw[arr] (P)  to [out=62,in=118,looseness=4.4] (P);
      \node at (0,-2.2) {$\C{G}^a(\C{M})$};
    \end{scope}
    \begin{scope}[xshift=5cm]
%      \node at (-2.5,2) {(c)};
      \node[exvar] (ED) at (1.5,0) {$E_D$};
      \node[var] (D) at (1.5,-1.5) {$X_D$};
      \node[var] (D2) at (1.5,1.5) {$X_D'$};
      \node[exvar] (ES) at (-1.5,0) {$E_S$};
      \node[var] (S) at (-1.5,-1.5) {$X_S$};
      \node[var] (S2) at (-1.5,1.5) {$X_S'$};
      \node[var] (P) at (0,-1.5) {$X_P$};
      \node[var] (P2) at (0,1.5) {$X_P'$};
      \draw[arr, bend left=20] (P) edge (S);
      \draw[arr, bend left=20] (P) edge (D);
      \draw[arr, bend left=20] (S) edge (P);
      \draw[arr, bend left=20] (D) edge (P);
      \draw[arr, bend left=20] (P2) edge (S2);
      \draw[arr, bend left=20] (P2) edge (D2);
      \draw[arr, bend left=20] (S2) edge (P2);
      \draw[arr, bend left=20] (D2) edge (P2);
      \draw[arr] (ED) to (D);
      \draw[arr] (ES) to (S);
      \draw[arr] (ED) to (D2);
      \draw[arr] (ES) to (S2);
      \draw[arr] (P)  to [out=62,in=118,looseness=4.4] (P);
      \draw[arr] (P2)  to [out=62,in=118,looseness=4.4] (P2);
      \node at (0,-2.2) {$\C{G}^a(\C{M}^\twin)$};
    \end{scope}
    \begin{scope}[xshift=10cm]
%      \node at (-2.5,2) {(c)};
      \node[exvar] (ED) at (1.5,0) {$E_D$};
      \node[var] (D) at (1.5,-1.5) {$X_D$};
      \node[var] (D2) at (1.5,1.5) {$X_D'$};
      \node[exvar] (ES) at (-1.5,0) {$E_S$};
      \node[var] (S) at (-1.5,-1.5) {$X_S$};
      \node[var] (S2) at (-1.5,1.5) {$X_S'$};
      \node[var] (P) at (0,-1.5) {$X_P$};
      \node[var] (P2) at (0,1.5) {$X_P'$};
%      \draw[arr, bend left=20] (P) edge (S);
      \draw[arr, bend left=20] (P) edge (D);
      \draw[arr, bend left=20] (S) edge (P);
      \draw[arr, bend left=20] (D) edge (P);
%      \draw[arr, bend left=20] (P2) edge (S2);
      \draw[arr, bend left=20] (P2) edge (D2);
      \draw[arr, bend left=20] (S2) edge (P2);
      \draw[arr, bend left=20] (D2) edge (P2);
      \draw[arr] (ED) to (D);
%      \draw[arr] (ES) to (S);
      \draw[arr] (ED) to (D2);
%      \draw[arr] (ES) to (S2);
      \draw[arr] (P)  to [out=62,in=118,looseness=4.4] (P);
      \draw[arr] (P2)  to [out=62,in=118,looseness=4.4] (P2);
      \node at (0,-2.2) {$\C{G}^a(\C{M}^\twin)_{\intervene(\{S,S'\})}$};
    \end{scope}
  \end{tikzpicture}}
  \end{center}
  \caption{The augmented graph of the SCM $\C{M}$ (left), its twin SCM $\C{M}^{\twin}$ (center) and the intervened twin SCM $(\C{M}^{\twin})_{\intervene(\{S,S'\},(s,s'))}$ (right) of Examples~\ref{ex:PriceSupplyDemandDynamical} and \ref{ex:PriceSupplyDemand}.}
   \label{fig:PriceSupplyDemand}
\end{figure}
\end{example}

%A typical counterfactual query is a question of the form ``What is $p(\B{X}_{M'} = \B{x}_{M'} \given \intervene(\B{X}_{I'} = \B{\xi }_{I'}), \intervene(\B{X}_J = \B{\eta}_J),\B{X}_{K'}=\B{x}_{K'}, \B{X}_L=\B{x}_L)$?'', where $J,L\subseteq\C{I}$ and $I',K', M' \subseteq\C{I}'$ are all disjoint and we assume that there exists a unique interventional distribution for which a density exists. Such a query reads ``Given that in the actual world we performed a perfect intervention $\intervene(\B{X}_J = \B{\eta}_J)$ and afterwards observed $\B{X}_L = \B{x}_L$, what would have been the probability of $\B{X}_{M'} = \B{x}_{M'}$ in the counterfactual world in which instead we had done the perfect intervention $\intervene(\B{X}_{I'} = \B{\xi}_{I'})$ and had afterwards observed $\B{X}_{K'}=\B{x}_{K'}$?''.

Next, we provide an example of how counterfactuals can be sensibly formulated for cyclic SCMs, namely for the price, supply and demand model at equilibrium.
\begin{example}[Price, supply and demand at equilibrium]
\label{ex:PriceSupplyDemand}
Consider the price, supply and demand model at equilibrium of Example~\ref{ex:PriceSupplyDemandDynamical} given by the SCM $\C{M}$. As an example of a counterfactual query, consider 
$$\Prb(X_P' \given \intervene(X_S = s, X_{S'} = s'), X_P = p)\,,$$
  which denotes the conditional distribution of $X_P'$ given $X_P=p$ of a solution of the intervened twin model $\C{M}^{\twin}_{\intervene(\{S,S'\},(s,s'))}$. In words: how would---\emph{ceteris paribus}---price have been distributed, 
had we intervened to set supplied quantities equal to $s'$, given that actually
we intervened to set supplied quantities equal to $s$ and observed that this
led to price $p$? A straightforward calculation shows that this counterfactual distribution
of price is the Dirac measure on $x_P' = p + (s' - s) / \beta_D$. The augmented graphs of the SCM, its twin graph, and its intervened twin graph are depicted
in Figure~\ref{fig:PriceSupplyDemand}.
\end{example}

%%%%%%%%%%%%%%%%%%%%%%%%%%%%%%%%%%%%%%%%%%%%%%%%%%
\subsection{Additional examples}
\label{app:AppendixAdditionalExamples}
%%%%%%%%%%%%%%%%%%%%%%%%%%%%%%%%%%%%%%%%%%%%%%%%%%

In this subsection, we provide additional examples that support the main text.

%%%%%%%%%%%%%%%%%%%%%%%%%%%%%%%%%%%%%%%%%%%%%%%%%%
\subsubsection*{Section 2}
\addcontentsline{toc}{subsubsection}{\protect\numberline{}Section 2}%
%%%%%%%%%%%%%%%%%%%%%%%%%%%%%%%%%%%%%%%%%%%%%%%%%%

\begin{example}[Structural equations up to almost sure equality]
%[Equivalent SCMs have the same solutions]
\label{ex:StructuralEquationsUpToASEquality}
Consider the SCM $\C{M} = \langle \B{1}, \B{1}, \C{X}, \C{E}, f, \Prb_{\C{E}} \rangle$ with 
$\C{X} = \C{E} = \{-1,0,1\}$, $\Prb_{\C{E}}(\{-1\}) = \Prb_{\C{E}}(\{1\}) = 
\frac{1}{2}$ and $f(x,e) = e^2 + e - 1$. Let $\tilde{\C{M}}$ be the SCM 
$\C{M}$ but with a different causal mechanism $\tilde{f}(x,e) = e$. Then the sets of solutions
of the structural equations agree for both SCMs for $e \in \{-1,+1\}$, while they differ
only for $e = 0$, which occurs with probability zero. Hence, a pair of random variables $(X,E)$ is a solution 
of $\C{M}$ if and only if it is a solution of $\tilde{\C{M}}$.
\end{example}

\begin{example}[The for-all and for-almost-every quantifier do not commute in general]
\label{ex:ForAlmostEveryForAllQuantifier}
  Consider the SCM $\C{M}=\langle \B{2}, \B{1}, \BC{X}, \C{E}, \B{f}, \Prb_{\C{E}} \rangle$ with $\BC{X}=(0,1)^2$, $\C{E}=(0,1)$, the causal mechanism $\B{f}$ given by
$$
f_1(\B{x},e) = x_1 \,,\quad f_2(\B{x},e) = \bm{1}_{\{0\}}(x_1-e)\cdot (x_2 + 1) \,,
$$
and $\Prb_{\C{E}}=\Prb^E$ with $E\sim \C{U}(0,1)$. Define the property
$$
P(\B{x},e) := \begin{cases} 1 &\text{if $\B{x}=\B{f}(\B{x},e)$ holds,} \\  0 &\text{otherwise.}\end{cases}
$$
Then, for all $\B{x}\in\BC{X}$ and for $\Prb_{\C{E}}$-almost every $e\in\C{E}$ the property $P(\B{x},e)$ holds, however for $\Prb_{\C{E}}$-almost every $e\in\C{E}$ and for all $\B{x}\in\BC{X}$ the property $P(\B{x},e)$ does not hold, since for $\Prb_{\C{E}}$-almost every $e\in\C{E}$ the equation $\B{x}=\B{f}(\B{x},e)$ does not hold for $x_1=e$. Hence, in general, for a property $P(\B{x},e)$ we have that for all $\B{x}\in\BC{X}$ and for $\Prb_{\C{E}}$-almost every $e\in\C{E}$ $P(\B{x},e)$ does not imply for $\Prb_{\C{E}}$-almost every $e\in\C{E}$ for all $\B{x}\in\BC{X}$ $P(\B{x},e)$ (see Lemma~\ref{lemm:AlmostAllQuantifierLogic} for additional properties of the for-almost-every quantifier).
\end{example}

%The following example showcases an SCM with two endogenous and three exogenous variables, for which there is no interventionally equivalent SCM (satisfying smoothness constraints) with one exogenous variable taking values in $\mathbb{R}^2$ whose first and second components enter in the first and second structural equation, respectively. In this sense, representing confounders with dependent exogenous variables can be non-trivial in nonlinear models.

\begin{example}[Representation of latent confounders]
  \label{ex:no_nice_reduction}
  Consider the SCM 
  $\C{M}=\langle \B{2}, \B{3}, \RN^{2}, \RN^{3}, \B{f}, \Prb_{\RN^3} \rangle$ with causal mechanism given by
  \begin{align*}
    f_1(e_1,e_3) &= e_1 + e_3 \\
    f_2(x_1,e_2,e_3) &= x_1 e_3 + e_2
\end{align*}
and $\Prb_{\RN^3}$ the standard-normal distribution on $\RN^3$; Figure~\ref{fig:counterexample} (left) shows the corresponding augmented graph.
Then there exists no SCM 
$\C{M}^*=\langle \B{2}, \B{1}, \RN^{2}, \RN^{2}, \B{f}^*, \Prb^*_{\RN^2} \rangle$
that satisfies the following conditions:
  \begin{enumerate}
    \item $\C{M^*}$ is interventionally equivalent to $\C{M}$,
    \item its structural equations have the form
      \begin{align*}
        x_1 &= f^*_1(e_1^*) \\
        x_2 &= f^*_2(x_1,e_2^*), 
      \end{align*}
     where $e_1^*,e_2^*$ are the two components of $e^* = (e_1^*, e_2^*) \in \RN^{2}$,
    \item the function
$e_2^* \mapsto f^*_2(x_1,e_2^*)$ is strictly monotonically increasing for all $x_1 \in \RN$,  
\item the cumulative distribution function $F^*_{2}$ of the second component of $\Prb^*_{\RN^2}$
is continuous and strictly monotonically increasing.
  \end{enumerate}
The augmented graph of such an SCM is shown in Figure~\ref{fig:counterexample} (right).
%\Jonas{Add assumptions}
\begin{figure}
  \adjustbox{scale=0.8,center}{%
  \centerline{
    \begin{tikzpicture}
      \begin{scope}
        \node[var] (X) at (0,0) {$X_1$};
        \node[var] (Y) at (2,0) {$X_2$};
        \node[exvar] (EX) at (-1,1) {$E_1$};
        \node[exvar] (EY) at (3,1) {$E_2$};
        \node[exvar] (EZ) at (1,1) {$E_3$};
        \draw[arr] (X) -- (Y);
        \draw[arr] (EX) -- (X);
        \draw[arr] (EY) -- (Y);
        \draw[arr] (EZ) -- (X);
        \draw[arr] (EZ) -- (Y);
      \end{scope}
      \begin{scope}[xshift=6cm]
        \node[var] (X) at (0,0) {$X_1$};
        \node[var] (Y) at (2,0) {$X_2$};
        \node[exvar] (EX) at (1,1) {$E$};
%        \node[exvar] (EY) at (2,1.75) {$\tilde E_{2}$};
        \draw[arr] (X) -- (Y);
        \draw[arr] (EX) -- (X);
        \draw[arr] (EX) -- (Y);
      \end{scope}
    \end{tikzpicture}
}}
  \caption{\label{fig:counterexample}Augmented graphs of the SCMs $\C{M}$ (left) and $\C{M}^*$ (right) in Example~\ref{ex:no_nice_reduction}. For SCM $\C{M}^*$, the exogenous variable $E$ consists of two real-valued components; the structural equation for $X_1$ depends only on the first, while the structural equation for $X_2$ depends only on the second component.}
\end{figure}

The proof of this statement proceeds by contradiction. Assume that such an SCM $\C{M}^*$ exists.
For any uniquely solvable SCM $\bar{\C{M}}$ and any endogenous variable $i$ appearing in $\bar{\C{M}}$,
we denote with $F^{\bar{\C{M}}}_{X_i}$ the marginal cumulative distribution function of the $i^{\text{th}}$ component
of the observational distribution of $\bar{\C{M}}$.
For all $\xi\in\RN$, we have for all $x_2\in\RN$
  \begin{equation}\label{eq:interventional_cdf_of_counterexample}
    F^{\C{M}_{\intervene(\{1\},\xi)}}_{X_2}(x_2) = \Prb(\xi E_3 + E_2 \le x_2) = \Phi\left({x_2}/{\sqrt{1+\xi^2}}\right),
\end{equation}
where $\Phi$ denotes the (invertible) cdf of the standard-normal distribution.
Now define 
$\phi: \mathbb{R} \rightarrow \mathbb{R}$ with
$\phi(e_2) := \Phi^{-1}(F^*_2(e_2))$
and define the SCM
$\tilde{\C{M}} :=\langle \B{2}, \B{1}, \RN^{2}, \RN^{2}, \tilde{\B{f}}, \tilde{\Prb}_{\RN^2} \rangle$
such that the causal mechanism $\tilde{\B{f}}$ is given by
  \begin{align*}
    \tilde{f}_1(e_1) &= f^*_1(e_1), \\
    \tilde{f}_2(x_1, e_2) &= {f^*_2}(x_1, \phi^{-1}(e_2)),
  \end{align*}
and $\tilde{\Prb}_{\RN^2}$ is the push-forward measure of ${\Prb}^*_{\RN^2}$ using $(\id_\RN,\phi)$.
Then, $\tilde{\C{M}}$ is interventionally equivalent to $\C{M}^*$ by construction,
and the second component of $\tilde{\Prb}_{\RN^2}$ has a standard-normal distribution. 
%\Jonas{add details}
Let $(\tilde{X}_1, \tilde{X}_2, \tilde{E})$ be a solution of $\tilde{\C{M}}$ and let us write
$\tilde{E} = (\tilde{E}_1, \tilde{E}_2)$. 
Then, 
for all $\xi\in\RN$ and $\tilde e_2\in\RN$,
%  \begin{align*}
%  F^{\bar{\C{M}}_{\intervene(\setminus i,\B{\xi}_{\setminus i})}}_{X_i} \big(\bar f_i(\B{\xi}_{\setminus i}, \bar e_i)\big) 
% & = \Prb(\bar f_i(\B{\xi}_{\setminus i},\bar E_i) \le \bar f_i(\B{\xi}_{\setminus i},\bar e_i)) \\
%  & = \Prb(\bar E_i \le \bar e_i) 
%     = F^{\bar{\C{M}}}_{\bar E_i}(\bar e_i).
%\end{align*}
  \begin{equation*}
    F_{X_2}^{\tilde{\C{M}}_{\intervene(\{1\},\xi)}}(\tilde f_2(\xi,\tilde e_2))
  = \Prb(\tilde f_2(\xi,\tilde E_2) 
 \le \tilde f_2(\xi,\tilde e_2))
   = \Prb(\tilde E_2 \le \tilde e_2) 
     %= F^{\tilde{\C{M}}}_{2}(\tilde e_2)
     = \Phi(\tilde e_2),
\end{equation*}
using that $\tilde e_2 \mapsto \tilde f_2(\xi,\tilde e_2)$, too, is strictly monotonically increasing for all $\xi$.
This implies that, 
for all $\xi\in\RN$ and $\tilde{e}_2\in\RN$,
%for $\Prb_{\RN^2}$-almost every $\tilde{\B{e}}$ we have
$$
\tilde f_2(\xi,\tilde e_2) = (F_{X_2}^{\C{M}_{\intervene(\{1\},\xi)}})^{-1}\big(\Phi(\tilde e_2)\big) = \sqrt{1+\xi^2} \; \tilde e_2\,,
$$
where we used interventional equivalence of $\C{M}$ and $\tilde{\C{M}}$, and \eref{eq:interventional_cdf_of_counterexample} for the second equality.
Furthermore,
$\tilde{X}_2 = \tilde f_2(\tilde{X}_1, \tilde{E}_2) = \sqrt{1 + \tilde{X}_1^2}\; \tilde{E}_2$ a.s., so
$\tilde E_2 = {\tilde X_2}/{\sqrt{1 + \tilde X_1^2}} \text{\ a.s.}$.
Now let $(X_1,X_2,E_1,E_2,E_3)$ be a solution of $\C{M}$. By observational equivalence, $(\tilde X_1, \tilde X_2)$ has the same distribution as $(X_1,X_2)$, and thus $\tilde E_2$ is distributed as
$$\frac{X_2}{\sqrt{1 + X_1^2}} = \frac{(E_1 + E_3)E_3 + E_2}{\sqrt{1 + (E_1 + E_3)^2}} \text{\ a.s.}.$$
This contradicts the fact that $\tilde E_2$ 
has a standard-normal distribution as, for example, the mean of the right-hand side is nonzero.
%$$\Exp \left(\frac{(E_1 + E_3)E_3 + E_2}{\sqrt{1 + (E_1 + E_3)^2}}\right) = \Exp \left(\frac{(E_1 + E_3)E_3}{\sqrt{1 + (E_1 + E_3)^2}}\right) \ne 0$$ as one can check by numerical integration, for example.
\end{example}

%In general, interventional equivalence does not imply counterfactual equivalence. Even interventionally equivalent SCMs with the same causal mechanism (that differ only in their exogenous distribution) may not be counterfactually equivalent. For example, the SCMs $\C{M}_{\rho}$ and $\C{M}_{\rho'}$ with $\rho\neq\rho'$ in the following example (due to Dawid~\citep{Daw02}) are interventionally but not counterfactually equivalent.
\begin{example}[Counterfactual density unidentifiable from observational and interventional densities \citep{Daw02}]
\label{ex:Counterfactuals}
Let $\rho\in\RN$ and 
$$
\C{M}_{\rho}=\langle
\B{2},\B{2},\{0,1\}\times\RN,\{0,1\}\times\RN^2,\B{f},\Prb_{\BC{E}}\rangle
$$
be the SCM with causal mechanism given by
$$
  f_1(\B{x},\B{e}) =e_1 \,, \quad f_2(\B{x},\B{e}) =e_{21}(1 - x_1) + e_{22}x_1
$$
and $\Prb_{\BC{E}}=\Prb^{(E_1,\B{E}_2)}$ with $E_1\sim\mathrm{Bernoulli}(1/2)$,
$$
\B{E}_2 := \begin{pmatrix} E_{21} \\ E_{22} \end{pmatrix}\sim\C{N}\left(\B{0} ,\begin{pmatrix}1&\rho\\ \rho&1\end{pmatrix}\right) \,
$$
normally distributed and $E_1 \indep \B{E}_2$. In an epidemiological setting, this SCM could be used to model whether a patient was treated or not ($X_1$) and the corresponding outcome for that patient ($X_2$).

Suppose in the actual world we did not assign treatment to a patient ($X_1=0$) and the outcome was $X_2=c \in\RN$. Consider the counterfactual query ``What would the outcome have been, if we had assigned treatment to this patient?''. We can answer this question by introducing a parallel counterfactual world that is modeled by the twin SCM $\C{M}_{\rho}^{\twin}$, as depicted in Figure~\ref{fig:Counterfactuals}. The counterfactual query then asks for $p(X_{2'}=x_{2'}\mid \intervene(X_{1'}=1, X_1=0),X_2=c)$. One can calculate that
$$
  \begin{pmatrix}X_{2'} \\X_2\end{pmatrix} \mid
  \intervene(X_{1'}=1, X_1=0) \sim \C{N}\left(\B{0}, \begin{pmatrix} 1 & \rho \\ \rho & 1 \end{pmatrix}\right)
$$
and hence $X_{2'} \mid \intervene(X_{1'}=1, X_1=0),X_2=c \sim \C{N}(\rho c,1-\rho^2)$. Note that the answer to the counterfactual query depends on a quantity $\rho$ that we cannot identify from the observational density $p(X_1,X_2)$ or the interventional densities $p(X_2\given \intervene(X_1=0))$ and $p(X_2\given \intervene(X_1=1))$, none of which depends on $\rho$. Therefore, even data from randomized controlled trials combined with observational data would not suffice to determine the value of this particular counterfactual query.
Indeed, SCMs $\C{M}_{\rho}$ and $\C{M}_{\rho'}$ with $\rho\neq\rho'$ are interventionally equivalent, but not counterfactually equivalent.
\end{example}

\begin{figure}[t]
  \begin{center}
  \adjustbox{scale=0.8,center}{%
  \begin{tikzpicture}
    \begin{scope}
    \node[exvar] (E2) at (1.25,0.75) {$\B{E}_2$};
    \node[exvar] (E1) at (1.25,2.0) {$E_1$};
    \node[var] (X1) at (0,1.25) {$X_1$}; 
    \node[var] (X2) at (0,0) {$X_2$}; 
    \draw[arr] (X1) -- (X2);
    \draw[arr] (E1) -- (X1);
    \draw[arr] (E2) -- (X2);
     \end{scope}
  \begin{scope}[shift={(4cm,0cm)}]
    \node[exvar] (E2) at (1.25,0.75) {$\B{E}_2$};
    \node[exvar] (E1) at (1.25,2.0) {$E_1$};
    \node[var] (X1) at (0,1.25) {$X_1$}; 
    \node[var] (X2) at (0,0) {$X_2$}; 
    \node[var] (X1a) at (2.5,1.25) {$X_{1'}$};
    \node[var] (X2a) at (2.5,0) {$X_{2'}$};
    \draw[arr] (X1) -- (X2);
    \draw[arr] (E2) -- (X2);
    \draw[arr] (E1) -- (X1);
    \draw[arr] (E2) -- (X2a);
    \draw[arr] (E1) -- (X1a);
    \draw[arr] (X1a) -- (X2a);
  \end{scope}
  \begin{scope}[shift={(9cm,0cm)}]
    \node[exvar] (E2) at (1.25,0.75) {$\B{E}_2$};
    \node[var] (X1) at (0,1.25) {$X_1$}; 
    \node[var] (X2) at (0,0) {$X_2$}; 
    \node[var] (X1a) at (2.5,1.25) {$X_{1'}$};
    \node[var] (X2a) at (2.5,0) {$X_{2'}$};
    \draw[arr] (X1) -- (X2);
    \draw[arr] (E2) -- (X2);
    \draw[arr] (E2) -- (X2a);
    \draw[arr] (X1a) -- (X2a);
  \end{scope}
  \end{tikzpicture}}
  \end{center}
  \caption{The augmented graph of the SCM $\C{M}_{\rho}$ (left), its twin SCM $\C{M}_{\rho}^{\twin}$ (center) and the intervened twin SCM $(\C{M}_{\rho}^{\twin})_{\intervene(\{1',1\},(1,0))}$ (right) of Example~\ref{ex:Counterfactuals}.}
   \label{fig:Counterfactuals}
\end{figure}

\subsubsection*{Section 3}
\addcontentsline{toc}{subsubsection}{\protect\numberline{}Section 3}%
%%%%%%%%%%%%%%%%%%%%%%%%%%%%%%%%%%%%%%%%%%%%%%%%%%

\begin{example}[Mixtures of solutions are solutions]
\label{ex:MixturesOfSolutions}
Let $\C{M} = \langle \B{1}, \emptyset, \RN, \B{1}, f, \Prb_{\B{1}} \rangle$ be an SCM with causal mechanism $f:\C{X}\times\C{E}\to\C{X}$ defined by $f(x,e)=x - x^2 + 1$. There exist only two measurable solution functions $g_{\pm}:\C{E}\to\C{X}$ for $\C{M}$, defined by $g_{\pm}(e)=\pm 1$. Let $X:\Omega\to\RN$ be a random variable that is a nontrivial mixture of point masses on $\{-1,+1\}$. Then $X$ is a solution of $\C{M}$, however neither $g_{+}(E) = X$ a.s., nor $g_{-}(E) = X$ a.s., for any random variable $E$ such that $\Prb^E = \Prb_{\C{E}}$.
\end{example}

\begin{example}[Solvability is not preserved under perfect intervention]
\label{ex:InterventionUniqueSolvability}
Consider the SCM $\C{M} = \langle \B{2}, \emptyset, \RN^2, \B{1}, \B{f}, \Prb_{\B{1}} \rangle$ with the 
following causal mechanism
$$
	f_1(\B{x}) =x_1+x_1^2-x_2+1 \,, \quad f_2(\B{x}) =x_2(1-\bm{1}_{\{0\}}(x_1))+1 \,.
$$
This SCM has a unique solution $(0,1)$. Doing a perfect intervention $\intervene(\{1\}, \xi_1)$ for some $\xi_1 \neq 0$, however, leads to an intervened model 
$\C{M}_{\intervene(\{1\}, \xi_1)}$ that is not solvable. Performing instead the perfect 
intervention $\intervene(\{2\}, \xi_2)$ for some $\xi_2 > 1$ leads also to a nonuniquely 
  solvable SCM $\C{M}_{\intervene(\{2\}, \xi_2)}$ which has solutions with multiple induced distributions, for example, $(X_1,X_2) = (\phi(\xi_2) \sqrt{\xi_2-1},\xi_2)$ with some measurable $\phi : \RN \to \{-1,+1\}$, but also mixtures of those.
\end{example}

%%%%%%%%%%%%%%%%%%%%%%%%%%%%%%%%%%%%%%%%%%%%%%%%%%
\subsubsection*{Section 4}
\addcontentsline{toc}{subsubsection}{\protect\numberline{}Section 4}%
%%%%%%%%%%%%%%%%%%%%%%%%%%%%%%%%%%%%%%%%%%%%%%%%%%

\begin{figure}
\adjustbox{scale=0.8,center}{%
\begin{tikzpicture}
  \begin{scope}[shift={(0cm,0cm)}]
    \node[var] (X1) at (0,0) {$X_1$};
    \node[var] (X2) at (1.5,0) {$X_2$}; 
    \node[exvar] (E1) at (0,1.25) {$E_1$};
    \node[exvar] (E2) at (1.5,1.25) {$E_2$}; 
    \draw[arr] (X1) -- (X2);
    \draw[arr] (E1) -- (X1);
    \draw[arr] (E2) -- (X2);
    \node (Mb) at (-1.0,-0.2) {$\C{G}^a(\bar{\C{M}})$};
  \end{scope}
  \begin{scope}[shift={(4cm,0cm)}]
    \node[var] (X1) at (0,0) {$X_1$};
    \node[var] (X2) at (1.5,0) {$X_2$}; 
    \node[exvar] (E1) at (0,1.25) {$E_1$};
    \node[exvar] (E2) at (1.5,1.25) {$E_2$}; 
    \draw[arr] (E1) -- (X1);
    \draw[arr] (E2) -- (X2);
    \node (Mh) at (-1.0,-0.2) {$\C{G}^a(\hat{\C{M}})$};
  \end{scope}
  \begin{scope}[shift={(8cm,0cm)}]
    \node[var] (X1) at (0,0) {$X_1$};
    \node[var] (X2) at (1.5,0) {$X_2$}; 
    \node[exvar] (E1) at (0,1.25) {$E_1$};
    \node[exvar] (E2) at (1.5,1.25) {$E_2$}; 
    \draw[arr] (E1) -- (X2);
    \draw[arr] (E1) -- (X1);
    \draw[arr] (E2) -- (X2);
    \node (Mb) at (-1.0,-0.2) {$\C{G}^a(\C{M})$};
  \end{scope}
  \begin{scope}[shift={(12cm,0cm)}]
    \node[var] (X1) at (0,0) {$X_1$};
    \node[var] (X2) at (1.5,0) {$X_2$}; 
    \node[exvar] (E1) at (0,1.25) {$E_1$};
    \node[exvar] (E2) at (1.5,1.25) {$E_2$}; 
    \draw[arr] (E1) -- (X1);
    \draw[arr] (X1) -- (X2);
    \draw[arr] (E1) -- (X2);
    \node (Mh) at (-1.0,-0.2) {$\C{G}^a(\tilde{\C{M}})$};
  \end{scope}
\end{tikzpicture}}
\caption{The augmented graphs of SCMs $\bar{\C{M}}$, $\hat{\C{M}}$, $\C{M}$, and $\tilde{\C{M}}$ that appear in Examples~\ref{ex:InterventionalEquivalenceDoesNotImplySameAugmentedGraph}, \ref{ex:CounterfactualEquivalenceDoesNotImplySameAugmentedGraph}, and \ref{ex:DetectingBiDirectedEdges}.\label{fig:SimpleExamples}}
\end{figure}

\begin{example}[Counterfactually equivalent SCMs with different graphs]
\label{ex:CounterfactualEquivalenceDoesNotImplySameAugmentedGraph}
Consider the SCM $\hat{\C{M}} = 
\langle \B{2},\B{2},\{-1,1\}^2,\{-1,1\}^2,\hat{\B{f}},\Prb_{\BC{E}}\rangle$ with causal mechanism given by
$\hat{f}_1(\B{x},\B{e}) = e_1$ and $\hat{f}_2(\B{x},\B{e}) = e_2$,
and $\Prb_{\BC{E}} = \Prb^{\B{E}}$ with $E_1,E_2 \sim \C{U}(\{-1,1\})$ uniformly
distributed and $E_1 \indep E_2$.
Consider also the SCM $\C{M}$ that is the same as $\hat{\C{M}}$ except for its causal mechanism, which is given by $f_1(\B{x},\B{e}) = e_1$ and $f_2(\B{x},\B{e}) = e_1 e_2$. 
Then $\C{M}$ and $\hat{\C{M}}$ are counterfactually equivalent although $\C{G}(\C{M})$ is 
not equal to $\C{G}(\hat{\C{M}})$ (see Figure~\ref{fig:SimpleExamples}).
\end{example}

%%%%%%%%%%%%%%%%%%%%%%%%%%%%%%%%%%%%%%%%%%%%%%%%%%
\subsubsection*{Section 5}
\addcontentsline{toc}{subsubsection}{\protect\numberline{}Section 5}%
%%%%%%%%%%%%%%%%%%%%%%%%%%%%%%%%%%%%%%%%%%%%%%%%%%

\begin{example}[Marginalization condition of an SCM is not a necessary condition]
\label{ex:MarginalizationSufficientCondition}
Consider the SCM $\C{M} = \langle \B{4}, \B{1}, \RN^4, \RN, \B{f}, \Prb_{\RN} \rangle$ with causal
mechanism given by
$$
  f_1(\B{x},e) = e \,, \quad f_2(\B{x},e) = x_1 \,, \quad f_3(\B{x},e) = x_2 \,,
  \quad f_4(\B{x},e) = x_4 
$$
and $\Prb_{\RN}$ is the standard-normal measure on $\RN$. This SCM is solvable w.r.t.\ $\C{L}=\{2,4\}$, but not
uniquely solvable w.r.t.\ $\C{L}$, and hence we cannot apply Definition~\ref{def:MarginalSCM} to
$\C{L}$. However, the SCM $\tilde{\C{M}}$ on the endogenous variables $\{1,3\}$ with the causal
mechanism $\tilde{\B{f}}$ given by $\tilde{f}_1(\B{x},e) = e$ and $\tilde{f}_3(\B{x},e) = x_1$
%$$
%  \tilde{f}_1(\B{x},e) = e \,, \quad \tilde{f}_3(\B{x},e) = x_1
%$$
is counterfactually equivalent to $\C{M}$ w.r.t.\ $\{1,3\}$, which can be checked easily.  \end{example}

\begin{example}[Graph of the marginal SCM is a strict subgraph of the latent projection]
\label{ex:LatentProjection}
Consider the SCM $\C{M}=\langle \B{3},\B{1},\RN^3,\RN,\B{f},\Prb_{\RN}\rangle$ with causal mechanism given by
$$
  f_1(\B{x},\B{e}) = e_1 \,, \quad
  f_2(\B{x},\B{e}) = x_1 - x_3 \,, \quad
  f_3(\B{x},\B{e}) = x_1
$$
and take for $\Prb_{\RN}$ the standard-normal measure on $\RN$. In contrast, to the (augmented) graph of $\C{M}$, there is no directed path in the (augmented) graph of the marginal SCM $\C{M}_{\marg(\{3\})}$.
\end{example}

%%%%%%%%%%%%%%%%%%%%%%%%%%%%%%%%%%%%%%%%%%%%%%%%%%
\subsubsection*{Section 7}
\addcontentsline{toc}{subsubsection}{\protect\numberline{}Section 7}%
%%%%%%%%%%%%%%%%%%%%%%%%%%%%%%%%%%%%%%%%%%%%%%%%%%

\begin{example}[Detecting a bidirected edge in the graph of an SCM]
\label{ex:DetectingBiDirectedEdges}
Consider the SCM $\bar{\C{M}} = 
\langle \B{2},\B{2},\{-1,1\}^2,\{-1,1\}^2,\bar{\B{f}},\Prb_{\BC{E}}\rangle$ with causal mechanism given by
$\bar{f}_1(\B{x},\B{e}) = e_1$ and $\bar{f}_2(\B{x},\B{e}) = x_1e_2$,
and $\Prb_{\BC{E}} = \Prb^{\B{E}}$ with $E_1,E_2 \sim \C{U}(\{-1,1\})$ uniformly
distributed and $E_1 \indep E_2$.
%Consider the acyclic SCM $\bar{\C{M}}$ of Example~\ref{ex:InterventionalEquivalenceDoesNotImplySameAugmentedGraph} 
Consider also the SCM $\tilde{\C{M}}$ that is the same as $\bar{\C{M}}$ except for its causal mechanism, which is given by $\tilde{f}_1(\B{x},\B{e}) = e_1$ and $\tilde{f}_2(\B{x},\B{e}) = x_1 e_1$. See Figure~\ref{fig:SimpleExamples} for their augmented graphs. For the SCM $\tilde{\C{M}}$ we observe that the marginal interventional distribution $\Prb_{\tilde{\C{M}}_{\intervene(\{1\},\xi_1)}}(X_2=-1)$ is not equal to the conditional distribution $\Prb_{\tilde{\C{M}}}(X_2=-1 \given X_1=\xi_1)$ for both $\xi_1=-1$ and $\xi_1=1$. This observation suffices to identify the presence of the bidirected edge $1\oto 2$ in the graph $\C{G}(\tilde{\C{M}})$. For the SCM $\bar{\C{M}}$, whose graph does not contain the bidirected edge $1\oto 2$, the marginal interventional distribution and conditional distribution coincide.
\end{example}

\section{Proofs}
\label{app:AppendixProofs}
\setcounter{section}{5}
\setcounter{theorem}{0}
%%%%%%%%%%%%%%%%%%%%%%%%%%%%%%%%%%%%%%%%%%%%%%%%%%

This appendix contains the proofs of all the theoretical results in the appendices \ref{app:AppendixCausalGraphicalModels}, \ref{app:AppendixSolvabilityResults} and \ref{app:AppendixLinSCMs}, and the main text. Some of the proofs will rely on the measure theoretic terminology and results of Appendix~\ref{app:AppendixMST}.

%%%%%%%%%%%%%%%%%%%%%%%%%%%%%%%%%%%%%%%%%%%%%%%%%%
\subsection{Proofs of the appendices}
%%%%%%%%%%%%%%%%%%%%%%%%%%%%%%%%%%%%%%%%%%%%%%%%%%

%%%%%%%%%%%%%%%%%%%%%%%%%%%%%%%%%%%%%%%%%%%%%%%%%%
\subsubsection*{Appendix A}
\addcontentsline{toc}{subsubsection}{\protect\numberline{}Appendix A}%
%%%%%%%%%%%%%%%%%%%%%%%%%%%%%%%%%%%%%%%%%%%%%%%%%%

\begin{proof}[Proof of Lemma~\ref{lemm:DWalksPaths}]
It suffices to show that for every $C$-$d$-open walk between $i$ and
$j$ in $\C{G}$, there exists a $C$-$d$-open path between $i$ and
$j$ in $\C{G}$. Take a $C$-$d$-open walk $\pi = (i=i_0, \dots ,i_n=j)$. 
If a node $\ell$ occurs more than once in $\pi$, let $i_j$
be the first occurrence of $\ell$ in $\pi$ and $i_k$ the last occurrence of $\ell$ in $\pi$. 
We now construct a new walk $\pi'$ from $\pi$ by removing the subwalk between $i_j$ and
$i_k$ of $\pi$ from $\pi$. It is easy to check that the new walk $\pi'$ is still $C$-$d$-open. 
If $\ell$ is an endpoint on $\pi'$, then $i_j$ or $i_k$ must be endpoint of $\pi$, and hence
$\ell \notin C$. 
If $\ell$ is a non-endpoint non-collider on $\pi'$, then also $i_j$ or $i_k$ must have been a non-endpoint non-collider on $\pi$,
and hence $\ell \notin C$. If $\ell$ is a collider on $\pi'$, then either
(i) $i_j$ or $i_k$ are both colliders on $\pi$, and hence $\ell$ is ancestor of
$C$ in $\C{G}$, or
(ii) on the subwalk between $i_j$ and $i_k$ that was removed, there must be a directed path in $\C{G}$ from 
$i_j$ or $i_k$ to a collider in $\an_{\C{G}}(C)$, and hence, $\ell$ is in $\an_{\C{G}}(C)$.
The other nodes on $\pi'$ cannot be responsible for $C$-$d$-blocking the walk, since they also occur 
(together with their adjacent edges) on $\pi$ and they do not $C$-$d$-block $\pi$.

In $\pi'$, the number of nodes that occur multiple times is at least one less than in $\pi$. Repeat this procedure until no
repeated nodes are left.
\end{proof}

\begin{proof}[Proof of Theorem~\ref{thm:dgMarkovPropertySCMThreeSpecialCases}]
  The first case is a well known result. An elementary proof is obtained by noting that an acyclic
  system of structural equations trivially satisfies the local directed Markov property,
  and then apply \citep[Proposition 4]{LDLL90}, followed by applying the stability of $d$-separation
  with respect to (graphical) marginalization \citep[Lemma 2.2.15]{FM17}. 
  Alternatively, the result also follows from sequential application of Theorems~3.8.2, 3.8.11, 3.7.7, 3.7.2 and 3.3.3
  (using Remark~3.3.4) in \citep{FM17}.

  The discrete case is proved by the series of results Theorem~3.8.12,
  Remark~3.7.2, Theorem~3.6.6 and 3.5.2 in \citep{FM17}. 
  
  %Here, Theorem~3.8.12 states that an SCM with the ancestral unique solvability property and discrete probability distributions implies the marginal ancestral factorization property, Remark~3.7.2 states that this in turn implies
%the ancestral factorization property, Theorem~3.6.6 states that this in turn implies the ancestral
%undirected global Markov property and Theorem~3.5.2 states that this implies the
%directed global Markov property. 
  
  The linear case is proved in Example~3.8.17 in \citep{FM17}. To connect the assumptions made there
  with the ones we state here, 
  observe that under the linear transformation rule for Lebesgue measures, the image measure of $\Prb_{\BC{E}}$ under the linear mapping $\RN^{\C{J}} \to \RN^{\C{I}} : \B{e} \mapsto \Gamma_{\C{IJ}}\B{e}$ gives a measure on $\BC{X} = \RN^{\C{I}}$ with a density w.r.t.\ the Lebesgue measure on $\RN^{\C{I}}$, as long as the image of the linear mapping is the entire $\RN^{\C{I}}$. This is guaranteed
  if each causal mechanism has a nontrivial dependence on some exogenous variable(s), that is, for each
  $i \in \C{I}$ there is some $j \in \C{J}$ with $\Gamma_{ij } \ne 0$.
  % NOTES: 1. in Example 3.8.17, $\#F$-dimensional exogenous variables are used, but it also
  % works with 1-dimensional exogenous variables like we assume here; 2. in Example 3.8.17, the
  % invertability of $B_A$ for an ancestral subgraph $H$ of $G$ follows since there are no edges
  % from the complement of $A$ into $A$, and hence we get a block structure in $B$ with a block
  % of zeros on one off-diagonal, which directly implies the ancestral solvability;
  % IDEA: we could instead assume that the observational distribution has a density w.r.t.\ Lebesgue measure
\end{proof}

\begin{proof}[Proof of Proposition~\ref{prop:Acyclification}]
This follows directly from the fact that the strongly connected components of
$\C{G}^a(\C{M})$ form a DAG by
Lemma~\ref{lemm:GraphOfStronglyConnectedComponentsDAG} and that the directed edges in
$\C{G}^a(\acy(\C{M}))$ by construction respect every topological ordering of that DAG. 
Both SCMs are observationally equivalent by construction.
\end{proof}

\begin{proof}[Proof of Proposition~\ref{prop:CompatibilityAcyclification}]
This follows immediately from the Definitions~\ref{def:Acyclification} and \ref{def:GraphicalAcyclification}.
\end{proof}

\begin{proof}[Proof of Lemma~\ref{lemm:SigmaWalksPaths}]
It suffices to show that for every $C$-$\sigma$-open walk between $i$ and
$j$ in $\C{G}$, there exists a $C$-$\sigma$-open path between $i$ and
$j$ in $\C{G}$. Let $\pi = (i=i_0,\dots,i_n=j)$ be a $C$-$\sigma$-open walk in $\C{G}$. 
If a node $\ell$ occurs more than once in $\pi$, let $i_j$
be the first node in $\pi$ and $i_k$ the last node in $\pi$ that are in the same 
strongly connected component as $\ell$. 
Since $i_j$ and $i_k$ are in the same strongly connected
component, there are directed paths $i_j \to \dots \to i_k$ and $i_k \to \dots \to i_j$ in $\C{G}$.
We now construct a new walk $\pi'$ from $\pi$ by replacing the subwalk between $i_j$ and
$i_k$ of $\pi$ by a particular directed path between $i_j$ and $i_k$:
(i) If $k = n$, or if $k < n$ and $i_k \to i_{k+1}$ on $\pi$, we replace it by a shortest directed path
$i_j \to \dots \to i_k$, otherwise (ii) we replace it by a shortest directed path 
$i_j \leftarrow \dots \leftarrow i_k$. 
We now show that the new walk $\pi'$ is still $C$-$\sigma$-open. 

$\pi'$ cannot become $C$-$\sigma$-blocked through one of the initial nodes $i_0 \dots i_{j-1}$ or 
one of the final nodes $i_{k+1} \dots i_n$ on $\pi'$, since these nodes occur in the same local
configuration on $\pi$ and do not $C$-$\sigma$-block $\pi$ by assumption.
Furthermore, $\pi'$ cannot become $C$-$\sigma$-blocked through one of the nodes
strictly between $i_j$ and $i_k$ on $\pi'$ (if there are any), since these nodes are 
all non-endpoint non-colliders that only point to nodes in the same strongly connected component on $\pi'$.
Because $\pi$ is $C$-$\sigma$-open, $i_k \notin C$ if $k = n$ or if $i_k \to i_{k+1}$ on $\pi$. 
This holds in particular in case (i). 
Similarly, $i_j \notin C$ if $j = 0$ or $i_{j-1} \ot i_j$ on $\pi$.

In case (i), $\pi'$ is not $C$-$\sigma$-blocked by $i_k$ because $i_k$ is a non-collider on $\pi'$ 
but $i_k \notin C$. Also $i_j$ does not $C$-$\sigma$-block $\pi'$. Assume $i_j \ne i_k$ (otherwise
there is nothing to prove). If $j = 0$, or if $j > 0$ and $i_{j-1} \ot i_j$ on $\pi'$, then the same 
holds for $\pi$ and hence $i_j \notin C$; $i_j$ is then a non-collider on $\pi'$, but $i_j \notin C$. 
If $j > 0$ and $i_{j-1} \oto i_j$ or $i_{j-1} \to i_j$ on $\pi'$ then $i_j$ is a non-endpoint non-collider 
on $\pi'$ that does not point to a node in another strongly connected component.

Now consider case (ii). If $j = 0$ or $i_{j-1} \ot i_j$ on $\pi'$ then this case is analogous
to case (i). So assume $j > 0$ and $i_{j-1} \to i_j$ or $i_{j-1} \oto i_j$ on $\pi'$. 
If $i_j$ is an endpoint of $\pi'$, then $i_j = i_k$ and $k = n$ and therefore $i_k \notin C$, 
and hence $i_j$ and $i_k$ do not $C$-$\sigma$-block $\pi'$. Otherwise, $i_j$ must be a 
collider on $\pi'$ (whether $i_j = i_k$ or not). Then 
on the subwalk of $\pi$ between $i_j$ and $i_k$ there must be a directed path from $i_j$ 
to a collider that is ancestor of $C$, which implies that $i_j$ is itself ancestor of $C$, 
and hence $i_j$ does not $C$-$\sigma$-block $\pi'$. Also $i_k$ cannot $C$-$\sigma$-block $\pi'$.
Assume $i_j \ne i_k$ (otherwise there is nothing to prove). Since $i_k \ot i_{k+1}$ or 
$i_k \oto i_{k+1}$ on $\pi'$, $i_k$ is a non-endpoint non-collider on $\pi'$ that does not point to
a node in another strongly connected component.

Now in $\pi'$, the number of nodes that occurs more than once is at least one less than in $\pi$. Repeat this procedure until no
nodes occur more than once.
\end{proof}

\begin{proof}[Proof of Proposition~\ref{prop:SigmaSeparationAsDSeparation}]
This follows directly as a special case of Corollary~2.8.4 in \citep{FM17}.
\end{proof}

\begin{proof}[Proof of Theorem~\ref{thm:gdgMarkovPropertySCM}]
An SCM $\C{M}$ that is uniquely solvable w.r.t.\ each strongly connected component is
uniquely solvable and hence, by Theorem~\ref{thm:UniqueSolvabilityIffCondition}, all its
solutions have the same observational distribution. 
The last statement follows from the series of results Theorem~3.8.2,
3.8.11, Lemma~3.7.7 and Remark~3.7.2 in \citep{FM17}. Alternatively,
we give here a shorter proof: Under the stated conditions one can
always construct the acyclification $\acy(\C{M})$ which is observationally
equivalent to $\C{M}$ and is acyclic (see Proposition~\ref{prop:Acyclification}) and
hence we can apply Theorem~\ref{thm:dgMarkovPropertySCMThreeSpecialCases} to
$\acy(\C{M})$.
Together with Proposition~\ref{prop:CompatibilityAcyclification} and \ref{prop:SigmaSeparationAsDSeparation} this gives
$$
A \mathrel{\mathop{\sigmablocked}^{\sigma}_{\C{G}(\C{M})}} B \given C
 \ \iff \
A \mathrel{\mathop{\sigmablocked}^{d}_{\acy(\C{G}(\C{M}))}} B \given C 
 \ \implies \
A \mathrel{\mathop{\sigmablocked}^{d}_{\C{G}(\acy(\C{M}))}} B \given C 
 \ \implies \
\B{X}_A \mathrel{\mathop{\indep}_{\Prb_{\C{M}}^{\B{X}}}} \B{X}_B \given \B{X}_C \,,
$$
for $A,B,C \subseteq \C{I}$ and $\B{X}$ a solution of $\C{M}$.
\end{proof}

\begin{proof}[Proof of Corollary~\ref{coro:gdgMarkovPropertyInterventionalSCM}]
First observe that simplicity is preserved under both perfect intervention and the
twin operation (see
Proposition~\ref{prop:SimplenessClosedUnderResults}). Now the first statement follows from
Theorem~\ref{thm:gdgMarkovPropertySCM} if one takes into account the identities of Proposition~\ref{prop:InterventionOnGraph} and \ref{prop:CommuteTwinGraph}. Similarly, the last statement follows from Theorem~\ref{thm:dgMarkovPropertySCMThreeSpecialCases}.
\end{proof}

\begin{proof}[Proof of Proposition~\ref{prop:ExistenceOfCompatibleSystemOfMeasurableSolutionFunctions}]
Let $\tilde{\C{M}} =: \langle \C{V}, \hat{\C{H}}, \BC{X}, \BC{E}, \tilde{\B{f}}, \Prb_{\BC{E}} \rangle$ be the induced SCM. Observe that every loop $\C{O}\in\C{L}(\C{G}(\tilde{\C{M}}))$ is a loop in $\C{L}(\C{G})$. Fix $\check{\B{x}}\in\BC{X}$ and $\check{\B{e}}\in\BC{E}$. For every $\C{O}\in\C{L}(\C{G}(\tilde{\C{M}}))$, define 
$$
I_{\C{O}}:=(\pa_{\C{G}}(\C{O})\setminus\C{O})\setminus(\pa(\C{O})\setminus\C{O})\subseteq\tilde{\C{I}}
$$ 
and
$$
J_{\C{O}}:= \{ \C{F}\in\tilde{\C{J}} \,:\, \C{F}\cap\C{O} \neq\emptyset \} \setminus \pa(\C{O})\subseteq\tilde{\C{J}} \,.
$$
Now, define the family of measurable mappings $(\tilde{\B{g}}_{\C{O}})_{\C{O}\in\C{L}(\C{G}(\tilde{\C{M}}))}$, where the mapping $\tilde{\B{g}}_{\C{O}}:\BC{X}_{\pa(\C{O})\setminus\C{O}}\times\BC{E}_{\pa(\C{O})}\to\BC{X}_{\C{O}}$ is given by 
$$
\tilde{\B{g}}_{\C{O}}(\B{x}_{\pa(\C{O})\setminus\C{O}},\B{e}_{\pa(\C{O})}) :=
\B{g}_{\C{O}}(\B{x}_{\pa(\C{O})\setminus\C{O}},\check{\B{x}}_{I_{\C{O}}},\B{e}_{\pa(\C{O})},\check{\B{e}}_{J_{\C{O}}})
$$
where $\B{x}_{\pa_{\C{G}}(\C{O})\setminus\C{O}}=(\B{x}_{\pa(\C{O})\setminus\C{O}},\check{\B{x}}_{I_{\C{O}}})$ and $\widehat{\B{e}}_{\C{O}} = (\B{e}_{\pa(\C{O})},\check{\B{e}}_{J_{\C{O}}})$. Observe that from the definition of the parents (see Definition~\ref{def:Parents}) it follows that for $\Prb_{\BC{E}}$-almost every $\B{e}\in\BC{E}$ and for all $\B{x}\in\BC{X}$ we have
$$
\B{x}_{\C{O}} = \tilde{\B{f}}_{\C{O}}(\B{x}_{\setminus I_{\C{O}}},\check{\B{x}}_{I_{\C{O}}},\B{e}_{\setminus J_{\C{O}}},\check{\B{e}}_{J_{\C{O}}}) \quad\iff\quad
\B{x}_{\C{O}} = \tilde{\B{f}}_{\C{O}}(\B{x},\B{e}) \,.
$$
This, together with the fact that the family of mappings $(\B{g}_{\C{O}})_{\C{O}\in\C{L}(\C{G})}$ is a compatible system of solution functions, implies that for $\Prb_{\BC{E}}$-almost every $\B{e}\in\BC{E}$ and for all $\B{x}\in\BC{X}$ we have
$$
  \B{x}_{\C{O}} = \tilde{\B{g}}_{\C{O}}(\B{x}_{\pa(\C{O})\setminus\C{O}},\B{e}_{\pa(\C{O})}) \quad\implies\quad
  \B{x}_{\C{O}} = \tilde{\B{f}}_{\C{O}}(\B{x},\B{e}) \,.
$$
Hence, $\iota(\widehat{\C{M}})$ is loop-wisely solvable and thus $(\tilde{\B{g}}_{\C{O}})_{\C{O}\in\C{L}(\C{G}(\tilde{\C{M}}))}$ is a family of measurable solution functions. In particular, for all $\C{O},\tilde{\C{O}}\in\C{L}(\C{G}(\tilde{\C{M}}))$ with $\tilde{\C{O}}\subseteq\C{O}$ and for $\Prb_{\BC{E}}$-almost every $\B{e}\in\BC{E}$ and for all $\B{x}\in\BC{X}$ we have
$$
\B{x}_{\C{O}} = \tilde{\B{g}}_{\C{O}}(\B{x}_{\pa(\C{O})\setminus\C{O}},\B{e}_{\pa(\C{O})}) \quad\implies\quad
\B{x}_{\tilde{\C{O}}} = \tilde{\B{g}}_{\tilde{\C{O}}}(\B{x}_{\pa(\tilde{\C{O}})\setminus\tilde{\C{O}}},\B{e}_{\pa(\tilde{\C{O}})})  \,.
$$
From this we conclude that $(\tilde{\B{g}}_{\C{O}})_{\C{O}\in\C{L}(\C{G}(\tilde{\C{M}}))}$ is a compatible system of solution functions.
\end{proof}

\begin{proof}[Proof of Lemma~\ref{lemm:LoopwiseUniqueSolvabilityAndSimple}]
Suppose $\C{M}$ is loop-wisely uniquely solvable and consider a subset $\C{O}\subseteq\C{I}$. Consider the induced subgraph 
$\C{G}^a(\C{M})_{\C{O}}$ of $\C{G}^a(\C{M})$ on the nodes $\C{O}$. Then every strongly connected
component of $\C{G}^a(\C{M})_{\C{O}}$ is
an element of $\C{L}(\C{G}(\C{M}))$. Let $\C{C}$ be such a strongly connected component in $\C{G}^a(\C{M})_{\C{O}}$, and let $\B{g}_{\C{C}}:\BC{X}_{\pa(\C{C})\setminus\C{C}}\times\BC{E}_{\pa(\C{C})}\to\BC{X}_{\C{C}}$ be a measurable solution function for $\C{M}$ w.r.t.\ $\C{C}$. Since $\C{G}^a(\C{M})_{\C{O}}$
partitions into strongly connected components, we can recursively (by following a topological
ordering of the DAG $\C{G}^a(\C{M})_{\C{O}}^{\scc}$ from Lemma~\ref{lemm:GraphOfStronglyConnectedComponentsDAG}) insert these
mappings into each other to obtain a mapping
$\B{g}_{\C{O}}:\BC{X}_{\pa(\C{O})\setminus\C{O}}\times\BC{E}_{\pa(\C{O})}\to\BC{X}_{\C{O}}$ that
makes $\C{M}$ uniquely solvable w.r.t.\ $\C{O}$. 
\end{proof}

\begin{proof}[Proof of Proposition~\ref{prop:SimpleSCMsCompatibleSysOfSolFunctions}]
  Let $(\B{g}_{\C{O}})_{\C{O}\in\C{L}(\C{G}(\C{M}))}$ be any family of measurable solution functions, where $\B{g}_{\C{O}}$ is measurable solution function of $\C{M}$ w.r.t.\ $\C{O}$. Then, for $\C{O},\tilde{\C{O}}\in\C{L}(\C{G}(\C{M}))$ such that $\tilde{\C{O}}\subseteq\C{O}$, we have that for $\Prb_{\BC{E}}$-almost every $\B{e}\in\BC{E}$ and for all $\B{x}\in\BC{X}$
$$
  \B{x}_{\C{O}} = \B{f}_{\C{O}}(\B{x},\B{e}) 
  \quad\implies\quad
  \B{x}_{\tilde{\C{O}}} = \B{f}_{\tilde{\C{O}}}(\B{x},\B{e}) \,.
$$
This implies that for $\Prb_{\BC{E}}$-almost every $\B{e}\in\BC{E}$ and for all $\B{x}\in\BC{X}$
$$
\B{x}_{\C{O}} = \B{g}_{\C{O}}(\B{x}_{\pa(\C{O})\setminus\C{O}},\B{e}_{\pa(\C{O})}) 
\quad\implies\quad
\B{x}_{\tilde{\C{O}}} = \B{g}_{\tilde{\C{O}}}(\B{x}_{\pa(\tilde{\C{O}})\setminus\tilde{\C{O}}},\B{e}_{\pa(\tilde{\C{O}})}) \,.
$$
\end{proof}

\begin{proof}[Proof of Corollary~\ref{coro:SufficientConditionCausesConfounder}]
This follows directly from Proposition~\ref{prop:DirectedPathEdges} and \ref{prop:BidirectedEdges}.
\end{proof}

%%%%%%%%%%%%%%%%%%%%%%%%%%%%%%%%%%%%%%%%%%%%%%%%%%
\subsubsection*{Appendix B}
\addcontentsline{toc}{subsubsection}{\protect\numberline{}Appendix B}%
%%%%%%%%%%%%%%%%%%%%%%%%%%%%%%%%%%%%%%%%%%%%%%%%%%

\begin{proof}[Proof of Proposition~\ref{prop:SolvabilityIfCondition}]
Let $\tilde{\B{f}} : \BC{E} \times \BC{X} \to \BC{X}$ be the causal mechanism of a
structurally minimal SCM that is equivalent to $\C{M}$ (see Proposition~\ref{prop:StructMinimalRepresentation}).
In particular, for any $\B{\epsilon}_{\setminus\pa(\C{O})} \in \BC{E}_{\setminus\pa(\C{O})}$ and
$\B{\xi}_{\setminus\pa(\C{O})} \in \BC{X}_{\setminus\pa(\C{O})}$, we have that for all $\B{x} \in \BC{X}$ and all $\B{e} \in \BC{E}$,
$\tilde{\B{f}}(\B{x},\B{e}) = \tilde{\B{f}}(\B{x}_{\pa(\C{O})},\B{\xi}_{\setminus\pa(\C{O})},\B{e}_{\pa(\C{O})},\B{\epsilon}_{\setminus\pa(\C{O})})$.
%Pick $\B{\epsilon}_{\setminus\pa(\C{O})} \in \BC{E}_{\setminus\pa(\C{O})}$ and
%$\B{\xi}_{\setminus\pa(\C{O})} \in \BC{X}_{\setminus\pa(\C{O})}$ and define the mapping
%$$\hat{\B{f}} : \BC{X}_{\pa(\C{O})} \times \BC{E}_{\pa(\C{O})} \to \BC{X} : 
%(\B{x}_{\pa(\C{O})},\B{e}_{\pa(\C{O})}) \mapsto \tilde{\B{f}}(\B{x}_{\pa(\C{O})},\B{\xi}_{\setminus\pa(\C{O})},\B{e}_{\pa(\C{O})},\B{\epsilon}_{\setminus\pa(\C{O})}).$$
This means that we may also consider $\tilde{\B{f}}$ as a mapping $\tilde{\B{f}}:\BC{X}_{\pa(\C{O})} \times \BC{E}_{\pa(\C{O})} \to \BC{X}$.

Consider the set 
$$
%  \begin{aligned}
  \tilde{\BC{S}}:=\{ %&
  (\B{e}_{\pa(\C{O})},\B{x}_{\pa(\C{O})\setminus\C{O}},\B{x}_{\C{O}}) \in \BC{E}_{\pa(\C{O})}\times\BC{X}_{\pa(\C{O})\setminus\C{O}}\times\BC{X}_{\C{O}} \,:\, %\\ &
  \B{x}_{\C{O}} = \tilde{\B{f}}_{\C{O}}(\B{x}_{\pa(\C{O})},\B{e}_{\pa(\C{O})}) \} \,.
%  \end{aligned}
$$
By similar reasoning as in the proof of Theorem~\ref{thm:SolvabilityIffCondition}, $\tilde{\BC{S}}$ is measurable.

By assumption, for $\Prb_{\BC{E}}$-almost every $\B{e}\in\BC{E}$ and for all $\B{x}_{\setminus \C{O}}\in\BC{X}_{\setminus\C{O}}$ the space
$\{ \B{x}_{\C{O}}\in\BC{X}_{\C{O}} : \B{x}_{\C{O}} = \B{f}_{\C{O}}(\B{x},\B{e}) \}$
is nonempty and $\sigma$-compact.
By applying Lemma~\ref{lemm:MeasurableMapsBetweenStandardSpaces} to the canonical projection $\B{pr}_{\BC{E}_\pa(\C{O})}:\BC{E} \to \BC{E}_{\pa(\C{O})}$
%$\B{pr}_{\BC{E}_{\pa(\C{O})}} : \BC{E} \to \BC{E}_{\pa(\C{O})}$
and using the equivalence of $\B{f}$ and $\tilde{\B{f}}$, we obtain that 
for $\Prb_{\BC{E}_{\pa(\C{O})}}$-almost every $\B{e}_{\pa(\C{O})}\in\BC{E}_{\pa(\C{O})}$ and for all $\B{x}_{\pa(\C{O})\setminus \C{O}}\in\BC{X}_{\pa(\C{O})\setminus\C{O}}$ the space
  $$
  \tilde{\BC{S}}_{(\B{e}_{\pa(\C{O})},\B{x}_{\pa(\C{O})\setminus\C{O}})} := 
  \{ \B{x}_{\C{O}}\in\BC{X}_{\C{O}} : \B{x}_{\C{O}} = \tilde{\B{f}}_{\C{O}}(\B{x}_{\pa(\C{O})},\B{e}_{\pa(\C{O})}) \}$$
is nonempty and $\sigma$-compact.

The second measurable selection theorem, Theorem~\ref{thm:MeasurableSelectionThmSigmaCompact}, now implies that there exists
a measurable $\B{g}_{\C{O}} : \BC{X}_{\pa(\C{O})\setminus\C{O}} \times \BC{E}_{\pa(\C{O})} \to \BC{X}_{\C{O}}$ such that for
  $\Prb_{\BC{E}_{\pa(\C{O})}}$-almost every $\B{e}_{\pa(\C{O})}\in\BC{E}_{\pa(\C{O})}$ and for all $\B{x}_{\pa(\C{O})\setminus\C{O}} \in \BC{X}_{\pa(\C{O})\setminus\C{O}}$
  $$\B{g}_{\C{O}}(\B{x}_{\pa(\C{O})\setminus\C{O}},\B{e}_{\pa(\C{O})}) = \tilde{\B{f}}_{\C{O}}\big(\B{x}_{\pa(\C{O})\setminus\C{O}},\B{g}_{\C{O}}(\B{x}_{\pa(\C{O})\setminus\C{O}},\B{e}_{\pa(\C{O})}),\B{e}_{\pa(\C{O})}\big).$$
Once more applying Lemma~\ref{lemm:MeasurableMapsBetweenStandardSpaces}, we obtain that for $\Prb_{\BC{E}}$-almost every $\B{e}\in\BC{E}$ and for all $\B{x} \in \BC{X}$ 
  $$\B{x}_{\C{O}} = \B{g}_{\C{O}}(\B{x}_{\pa(\C{O})\setminus\C{O}},\B{e}_{\pa(\C{O})}) \implies \B{x}_{\C{O}} = \B{f}_{\C{O}}(\B{x},\B{e}).$$
Hence $\C{M}$ is solvable w.r.t.\ $\C{O}$.
\end{proof}

\begin{proof}[Proof of Proposition~\ref{prop:UniqueSolvabilityAncestralUnion}]
\Joris{Fixed a bug in the previous proof. We need the additional assumption that
$\C{M}$ is uniquely solvable w.r.t.\ the intersection $\C{A} \cap \tilde{\C{A}}$.}
Without loss of generality, we assume that $\C{M}$ is structurally minimal (see Proposition~\ref{prop:StructMinimalRepresentation}).
Define $\C{C} := \C{A} \cap \tilde{\C{A}}$ and $\C{D} := \C{A} \cup \tilde{\C{A}}$.
Let $\B{g}_{\C{A}}$, $\B{g}_{\tilde{\C{A}}}$ be measurable solution functions for $\C{M}$ 
w.r.t.\ $\C{A}$ and $\tilde{\C{A}}$, respectively. 
Note that $\pa(\C{C}) \setminus \C{C} \subseteq \pa(\C{A}) \setminus \C{A}$ and
similarly $\pa(\C{C}) \setminus \C{C} \subseteq \pa(\tilde{\C{A}}) \setminus \tilde{\C{A}}$.
Indeed, for $c \in \pa(\C{C})$: if $c \in \C{O}$ then $c \in \C{C}$ because
$\C{A}$ and $\tilde{\C{A}}$ are both ancestral in $\C{G}(\C{M})_{\C{O}}$, 
while if $c \notin \C{O}$ then $c \notin \C{A}$ and $c \notin \tilde{\C{A}}$.
Hence by Lemma~\ref{lemm:ConsistentSolutionFunctions}, for
$\Prb_{\BC{E}}$-almost all $\B{e}\in\BC{E}$ and for all $\B{x}\in\BC{X}$
$$(\B{g}_{\C{A}})_{\C{C}}(\B{x}_{\pa(\C{A})\setminus \C{A}},\B{e}_{\pa(\C{A})}) = 
  (\B{g}_{\tilde{\C{A}}})_{\C{C}}(\B{x}_{\pa(\tilde{\C{A}})\setminus \tilde{\C{A}}},\B{e}_{\pa(\tilde{\C{A}})}) \,.$$
Hence for $\Prb_{\BC{E}}$-almost every $\B{e}\in\BC{E}$ and for all $\B{x}\in\BC{X}$
$$ 
  \begin{aligned}
  &\B{x}_{\C{D}} = \B{f}_{\C{D}}(\B{x},\B{e}) \\
  \iff
  &\begin{cases}
    \B{x}_{\C{A} \setminus \C{C}} &= \B{f}_{\C{A} \setminus \C{C}}(\B{x},\B{e}) \\
    \B{x}_{\C{C}} &= \B{f}_{\C{C}}(\B{x},\B{e}) \\
    \B{x}_{\C{C}} &= \B{f}_{\C{C}}(\B{x},\B{e}) \\
    \B{x}_{\tilde{\C{A}} \setminus \C{C}} &= \B{f}_{\tilde{\C{A}} \setminus \C{C}}(\B{x},\B{e})
  \end{cases} \\
  \iff
  &\begin{cases}
    \B{x}_{\C{A} \setminus \C{C}} &= (\B{g}_{\C{A}})_{\C{A}\setminus \C{C}}(\B{x}_{\pa(\C{A})\setminus \C{A}},\B{e}_{\pa(\C{A})}) \\
    \B{x}_{\C{C}} &= (\B{g}_{\C{A}})_{\C{C}}(\B{x}_{\pa(\C{A})\setminus \C{A}},\B{e}_{\pa(\C{A})}) \\
    \B{x}_{\C{C}} &= (\B{g}_{\tilde{\C{A}}})_{\C{C}}(\B{x}_{\pa(\tilde{\C{A}})\setminus \tilde{\C{A}}},\B{e}_{\pa(\tilde{\C{A}})}) \\
    \B{x}_{\tilde{\C{A}}\setminus \C{C}} &= (\B{g}_{\tilde{\C{A}}})_{\tilde{\C{A}}\setminus \C{C}}(\B{x}_{\pa(\tilde{\C{A}})\setminus \tilde{\C{A}}},\B{e}_{\pa(\tilde{\C{A}})})
  \end{cases}  \\
  \iff
  &\begin{cases}
    \B{x}_{\C{A}} &= \B{g}_{\C{A}}(\B{x}_{\pa(\C{A})\setminus \C{A}},\B{e}_{\pa(\C{A})}) \\
    \B{x}_{\tilde{\C{A}}} &= \B{g}_{\tilde{\C{A}}}(\B{x}_{\pa(\tilde{\C{A}})\setminus \tilde{\C{A}}},\B{e}_{\pa(\tilde{\C{A}})}) \,.
  \end{cases} 
\end{aligned}
$$
Now $\pa(\C{A}) \setminus \C{A} \subseteq \pa(\C{D}) \setminus \C{D}$, and similarly,
$\pa(\tilde{\C{A}}) \setminus \tilde{\C{A}} \subseteq \pa(\C{D}) \setminus \C{D}$.
%Now, since 
%$\pa(\bar{A})\setminus \bar{A},\pa(\bar{B})\setminus \bar{B}\subseteq
%\pa(D)\setminus D$ 
%(For $V \subseteq U \subseteq \C{O}$:
%$\pa(V) \setminus \C{O} \subseteq \pa(U) \setminus U$, and if $V$ is ancestral in $\C{G}_{\C{O}}$
%then $(\pa(V) \setminus V) \cap \C{O} = \emptyset$.)
Hence, we conclude that the mapping
$\B{h}_{\C{D}}:\BC{X}_{\pa(\C{D})\setminus \C{D}}\times\BC{E}_{\pa(\C{D})}\to\BC{X}_{\C{D}}$ defined by 
$$
  \begin{aligned}
    \B{h}_{\C{D}}(&\B{x}_{\pa(\C{D})\setminus \C{D}},\B{e}_{\pa(\C{D})}):=
\\
    \big(&(\B{g}_{\C{A}})_{\C{A}\setminus \C{C}}(\B{x}_{\pa(\C{A})\setminus \C{A}},\B{e}_{\pa(\C{A})}), (\B{g}_{\C{A}})_{\C{C}}(\B{x}_{\pa(\C{A})\setminus \C{A}},\B{e}_{\pa(\C{A})}), (\B{g}_{\tilde{\C{A}}})_{\tilde{\C{A}}\setminus \C{C}}(\B{x}_{\pa(\tilde{\C{A}})\setminus \tilde{\C{A}}},\B{e}_{\pa(\tilde{\C{A}})})
    \big)
  \end{aligned}
$$
is a measurable solution function for $\C{M}$ w.r.t.\ $\C{D}$, and that $\C{M}$ is uniquely solvable w.r.t.\ $\C{D}$.
\end{proof}

\begin{proof}[Proof of Corollary~\ref{coro:UniqueSolvabilityAncestralUniqueSolvability}]
%  \Joris{Fixed a bug. The old proof was incorrect because of the mistake in Proposition~\ref{prop:UniqueSolvabilityAncestralUnion}, which failed to state the important condition that the SCM is uniquely solvable w.r.t.\ the intersection of the two sets.}
%  It suffices to show the implication to the left. Applying several times Proposition~\ref{prop:UniqueSolvabilityAncestralUnion} gives that $\C{M}$ is uniquely solvable w.r.t.\ $\an_{\C{G}(\C{M})_{\C{O}}}(\C{A}) = \bigcup_{i \in \C{A}} \an_{\C{G}(\C{M})_{\C{O}}}(i)$ for every $\C{A}\subseteq\C{O}$. \Joris{This proof is incomplete. To apply Proposition~\ref{prop:UniqueSolvabilityAncestralUnion} to a pair of ancestral subsets, one needs to know that the intersection of the two is also uniquely solvable. Here is a proof sketch.
  It suffices to show the implication to the left. We have to show that $\C{M}$ is uniquely solvable w.r.t.\ each ancestral subset of $\C{G}(\C{M})_{\C{O}}$. The proof proceeds via induction with respect to the size of the ancestral subset. For ancestral subsets of size 0, the claim is trivially true. Ancestral subsets of size 1 must be of the form $\{i\} = \an_{\C{G}(\C{M})_{\C{O}}}(i)$ for $i\in\C{O}$ and hence the claim is true by assumption. Assume that the claim holds for all ancestral subsets of size $\le n$. Let $\C{A}$ be an ancestral subset of $\C{G}(\C{M})_{\C{O}}$ of size $n+1$. If $\C{A} = \an_{\C{G}(\C{M})_{\C{O}}}(i)$ for some $i \in \C{O}$ then the claim holds for $\C{A}$ by assumption. Otherwise, $\C{A} = \bigcup_{i\in\C{A}} \an_{\C{G}(\C{M})_{\C{O}}}(i)$ is a union of ancestral subsets of size $\le n$. Choose distinct elements $\{i_1,\dots,i_k\} \subseteq \C{A}$ where $k$ is the smallest integer such that $\bigcup_{j=1}^k \an_{\C{G}(\C{M})_{\C{O}}}(i_j) = \C{A}$. By applying Proposition~\ref{prop:UniqueSolvabilityAncestralUnion} to $\bigcup_{j=1}^{k-1} \an_{\C{G}(\C{M})_{\C{O}}}(i_j)$ and $\an_{\C{G}(\C{M})_{\C{O}}}(i_k)$, thereby noting that the intersection of these two sets is an ancestral subset of size $\le n$ and making use of the induction hypothesis, we arrive at the conclusion that $\C{M}$ is uniquely solvable w.r.t.\ $\C{A}$.
\end{proof}

%%%%%%%%%%%%%%%%%%%%%%%%%%%%%%%%%%%%%%%%%%%%%%%%%%
\subsubsection*{Appendix C}
\addcontentsline{toc}{subsubsection}{\protect\numberline{}Appendix C}%
%%%%%%%%%%%%%%%%%%%%%%%%%%%%%%%%%%%%%%%%%%%%%%%%%%

\begin{proof}[Proof of Proposition~\ref{prop:LinearSolvable}]
Let $\B{e} \in \BC{E}$ and $\B{x}_{\C{O}}\in\BC{X}_{\C{O}}$. For $\B{x}_{\C{L}}\in\BC{X}$,
$$
  \begin{aligned}
    & \B{x}_{\C{L}}=\B{f}_{\C{L}}(\B{x},\B{e}) \\
    &\iff \B{x}_{\C{L}}=B_{\C{LL}}\B{x}_{\C{L}}+B_{\C{LO}}\B{x}_{\C{O}}+\Gamma_{\C{LJ}}\B{e} \\
%    &\iff (\mathbb{I}_{\C{L}}-B_{\C{LL}})\B{x}_{\C{L}}=B_{\C{LO}}\B{x}_{\C{O}}+\Gamma_{\C{LJ}}\B{e} \\
%    &\iff \C{A}_{\C{L}\C{L}}\B{x}_{\C{L}}=\B{b}_{\C{L}} \\
    &\iff \C{A}_{\C{L}\C{L}}\B{x}_{\C{L}}=B_{\C{LO}}\B{x}_{\C{O}}+\Gamma_{\C{LJ}}\B{e} \\
    &\iff \begin{cases}
      A_{\C{L}\C{L}} A_{\C{L}\C{L}}^+ (B_{\C{LO}}\B{x}_{\C{O}}+\Gamma_{\C{LJ}}\B{e})=B_{\C{LO}}\B{x}_{\C{O}}+\Gamma_{\C{LJ}}\B{e}  & \\
      \exists_{\B{v}\in\BC{X}_{\C{L}}} : \B{x}_{\C{L}} = A_{\C{L}\C{L}}^+(B_{\C{LO}}\B{x}_{\C{O}}+\Gamma_{\C{LJ}}\B{e}) + [\mathbb{I}_{\C{L}} - A_{\C{L}\C{L}}^+A_{\C{L}\C{L}}]\B{v} \,, &
    \end{cases}
  \end{aligned}
$$
  where the last equivalence follows from \citep[Theorem~2,][]{Pen55}.
  %where $\B{b}_{\C{L}}:=B_{\C{LO}}\B{x}_{\C{O}}+\Gamma_{\C{LJ}}\B{e}$. The equation $\C{A}_{\C{L}\C{L}}\B{x}_{\C{L}}=\B{b}_{\C{L}}$ has a solution iff $A_{\C{L}\C{L}} A_{\C{L}\C{L}}^+ \B{b}_{\C{L}}=\B{b}_{\C{L}}$ holds (see Theorem~2 in \citep{Pen55}), in which case the general solution is given by
%$$
%\B{x}_{\C{L}} = A_{\C{L}\C{L}}^+\B{b}_{\C{L}} + [\mathbb{I}_{\C{L}} - A_{\C{L}\C{L}}^+A_{\C{L}\C{L}}]\B{v} \,,
%$$
%for any arbitrary vector $\B{v}\in\BC{X}_{\C{L}}$.
\end{proof}

\begin{proof}[Proof of Proposition~\ref{prop:LinearUniquelySolvable}]
$\C{M}$ is uniquely solvable w.r.t.\ $\C{L}$ if and only if for $\Prb_{\BC{E}}$-almost every 
  $\B{e}\in\BC{E}$ and for all $\B{x}_{\C{O}}\in\BC{X}_{\C{O}}$ the linear system of equations
$$
  \begin{aligned}
    & \B{x}_{\C{L}}=\B{f}_{\C{L}}(\B{x},\B{e}) \\
    &\iff \B{x}_{\C{L}}=B_{\C{LL}}\B{x}_{\C{L}}+B_{\C{LO}}\B{x}_{\C{O}}+\Gamma_{\C{LJ}}\B{e} \\
    &\iff A_{\C{LL}}\B{x}_{\C{L}}=B_{\C{LO}}\B{x}_{\C{O}}+\Gamma_{\C{LJ}}\B{e}
  \end{aligned}
$$
has a unique solution $\B{x}_{\C{L}}\in\C{X}_{\C{L}}$. Hence, $\C{M}$ is uniquely solvable w.r.t.\ $\C{L}$ if and only if $A_{\C{LL}}$ is invertible.
\end{proof}

\begin{proof}[Proof of Proposition~\ref{prop:LinearEquivalentUniqueSolvability}]
  It suffices to show $(1) \implies (2)$ and $(1) \iff (3)$. We start by showing that $(1) \implies (2)$. Let $\C{V}\subseteq\C{L}$ and denote $\C{U}:=\an_{\C{G}(\C{M})_{\C{L}}}(\C{V})$, then we need to show that $\C{M}$ is uniquely solvable w.r.t.\ $\C{U}$. From Proposition~\ref{prop:LinearUniquelySolvable} we know that $\C{M}$ is uniquely solvable w.r.t.\ $\C{L}$ if and only if the matrix $A_{\C{LL}} = \mathbb{I}_{\C{L}} - B_{\C{LL}}$ is invertible. The matrix $A_{\C{LL}}$ is invertible if and only if the rows of $A_{\C{LL}}$ are all linearly independent. In particular, the rows of $A_{\C{U}\C{L}}$ are all linearly independent. Because $A_{\C{U}\C{L}} = [A_{\C{U}\C{U}} \, Z_{\C{U}\C{L}}]$, where $Z_{\C{U}\C{L}}$ is the zero matrix, we know that the rows of $A_{\C{U}\C{U}}=\mathbb{I}_{\C{U}} - B_{\C{U}\C{U}}$ are also all linearly independent, and hence $A_{\C{U}\C{U}}$ is invertible.
%(see Section~2.4 in \citep{Str05} or the paper "The Fundamental Theorem of Linear Algebra" by Gilbert Strang)

  Next, we show that $(1) \iff (3)$. Observe that the strongly connected components of $\C{G}(\C{M})_{\C{L}}$ form a partition of the set $\C{L}$ and that the directed mixed graph $\C{G}(\C{M})_{\C{L}}$ and the directed graph $\C{G}^a(\C{M})_{\C{L}}$ have the same strongly connected components. Because, by Lemma~\ref{lemm:GraphOfStronglyConnectedComponentsDAG}, the graph of strongly connected components $\C{G}^{\scc}$ of the directed graph $\C{G}^a(\C{M})_{\C{L}}$ is a DAG, the square matrix $B_{\C{LL}}$ can be permuted to an upper triangular block matrix $\tilde{B}_{\C{LL}}$, where for each diagonal block $\tilde{B}_{\C{VV}}$ of $\tilde{B}_{\C{LL}}$ the set of nodes $\C{V}$ is a strongly connected component in $\C{G}(\C{M})_{\C{L}}$.

  Without loss of generality we assume now that $B_{\C{LL}}$ is an upper triangular block matrix. From Proposition~\ref{prop:LinearUniquelySolvable} it follows that $\C{M}$ is uniquely solvable w.r.t.\ $\C{L}$ if and only if the matrix $A_{\C{LL}}=\mathbb{I}_{\C{L}} - B_{\C{LL}}$ is invertible. Because $B_{\C{LL}}$ is an upper triangular block matrix, we know that $A_{\C{LL}}$ is an upper triangular block matrix, where for each diagonal block $A_{\C{VV}}$ of $A_{\C{LL}}$ the set of nodes $\C{V}$ is a strongly connected component in $\C{G}(\C{M})_{\C{L}}$. Since an upper triangular block matrix $A_{\C{LL}}$ is invertible if and only if every diagonal block in $A_{\C{LL}}$ is invertible, we have that $\C{M}$ is uniquely solvable w.r.t.\ $\C{L}$ if and only if $\C{M}$ is uniquely solvable w.r.t.\ each strongly connected component in $\C{G}(\C{M})_{\C{L}}$. 
\end{proof}

\begin{proof}[Proof of Proposition~\ref{prop:MarginalLinearModel}]
By the definition of marginalization and Proposition~\ref{prop:LinearUniquelySolvable} the marginal causal mechanism $\tilde{\B{f}}$ is given by
$$
  \begin{aligned}
    \tilde{\B{f}}(\B{x}_{\C{O}},\B{e}) &:= 
    \B{f}_{\C{O}}(\B{x}_{\C{O}},\B{g}_{\C{L}}(\B{x}_{\C{O}},\B{e}),\B{e}) \\
    &= B_{\C{OO}}\B{x}_{\C{O}}+B_{\C{OL}}\B{g}_{\C{L}}(\B{x}_{\C{O}},\B{e})+\Gamma_{\C{OJ}}\B{e} \\
    &= [B_{\C{OO}}+B_{\C{OL}}A_{\C{LL}}^{-1}B_{\C{LO}}]\B{x}_{\C{O}} + [B_{\C{OL}}A_{\C{LL}}^{-1}\Gamma_{\C{LJ}}+\Gamma_{\C{OJ}}]\B{e} \,.
  \end{aligned}
$$
From Propositions~\ref{prop:LinearEquivalentUniqueSolvability} and \ref{prop:LatentProjection} it follows that the marginalization respects the latent projection.
\end{proof}

%%%%%%%%%%%%%%%%%%%%%%%%%%%%%%%%%%%%%%%%%%%%%%%%%%
\subsection{Proofs of the main text}
%%%%%%%%%%%%%%%%%%%%%%%%%%%%%%%%%%%%%%%%%%%%%%%%%%

%%%%%%%%%%%%%%%%%%%%%%%%%%%%%%%%%%%%%%%%%%%%%%%%%%
\subsubsection*{Section 2}
\addcontentsline{toc}{subsubsection}{\protect\numberline{}Section 2}%
%%%%%%%%%%%%%%%%%%%%%%%%%%%%%%%%%%%%%%%%%%%%%%%%%%

\begin{proof}[Proof of Proposition~\ref{prop:StructMinimalRepresentation}]
Let $i\in\C{I}$. Note that Definition~\ref{def:Parents} can alternatively be formulated as follows:
for $k \in \C{I} \cup \C{J}$, $k \not\in \pa(i)$ if and only if there exists a measurable mapping $\hat f_i : \BC{X} \times \BC{E} \to \C{X}_i$ such that
for $\Prb_{\BC{E}}$-almost every $\B{e} \in \BC{E}$ and for all $\B{x}\in\BC{X}$,
  $$x_i = f_i(\B{x},\B{e}) \iff x_i = \hat f_i(\B{x},\B{e})$$
and either $k \in \C{I}$ and there exists $\hat x_k \in \C{X}_k$ such that $\hat f_i(\B{x},\B{e}) = \hat f_i(\B{x}_{\setminus k},\hat x_k,\B{e})$ for all $\B{x} \in \BC{X}, \B{e} \in \BC{E}$, or $k \in \C{J}$ and
there exists $\hat e_k \in \C{E}_k$ such that $\hat f_i(\B{x},\B{e}) = \hat f_i(\B{x},\B{e}_{\setminus k},\hat e_k)$ for all $\B{x} \in \BC{X}, \B{e} \in \BC{E}$.
By repeatedly applying (this formulation of) Definition~\ref{def:Parents} to all $k \notin \pa(i)$, 
we obtain the existence of a measurable mapping $\tilde f_i : \BC{X} \times \BC{E} \to \C{X}_i$ 
and $\B{\hat x}_{\setminus \pa(i)} \in \BC{X}_{\setminus \pa(i)}$, $\B{\hat e}_{\setminus \pa(i)} \in \BC{E}_{\setminus\pa(i)}$ such that
for $\Prb_{\BC{E}}$-almost every $\B{e} \in \BC{E}$ and for all $\B{x}\in\BC{X}$,
  $$x_i = f_i(\B{x},\B{e}) \iff x_i = \tilde f_i(\B{x},\B{e}),$$
and for all $\B{e} \in \BC{E}$ and all $\B{x} \in \BC{X}$,
  $$\tilde f_i(\B{x},\B{e}) = \tilde f_i(\B{x}_{\pa(i)},\B{\hat x}_{\setminus \pa(i)},\B{e}_{\pa(i)},\B{\hat e}_{\setminus \pa(i)}).$$
Define the SCM $\tilde{\C{M}}$ as $\C{M}$ except that its causal mechanism is $\B{\tilde{f}}$ instead of $\B{f}$. 
Then $\tilde{\C{M}}$ is structurally minimal and equivalent to $\C{M}$.
\end{proof}

\begin{proof}[Proof of Proposition~\ref{prop:InterventionOnGraph}]
The $\intervene(I, \B{\xi}_I)$ operation on $\C{M}$ completely removes the functional dependence 
on $\B{x}$ and $\B{e}$ from the $f_i$ components for $i\in I$ and hence the corresponding 
incoming directed and bidirected edges on nodes in $I$ from the (augmented) graph.
\end{proof}

\begin{proof}[Proof of Proposition~\ref{prop:IntDisjSubsetAndAcyclicity}]
The first statement follows from Definitions~\ref{def:intervenedSCM} and \ref{def:InterventionOnGraph}. For the second statement, 
note that a perfect intervention can only remove parental relations, and therefore will 
never introduce a cycle.
\end{proof}

\begin{proof}[Proof of Proposition~\ref{prop:CommuteTwinGraph}]
This follows directly from Definitions~\ref{def:TwinSCM} and \ref{def:TwinGraph}.
\end{proof}

\begin{proof}[Proof of Proposition~\ref{prop:TwinAcyclicityPreserved}]
The additional edges introduced by the twin operation cannot lead to a directed cycle involving both copied and
original nodes, because there are no edges pointing from copied nodes to original nodes (i.e., of the form $i' \to v$ with $i' \in I'$ and $v \in \C{V}$). Directed cycles involving only original nodes are absent by assumption, and directed cycles involving only copied nodes as well since they would correspond with a directed cycle in the original directed graph.
\end{proof}

\begin{proof}[Proof of Proposition~\ref{prop:CommuteInterveneTwin}]
It suffices to prove the property for directed graphs, since the property for SCMs follows directly from Definitions~\ref{def:intervenedSCM} and \ref{def:TwinSCM}.

Applying the intervention $\intervene(I)$ on the graph $\C{G}$ removes all the incoming edges from the nodes in $I$. Now, if we perform the twin operation w.r.t.\ $\C{I}$ on this graph $\intervene(I)(\C{G})$, then we copy the same edges as if we had twinned the graph $\C{G}$ w.r.t.\ $\C{I}$, except those edges that do point to one of the nodes in $I$. Hence, if we apply the intervention $\intervene(I\cup I')$ on the graph $\twin(\C{I})(\C{G})$, which removes all incoming edges of both $I$ and its copy $I'$, then we clearly obtain the same graph.
\end{proof}

%\begin{proof}[Proof of Proposition~\ref{prop:CommuteInterveneTwin}]
%This follows directly from Definitions~\ref{def:intervenedSCM} and \ref{def:TwinSCM}.
%\end{proof}
%
%\begin{proof}[Proof of Proposition~\ref{prop:CommuteGraphInterveneTwin}]
%Applying the intervention $\intervene(I)$ on the graph $\C{G}$ removes all the incoming edges from the nodes in $I$. Now, if we perform the twin operation w.r.t.\ $K$ on this graph $\intervene(I)(\C{G})$, then we copy the same edges as if we had twinned the graph $\C{G}$ w.r.t.\ $K$, except those edges that do point to one of the nodes in $I$. Hence, if we apply the intervention $\intervene(I\cup I')$ on the graph $\twin(K)(\C{G})$, which removes all incoming edges of both $I$ and its copy $I'$, then we clearly obtain the same graph.
%\end{proof}

%%%%%%%%%%%%%%%%%%%%%%%%%%%%%%%%%%%%%%%%%%%%%%%%%%
\subsubsection*{Section 3}
\addcontentsline{toc}{subsubsection}{\protect\numberline{}Section 3}%
%%%%%%%%%%%%%%%%%%%%%%%%%%%%%%%%%%%%%%%%%%%%%%%%%%

\begin{proof}[Proof of Theorem~\ref{thm:SolvabilityIffCondition}]
First we define the solution space $\BC{S}(\C{M})$ of $\C{M}$ by
$$
  \BC{S}(\C{M}) := \{(\B{e},\B{x}) \in 
  \BC{E} \times \BC{X} : \B{x} 
  = \B{f}(\B{x},\B{e}) \} \,.
$$
This is a measurable set, since $\BC{S}(\C{M}) = \B{h}^{-1}(\Delta)$, where $\B{h} : \BC{E} \times \BC{X} \to \BC{X} \times \BC{X}$ is the measurable mapping defined by $\B{h}(\B{e}, \B{x}) = (\B{x},\B{f}(\B{x}, \B{e}))$ and $\Delta$ is the set defined by $\{ (\B{x}, \B{x}) : \B{x} \in \BC{X}\}$, which is measurable since $\BC{X}$ is Hausdorff.
Note that
$$\begin{aligned}
    \BC{A} := \B{pr}_{\BC{E}}(\BC{S}(\C{M})) &= \{ \B{e} \in \BC{E}  : \exists \B{x}\in\BC{X} \text{ s.t. } \B{x} = \B{f}(\B{x},\B{e}) \} \,,
  \end{aligned}$$
is an analytic set because the projection $\B{pr}_{\BC{E}} : \BC{X} \times \BC{E} \to \BC{E}$ is a measurable mapping between standard measurable spaces (Lemma~\ref{lemm:AnalyticSetProperties}).

Suppose that (1) holds, that is, $\C{M}$ has a solution. Then there exists a pair of random variables 
$(\B{E}, \B{X}) : \Omega \to \BC{E} \times \BC{X}$ such that 
$\B{X} = \B{f}(\B{X},\B{E})$ $\Prb$-a.s.. 
Note that
  \begin{equation*}\begin{split}
    \{ \omega \in \Omega : \B{X}(\omega) = \B{f}\big(\B{X}(\omega),\B{E}(\omega)\big) \} 
    & {} \subseteq \{ \omega \in \Omega : \exists \B{x}\in\BC{X} \text{ s.t. } \B{x} = \B{f}\big(\B{x},\B{E}(\omega)\big) \} \\
    & {} \subseteq \B{E}^{-1}\Big( \{ \B{e} \in \BC{E} : \exists \B{x} \in \BC{X} \text{ s.t. } \B{x} = \B{f}(\B{x},\B{e}) \} \Big) \\
    & {} = \B{E}^{-1}(\BC{A}).
  \end{split}\end{equation*}
By Lemma~\ref{lemm:AnalyticSetMeasurableApproxUpPset}, $\BC{A}$ is $\Prb^{\B{E}}$-measurable because it is
analytic, and we can write $\BC{A} = \BC{B} \,\dot\cup\, \BC{N}$ with $\BC{B} \subseteq \BC{E}$ measurable and $\BC{N}$ a $\Prb^{\B{E}}$-null set.
Hence $\B{E}^{-1}(\BC{A}) = \B{E}^{-1}(\BC{B}) \cup \B{E}^{-1}(\BC{N})$ where $\B{E}^{-1}(\BC{N})$ is a $\Prb$-null set.
Therefore,
  $$\B{E}^{-1}(\BC{B}) \supseteq \{ \omega \in \Omega : \B{X}(\omega) = \B{f}\big(\B{X}(\omega),\B{E}(\omega)\big) \} \setminus \B{E}^{-1}(\BC{N})$$
which implies that $\Prb(\B{E}^{-1}(\BC{B})) = 1$. Hence, $\BC{E} \setminus \BC{A}$ is a $\Prb_{\BC{E}}$-null set.
In other words, for $\Prb_{\BC{E}}$-almost every $\B{e}\in\BC{E}$ the structural equations 
$\B{x} = \B{f}(\B{x},\B{e})$
have a solution $\B{x}\in\BC{X}$, that is, (2) holds.

Suppose that (2) holds. Then $\BC{E} \setminus \B{pr}_{\BC{E}}(\BC{S}(\C{M}))$ is a $\Prb_{\BC{E}}$-null set.
%We first assume that $\C{J} \subseteq \pa(\C{I})$. 
By application of the measurable selection theorem
\ref{thm:MeasurableSelectionThm}, there exists a measurable $\B{g} : \BC{E} \to \BC{X}$ such that
for $\Prb_{\BC{E}}$-almost all $\B{e}\in\BC{E}$, $\B{g}(\B{e}) = \B{f}(\B{g}(\B{e}),\B{e})$. Hence, there exists a measurable mapping $\B{g}:\BC{E}\to\BC{X}$ such
that for $\Prb_{\BC{E}}$-almost every $\B{e}\in\BC{E}$ and for all $\B{x}\in\BC{X}$
$$
  \B{x}=\B{g}(\B{e}) \quad\implies\quad \B{x}=\B{f}(\B{x},\B{e}) \,,
$$
which we call property (A). Let $\tilde{\B{f}} : \BC{E} \times \BC{X} \to \BC{X}$ be the causal mechanism of a
structurally minimal SCM that is equivalent to $\C{M}$ (see Proposition~\ref{prop:StructMinimalRepresentation}).
In particular, for any $\B{\epsilon}_{\setminus\pa(\C{I})} \in \BC{E}_{\setminus\pa(\C{I})}$, we have that
$\tilde{\B{f}}(\B{x},\B{e}) = \tilde{\B{f}}(\B{x},\B{e}_{\pa(\C{I})},\B{\epsilon}_{\setminus\pa(\C{I})})$ for all $\B{x} \in \BC{X}$ and all $\B{e} \in \BC{E}$.
This means that we may also consider $\tilde{\B{f}}$ as a mapping $\tilde{\B{f}}:\BC{X} \times \BC{E}_{\pa(\C{I})} \to \BC{X}$.
%Let $\B{\epsilon}_{\setminus\pa(\C{I})} \in \BC{E}_{\setminus\pa(\C{I})}$ and define the mapping $\hat{\B{f}}:\BC{X}\times\BC{E}_{\pa(\C{I})}\to\BC{X}$ by $\hat{\B{f}}(\B{x},\B{e}_{\pa(\C{I})}) := \tilde{\B{f}}(\B{x},\B{e}_{\pa(\C{I})},\B{\epsilon}_{\setminus\pa(\C{I})})$. 
By applying Lemma~\ref{lemm:MeasurableMapsBetweenStandardSpaces} to the canonical projection $\B{pr}_{\BC{E}_\pa(\C{I})}:\BC{E} \to \BC{E}_{\pa(\C{I})}$
%$\B{pr}_{\BC{E}_{\pa(\C{I})}} : \BC{E} \to \BC{E}_{\pa(\C{I})}$
and using the equivalence of $\B{f}$ and $\tilde{\B{f}}$,
we obtain that for $\Prb_{\BC{E}_{\pa(\C{I})}}$-almost all $\B{e}_{\pa(\C{I})} \in \BC{E}_{\pa(\C{I})}$
there exists $\B{x} \in \BC{X}$ with $\B{x} = \tilde{\B{f}}(\B{x},\B{e}_{\pa(\C{I})})$. 
By applying the implication (2) $\implies$ (A) to $\BC{E}_{\pa(\C{I})}$ and $\tilde{\B{f}}$,
we conclude the existence of a measurable $\B{g} : \BC{E}_{\pa(\C{I})} \to \BC{X}$ such that
for $\Prb_{\BC{E}_{\pa(\C{I})}}$-almost all $\B{e}_{\pa(\C{I})}\in\BC{E}_{\pa(\C{I})}$, $\B{g}(\B{e}_{\pa(\C{I})}) = \tilde{\B{f}}(\B{g}(\B{e}_{\pa(\C{I})}),\B{e}_{\pa(\C{I})})$. 
Once more using Lemma~\ref{lemm:MeasurableMapsBetweenStandardSpaces}, we obtain that for $\Prb_{\BC{E}}$-almost all $\B{e}\in\BC{E}$, $\B{g}(\B{e}_{\pa(\C{I})}) = \B{f}(\B{g}(\B{e}_{\pa(\C{I})}),\B{e})$. In other words, (3) holds.

Lastly, suppose that (3) holds, that is there exists a measurable solution function $\B{g}:\BC{E}_{\pa(\C{I})}\to\BC{X}$. 
Then the measurable mappings 
$\B{E} : \BC{E} \to \BC{E}$ and $\B{X} : 
\BC{E} \to \BC{X}$, defined by $\B{E}(\B{e}) 
:= \B{e}$ and $\B{X}(\B{e}) := 
\B{g}(\B{e}_{\pa(\C{I})})$, respectively, define a pair of random variables 
$(\B{X},\B{E})$ such that $\B{X} = 
\B{f}(\B{X},\B{E})$ holds a.s.\ and hence 
  $(\B{X},\B{E})$ is a solution. Hence (1) holds.
\end{proof}

\begin{proof}[Proof of Proposition~\ref{prop:AcyclicSCMUniquelySolvable}]
Let $\tilde{\B{f}} : \BC{E} \times \BC{X} \to \BC{X}$ be the causal mechanism of a
structurally minimal SCM $\tilde{\C{M}}$ that is equivalent to $\C{M}$ (see Proposition~\ref{prop:StructMinimalRepresentation}).
For a subset $\C{O}\subseteq\C{I}$ consider the induced subgraph $\C{G}^a(\C{M})_{\C{O}}$ of the augmented graph $\C{G}^a(\C{M})$ on $\C{O}$. Then the acyclicity of $\C{G}^a(\C{M})$ implies that the induced subgraph $\C{G}^a(\C{M})_{\C{O}}$ is acyclic, and hence there exists a topological ordering on the nodes $\C{O}$. We can substitute the components $\tilde{f}_i$ of the causal mechanism $\tilde{\B{f}}$ for $i\in\C{O}$ into each other along this topological ordering. This gives a measurable solution function 
$\B{g}_{\C{O}} : \BC{X}_{\pa(\C{O})\setminus\C{O}} \times \BC{E}_{\pa(\C{O})}  \to \BC{X}_{\C{O}}$ for
$\tilde{\C{M}}$, and hence for $\C{M}$. It is clear from the acyclic structure that 
this mapping $\B{g}_{\C{O}}$ is independent of the choice of the topological ordering and is the only
  solution function for $\C{M}$. Therefore, $\tilde{\C{M}}$ is uniquely solvable w.r.t.\ $\C{O}$, and so is $\C{M}$.
\end{proof}

\begin{proof}[Proof of Proposition~\ref{prop:ObstructionToSelfCycles}]
This follows immediately from Definitions~\ref{def:Graphs} and \ref{def:UniqueSolvability}.
%Unique solvability w.r.t.\ $\{i\}$ means that there exists a $g_i : \BC{X}_{\pa(i)\setminus i} \times \BC{E}_{\pa(i)} \to X_i$
%such that for $\Prb_{\BC{E}}$-almost every $\B{e}\in\BC{E}$ for all $\B{x}\in\BC{X}$:
%$$
%  x_i = g_i(\B{x}_{\pa(i)\setminus i},\B{e}_{\pa(i)})
%  \quad\iff\quad
%  x_i = f_i(\B{x},\B{e}) \,.
%$$
%The result is thus immediate.
\end{proof}

\begin{proof}[Proof of Theorem~\ref{thm:UniqueSolvabilityIffCondition}]
\Joris{Rewrote proof.}
Suppose that (1) holds. By Proposition~\ref{prop:SolvabilityIfCondition} there exists a measurable solution function
$\B{g}_{\C{O}} : \BC{X}_{\pa(\C{O})\setminus\C{O}} \times \BC{E}_{\pa(\C{O})} \to \BC{X}_{\C{O}}$ for
$\C{M}$ w.r.t.\ $\C{O}$. Then for $\Prb_{\BC{E}}$-almost every $\B{e}\in\BC{E}$ and for all
$\B{x}_{\setminus\C{O}} \in \BC{X}_{\setminus\C{O}}$ we have that 
$\B{g}_{\C{O}}(\B{x}_{\pa(\C{O})\setminus\C{O}}, \B{e}_{\pa(\C{O})})$
is a solution of $\B{x}_{\C{O}} = \B{f}_{\C{O}}(\B{x},\B{e})$. Hence, because of (1),
for $\Prb_{\BC{E}}$-almost every $\B{e}\in\BC{E}$ and for all
$\B{x}_{\setminus\C{O}} \in \BC{X}_{\setminus\C{O}}$ we have that
$\B{x}_{\C{O}} = \B{f}_{\C{O}}(\B{x},\B{e})$ implies
$\B{x}_{\C{O}} = \B{g}_{\C{O}}(\B{x}_{\pa(\C{O})\setminus\C{O}}, \B{e}_{\pa(\C{O})})$.
Thus, $\C{M}$ is
uniquely solvable w.r.t.\ $\C{O}$, that is, (2) holds.

Suppose that (2) holds. Let $\B{g}_{\C{O}} : \BC{X}_{\pa(\C{O})\setminus\C{O}} \times \BC{E}_{\pa(\C{O})} \to \BC{X}_{\C{O}}$ be a measurable solution function for $\C{M}$ w.r.t.\ $\C{O}$.
Then, for $\Prb_{\BC{E}}$-almost every $\B{e}\in\BC{E}$ and for all $\B{x}\in\BC{X}$
$$
    \B{x}_{\C{O}} = \B{g}_{\C{O}}(\B{x}_{\pa(\C{O})\setminus\C{O}}, \B{e}_{\pa(\C{O})}) 
    \quad\iff\quad
    \B{x}_{\C{O}} =
    \B{f}_{\C{O}}(\B{x},\B{e}) \,.
$$
This implies (1).

For the last statement, assume that $\C{M}$ is uniquely solvable. Let $\B{g} : \BC{E}_{\pa(\C{I})} \to \BC{X}$ be
a measurable solution function. 
Then there exists a measurable set $\B{B} \subseteq \BC{E}$ with $\Prb_{\BC{E}}(\B{B}) = 1$ and for all $\B{e} \in \B{B}$,
  $$\forall \B{x} \in \BC{X} : \B{x} = \B{f}(\B{x},\B{e}) \implies \B{x} = \B{g}(\B{e}_{\pa(\C{I})}).$$
The existence of a solution for $\C{M}$ follows directly from Theorem~\ref{thm:SolvabilityIffCondition}. 
Each solution $(\B{X}, \B{E}) : \Omega \to \BC{X} \times \BC{E}$ of $\C{M}$ satisfies $\B{X}(\omega) = \B{f}(\B{X}(\omega),\B{E}(\omega))$ $\Prb$-a.s..
In addition, it satisfies $\B{E}(\omega) \in \B{B}$ $\Prb$-a.s., since $\Prb \circ \B{E}^{-1} = \Prb_{\BC{E}}$.
Hence, it satisfies $\B{X}(\omega)=\B{g}(\B{E}(\omega)_{\pa(\C{I})})$ $\Prb$-a.s..
Thus for every solution $(\B{X},\B{E})$ the associated observational 
distribution is the push-forward of $\Prb_{\BC{E}}$ under $\B{g} \circ \B{pr}_{\pa(\C{I})}$.
\end{proof}

\begin{proof}[Proof of Proposition~\ref{prop:SolvabilityPreservedIntervenedModel}]
Let $\B{g}_{\C{O}}:\BC{X}_{\pa(\C{O})\setminus\C{O}}\times\BC{E}_{\pa(\C{O})}\to\BC{X}_{\C{O}}$ be 
a measurable solution function for $\C{M}$ w.r.t.\ $\C{O}$. Then the mapping
$\tilde{\B{g}}_{\C{O}\cup I} : \BC{E}_{\pa(\C{O})} \to \BC{X}_{\C{O}\cup I}$ 
defined by
$\tilde{\B{g}}_{\C{O}\cup I}(\B{e}_{\pa(\C{O})}):=(\B{g}_{\C{O}}(\B{\xi}_{\pa(\C{O})\setminus\C{O}},\B{e}_{\pa(\C{O})}),\B{\xi}_I)$ is a measurable solution function for the SCM
$\C{M}_{\intervene(I,\B{\xi}_I)}$
w.r.t.\ $\C{O}\cup I$. If $\C{M}$ is (uniquely) solvable w.r.t.\ $\C{O}$, then it follows that $\C{M}_{\intervene(I,\B{\xi}_I)}$ is (uniquely) solvable w.r.t.\ $\C{O}\cup I$.
\end{proof}
%\begin{proof}[Proof of Proposition~\ref{prop:SolvabilityPreservedIntervenedModel}]
%Let $\B{g}_{\C{O}}:\BC{X}_{\pa(\C{O})\setminus\C{O}}\times\BC{E}_{\pa(\C{O})}\to\BC{X}_{\C{O}}$ be 
%a measurable solution function for $\C{M}$ w.r.t.\ $\C{O}$. Then the mapping
%$\tilde{\B{g}}_{\pa(\C{O})\cup\C{O}} : \BC{E}_{\pa(\C{O})} \to \BC{X}_{\pa(\C{O})\cup\C{O}}$ 
%defined by
%$\tilde{\B{g}}_{\pa(\C{O})\cup\C{O}}(\B{e}_{\pa(\C{O})}):=(\B{\xi}_{\pa(\C{O})\setminus\C{O}},\B{g}_{\C{O}}(\B{\xi}_{\pa(\C{O})\setminus\C{O}},\B{e}_{\pa(\C{O})}))$ is a measurable solution function for the SCM
%$\C{M}_{\intervene(\pa(\C{O})\setminus\C{O},\B{\xi}_{\pa(\C{O})\setminus\C{O}})}$
%w.r.t.\ $\pa(\C{O})\cup\C{O}$. Moreover, the mapping $\hat{\B{g}} :
%\BC{E}_{\pa(\C{I})}\to\BC{X}$ defined by
%$\hat{\B{g}}(\B{e}_{\pa(\C{I})}) := (\B{\xi}_{\C{I}\setminus\C{O}},\B{g}_{\C{O}}(\B{\xi}_{\pa(\C{O})\setminus\C{O}},\B{e}_{\pa(\C{O})}))$ is a measurable solution function for
%$\C{M}_{\intervene(\C{I}\setminus\C{O},\B{\xi}_{\C{I}\setminus\C{O}})}$.
%If $\C{M}$ w.r.t.\ $\C{O}$ is uniquely solvable w.r.t.\ $\C{O}$, it follows that also
%$\C{M}_{\intervene(\pa(\C{O})\setminus\C{O},\B{\xi}_{\pa(\C{O})\setminus\C{O}})}$ is
%uniquely solvable w.r.t.\ $\C{O}$ and that
%$\C{M}_{\intervene(\C{I}\setminus\C{O},\B{\xi}_{\C{I}\setminus\C{O}})}$ is uniquely solvable.
%\end{proof}

\begin{proof}[Proof of Proposition~\ref{prop:SolvabilityAncestralSolvability}]
It suffices to show that solvability of $\C{M}$ w.r.t.\ $\C{O}$ implies
ancestral solvability w.r.t.\ $\C{O}$. Solvability of $\C{M}$ w.r.t.\ $\C{O}$ implies that there exists a measurable mapping 
$\B{g}_{\C{O}}:\BC{X}_{\pa(\C{O})\setminus\C{O}}\times\BC{E}_{\pa(\C{O})}\to\BC{X}_{\C{O}}$ 
such that for $\Prb_{\BC{E}}$-almost every $\B{e}\in\BC{E}$ and for all 
$\B{x}\in\BC{X}$
$$
  \B{x}_{\C{O}} = \B{g}_{\C{O}}(\B{x}_{\pa(\C{O})\setminus\C{O}},\B{e}_{\pa(\C{O})})
  \quad\implies\quad
  \B{x}_{\C{O}} = \B{f}_{\C{O}}(\B{x},\B{e}) \,.
$$
Let $\tilde{\B{f}} : \BC{E} \times \BC{X} \to \BC{X}$ be the causal mechanism of a
structurally minimal SCM $\tilde{\C{M}}$ that is equivalent to $\C{M}$ (see Proposition~\ref{prop:StructMinimalRepresentation}).
Let $\C{P}:=\an_{\C{G}(\C{M})_{\C{O}}}(\C{A})$ for some $\C{A}\subseteq\C{O}$. 
Then 
%by choosing equivalent causal mechanisms $\tilde{\B{f}}_{\C{P}}$ and $\tilde{\B{f}}_{\C{O}\setminus\C{P}}$ that only depend on their parents we have that
for $\Prb_{\BC{E}}$-almost every $\B{e}\in\BC{E}$ and for all $\B{x}\in\BC{X}$
$$
  \begin{cases}
	  \B{x}_{\C{P}} &= (\B{g}_{\C{O}})_{\C{P}}(\B{x}_{\pa(\C{O})\setminus\C{O}},\B{e}_{\pa(\C{O})}) \\
	  \B{x}_{\C{O}\setminus\C{P}} &= (\B{g}_{\C{O}})_{\C{O}\setminus\C{P}}(\B{x}_{\pa(\C{O})\setminus\C{O}},\B{e}_{\pa(\C{O})})
  \end{cases} 
  \,\implies\,
  \begin{cases}
  \B{x}_{\C{P}} &= \tilde{\B{f}}_{\C{P}}(\B{x}_{\pa(\C{P})},\B{e}_{\pa(\C{P})}) \\
  \B{x}_{\C{O}\setminus\C{P}} &=
\tilde{\B{f}}_{\C{O}\setminus\C{P}}(\B{x}_{\pa(\C{O}\setminus\C{P})},\B{e}_{\pa(\C{O}\setminus\C{P})}) \,.
  \end{cases}
$$
%Since for the endogenous variables $\pa(\C{P})\subseteq\C{P}\cup(\pa(\C{O})\setminus\C{O})$,
Since $\pa(\C{P})\setminus\C{P}\subseteq\pa(\C{O})\setminus\C{O}$,
we have that in particular for $\Prb_{\BC{E}}$-almost every $\B{e}\in\BC{E}$ and for all $\B{x}\in\BC{X}$
$$
\B{x}_{\C{P}} = (\B{g}_{\C{O}})_{\C{P}}(\B{x}_{\pa(\C{O})\setminus\C{O}},\B{e}_{\pa(\C{O})}) 
  \quad\implies\quad
  \B{x}_{\C{P}} = \tilde{\B{f}}_{\C{P}}(\B{x}_{\pa(\C{P})},\B{e}_{\pa(\C{P})}) \,.
$$
This implies that the mapping $(\B{g}_{\C{O}})_{\C{P}}$ cannot depend on elements different from $\pa(\C{P})$. Moreover, it follows from the definition of 
$\C{P}$ that
$(\pa(\C{O})\setminus\C{O})\cap\pa(\C{P})=\pa(\C{P})\setminus\C{P}$ and thus we have $\pa(\C{O})\setminus\C{O} = (\pa(\C{P})\setminus\C{P}) \cup (\pa(\C{O})\setminus(\C{O}\cup\pa(\C{P})))$. Now, pick an element $\hat{\B{x}}_{\pa(\C{O})\setminus(\C{O}\cup\pa(\C{P}))} \in \BC{X}_{\pa(\C{O})\setminus(\C{O}\cup\pa(\C{P}))}$ and define the mapping 
$\tilde{\B{g}}_{\C{P}}:\BC{X}_{\pa(\C{P})\setminus\C{P}}\times\BC{E}_{\pa(\C{P})}\to\BC{X}_{\C{P}}$ by 
$$
  \tilde{\B{g}}_{\C{P}}(\B{x}_{\pa(\C{P})\setminus\C{P}},\B{e}_{\pa(\C{P})}) := 
  (\B{g}_{\C{O}})_{\C{P}}(\B{x}_{\pa(\C{P})\setminus\C{P}},\hat{\B{x}}_{\pa(\C{O})\setminus(\C{O}\cup\pa(\C{P}))},\B{e}_{\pa(\C{O})}) \,.
$$
Then, for $\Prb_{\BC{E}}$-almost every
$\B{e}\in\BC{E}$ and for all $\B{x}\in\BC{X}$
$$
  \B{x}_{\C{P}} = \tilde{\B{g}}_{\C{P}}(\B{x}_{\pa(\C{P})\setminus\C{P}},\B{e}_{\pa(\C{P})}) 
  \quad\iff\quad
  \B{x}_{\C{P}} = (\B{g}_{\C{O}})_{\C{P}}(\B{x}_{\pa(\C{O})\setminus\C{O}},\B{e}_{\pa(\C{O})}) \,.
$$
Together this gives that for $\Prb_{\BC{E}}$-almost every
$\B{e}\in\BC{E}$ and for all $\B{x}\in\BC{X}$
$$
  \B{x}_{\C{P}} = \tilde{\B{g}}_{\C{P}}(\B{x}_{\pa(\C{P})\setminus\C{P}},\B{e}_{\pa(\C{P})}) 
  \quad\implies\quad
  \B{x}_{\C{P}} = \tilde{\B{f}}_{\C{P}}(\B{x}_{\pa(\C{P})},\B{e}_{\pa(\C{P})}) \,.
$$
which is equivalent to the statement that $\C{M}$ is solvable w.r.t.\ 
$\an_{\C{G}(\C{M})_{\C{O}}}(\C{A})$.
\end{proof}

\subsubsection*{Section 4}
\addcontentsline{toc}{subsubsection}{\protect\numberline{}Section 4}%
%%%%%%%%%%%%%%%%%%%%%%%%%%%%%%%%%%%%%%%%%%%%%%%%%%

\begin{lemma}
  \label{lemm:ConsistentSolutionFunctions}
  \Joris{NEW!}
  Let $\C{M}$ be an SCM that is uniquely solvable w.r.t.\ two subsets $A,B \subseteq \C{I}$
  that satisfy $A \subseteq B$ and $\pa(A)\setminus A \subseteq \pa(B) \setminus B$. Let
  $\B{g}_A : \BC{X}_{\pa(A)\setminus A} \times \BC{E}_{\pa(A)} \to \BC{X}_A$
  and
  $\B{g}_B : \BC{X}_{\pa(B)\setminus B} \times \BC{E}_{\pa(B)} \to \BC{X}_B$
  be measurable solution functions for $\C{M}$ w.r.t.\ $A$ and $B$, respectively. Then 
  for $\Prb_{\BC{E}}$-almost every $\B{e}\in\BC{E}$ and for all $\B{x}\in\BC{X}$
  $$\B{g}_A(\B{x}_{\pa(A)\setminus A},\B{e}_{\pa(A)}) = (\B{g}_B)_A(\B{x}_{\pa(B)\setminus B},\B{e}_{\pa(B)}) \,.$$
\end{lemma}
\begin{proof}
  Without loss of generality, we assume that $\C{M}$ is structurally minimal (see Proposition~\ref{prop:StructMinimalRepresentation}).
  Let $\bar{\BC{E}} \subseteq \BC{E}$ be a measurable set with $\Prb_{\BC{E}}(\bar{\BC{E}}) = 1$ such that for all $\B{e} \in \bar{\BC{E}}$ for all $\B{x}\in\BC{X}$:
  $$\B{x}_{A} = \B{g}_{A}(\B{x}_{\pa(A)\setminus A},\B{e}_{\pa(A)}) \iff \B{x}_{A} = \B{f}_{A}(\B{x}_{\pa(A)},\B{e}_{\pa(A)})$$
  and
  $$\B{x}_{B} = \B{g}_{B}(\B{x}_{\pa(B)\setminus B},\B{e}_{\pa(B)}) \iff \B{x}_{B} = \B{f}_{B}(\B{x}_{\pa(B)},\B{e}_{\pa(B)})\,.$$
  
  Now let $\B{e}\in\bar{\BC{E}}$ and let $\B{x}_{A \cup \pa(B)\setminus B}\in\BC{X}_{A \cup \pa(B)\setminus B}$. Then
  $$\begin{aligned}
    & \B{x}_A = (\B{g}_B)_A(\B{x}_{\pa(B)\setminus B},\B{e}_{\pa(B)}) \\
    & \implies \begin{cases}
            &\B{x}_A = (\B{g}_B)_A(\B{x}_{\pa(B)\setminus B},\B{e}_{\pa(B)}) \\
            \exists \B{x}_{B\setminus A}\in\BC{X}_{B\setminus A}: &\B{x}_{B\setminus A} = (\B{g}_B)_{B\setminus A}(\B{x}_{\pa(B)\setminus B},\B{e}_{\pa(B)})
           \end{cases} \\
    & \implies \exists \B{x}_{B\setminus A}\in\BC{X}_{B\setminus A}: \quad
            \B{x}_B = \B{g}_B(\B{x}_{\pa(B)\setminus B},\B{e}_{\pa(B)}) \\
    & \implies \exists \B{x}_{B\setminus A}\in\BC{X}_{B\setminus A}: \quad
            \B{x}_B = \B{f}_B(\B{x}_{\pa(B)},\B{e}_{\pa(B)}) \\
    & \implies \exists \B{x}_{B\setminus A}\in\BC{X}_{B\setminus A}: \quad
            \B{x}_A = \B{f}_A(\B{x}_{\pa(A)},\B{e}_{\pa(A)}) \\
    & \implies \B{x}_A = \B{f}_A(\B{x}_{\pa(A)},\B{e}_{\pa(A)}) \\
    & \implies \B{x}_A = \B{g}_A(\B{x}_{\pa(A)\setminus A},\B{e}_{\pa(A)}) \,,
  \end{aligned}$$
  where the exists-quantifier could be omitted because the expression it binds to does not depend on $\B{x}_{B\setminus A}$ (from the assumptions it follows that $(A \cup \pa(A)) \cap (B \setminus A) = \emptyset$).
  Hence, for all $\B{e}\in\bar{\BC{E}}$ and all $\B{x}_{A \cup \pa(B)\setminus B}\in\BC{X}_{A \cup \pa(B)\setminus B}$
    $$\B{x}_A = (\B{g}_B)_A(\B{x}_{\pa(B)\setminus B},\B{e}_{\pa(B)}) \implies
      \B{x}_A = \B{g}_A(\B{x}_{\pa(A)\setminus A},\B{e}_{\pa(A)}) \,.$$
  Hence, for all $\B{e}\in\bar{\BC{E}}$ and all $\B{x}_{A \cup \pa(B)\setminus B}\in\BC{X}_{A \cup \pa(B)\setminus B}$
    $$(\B{g}_B)_A(\B{x}_{\pa(B)\setminus B},\B{e}_{\pa(B)}) = \B{g}_A(\B{x}_{\pa(A)\setminus A},\B{e}_{\pa(A)}) \,.$$
  Since this expression does not depend on $\B{x}_{(B \setminus A) \cup \C{I} \setminus (B \cup \pa(B))}$, from Lemma~\ref{lemm:AlmostAllQuantifierLogic2} we conclude that
  for all $\B{e}\in\bar{\BC{E}}$ and all $\B{x}\in\BC{X}$
    $$(\B{g}_B)_A(\B{x}_{\pa(B)\setminus B},\B{e}_{\pa(B)}) = \B{g}_A(\B{x}_{\pa(A)\setminus A},\B{e}_{\pa(A)}) \,.$$
\end{proof}

%\begin{proof}[Proof of Proposition~\ref{prop:EquivalentSCMsAreCounterfactualEq}]
%The twin operation preserves the equivalence relation on SCMs and since equivalent SCMs are interventionally equivalent w.r.t.\ every subset, the two equivalent twin SCMs have to be interventionally equivalent w.r.t.\ $\C{O}\cup\C{O}'$ for every $\C{O}\subseteq\C{I}$ with $\C{O}'$ the copy of $\C{O}$ in $\C{I}'$.
%\end{proof}

%\begin{lemma}
%\label{lemm:ObservationalEquivalentSubsets}
%If two SCMs $\C{M}$ and $\tilde{\C{M}}$ are observationally equivalent w.r.t.\ $\C{O}$, 
%then they are observationally equivalent w.r.t.\ every 
%subset $\C{V}\subseteq\C{O}$.
%\end{lemma}
%\begin{proof}
%Let $\C{V}\subsetneq\C{O}$ and assume without loss of generality that for all solutions 
%$\B{X}$ of $\C{M}$ there exists a solution $\tilde{\B{X}}$ of $\tilde{\C{M}}$ such that 
%$\Prb^{\B{X}_{\C{O}}}=\Prb^{\tilde{\B{X}}_{\C{O}}}$, then in particular
%$\Prb^{\B{X}_{\C{V}}}=\Prb^{\tilde{\B{X}}_{\C{V}}}$.
%\end{proof}

\begin{lemma}
\label{lemm:SCMandTwinSCMObservationallyEquivalence}
An SCM $\C{M}$ is observationally equivalent to $\C{M}^{\twin}$ w.r.t.\ $\C{O}\subseteq\C{I}$.
\end{lemma}
\begin{proof}
Let $(\B{X},\B{E})$ be a solution of $\C{M}$, then $((\B{X},\B{X}),\B{E})$ is a solution of $\C{M}^{\twin}$. 
Conversely, let $((\B{X},\B{X}'),\B{E})$ be a solution of $\C{M}^{\twin}$, then $(\B{X},\B{E})$ is a solution 
of $\C{M}$.
\end{proof}

%\begin{proof}[Proof of Proposition~\ref{prop:CounterfactuallyEquivalentSCMsAreInterventionallyEq}]
%Let $\C{M}$ and $\tilde{\C{M}}$ be counterfactually equivalent w.r.t.\ $\C{O}$. Then 
%$\C{M}^{\twin}$ and $\tilde{\C{M}}^{\twin}$ are interventionally equivalent w.r.t.\ 
%$\C{O} \cup \C{O}'$. Thus for $I \subseteq \C{O}$,
%$I'\subseteq\C{O}'$ the copy of $I$ and $\B{\xi}_{I'}=\B{\xi}_I\in\BC{X}_{I}$,
%$\C{M}^{\twin}_{\intervene(I\cup I',\B{\xi}_{I\cup I'})}$ and 
%$\tilde{\C{M}}^{\twin}_{\intervene(I\cup I',\B{\xi}_{I\cup I'})}$ are observationally 
%equivalent w.r.t.\ $\C{O}\cup\C{O}'$. In particular, they are observationally equivalent w.r.t.\ $\C{O}$. From Proposition~\ref{prop:CommuteInterveneTwin} we have that $\C{M}^{\twin}_{\intervene(I\cup I',\B{\xi}_{I\cup I'})} = 
%(\C{M}_{\intervene(I,\B{\xi}_I)})^{\twin}$ and 
%$\tilde{\C{M}}^{\twin}_{\intervene(I\cup I',\B{\xi}_{I\cup I'})} = 
%(\tilde{\C{M}}_{\intervene(I,\B{\xi}_I)})^{\twin}$, and together with Lemma~\ref{lemm:SCMandTwinSCMObservationallyEquivalence} this gives that  
%$\C{M}_{\intervene(I,\B{\xi}_I)}$ and 
%$\tilde{\C{M}}_{\intervene(I,\B{\xi}_I)}$ are observationally equivalent w.r.t.\ $\C{O}$.
%\end{proof}

\begin{proof}[Proof of Proposition~\ref{prop:CounterfactualEqRelatedToOtherEqs}]
First we show that equivalence implies counterfactual equivalence w.r.t.\ $\C{O}$. The twin operation preserves the equivalence relation on SCMs and since equivalent SCMs are interventionally equivalent w.r.t.\ every subset, the two equivalent twin SCMs have to be interventionally equivalent w.r.t.\ $\C{O}\cup\C{O}'$ for every $\C{O}\subseteq\C{I}$ with $\C{O}'$ the copy of $\C{O}$ in $\C{I}'$.

Now, let $\C{M}$ and $\tilde{\C{M}}$ be counterfactually equivalent w.r.t.\ $\C{O}$. Then 
$\C{M}^{\twin}$ and $\tilde{\C{M}}^{\twin}$ are interventionally equivalent w.r.t.\ 
$\C{O} \cup \C{O}'$. Thus for $I \subseteq \C{O}$,
$I'\subseteq\C{O}'$ the copy of $I$ and $\B{\xi}_{I'}=\B{\xi}_I\in\BC{X}_{I}$,
$\C{M}^{\twin}_{\intervene(I\cup I',\B{\xi}_{I\cup I'})}$ and 
$\tilde{\C{M}}^{\twin}_{\intervene(I\cup I',\B{\xi}_{I\cup I'})}$ are observationally 
equivalent w.r.t.\ $\C{O}\cup\C{O}'$. In particular, they are observationally equivalent w.r.t.\ $\C{O}$. From Proposition~\ref{prop:CommuteInterveneTwin} we have that $\C{M}^{\twin}_{\intervene(I\cup I',\B{\xi}_{I\cup I'})} = 
(\C{M}_{\intervene(I,\B{\xi}_I)})^{\twin}$ and 
$\tilde{\C{M}}^{\twin}_{\intervene(I\cup I',\B{\xi}_{I\cup I'})} = 
(\tilde{\C{M}}_{\intervene(I,\B{\xi}_I)})^{\twin}$, and together with Lemma~\ref{lemm:SCMandTwinSCMObservationallyEquivalence} this gives that  
$\C{M}_{\intervene(I,\B{\xi}_I)}$ and 
$\tilde{\C{M}}_{\intervene(I,\B{\xi}_I)}$ are observationally equivalent w.r.t.\ $\C{O}$.
\end{proof}

\subsubsection*{Section 5}
\addcontentsline{toc}{subsubsection}{\protect\numberline{}Section 5}%
%%%%%%%%%%%%%%%%%%%%%%%%%%%%%%%%%%%%%%%%%%%%%%%%%%

\Joris{Added the following lemma.}
%The following lemma shows that it does not matter if a candidate measurable solution function for concluding unique solvability has ``redundant'' inputs.
\begin{lemma}
\label{lemm:SolutionFunctionRedundantInputs}
\Joris{NEW!}
Let $\C{M}$ be an SCM. Let $B \subseteq \C{I}$ and $A \subseteq \C{I} \cup \C{J}$ such that $(\pa(B)\setminus B) \subseteq A$
and $B \cap A = \emptyset$.
Assume that $\B{g}_B : \BC{X}_A \times \BC{E}_A \to \BC{X}_B$ is a measurable function such that
for $\Prb_{\BC{E}}$-almost every $\B{e}\in\BC{E}$ and for all $\B{x}\in\BC{X}$
$$ 
  \B{x}_{B} = \B{f}_{B}(\B{x}_{\pa(B)},\B{e}_{\pa(B)}) \iff
  \B{x}_{B} = \B{g}_{B}(\B{x}_A,\B{e}_A) \,.
$$
Then $\C{M}$ is uniquely solvable w.r.t.\ $B$.
\end{lemma}
\begin{proof}
Assume that for $\Prb_{\BC{E}}$-almost every $\B{e}\in\BC{E}$ and for all $\B{x}\in\BC{X}$
$$ 
  \B{x}_{B} = \B{f}_{B}(\B{x}_{\pa(B)},\B{e}_{\pa(B)}) \iff
  \B{x}_{B} = \B{g}_{B}(\B{x}_A,\B{e}_A) \,.
$$
Let $C:=A\setminus (\pa(B)\setminus B)$, then by Lemma~\ref{lemm:AlmostAllQuantifierLogic7} we have that there exists $\hat{\B{e}}_C\in\BC{E}_C$ and $\hat{\B{x}}_C\in\BC{X}_C$ such that for $\Prb_{\BC{E}_{\C{J}\setminus C}}$-almost every $\B{e}_{\C{J}\setminus C}\in\BC{E}_{\C{J}\setminus C}$ and for all $\B{x}_{\C{I}\setminus C}\in\BC{X}_{\C{I}\setminus C}$ $$
    \B{x}_{B} = \B{f}_{B}(\B{x}_{\pa(B)},\B{e}_{\pa(B)})
      \iff
  \B{x}_{B} = \B{g}_{B}(\B{x}_{\pa(B)\setminus B},\hat{\B{x}}_C,
\B{e}_{\pa(B)},\hat{\B{e}}_C) \,.
$$
Defining the mapping
$\B{h}_{B} : \BC{X}_{\pa(B)\setminus B}\times\BC{E}_{\pa(B)}\to\BC{X}_{B}$
by
$$\B{h}_{B}(\B{x}_{\pa(B)\setminus B},\B{e}_{\pa(B)}) :=\B{g}_{B}(\B{x}_{\pa(B)\setminus B},\hat{\B{x}}_C, \B{e}_{\pa(B)},\hat{\B{e}}_C)\,,$$
where we picked $\hat{\B{e}}_C\in\BC{E}_C$ and $\hat{\B{x}}_C\in\BC{X}_C$ such that the above equivalence holds,
and applying Lemma~\ref{lemm:AlmostAllQuantifierLogic6} we get that for $\Prb_{\BC{E}}$-almost every $\B{e}\in\BC{E}$ and for all $\B{x}\in\BC{X}$
$$
  \B{x}_{B} = \B{f}_{B}(\B{x}_{\pa(B)},\B{e}_{\pa(B)})
  \iff
  \B{x}_{B} = \B{h}_{B}(\B{x}_{\pa(B)\setminus B},\B{e}_{\pa(B)})
$$
holds. Thus, $\C{M}$ is uniquely solvable w.r.t.\ $B$.
\end{proof}

\begin{proof}[Proof of Proposition~\ref{prop:MarginalizationCommutes}]
\Joris{Rewritten, and added converse of original statement.}
From unique solvability of $\C{M}$ w.r.t.\ $\C{L}_1$ it follows that there exists a mapping 
$\B{g}_{\C{L}_1} : \BC{X}_{\pa(\C{L}_1)\setminus(\C{L}_1)} \times 
\BC{E}_{\pa(\C{L}_1)} \to \BC{X}_{\C{L}_1}$ such that for $\Prb_{\BC{E}}$-almost every $\B{e}\in\BC{E}$ and for all $\B{x}\in\BC{X}$
$$ 
  \B{x}_{\C{L}_1} = \B{g}_{\C{L}_1}(\B{x}_{\pa(\C{L}_1)\setminus\C{L}_1},\B{e}_{\pa(\C{L}_1)})
  \quad\iff\quad
  \B{x}_{\C{L}_1} = \B{f}_{\C{L}_1}(\B{x},\B{e}) \,.
$$
Let $\widehat{\pa}$ denotes the parents in $\C{G}^a(\C{M}_{\marg(\C{L}_1)})$.
Note that %for the endogenous parents 
$\widehat{\pa}(\C{L}_2)\setminus\C{L}_2\subseteq\pa(\C{L}_1\cup\C{L}_2)\setminus(\C{L}_1\cup\C{L}_2)$.
%and for the exogenous parents 
%$\widehat{\pa}(\C{L}_2)\subseteq\pa(\C{L}_1\cup\C{L}_2)$ 
Let $\tilde{\B{f}}$ denote the marginal causal mechanism of a structurally minimal
SCM that is equivalent to the marginalization $\C{M}_{\marg(\C{L}_1)}$ constructed
from $\B{g}_{\C{L}_1}$ (see Proposition~\ref{prop:StructMinimalRepresentation}).

$\implies$: If $\C{M}_{\marg(\C{L}_1)}$ is uniquely solvable w.r.t.\ $\C{L}_2$, then there exists a mapping 
$\tilde{\B{g}}_{\C{L}_2} : \BC{X}_{\widehat{\pa}(\C{L}_2)\setminus\C{L}_2} \times 
\BC{E}_{\widehat{\pa}(\C{L}_2)} \to \BC{X}_{\C{L}_2}$
such that for $\Prb_{\BC{E}}$-almost every $\B{e}\in\BC{E}$ and for all $\B{x}_{\C{I}\setminus\C{L}_1} \in \BC{X}_{\C{I}\setminus\C{L}_1}$
$$
  \B{x}_{\C{L}_2} = \tilde{\B{g}}_{\C{L}_2}(\B{x}_{\widehat{\pa}(\C{L}_2)\setminus\C{L}_2}, 
  \B{e}_{\widehat{\pa}(\C{L}_2)}) 
  \iff
  \B{x}_{\C{L}_2} = \B{f}_{\C{L}_2}(\B{g}_{\C{L}_1}(
  \B{x}_{\pa(\C{L}_1)\setminus\C{L}_1},\B{e}_{\pa(\C{L}_1)}), 
  \B{x}_{\C{I}\setminus\C{L}_1},\B{e}) \,.
$$
%$$
%  \begin{aligned}
%  &\B{x}_{\C{L}_2} = \tilde{\B{g}}_{\C{L}_2}(\B{x}_{\widehat{\pa}(\C{L}_2)\setminus\C{L}_2}, 
%  \B{e}_{\widehat{\pa}(\C{L}_2)}) \\
%  \iff
%  &\B{x}_{\C{L}_2} = \B{f}_{\C{L}_2}(\B{g}_{\C{L}_1}(
%  \B{x}_{\pa(\C{L}_1)\setminus\C{L}_1},\B{e}_{\pa(\C{L}_1)}), 
%  \B{x}_{\C{I}\setminus\C{L}_1},\B{e}) \,.
%  \end{aligned}
%$$
Define the mapping $\B{h} : 
\BC{X}_{\pa(\C{L}_1\cup\C{L}_2)\setminus(\C{L}_1\cup\C{L}_2)} \times 
\BC{E}_{\pa(\C{L}_1\cup\C{L}_2)} \to  \BC{X}_{\C{L}_1\cup\C{L}_2}$ by
%$$
%  \begin{aligned}
%    &\tilde{\B{g}}_{\C{L}_1}(\B{x}_{\pa(\C{L}_1\cup\C{L}_2)\setminus(\C{L}_1\cup\C{L}_2)}, \B{e}_{\pa(\C{L}_1\cup\C{L}_2)}) := 
%    \B{g}_{\C{L}_1}\big((\tilde{\B{g}}_{\C{L}_2})_{\pa(\C{L}_1)}(
%    \B{x}_{\widehat{\pa}(\C{L}_2)\setminus\C{L}_2}, 
%    \B{e}_{\widehat{\pa}(\C{L}_2)}), \B{x}_{\pa(\C{L}_1)\setminus(\C{L}_1\cup\C{L}_2)},
%    \B{e}_{\pa(\C{L}_1)}\big) \\
%    &\tilde{\B{g}}_{\C{L}_2}(\B{x}_{\pa(\C{L}_1\cup\C{L}_2)\setminus(\C{L}_1\cup\C{L}_2)}, \B{e}_{\pa(\C{L}_1\cup\C{L}_2)}) := 
%    \tilde{\B{g}}_{\C{L}_2}(\B{x}_{\widehat{\pa}(\C{L}_2)\setminus\C{L}_2}, \B{e}_{\widehat{\pa}(\C{L}_2)})\,.
%  \end{aligned}
%$$
$$
  \begin{aligned}
    (&\B{h}_{\C{L}_1}, \B{h}_{\C{L}_2})(
    \B{x}_{\pa(\C{L}_1\cup\C{L}_2)\setminus(\C{L}_1\cup\C{L}_2)}, 
    \B{e}_{\pa(\C{L}_1\cup\C{L}_2)}) := \\
    \Big(&\B{g}_{\C{L}_1}\big((\tilde{\B{g}}_{\C{L}_2})_{\pa(\C{L}_1)}(
    \B{x}_{\widehat{\pa}(\C{L}_2)\setminus\C{L}_2}, 
    \B{e}_{\widehat{\pa}(\C{L}_2)}), \B{x}_{\pa(\C{L}_1)\setminus(\C{L}_1\cup\C{L}_2)},
    \B{e}_{\pa(\C{L}_1)}\big), \tilde{\B{g}}_{\C{L}_2}(\B{x}_{\widehat{\pa}(\C{L}_2)\setminus\C{L}_2}, 
    \B{e}_{\widehat{\pa}(\C{L}_2)})\Big) \,.
  \end{aligned}
$$
Then for $\Prb_{\BC{E}}$-almost every $\B{e}\in\BC{E}$ 
and for all $\B{x} \in \BC{X}$
$$
  \begin{aligned}
  &\begin{cases}
    \B{x}_{\C{L}_1} &= \B{f}_{\C{L}_1}(\B{x},\B{e}) \\
    \B{x}_{\C{L}_2} &= \B{f}_{\C{L}_2}(\B{x},\B{e}) 
  \end{cases} \\
  \iff
  &\begin{cases}
    \B{x}_{\C{L}_1} &= \B{g}_{\C{L}_1}(\B{x}_{\pa(\C{L}_1)\setminus\C{L}_1},\B{e}_{\pa(\C{L}_1)}) \\
    \B{x}_{\C{L}_2} &= \B{f}_{\C{L}_2}(\B{x},\B{e}) 
  \end{cases} \\
  \iff
  &\begin{cases}
    \B{x}_{\C{L}_1} &= \B{g}_{\C{L}_1}(\B{x}_{\pa(\C{L}_1)\setminus\C{L}_1},\B{e}_{\pa(\C{L}_1)}) \\
    \B{x}_{\C{L}_2} &= \B{f}_{\C{L}_2}(\B{g}_{\C{L}_1}(\B{x}_{\pa(\C{L}_1)\setminus\C{L}_1},\B{e}_{\pa(\C{L}_1)}), \B{x}_{\C{I}\setminus\C{L}_1},\B{e})
  \end{cases} \\
  \iff
  &\begin{cases}
    \B{x}_{\C{L}_1} &= \B{g}_{\C{L}_1}(\B{x}_{\pa(\C{L}_1)\setminus\C{L}_1},\B{e}_{\pa(\C{L}_1)}) \\
    \B{x}_{\C{L}_2} &=
\tilde{\B{g}}_{\C{L}_2}(\B{x}_{\widehat{\pa}(\C{L}_2)\setminus\C{L}_2},\B{e}_{\widehat{\pa}(\C{L}_2)}) 
  \end{cases} \\
  \iff
  &\begin{cases}
    \B{x}_{\C{L}_1} &=
    \B{g}_{\C{L}_1}\big((\tilde{\B{g}}_{\C{L}_2})_{\pa(\C{L}_1)}(\B{x}_{\widehat{\pa}(\C{L}_2)\setminus\C{L}_2},\B{e}_{\widehat{\pa}(\C{L}_2)}), \B{x}_{\pa(\C{L}_1)\setminus(\C{L}_1\cup\C{L}_2)},\B{e}_{\pa(\C{L}_1)}\big) \\
    \B{x}_{\C{L}_2} &=
\tilde{\B{g}}_{\C{L}_2}(\B{x}_{\widehat{\pa}(\C{L}_2)\setminus\C{L}_2},\B{e}_{\widehat{\pa}(\C{L}_2)}) 
  \end{cases} \\
  \iff
  &\begin{cases}
    \B{x}_{\C{L}_1} &= \B{h}_{\C{L}_1}(\B{x}_{\pa(\C{L}_1\cup\C{L}_2)\setminus(\C{L}_1\cup\C{L}_2)},\B{e}_{\pa(\C{L}_1\cup\C{L}_2)}) \\
    \B{x}_{\C{L}_2} &= \B{h}_{\C{L}_2}(\B{x}_{\pa(\C{L}_1\cup\C{L}_2)\setminus(\C{L}_1\cup\C{L}_2)},\B{e}_{\pa(\C{L}_1\cup\C{L}_2)}) \,,
  \end{cases} 
  \end{aligned}
$$
where in the first equivalence we used unique solvability w.r.t.\ $\C{L}_1$ of $\C{M}$, in the second 
we used substitution, in the third we used unique solvability w.r.t.\ $\C{L}_2$ of 
$\C{M}_{\marg(\C{L}_1)}$, in the fourth we used again substitution and in the last equivalence 
we used the definition of $\B{h}$. From this we conclude that $\C{M}$ is uniquely solvable w.r.t.\ 
$\C{L}_1\cup\C{L}_2$. Hence, by definition it follows that $\marg(\C{L}_2)\circ\marg(\C{L}_1)(\C{M})=\marg(\C{L}_1\cup\C{L}_2)(\C{M})$.

$\impliedby$: 
If $\C{M}$ is uniquely solvable w.r.t.\ $\C{L}_1\cup\C{L}_2$, then there exists a
mapping $\B{h} : \BC{X}_{\pa(\C{L}_1\cup\C{L}_2)\setminus(\C{L}_1\cup\C{L}_2)} \times \BC{E}_{\C{L}_1\cup\C{L}_2} \to \BC{X}_{\C{L}_1\cup\C{L}_2}$ such that for $\Prb_{\BC{E}}$-almost every $\B{e}\in\BC{E}$ for all
$\B{x}\in\BC{X}$
$$ 
  \B{x}_{\C{L}_1 \cup \C{L}_2} = \B{h}(\B{x}_{\pa(\C{L}_1\cup\C{L}_2)\setminus(\C{L}_1\cup\C{L}_2)},\B{e}_{\pa(\C{L}_1\cup\C{L}_2)})
  \quad\iff\quad
  \B{x}_{\C{L}_1 \cup \C{L}_2} = \B{f}_{\C{L}_1\cup \C{L}_2}(\B{x},\B{e}) \,.
$$
Then, for $\Prb_{\BC{E}}$-almost every $\B{e}\in\BC{E}$
for all $\B{x}\in\BC{X}$
$$
  \begin{aligned}
  &\begin{cases}
    \B{x}_{\C{L}_1} &= \B{h}_{\C{L}_1}(\B{x}_{\pa(\C{L}_1\cup\C{L}_2)\setminus(\C{L}_1\cup\C{L}_2)},\B{e}_{\pa(\C{L}_1\cup\C{L}_2)}) \\
    \B{x}_{\C{L}_2} &= \B{h}_{\C{L}_2}(\B{x}_{\pa(\C{L}_1\cup\C{L}_2)\setminus(\C{L}_1\cup\C{L}_2)},\B{e}_{\pa(\C{L}_1\cup\C{L}_2)})
  \end{cases} \\
  \iff
  &\begin{cases}
    \B{x}_{\C{L}_1} &= \B{f}_{\C{L}_1}(\B{x},\B{e}) \\
    \B{x}_{\C{L}_2} &= \B{f}_{\C{L}_2}(\B{x},\B{e}) 
  \end{cases} \\
  \iff
  &\begin{cases}
    \B{x}_{\C{L}_1} &= \B{g}_{\C{L}_1}(\B{x}_{\pa(\C{L}_1)\setminus\C{L}_1},\B{e}_{\pa(\C{L}_1)}) \\
    \B{x}_{\C{L}_2} &= \B{f}_{\C{L}_2}(\B{x},\B{e}) 
  \end{cases} \\
  \iff
  &\begin{cases}
    \B{x}_{\C{L}_1} &= \B{g}_{\C{L}_1}(\B{x}_{\pa(\C{L}_1)\setminus\C{L}_1},\B{e}_{\pa(\C{L}_1)}) \\
    \B{x}_{\C{L}_2} &= \B{f}_{\C{L}_2}(\B{g}_{\C{L}_1}(\B{x}_{\pa(\C{L}_1)\setminus\C{L}_1},\B{e}_{\pa(\C{L}_1)}), \B{x}_{\C{I}\setminus\C{L}_1},\B{e})
  \end{cases} \\
  \iff
  &\begin{cases}
    \B{x}_{\C{L}_1} &= \B{g}_{\C{L}_1}(\B{x}_{\pa(\C{L}_1)\setminus\C{L}_1},\B{e}_{\pa(\C{L}_1)}) \\
    \B{x}_{\C{L}_2} &= \tilde{\B{f}}_{\C{L}_2}(\B{x}_{\widehat{\pa}(\C{L}_2)},\B{e}_{\widehat{\pa}(\C{L}_2)}) \,.
  \end{cases} \\
  \end{aligned}
$$
This gives for $\Prb_{\BC{E}}$-almost every $\B{e}\in\BC{E}$
for all $\B{x}_{\C{I}\setminus\C{L}_1}\in\BC{X}_{\C{I}\setminus\C{L}_1}$
$$ 
\begin{aligned}
    \B{x}_{\C{L}_2} &=
\B{h}_{\C{L}_2}(\B{x}_{\pa(\C{L}_1\cup\C{L}_2)\setminus(\C{L}_1\cup\C{L}_2)},\B{e}_{\pa(\C{L}_1\cup\C{L}_2)}) \\
 \iff
%    \B{x}_{\C{L}_2} &=
%\B{f}_{\C{L}_2}(\B{g}_{\C{L}_1}(\B{x}_{\pa(\C{L}_1)\setminus\C{L}_1},\B{e}_{\pa(\C{L}_1)}),
%\B{x}_{\C{I}\setminus\C{L}_1},\B{e}) \\
%\iff
    \B{x}_{\C{L}_2} &=
\tilde{\B{f}}_{\C{L}_2}(\B{x}_{\widehat{\pa}(\C{L}_2)},\B{e}_{\widehat{\pa}(\C{L}_2)}) \,.
  \end{aligned}
$$
Now apply Lemma~\ref{lemm:SolutionFunctionRedundantInputs} to conclude that $\C{M}_{\marg(\C{L}_1)}$ is uniquely solvable w.r.t.\ $\C{L}_2$.
\end{proof}

\begin{proof}[Proof of Proposition~\ref{prop:MarginalizationCommuteWithInterventionAndTwin}]
The commutation relation with the perfect intervention follows straightforwardly from the definitions of perfect intervention and marginalization 
and the fact that if $\C{M}$ is uniquely solvable w.r.t.\ $\C{L}$, then 
$\C{M}_{\intervene(I,\B{\xi}_I)}$ is also uniquely solvable w.r.t.\ $\C{L}$, since the 
structural equations for the variables $\C{L}$ are the same for $\C{M}$ and $\C{M}_{\intervene(I,\B{\xi}_I)}$.

The commutation relation with the twin operation follows straightforwardly from the definition of the twin operation and
marginalization and the fact that if $\C{M}$ is uniquely solvable w.r.t.\
$\C{L}$, then $\twin(\C{M})$ is uniquely solvable w.r.t.\ $\C{L}\cup\C{L}'$,
where $\C{L}'$ is the copy of $\C{L}$ in $\C{I}'$.
\end{proof}

%\begin{proof}[Proof of Proposition~\ref{prop:MarginalizationCommuteWithIntervention}]
%This follows straightforwardly from the definitions of perfect intervention and marginalization 
%and the fact that if $\C{M}$ is uniquely solvable w.r.t.\ $\C{L}$, then 
%$\C{M}_{\intervene(I,\B{\xi}_I)}$ is also uniquely solvable w.r.t.\ $\C{L}$, since the 
%structural equations for the variables $\C{L}$ are the same for $\C{M}$ and $\C{M}_{\intervene(I,\B{\xi}_I)}$.
%\end{proof}
%
%\begin{proof}[Proof of Proposition~\ref{prop:MarginalizationCommuteWithTwinning}]
%This follows straightforwardly from the definition of the twin operation and
%marginalization and the fact that if $\C{M}$ is uniquely solvable w.r.t.\
%$\C{L}$, then $\twin(\C{M})$ is uniquely solvable w.r.t.\ $\C{L}\cup\C{L}'$,
%where $\C{L}'$ is the copy of $\C{L}$ in $\C{I}'$.
%\end{proof}

\begin{lemma}
\label{lemm:SCMandMarginalSCMAreObservationalEquivalent}
Given an SCM $\C{M}$ and a subset $\C{L}\subseteq\C{I}$ such that $\C{M}$ is uniquely solvable 
w.r.t.\ $\C{L}$. Then $\C{M}$ and $\marg(\C{L})(\C{M})$ are observationally equivalent 
w.r.t.\ $\C{I}\setminus\C{L}$.
\end{lemma}
\begin{proof}
Let $\C{O}:=\C{I}\setminus\C{L}$. From unique solvability w.r.t.\ $\C{L}$ it follows that for $\Prb_{\BC{E}}$-almost every $\B{e}\in\BC{E}$ and for all $\B{x} \in \BC{X}$
$$
  \begin{aligned}
  &\begin{cases}
    \B{x}_{\C{L}} &= \B{f}_{\C{L}}(\B{x},\B{e}) \\
    \B{x}_{\C{O}} &= \B{f}_{\C{O}}(\B{x},\B{e}) 
  \end{cases} \\
  \iff
  &\begin{cases}
    \B{x}_{\C{L}} &= \B{g}_{\C{L}}(\B{x}_{\pa(\C{L})\setminus\C{L}},\B{e}_{\pa(\C{L})}) \\
    \B{x}_{\C{O}} &=
\B{f}_{\C{O}}(\B{g}_{\C{L}}(\B{x}_{\pa(\C{L})\setminus\C{L}},\B{e}_{\pa(\C{L})}), \B{x}_{\C{O}},\B{e})
  \end{cases} \\
  \iff 
  &\begin{cases}
    \B{x}_{\C{L}} &= \B{g}_{\C{L}}(\B{x}_{\pa(\C{L})\setminus\C{L}},\B{e}_{\pa(\C{L})}) \\
    \B{x}_{\C{O}} &= \tilde{\B{f}}(\B{x}_{\C{O}},\B{e}) \,,
  \end{cases} 
  \end{aligned}
$$
where $\tilde{\B{f}}$ is the marginal causal mechanism of $\C{M}_{\marg(\C{L})}$ constructed from a measurable solution function $\B{g}_{\C{L}}: \BC{X}_{\pa(\C{L})\setminus\C{L}}\times\BC{E}_{\pa(\C{L})} \to \BC{X}_{\C{L}}$ for $\C{M}$ w.r.t.\ $\C{L}$. Hence, a solution $(\B{X},\B{E})$ of $\C{M}$ satisfies $\B{X}_{\C{O}} = \tilde{\B{f}}(\B{X}_{\C{O}},\B{E})$ a.s.. Conversely, if $(\tilde{\B{X}}_{\C{O}},\B{E})$ is a solution of the marginal SCM $\C{M}_{\marg(\C{L})}$ then with $\tilde{\B{X}}_{\C{L}} := \B{g}_{\C{L}}(\tilde{\B{X}}_{\pa(\C{L})\setminus\C{L}},\B{E}_{\pa(\C{L})})$, the random variables $(\B{X},\B{E}):=(\tilde{\B{X}}_{\C{O}},\tilde{\B{X}}_{\C{L}},\B{E})$ are a solution of $\C{M}$.
\end{proof}

\begin{proof}[Proof of Theorem~\ref{thm:MarginalizationEquivalences}]
The observational equivalence follows from
Lemma~\ref{lemm:SCMandMarginalSCMAreObservationalEquivalent}. Using both
Lemma~\ref{lemm:SCMandMarginalSCMAreObservationalEquivalent} and Proposition~\ref{prop:MarginalizationCommuteWithInterventionAndTwin} we can prove the
interventional equivalence. Observe that from Proposition~\ref{prop:MarginalizationCommuteWithInterventionAndTwin} we know that for a subset 
$I\subseteq\C{I}\setminus\C{L}$ and a value $\B{\xi}_I\in\BC{X}_I$, 
$(\marg(\C{L})\circ\intervene(I,\B{\xi}_I))(\C{M})$ exists. By 
Lemma~\ref{lemm:SCMandMarginalSCMAreObservationalEquivalent} we know that 
$\intervene(I,\B{\xi}_I)(\C{M})$ and $(\marg(\C{L})\circ\intervene(I,\B{\xi}_I))(\C{M})$ 
are observationally equivalent w.r.t.\ $\C{O}$ and hence by applying again 
Proposition~\ref{prop:MarginalizationCommuteWithInterventionAndTwin}, $\intervene(I,\B{\xi}_I)(\C{M})$ 
and $(\intervene(I,\B{\xi})\circ\marg(\C{L}))(\C{M})$ are observationally equivalent w.r.t.\ 
$\C{O}$. This implies that $\C{M}$ and $\marg(\C{L})(\C{M})$ are
interventionally equivalent w.r.t.\ $\C{O}$. Lastly, we need to show that $\twin(\C{M})$ and
$(\twin\circ\marg(\C{L}))(\C{M})$ are interventionally equivalent w.r.t.\
$(\C{I}\cup\C{I}')\setminus(\C{L}\cup\C{L}')$, where $\C{L}'$ is the copy of
$\C{L}$ in $\C{I}'$. From Proposition~\ref{prop:MarginalizationCommuteWithInterventionAndTwin}
$(\twin\circ\marg(\C{L}))(\C{M})$ is equivalent to
$(\marg(\C{L}\cup\C{L}')\circ\twin)(\C{M})$ and since we proved that $(\marg(\C{L}\cup\C{L}')\circ\twin)(\C{M})$ and $\twin(\C{M})$ are
interventionally equivalent w.r.t.\ $(\C{I}\cup\C{I}')\setminus(\C{L}\cup\C{L}')$ the result follows.
\end{proof}

\begin{proof}[Proof of Proposition~\ref{prop:LatentProjectionCommutes}]
A similar proof as for Theorem~1 in \citep{Eva16} works.
\end{proof}

\begin{proof}[Proof of Proposition~\ref{prop:LatentProjectionCommuteWithInterventionAndTwin}]
First we prove the commutation relation of the perfect intervention. Observe that applying the $\intervene(I)$ operation to the latent projection
$\marg(\C{L})(\C{G})$ removes all the incoming edges on the nodes $I$.
Such an incoming edge at a node in $I$ in $\marg(\C{L})(\C{G})$ corresponds to a path in $\C{G}$ that points to that node. But since $\intervene(I)(\C{G})$ is
just $\C{G}$ with all the incoming edges on $I$ removed, the graph
$(\marg(\C{L})\circ\intervene(I))(\C{G})$ also has all the incoming edges on the
nodes $I$ removed.

Next, we will prove the commutation relation of the twin operation. We will denote the copy in $\C{I}'$ of any node $i \in \C{I}$ by $i'$, that is, $\C{I}' = \{i' : i \in \C{I}\}$.
The edges in $(\twin(\C{I}\setminus\C{L}) \circ \marg(\C{L}))(\C{G})$ can be partitioned into three cases:
$$\begin{cases}
v \to w & v \in \C{J} \cup \C{I} \setminus \C{L}, w \in \C{J} \cup \C{I} \setminus \C{L}, v \to w \in \marg(\C{L})(\C{G}) \,, \\
v \to w' & v \in \C{J}, w \in \C{I} \setminus \C{L}, v \to w \in \marg(\C{L})(\C{G}) \,, \\
v' \to w' & v \in \C{I} \setminus \C{L}, w \in \C{I} \setminus \C{L}, v \to w \in \marg(\C{L})(\C{G}) \,,
\end{cases}$$
where $\C{J}:=\C{V}\setminus \C{I}$.

Note that in $\twin(\C{I})(\C{G})$, there are no directed edges of the form $v' \to w$ by definition.
Therefore, the edges in $(\marg(\C{L} \cup \C{L}') \circ \twin(\C{I}))(\C{G})$ can be partitioned into three cases:
$$\begin{cases}
  v \to w & v \in \C{J} \cup \C{I} \setminus \C{L}, w \in \C{J} \cup \C{I} \setminus \C{L}, v \to \ell_1 \to \dots \to \ell_n \to w \in \twin(\C{I})(\C{G}) \,, \\
  v \to w' & v \in \C{J}, w \in \C{I} \setminus \C{L}, v \to \ell_1' \to \dots \to \ell_n' \to w' \in \twin(\C{I})(\C{G}) \,, \\
  v' \to w' & v \in \C{I} \setminus \C{L}, w \in \C{I} \setminus \C{L}, v' \to \ell_1' \to \dots \to \ell_n' \to w' \in \twin(\C{I})(\C{G}) \,,
\end{cases}$$
where all $\ell_1,\dots,\ell_n \in \C{L}$ and $\ell'_1,\dots,\ell'_n \in \C{L}'$.
Thus, the non-endpoint nodes on the directed paths in $\twin(\C{I})(\C{G})$ must either all lie in $\C{L}$ or in $\C{L}'$.
With the definition of $\twin(\C{I})(\C{G})$ we can rewrite this as follows:
$$\begin{cases}
  v \to w & v \in \C{J} \cup \C{I} \setminus \C{L}, w \in \C{J} \cup \C{I} \setminus \C{L}, v \to \ell_1 \to \dots \to \ell_n \to w \in \C{G} \,, \\
  v \to w' & v \in \C{J}, w \in \C{I} \setminus \C{L}, v \to \ell_1 \to \dots \to \ell_n \to w \in \C{G} \,, \\
  v' \to w' & v \in \C{I} \setminus \C{L}, w \in \C{I} \setminus \C{L}, v \to \ell_1 \to \dots \to \ell_n \to w \in \C{G} \,,
\end{cases}$$
where all intermediate $\ell_1,\dots,\ell_n$ must lie in $\C{L}$.
This corresponds exactly with the edges in $(\twin(\C{I}\setminus\C{L}) \circ \marg(\C{L}))(\C{G})$.
%Marginalizing w.r.t.\ $\C{L}$ on $\C{G}$ replaces directed paths between $i$ and
%$j$ in $\C{V}\setminus\C{L}$ that pass only through nodes in $\C{L}$ by a directed
%edge $i \to j$. By applying the $\twin(I)$ operation on
%$\C{G}$ first, each such directed path between $i$ and $j$ is copied to a corresponding
%path between $i$ and $j$ that pases through the corresponding copied nodes
%$I'$, where we take for $i$ and $j$ the corresponding copied nodes in $I'$ if $i,j\in I$. It is clear that applying first $\marg(\C{L})$, which replaces these
%paths by edges (or no edges), and then applying the $\twin(I\setminus\C{L})$
%operation, which copies the replaced edges, is the same as first applying the
%$\twin(I)$ operation, which copies these paths, and then applying
%$\marg(\C{L}\cup\C{L}')$, which replaces the original paths and the copied paths
%by its corresponding edges (or no edges).

\end{proof}

\begin{proof}[Proof of Proposition~\ref{prop:LatentProjection}]
  \sloppy \Joris{Refined} Without loss of generality, we assume that $\C{M}$ is structurally minimal (see Proposition~\ref{prop:StructMinimalRepresentation}).
  Let $\B{g}_{\C{L}}$ be a measurable solution function for $\C{M}$
  w.r.t.\ $\C{L}$ and denote by $\C{M}_{\marg(\C{L})}$ the marginal SCM 
  constructed from $\B{g}_{\C{L}}$. For $j\in\C{I}\setminus\C{L}$, define $A_j :=
  \an_{\C{G}(\C{M})_{\C{L}}}(\pa(j)\cap\C{L}) \subseteq \C{L}$ and let 
$\tilde{\B{g}}_{A_j}$ be a measurable solution function for $\C{M}$ 
w.r.t.\ $A_j$. Because $A_j \subseteq \C{L}$ and $\pa(A_j) \setminus A_j \subseteq \pa(\C{L}) \setminus \C{L}$,
by Lemma~\ref{lemm:ConsistentSolutionFunctions},
for $\Prb_{\BC{E}}$-almost every $\B{e}\in\BC{E}$ and for all $\B{x}\in\BC{X}$
$$
  (\B{g}_{\C{L}})_{A_j}(\B{x}_{\pa(\C{L})\setminus\C{L}},\B{e}_{\pa(\C{L})})
  =
  \tilde{\B{g}}_{A_j}(\B{x}_{\pa(A_j)\setminus
  A_j},\B{e}_{\pa(A_j)}) \,.
$$
Therefore, the component $\tilde{f}_j$ of the marginal causal mechanism $\tilde{\B{f}}$ of
$\C{M}_{\marg(\C{L})}$ satisfies
for $\Prb_{\BC{E}}$-almost every $\B{e}\in\BC{E}$ and for all $\B{x}\in\BC{X}$
  $$
  \begin{aligned}
    \tilde{f}_j(\B{x}_{\C{I}\setminus\C{L}},\B{e}) & := f_j\big((\B{g}_{\C{L}})_{\pa(j)}(\B{x}_{\pa(\C{L})\setminus\C{L}},\B{e}_{\pa(\C{L})}),\B{x}_{\pa(j)\setminus\C{L}},\B{e}_{\pa(j)}\big) \\
    & = f_j\big((\tilde{\B{g}}_{A_j})_{\pa(j)\cap\C{L}} (\B{x}_{\pa(A_j)\setminus A_j},\B{e}_{\pa(A_j)}),\B{x}_{\pa(j)\setminus\C{L}},\B{e}_{\pa(j)}\big) \,.
  \end{aligned}
  $$
Hence, the endogenous parents of $j$ in $\C{M}_{\marg(\C{L})}$ are a subset of 
$\big((\pa(A_j)\setminus A_j) \cup (\pa(j)\setminus\C{L})\big) \cap \C{I}$
and the exogenous parents of $j$ in $\C{M}_{\marg(\C{L})}$ are a subset of 
$(\pa(A_j) \cup \pa(j)) \cap \C{J}$. Hence, all parents of $j$ in $\C{M}_{\marg(\C{L})}$ 
are a subset of those $k\in(\C{I}\setminus\C{L})\cup\C{J}$ such that there exists a path 
$k\to\ell_1\to\dots\to\ell_n\to j\in\C{G}^a(\C{M})$ for $n\geq 0$ and 
$\ell_1,\dots,\ell_n\in\C{L}$. Therefore, the augmented graph 
$\C{G}^a\big(\marg(\C{L})(\C{M})\big)$ is a subgraph of the latent projection 
$\marg(\C{L})\big(\C{G}^a(\C{M})\big)$. Hence, 
$$\begin{aligned}
  \C{G}\big(\marg(\C{L})(\C{M})\big) & = \marg(\C{J})\Big(\C{G}^a\big(\marg(\C{L})(\C{M})\big)\Big) \\
                                     & \subseteq \marg(\C{J})\Big(\marg(\C{L})\big(\C{G}^a(\C{M})\big)\Big) \\
                                     & = \marg(\C{L})\Big(\marg(\C{J})\big(\C{G}^a(\C{M})\big)\Big) \\
                                     & = \marg(\C{L})\big(\C{G}(\C{M})\big)
  \end{aligned}
$$
and we conclude that also the graph $\C{G}\big(\marg(\C{L})(\C{M})\big)$
is a subgraph of the latent projection $\marg(\C{L})\big(\C{G}(\C{M})\big)$.
\end{proof}

%\begin{proof}[Proof of Proposition~\ref{prop:MarginalizingSCMUniquelySolvableWRTeverySubset}]
%Take two disjoint subsets $\C{L}_1$ and $\C{L}_2$ in $\C{I}$. Then, it suffices to show that
%$\C{M}_{\marg(\C{L}_1)}$ is uniquely solvable w.r.t.\ $\C{L}_2$. This follows directly from Proposition~\ref{prop:MarginalizationCommutes}.
%\end{proof}

%%%%%%%%%%%%%%%%%%%%%%%%%%%%%%%%%%%%%%%%%%%%%%%%%%
\subsubsection*{Section 6}
\addcontentsline{toc}{subsubsection}{\protect\numberline{}Section 6}%
%%%%%%%%%%%%%%%%%%%%%%%%%%%%%%%%%%%%%%%%%%%%%%%%%%

\begin{proof}[Proof of Theorem~\ref{thm:globalMarkovPropertiesSCMs}]
This follows directly from Theorems~\ref{thm:dgMarkovPropertySCMThreeSpecialCases} and \ref{thm:gdgMarkovPropertySCM}.
\end{proof}

%%%%%%%%%%%%%%%%%%%%%%%%%%%%%%%%%%%%%%%%%%%%%%%%%%
\subsubsection*{Section 7}
\addcontentsline{toc}{subsubsection}{\protect\numberline{}Section 7}%
%%%%%%%%%%%%%%%%%%%%%%%%%%%%%%%%%%%%%%%%%%%%%%%%%%

\begin{proof}[Proof of Proposition~\ref{prop:DirectedPathEdges}]
We define $\tilde{\C{M}}:=\C{M}_{\intervene(I,\B{\xi}_I)}$, $\widetilde{\pa} := \pa_{\C{G}^a(\tilde{\C{M}})}$ and $\C{A} := \an_{\C{G}(\tilde{\C{M}})_{\setminus i}}(j)$. Suppose that $i\to j \notin \marg(\C{I}\setminus\C{O})(\C{G}(\C{M}))$ and assume that the two induced distributions do not coincide. Because $i\to j \notin \marg(\C{I}\setminus\C{O})(\C{G}(\C{M}))$ it follows that $(\widetilde{\pa}(\C{A})\setminus\C{A}) \cap \C{I}= \emptyset$. Let now $\tilde{\B{g}}_{\C{A}}:\BC{E}_{\widetilde{\pa}(\C{A})}\to\BC{X}_{\C{A}}$
be a measurable solution function for $\tilde{\C{M}}$ w.r.t.\ $\C{A}$, that is, we have 
for $\Prb_{\BC{E}}$-almost every $\B{e}\in\BC{E}$ and for all $\B{x}\in\BC{X}$
$$
  \B{x}_{\C{A}} = \tilde{\B{f}}_{\C{A}}(\B{x},\B{e})  \quad\iff\quad
  \B{x}_{\C{A}} = \tilde{\B{g}}_{\C{A}}(\B{e}_{\widetilde{\pa}(\C{A})}) \,,
$$
where $\tilde{\B{f}}$ is the ausal mechanism of $\tilde{\C{M}}$. Because
$i\notin\C{A}$ and $j\in \C{A}$, it follows that for the intervened model
$(\C{M}_{\intervene(I,\B{\xi}_I)})_{\intervene(\{i\},\xi_i)}$
the marginal solution $X_j$ is also a marginal solution of
$(\C{M}_{\intervene(I,\B{\xi}_I)})_{\intervene(\{i\},\tilde{\xi}_i)}$
and vice versa, which is in contradiction with the assumption.
\end{proof}

\begin{proof}[Proof of Proposition~\ref{prop:BidirectedEdges}]
  Let's define $\tilde{\C{M}}:=\C{M}_{\intervene(I,\B{\xi}_I)}$, $\widetilde{\pa} := \pa_{\C{G}^a(\tilde{\C{M}})}$, $\C{A}_i := \an_{\C{G}(\tilde{\C{M}})}(i)$ and $\C{A}_{j}^{\setminus i} :=
  \an_{\C{G}(\tilde{\C{M}})_{\setminus i}}(j)$. Suppose that there does not exist a bidirected edge $i\oto j$ in the latent projection $\marg(\C{I}\setminus\C{O})(\C{G}(\C{M}))$. Because $i\oto j \notin \marg(\C{I}\setminus\C{O})(\C{G}(\tilde{\C{M}}))$, where here $\tilde{\C{M}}$ is the intervened model $\C{M}_{\intervene(I,\B{\xi}_I)}$, we have that $\an_{\C{G}^a(\tilde{\C{M}})_{\setminus j}}(i)\cap\an_{\C{G}^a(\tilde{\C{M}})_{\setminus i}}(j)\cap\C{J} = \emptyset$. From $j\notin \an_{\C{G}(\tilde{\C{M}})}(i)$ it follows that $\an_{\C{G}(\tilde{\C{M}})_{\setminus j}}(i) = \an_{\C{G}(\tilde{\C{M}})}(i)$, and hence $\an_{\C{G}^a(\tilde{\C{M}})}(i)\cap\an_{\C{G}^a(\tilde{\C{M}})_{\setminus i}}(j)\cap\C{J} = \emptyset$. Observe that $\widetilde{\pa}(\C{A}_i)\subseteq\an_{\C{G}^a(\tilde{\C{M}})}(i)$ and $\widetilde{\pa}(\C{A}_{j}^{\setminus i})\subseteq\an_{\C{G}^a(\tilde{\C{M}})_{\setminus i}}(j)\cup\{i\}$, and thus $\widetilde{\pa}(\C{A}_i)\cap\widetilde{\pa}(\C{A}_{j}^{\setminus i})\cap \C{J} = \emptyset$. Let $\B{g}_{\C{A}_i}:\BC{E}_{\widetilde{\pa}(\C{A}_i)}\to\BC{X}_{\C{A}_i}$
be a measurable solution function for $\tilde{\C{M}}$ w.r.t.\ $\C{A}_i$, that is, we have 
for $\Prb_{\BC{E}}$-almost every $\B{e}\in\BC{E}$ and for all $\B{x}\in\BC{X}$
$$
  \B{x}_{\C{A}_i} = \tilde{\B{f}}_{\C{A}_i}(\B{x},\B{e})  \quad\iff\quad
  \B{x}_{\C{A}_i} =
  \B{g}_{\C{A}_i}(\B{e}_{\widetilde{\pa}(\C{A}_i)}) \,,
$$
  where $\tilde{\B{f}}$ is the intervened causal mechanism of $\tilde{\C{M}}$. Because $\widetilde{\pa}(\C{A}_i)\cap\widetilde{\pa}(\C{A}_{j}^{\setminus i})\cap \C{J} =
  \emptyset$ and $i\in\C{A}_i$, we have that $X_i \indep \B{E}_{\widetilde{\pa}(\C{A}_{j}^{\setminus i})}$ for every solution $(\B{X},\B{E})$ of $\tilde{\C{M}}$. 
  
Assume for the moment that $i\in\widetilde{\pa}(\C{A}_{j}^{\setminus i})\setminus\C{A}_{j}^{\setminus i}$, then
  $(\widetilde{\pa}(\C{A}_{j}^{\setminus i})\setminus\C{A}_{j}^{\setminus i}) \cap \C{I}= \{i\}$. Let
  $\B{g}_{\C{A}_{j}^{\setminus i}}:\C{X}_i\times\BC{E}_{\widetilde{\pa}(\C{A}_{j}^{\setminus i})}\to\BC{X}_{\C{A}_{j}^{\setminus i}}$
  be a measurable solution function for $\tilde{\C{M}}$ w.r.t.\ $\C{A}_{j}^{\setminus i}$, that is, we
have for $\Prb_{\BC{E}}$-almost every $\B{e}\in\BC{E}$ and for all $\B{x}\in\BC{X}$
$$
  \B{x}_{\C{A}_{j}^{\setminus i}} =
  \tilde{\B{f}}_{\C{A}_{j}^{\setminus i}}(\B{x},\B{e}) \iff \B{x}_{\C{A}_{j}^{\setminus i}} =
  \B{g}_{\C{A}_{j}^{\setminus i}}(x_i,\B{e}_{\widetilde{\pa}(\C{A}_{j}^{\setminus i})}) \,.
$$
For every measurable set $\C{B}_j\subseteq\C{X}_j$ there exists a version of the regular conditional
probability $\Prb_{\C{M}_{\intervene(I,\B{\xi}_I)}}(X_j \in \C{B} \given X_i = \xi_i)$ such that for every value $\xi_i\in\C{X}_i$ it satisfies
$$
  \begin{aligned}
    \Prb_{\C{M}_{\intervene(I,\B{\xi}_I)}}\big( X_j \in \C{B}_j \given X_i = \xi_i \big) &= \Prb_{\tilde{\C{M}}}\big( X_j \in \C{B}_j \given X_i = \xi_i \big) \\ 
    &= \Prb_{\tilde{\C{M}}}\big( (\B{g}_{\C{A}_{j}^{\setminus i}})_j(X_i, \B{E}_{\widetilde{\pa}(\C{A}_{j}^{\setminus i})}) \in \C{B}_j \given X_i = \xi_i \big) \\
    &= \Prb_{\tilde{\C{M}}}\big( (\B{g}_{\C{A}_{j}^{\setminus i}})_j(\xi_i, \B{E}_{\widetilde{\pa}(\C{A}_{j}^{\setminus i})}) \in \C{B}_j \given X_i = \xi_i \big) \\
    &= \Prb_{\tilde{\C{M}}}\big( (\B{g}_{\C{A}_{j}^{\setminus i}})_j(\xi_i, \B{E}_{\widetilde{\pa}(\C{A}_{j}^{\setminus i})}) \in \C{B}_j \big) \\
    &= \Prb_{\tilde{\C{M}}_{\intervene(\{i\},\xi_i)}}\big( (\B{g}_{\C{A}_{j}^{\setminus i}})_j(X_i, \B{E}_{\widetilde{\pa}(\C{A}_{j}^{\setminus i})}) \in \C{B}_j \big) \\
    &= \Prb_{\tilde{\C{M}}_{\intervene(\{i\},\xi_i)}}\big( X_j \in \C{B}_j \big) \\
    &= \Prb_{\big(\C{M}_{\intervene(I,\B{\xi}_I)}\big)_{\intervene(\{i\},\xi_i)}}\big( X_j \in \C{B}_j \big) \,,
  \end{aligned}
$$
  where we used $X_i \indep \B{E}_{\widetilde{\pa}(\C{A}_{j}^{\setminus i})}$ in the fourth
equality. 
  
If we assume
  $i\notin\widetilde{\pa}(\C{A}_{j}^{\setminus i})\setminus\C{A}_{j}^{\setminus i}$ instead of $i\in\pa(\C{A}_{j}^{\setminus i})\setminus\C{A}_{j}^{\setminus i}$, then we similarly arrive at the same conclusion.
\end{proof}

%%%%%%%%%%%%%%%%%%%%%%%%%%%%%%%%%%%%%%%%%%%%%%%%%%
\subsubsection*{Section 8}
\addcontentsline{toc}{subsubsection}{\protect\numberline{}Section 8}%
%%%%%%%%%%%%%%%%%%%%%%%%%%%%%%%%%%%%%%%%%%%%%%%%%%

\begin{proof}[Proof of Proposition~\ref{prop:SimplenessClosedUnderResults}]
We first show that the class of simple SCMs is closed under marginalization. Take two disjoint subsets $\C{L}_1$ and $\C{L}_2$ in $\C{I}$. Then, it suffices to show that
$\C{M}_{\marg(\C{L}_1)}$ is uniquely solvable w.r.t.\ $\C{L}_2$. This follows directly from Proposition~\ref{prop:MarginalizationCommutes}.

To show that the class of simple SCMs is closed under perfect intervention. Let $\C{M}$ be a simple SCM, $\C{O}\subseteq\C{I}$, $I\subseteq\C{I}$ and $\B{\xi}_I\in\BC{X}_I$. Define $\C{O}_1:=\C{O}\cap I$ and $\C{O}_2:=\C{O}\setminus I$, then $\C{O}=\C{O}_1\cup\C{O}_2$. Note that $\pa(\C{O}_2)\setminus\C{O}_2 = (\pa(\C{O}_2)\setminus (\C{O}_2\cup I)) \cup (\pa(\C{O}_2)\cap I)$ and $\pa(\C{O}_2)\setminus (\C{O}_2\cup I) \subseteq \pa(\C{O})\setminus\C{O}$. Let $\B{g}_{\C{O}_2} : \BC{X}_{\pa(\C{O}_2)\setminus\C{O}_2} \times \BC{E}_{\pa(\C{O}_2)} \to \BC{X}_{\C{O}_2}$ be a measurable solution function for $\C{M}$ w.r.t.\ $\C{O}_2$. The mapping $\tilde{\B{g}}_{\C{O}}:\BC{X}_{\pa(\C{O})\setminus\C{O}}\times\BC{E}_{\pa(\C{O})}\to\BC{X}_{\C{O}}$ defined by 
$$
\left\{
\begin{aligned}
  (\tilde{\B{g}}_{\C{O}})_{\C{O}_1}(\B{x}_{\pa(\C{O})\setminus\C{O}},\B{e}_{\pa(\C{O})}) &:= \B{\xi}_{\C{O}_1} \\
  (\tilde{\B{g}}_{\C{O}})_{\C{O}_2}(\B{x}_{\pa(\C{O})\setminus\C{O}},\B{e}_{\pa(\C{O})}) &:=   \B{g}_{\C{O}_2}(\B{x}_{\pa(\C{O}_2)\setminus (\C{O}_2\cup I)},\B{\xi}_{\pa(\C{O}_2)\cap I}, \B{e}_{\pa(\C{O}_2)})
\end{aligned}
\right.
$$
is a measurable solution function for $\C{M}_{\intervene(I,\B{\xi}_I)}$ w.r.t.\ $\C{O}$, and it is clear that
$\C{M}_{\intervene(I,\B{\xi}_I)}$ is uniquely solvable w.r.t.\ $\C{O}$.

Next, we show that the class of simple SCMs is closed under the twin operation. Let $\tilde{\C{O}}\subseteq\C{I}\cup\C{I}'$. Take $\C{O}_1=\tilde{\C{O}}\cap\C{I}$, $\C{O}_2'=\tilde{\C{O}}\cap\C{I}'$ and $\C{O}_2$ the original copy of $\C{O}_2'$ in $\C{I}$. Let $\B{g}_{\C{O}_1}:\BC{X}_{\pa(\C{O}_1)\setminus\C{O}_1}\times\BC{E}_{\pa(\C{O}_1)}\to\BC{X}_{\C{O}_1}$ and
$\B{g}_{\C{O}_2}:\BC{X}_{\pa(\C{O}_2)\setminus\C{O}_2}\times\BC{E}_{\pa(\C{O}_2)}\to\BC{X}_{\C{O}_2}$ be measurable solution functions for $\C{M}$ w.r.t.\ $\C{O}_1$ and $\C{O}_2$, respectively. Define now the mapping $\B{h}_{\tilde{\C{O}}}:\BC{X}_{\widetilde{\pa}(\tilde{\C{O}})\setminus\tilde{\C{O}}}\times\BC{E}_{\widetilde{\pa}(\tilde{\C{O}})}\to\BC{X}_{\tilde{\C{O}}}$ by 
$$
\begin{aligned} 
  (\B{h}_{\tilde{\C{O}}})_{\tilde{\C{O}}\cap\C{I}}(\B{x}_{\widetilde{\pa}(\tilde{\C{O}})\setminus\tilde{\C{O}}},\B{e}_{\widetilde{\pa}(\tilde{\C{O}})}) &:=
  \B{g}_{\C{O}_1}(\B{x}_{\widetilde{\pa}(\C{O}_1)\setminus\C{O}_1},\B{e}_{\widetilde{\pa}(\C{O}_1)}) \\
  (\B{h}_{\tilde{\C{O}}})_{\tilde{\C{O}}\cap\C{I}'}(\B{x}_{\widetilde{\pa}(\tilde{\C{O}})\setminus\tilde{\C{O}}},\B{e}_{\widetilde{\pa}(\tilde{\C{O}})}) &:= 
  \B{g}_{\C{O}_2}(\B{x}_{\widetilde{\pa}(\C{O}_2')\setminus\C{O}_2'},\B{e}_{\widetilde{\pa}(\C{O}_2')}) \,,
\end{aligned}
$$
where we define $\widetilde{\pa}:=\pa_{\C{G}^a(\C{M}^{\twin})}$ as the parents w.r.t.\ the twin graph $\C{G}^a(\C{M}^{\twin})$. Then by construction this mapping $\B{h}_{\tilde{\C{O}}}$ is a measurable solution function for $\C{M}^{\twin}$ w.r.t.\ $\tilde{\C{O}}$, and it is clear that $\C{M}^{\twin}$ is uniquely solvable w.r.t.\ $\tilde{\C{O}}$.

Lastly, it follows that the observational and all the intervened models of $\C{M}$ and $\C{M}^{\twin}$ are uniquely solvable. From Theorem~\ref{thm:UniqueSolvabilityIffCondition} we conclude that $\C{M}$ induces unique observational, interventional and counterfactual distributions.
\end{proof}

\begin{proof}[Proof of Corollary~\ref{coro:gdgMarkovPropertySimpleSCM}]
This follows from Corollary~\ref{coro:gdgMarkovPropertyInterventionalSCM}.
\end{proof}

\section{Measurable selection theorems}
\label{app:AppendixMST}
\setcounter{section}{6}
\setcounter{theorem}{0}
%%%%%%%%%%%%%%%%%%%%%%%%%%%%%%%%%%%%%%%%%%%%%%%%%%

In this appendix, we derive some lemmas and state two measurable selection theorems that are
used in several proofs in Appendix~\ref{app:AppendixProofs}. First, we introduce the measure theoretic notation and terminology needed to understand the results (see \citep{Kec95} for more details).

%\begin{definition}[Measures]
%Consider a measurable space $(\BC{X},\B{\Sigma})$. A \emph{measure} on $(\BC{X},\B{\Sigma})$ is a map $\B{\mu}:\B{\Sigma}\to [0,\infty]$ such that $\B{\mu}(\emptyset)=0$ and $\B{\mu}(\cup_{n\in\NN}\BC{A}_n) = \sum_{n\in\NN} \B{\mu}(\BC{A}_n)$ for any pairwise disjoint sets $\{\BC{A}_n\}_{n\in\NN}\subseteq\B{\Sigma}$. A \emph{measure space} is a triple $(\BC{X},\B{\Sigma},\B{\mu})$, where $(\BC{X},\B{\Sigma})$ is a measurable space and $\B{\mu}$ is a measure on $(\BC{X},\B{\Sigma})$. A measure $\B{\mu}$ on $(\BC{X},\B{\Sigma})$ is called \emph{$\sigma$-finite} if $\BC{X}=\cup_{n\in\NN}\BC{A}_n$, with $\BC{A}_n\in\B{\Sigma}$, $\B{\mu}(\BC{A}_n) < \infty$, \emph{finite} if $\B{\mu}(\BC{X})<\infty$, and a \emph{probability} measure if $\B{\mu}(\BC{X})=1$.
%\end{definition}
%See Section~17.A in Kechris

\begin{definition}[Standard measurable space]
\label{def:StandardMeasurableSpace}
A measurable space $(\BC{X},\B{\Sigma})$ is a \emph{standard measurable space} if it is isomorphic to $(\BC{Y},\C{B}(\BC{Y}))$, where $\BC{Y}$ is a Polish space, that is, a separable completely metrizable space,\footnote{A \emph{metrizable space} is a topological space $\BC{X}$ for which there exists a metric $d$ such that $(\BC{X},d)$ is a metric space and induces the topology on $\BC{X}$. For a metric space $(\BC{X},d)$, a \emph{Cauchy sequence} is a sequence $(x_n)_{n\in\NN}$ of elements of $\BC{X}$ such that for every $\epsilon > 0$ there exists an $N\in\NN$ such that for all natural numbers $p,q > N$ we have $d(x_n,x_m)<\epsilon$. We call $(\BC{X},d)$ \emph{complete} if every Cauchy sequence has a limit in $\BC{X}$. A \emph{completely metrizable space} is a topological space $\BC{X}$ for which there exists a metric $d$ such that $(\BC{X},d)$ is a complete metric space that induces the topology on $\BC{X}$. A topological space $\BC{X}$ is called \emph{separable} if it contains a countable dense subset, that is, there exists a sequence $(x_n)_{n\in\NN}$ of elements in $\BC{X}$ such that every nonempty open subset of $\BC{X}$ contains at least one element of the sequence. A separable completely metrizable space is called a \emph{Polish space} (see \citep{Coh13} and \citep{Kec95} for more details).} and $\C{B}(\BC{Y})$ are the Borel subsets of $\BC{Y}$, that is, the $\sigma$-algebra generated by the open sets in $\BC{Y}$. A measure space $(\BC{X},\B{\Sigma},\B{\mu})$ is a \emph{standard probability space} if $(\BC{X},\B{\Sigma})$ is a standard measurable space and $\B{\mu}$ is a probability measure.
\end{definition}
%See Definition~3.1 and 12.5 in Kechris
Examples of standard measurable spaces are the open and closed subsets of $\RN^d$, and the finite sets with the usual complete metric. If we say that $\BC{X}$ is a standard measurable space, then we implicitly assume that there exists a $\sigma$-algebra $\B{\Sigma}$ such that $(\BC{X},\B{\Sigma})$ is a standard measurable space. Similarly, if we say that $\BC{X}$ is a standard probability space with probability measure $\Prb_{\BC{X}}$, then we implicitly assume that there exists a $\sigma$-algebra $\B{\Sigma}$ such that $(\BC{X},\B{\Sigma},\Prb_{\BC{X}})$ is a standard probability space.

\begin{definition}[Analytic set]
Let $\BC{X}$ be a Polish space. A set $\BC{A}\subseteq\BC{X}$ is called \emph{analytic} if there exist a Polish space $\BC{Y}$ and a continuous mapping $\B{f}:\BC{Y}\to\BC{X}$ with $\B{f}(\BC{Y})=\BC{A}$.
\end{definition}
%See Definition~14.1 in Kechris

\begin{lemma}
\label{lemm:AnalyticSetProperties}
Let $\BC{X}$ and $\BC{Y}$ be standard measurable spaces and $\B{f}:\BC{X}\to\BC{Y}$ a measurable mapping. Then
\begin{enumerate}
	\item every measurable set $\BC{A}\subseteq\BC{X}$ is analytic;
	\item if the subsets $\BC{A}\subseteq\BC{X}$ and $\BC{\tilde{A}}\subseteq\BC{Y}$ are analytic, then the sets $\B{f}(\BC{A})$ and $\B{f}^{-1}(\BC{\tilde{A}})$ are analytic.
\end{enumerate}
\end{lemma}
\begin{proof}
From Proposition~13.7 in~\citep{Kec95} it follows that every measurable set $\BC{A}\subseteq\BC{X}$ is analytic. From Proposition~14.4.(ii) in~\citep{Kec95} it follows that the image and the preimage of an analytic set is an analytic set.
\end{proof}

\begin{definition}[$\B{\mu}$-measurability]
Let $(\BC{X},\B{\Sigma},\B{\mu})$ be a measure space. A set $\BC{E} \subseteq \BC{X}$ is called 
a \emph{$\B{\mu}$-null set} if there exists a $\BC{A} \in \B{\Sigma}$ with 
$\BC{E} \subseteq \BC{A}$ and $\B{\mu}(\BC{A}) = 0$. We denote the class of $\B{\mu}$-null sets 
by $\BC{N}$, and we denote the $\sigma$-algebra generated by 
$\B{\Sigma}\cup\BC{N}$ by $\bar{\B{\Sigma}}$, 
and its members are called the \emph{$\B{\mu}$-measurable} sets. Note that each member of
$\bar{\B{\Sigma}}$ is of the form $\BC{A}\cup\BC{E}$ with $\BC{A}\in\B{\Sigma}$ and $\BC{E}\in\BC{N}$. 
The measure $\B{\mu}$ is extended to a measure $\bar{\B{\mu}}$ on $\bar{\B{\Sigma}}$, by
$\bar{\B{\mu}}(\BC{A}\cup\BC{E})=\B{\mu}(\BC{A})$ for every $\BC{A}\in\B{\Sigma}$ and $\BC{E}\in\BC{N}$, 
and is called its \emph{completion}. A mapping $\B{f} : \BC{X} \to \BC{Y}$ between measurable 
spaces is called \emph{$\B{\mu}$-measurable} if the inverse image $\B{f}^{-1}(\BC{C})$ of every
 measurable set $\BC{C} \subseteq \BC{Y}$ is $\B{\mu}$-measurable. 
\end{definition}
%See Section~17.A in Kechris

\begin{definition}[Universal measurability]
Let $(\BC{X},\B{\Sigma})$ be a standard measurable space. A set $\BC{A}\subseteq\BC{X}$ is called \emph{universally measurable} if it is $\B{\mu}$-measurable for every $\sigma$-finite measure\footnote{A measure $\B{\mu}$ on a measurable space $(\BC{X},\B{\Sigma})$ is called \emph{$\sigma$-finite} if $\BC{X}=\cup_{n\in\NN}\BC{A}_n$, with $\BC{A}_n\in\B{\Sigma}$, $\B{\mu}(\BC{A}_n) < \infty$.} $\B{\mu}$ on $\BC{X}$ (i.e., in particular every probability measure). A mapping $\B{f}:\BC{X}\to\BC{Y}$ between standard measurable spaces is \emph{universally measurable} if it is $\B{\mu}$-measurable for every $\sigma$-finite measure $\B{\mu}$.
\end{definition}
%See Section~21.D (and Exercise 17.3.ii for the equivalence between probability measures and $\sigma$-finite measures) in Kechris

\begin{lemma}
\label{lemm:AnalyticSetMeasurableApproxUpPset}
Let $\BC{E}$ be a standard probability space with probability measure $\Prb_{\BC{E}}$ and $\BC{A}\subseteq\BC{E}$ an analytic set. Then $\BC{A}$ is $\Prb_{\BC{E}}$-measurable and there exist measurable sets $\BC{S},\BC{T}\subseteq\BC{E}$ such that $\BC{S}\subseteq\BC{A}\subseteq\BC{T}$ and $\Prb_{\BC{E}}(\BC{S}) = \bar{\Prb}_{\BC{E}}(\BC{A}) = \Prb_{\BC{E}}(\BC{T})$, where $\bar{\Prb}_{\BC{E}}$ is the completion of $\Prb_{\BC{E}}$.
\end{lemma}
\begin{proof}
  Let $\BC{A}\subseteq\BC{E}$ be an analytic set. Since every analytic set in a standard measurable space is a universally measurable set (see Theorem~21.10 in~\citep{Kec95}), we know that $\BC{A}$ is a universally measurable set, and hence it is in particular a $\Prb_{\BC{E}}$-measurable set. Thus, there exist a measurable set $\BC{S}\subseteq\BC{E}$ and a $\Prb_{\BC{E}}$-null set $\BC{C}\subseteq\BC{E}$ such that $\BC{A}=\BC{S}\cup\BC{C}$ and $\bar{\Prb}_{\BC{E}}(\BC{A}) = \Prb_{\BC{E}}(\BC{S})$, where $\bar{\Prb}_{\BC{E}}$ is the completion of $\Prb_{\BC{E}}$. Moreover, there exists a measurable set $\tilde{\BC{C}}\subseteq\BC{E}$ such that $\BC{C}\subseteq\tilde{\BC{C}}$ and $\Prb_{\BC{E}}(\tilde{\BC{C}})=0$. Let $\BC{T}:=\BC{S}\cup\tilde{\BC{C}}$, then $\BC{A}\subseteq\BC{T}$ and $\Prb_{\BC{E}}(\BC{T})=\Prb_{\BC{E}}(\BC{S})$.
\end{proof}

\begin{lemma}
\label{lemm:CountGenMeasurability}
Let $\B{f} : \BC{X} \to \BC{Y}$ be a $\B{\mu}$-measurable mapping. If $\BC{Y}$ is countably 
generated, then there exists a measurable mapping $\B{g} : \BC{X} \to \BC{Y}$  such that 
$\B{f}(\B{x}) = \B{g}(\B{x})$ holds $\B{\mu}$-a.e..
\end{lemma}
\begin{proof}
Let the $\sigma$-algebra of $\BC{Y}$ be generated by the countable generating set 
$\{ \BC{C}_n \}_{n \in \NN}$. The $\B{\mu}$-measurable set $\B{f}^{-1}(\BC{C}_n) = 
\BC{A}_n \cup \BC{E}_n$ for some $\BC{A}_n \in \B{\Sigma}$ and some $\BC{E}_n \in \BC{N}$ 
and hence there is some $\BC{E}_n \subseteq \BC{B}_n \in \B{\Sigma}$ such that 
$\B{\mu}(\BC{B}_n) = 0$. Let $\hat{\BC{B}} = \cup_{n \in \NN} \BC{B}_n$, 
$\hat{\BC{A}}_n = \BC{A}_n \setminus \hat{\BC{B}}$ and $\hat{\BC{A}} = \cup_{n \in \NN} 
\hat{\BC{A}}_n$, then $\B{\mu}(\hat{\BC{B}}) = 0$, $\hat{\BC{A}}$ and $\hat{\BC{B}}$ are 
disjoint and $\BC{X} = \hat{\BC{A}} \cup \hat{\BC{B}}$. Now define the mapping 
$\B{g} : \BC{X} \to \BC{Y}$ by
$$
  \begin{aligned}
    \B{g}(\B{x}) := 
    \begin{cases} 
      \B{f}(\B{x}) &\text{if } \B{x} \in \hat{\BC{A}}, \\ 
      \B{y}_0 &\text{otherwise,} 
    \end{cases}
  \end{aligned}
$$
where for $\B{y}_0$ we can take an arbitrary point in $\BC{Y}$. This mapping $\B{g}$ is 
measurable since for each generator $\BC{C}_n$ we have
$$
  \begin{aligned}
    \B{g}^{-1}(\BC{C}_n)=
    \begin{cases} 
      \hat{\BC{A}}_n &\text{if }\B{y}_0 \notin \BC{C}_n, \\ 
      \hat{\BC{A}}_n \cup \hat{\BC{B}} &\text{otherwise.} 
    \end{cases}
  \end{aligned}
$$
is in $\B{\Sigma}$. Moreover, $\B{f}(\B{x}) = \B{g}(\B{x})$ $\B{\mu}$-almost everywhere. 
\end{proof}
%See here also Section~17.A in Kechris

With this result at hand we can now prove the first measurable selection theorem.
\begin{theorem}[Measurable selection theorem]
\label{thm:MeasurableSelectionThm}
Let $\BC{E}$ be a standard probability space with probability measure $\Prb_{\BC{E}}$, $\BC{X}$ a standard 
measurable space and $\BC{S} \subseteq \BC{E} \times \BC{X}$ a measurable set such that 
$\BC{E} \setminus \B{pr}_{\BC{E}}(\BC{S})$ is a $\Prb_{\BC{E}}$-null set, where $\B{pr}_{\BC{E}}:\BC{E}\times\BC{X}\to\BC{E}$ is
the projection mapping on $\BC{E}$. Then there 
exists a measurable mapping $\B{g} : \BC{E} \to \BC{X}$ such that $(\B{e},\B{g}(\B{e})) 
\in \BC{S}$ for $\Prb_{\BC{E}}$-almost every $\B{e}\in\BC{E}$.
\end{theorem}
\begin{proof}
Take the subset $\hat{\BC{E}} := \BC{E} \setminus \BC{B}$, for some measurable set $\BC{B} \supseteq \BC{E} 
\setminus \B{pr}_{\BC{E}}(\BC{S})$ and $\Prb_{\BC{E}}(\BC{B}) = 0$, and note that $\hat{\BC{E}}$ 
is a standard measurable space (see Corollary~13.4 in \citep{Kec95}) and $\hat{\BC{E}} \subseteq \B{pr}_{\BC{E}}(\BC{S})$. Let $\hat{\BC{S}} = 
\BC{S} \cap (\hat{\BC{E}} \times \BC{X})$. Because the set $\hat{\BC{S}}$ is measurable, it is 
in particular analytic (see Lemma~\ref{lemm:AnalyticSetProperties}). It follows by the 
Jankov-von Neumann Theorem (see Theorem~18.8 or 29.9 in \citep{Kec95}) that $\hat{\BC{S}}$ has a 
universally measurable uniformizing function, that is, there exists a universally measurable 
mapping $\hat{\B{g}} : \hat{\BC{E}} \to \BC{X}$ such that for all $\B{e} \in \hat{\BC{E}}$, 
$(\B{e},\hat{\B{g}}(\B{e})) \in \hat{\BC{S}}$. Hence,
in particular, it is $\Prb_{\BC{E}} \big|_{\hat{\BC{E}}}$-measurable, where $\Prb_{\BC{E}}
\big|_{\hat{\BC{E}}}$ is the restriction of $\Prb_{\BC{E}}$ to $\hat{\BC{E}}$.

Now define the mapping $\B{g}^*:\BC{E}\to\BC{X}$ by
$$
  \begin{aligned}
    \B{g}^*(\B{e}) := 
    \begin{cases} 
      \hat{\B{g}}(\B{e}) &\text{if } \B{e} \in \hat{\BC{E}} \\ 
      \B{x}_0 &\text{otherwise,}
    \end{cases} 
  \end{aligned}
$$
where for $\B{x}_0$ we can take an arbitrary point in $\BC{X}$. Then this mapping $\B{g}^*$ is
$\Prb_{\BC{E}}$-measurable. To see this, take any measurable set $\BC{C}\subseteq\BC{X}$, then 
$$
  \begin{aligned}
    \B{g}^{*-1}(\BC{C}) = 
    \begin{cases} 
      \hat{\B{g}}^{-1}(\BC{C}) &\text{if } \B{x_0} \notin \BC{C} \\ 
      \hat{\B{g}}^{-1}(\BC{C})\cup\BC{B} &\text{otherwise.}
    \end{cases} 
  \end{aligned}
$$
Because $\hat{\B{g}}^{-1}(\BC{C})$ is $\Prb_{\BC{E}} \big|_{\hat{\BC{E}}}$-measurable it is also
$\Prb_{\BC{E}}$-measurable and thus $\B{g}^{*-1}(\BC{C})$ is $\Prb_{\BC{E}}$-measurable.

By Lemma~\ref{lemm:CountGenMeasurability} and the fact that standard measurable spaces are 
countably generated (see Proposition~12.1 in \citep{Kec95}), we prove the existence of a measurable mapping $\B{g} : \BC{E} \to \BC{X}$ 
such that $\B{g}^* = \B{g}$ $\Prb_{\BC{E}}$-a.e.\ and thus it satisfies 
$(\B{e},\B{g}(\B{e})) \in \BC{S}$ for $\Prb_{\BC{E}}$-almost every $\B{e}\in\BC{E}$.
\end{proof}

This theorem rests on the assumption that the standard measurable space $\BC{E}$ has a probability measure
$\Prb_{\BC{E}}$. If this space becomes the product space $\BC{Y}\times\BC{E}$, for some standard measurable space $\BC{Y}$ where only the space $\BC{E}$ has a probability measure, then in general this theorem does not hold
anymore. However, if we assume in addition that the fibers of $\BC{S}$ in $\BC{Y}$ are $\sigma$-compact for $\Prb_{\BC{E}}$-almost every $\B{e}\in\BC{E}$ and for all $\B{x}\in\BC{X}$, then we can prove a second measurable selection theorem. A topological space is $\sigma$-compact if it is the union of 
countably many compact subspaces. For example, all countable discrete
spaces, every interval of the real line, and moreover all the Euclidean spaces are
$\sigma$-compact spaces.

\begin{theorem}[Second measurable selection theorem]
\label{thm:MeasurableSelectionThmSigmaCompact}
Let $\BC{E}$ be a standard probability space with probability measure $\Prb_{\BC{E}}$, $\BC{X}$ and $\BC{Y}$ standard measurable spaces and $\BC{S} \subseteq \BC{X} \times 
\BC{E} \times \BC{Y}$ a measurable set such that $\BC{E} \setminus \BC{K}_{\sigma}$ is a $\Prb_{\BC{E}}$-null set, 
where
$$
  \BC{K}_{\sigma} := \{ \B{e} \in \BC{E} \,:\, \forall\B{x} \in \BC{X} (\BC{S}_{(\B{x},\B{e})} 
  \text{ is nonempty and $\sigma$-compact}) \} \,,
$$
with $\BC{S}_{(\B{x},\B{e})}$ denoting the fiber over $(\B{x},\B{e})$, that is
$$
\BC{S}_{(\B{x},\B{e})} := \{ \B{y}\in\BC{Y} \,:\, (\B{x},\B{e},\B{y})\in\BC{S} \} \,.
$$
Then there exists a measurable mapping $\B{g} : \BC{X} \times \BC{E} \to \BC{Y}$ such that for 
$\Prb_{\BC{E}}$-almost every $\B{e}\in\BC{E}$ and for all $\B{x} \in \BC{X}$ we have 
$(\B{x},\B{e},\B{g}(\B{x},\B{e})) \in \BC{S}$.
\end{theorem}
\begin{proof}
Take the subset $\hat{\BC{E}} := \BC{E} \setminus \BC{B}$, for some measurable set $\BC{B} \supseteq \BC{E} 
\setminus \BC{K}_{\sigma}$ and $\Prb_{\BC{E}}(\BC{B}) = 0$. Note that $\hat{\BC{E}}$ 
is a standard measurable space, $\hat{\BC{E}} \subseteq \BC{K}_{\sigma}$ and $\hat{\BC{S}} = 
\BC{S} \cap (\BC{X} \times \hat{\BC{E}} \times \BC{Y})$ is measurable. By assumption, for 
each $(\B{x}, \B{e}) \in \BC{X}\times\hat{\BC{E}}$ the fiber $\hat{\BC{S}}_{(\B{x}, \B{e})}$ is nonempty 
and $\sigma$-compact and hence by applying the Theorem of Arsenin-Kunugui 
(see Theorem 35.46 in \citep{Kec95}) it follows that the set $\hat{\BC{S}}$ has a 
measurable uniformizing function, that is, there exists a measurable mapping $\hat{\B{g}} : \BC{X} \times \hat{\BC{E}} \to \BC{Y}$ 
such that for all $(\B{x},\B{e}) \in \BC{X} \times \hat{\BC{E}}$, 
$(\B{x},\B{e},\hat{\B{g}}(\B{x},\B{e})) \in \hat{\BC{S}}$. 
%Similar to the reasoning used in the proof of Theorem~\ref{thm:MeasurableSelectionThm} we can construct a mapping $\B{g}:\BC{X} \times \BC{E} \to \BC{Y}$ that inherits the measurability from $\hat{\B{g}}$ and satisfies that for $\Prb_{\BC{E}}$-almost every 
%$\B{e}$ and for all $\B{x} \in \BC{X}$, $(\B{x},\B{e},\B{g}(\B{x},\B{e})) \in \BC{S}$.
Now define the mapping $\B{g} : \BC{X} \times \BC{E} \to \BC{Y}$ by 
$$
  \begin{aligned}
    \B{g}(\B{x}, \B{e}) := 
    \begin{cases} 
      \hat{\B{g}}(\B{x}, \B{e}) &\text{if } \B{e} \in \hat{\BC{E}} \\ 
      \B{y}_0 &\text{otherwise,}
    \end{cases} 
  \end{aligned} 
$$
where for $\B{y}_0$ we can take an arbitrary point in $\BC{Y}$. This mapping $\B{g}$ inherits 
the measurability from $\hat{\B{g}}$ and it satisfies for $\Prb_{\BC{E}}$-almost every 
$\B{e}\in\BC{E}$ and for all $\B{x} \in \BC{X}$ that $(\B{x},\B{e},\B{g}(\B{x},\B{e})) \in \BC{S}$.
\end{proof}

The next two lemmas provide some useful properties for the ``for $\Prb_{\BC{E}}$-almost every $\B{e}\in\BC{E}$'' quantifier.
\begin{lemma}\label{lemm:MeasurableMapsBetweenStandardSpaces}
  Let $\B{\phi} : \BC{E} \to \tilde{\BC{E}}$ be a measurable map between 
  two standard measurable spaces. Let $\Prb_{\BC{E}}$ be a probability measure on $\BC{E}$ and let $\Prb_{\tilde{\BC{E}}} = \Prb_{\BC{E}} \circ \B{\phi}^{-1}$ be its push-forward under $\B{\phi}$. 
  Let $\tilde P:\tilde{\BC{E}}\to\{0,1\}$ be a property, that is, a (measurable) boolean-valued function on $\tilde{\BC{E}}$. Then the property $P = \tilde P\circ\B{\phi}$ on $\BC{E}$ holds $\Prb_{\BC{E}}$-a.e.\ if and only if the property $\tilde P$ holds $\Prb_{\tilde{\BC{E}}}$-a.e..
\end{lemma}
\begin{proof}
Assume the property $P = \tilde P\circ\B{\phi}$ holds $\Prb_{\BC{E}}$-a.e., then $\BC{C}=\{\B{e}\in\BC{E} : P(\B{e}) = 1 \}$ contains a measurable set $\BC{C}^*$ with $\Prb_{\BC{E}}$-measure 1, that is, $\BC{C}^* \subseteq \BC{C}$ and $\Prb_{\BC{E}}(\BC{C}^*) = 1$.
By Lemma~\ref{lemm:AnalyticSetProperties}, $\B{\phi}(\BC{C}^*)$ is analytic. By Lemma~\ref{lemm:AnalyticSetMeasurableApproxUpPset}, there exist measurable sets $\BC{A}, \BC{B}$ such that $\BC{A} \subseteq \B{\phi}(\BC{C}^*) \subseteq \BC{B}$ and $\Prb_{\tilde{\BC{E}}}(\BC{A}) = \Prb_{\tilde{\BC{E}}}(\BC{B})$. Because $\B{\phi}$ is measurable, $\B{\phi}^{-1}(\BC{A})$ and $\B{\phi}^{-1}(\BC{B})$ are both measurable. Also, $\B{\phi}^{-1}(\BC{A}) \subseteq \B{\phi}^{-1}(\B{\phi}(\BC{C}^*)) \subseteq \B{\phi}^{-1}(\BC{B})$. As $\BC{C}^* \subseteq \B{\phi}^{-1}(\B{\phi}(\BC{C}^*))$, we must have that $\Prb_{\BC{E}}(\B{\phi}^{-1}(\BC{B})) \ge \Prb_{\BC{E}}(\BC{C}^*) = 1$. Hence $\Prb_{\tilde{\BC{E}}}(\BC{A}) = \Prb_{\tilde{\BC{E}}}(\BC{B}) = 1$. 
  Note that as $\BC{C}^* \subseteq \BC{C}$, $\BC{A} \subseteq \B{\phi}(\BC{C}^*) \subseteq \B{\phi}(\BC{C})\subseteq\{\tilde{\B{e}}\in\tilde{\BC{E}} : \tilde P(\tilde{\B{e}}) = 1 \}$. Hence the set $\tilde{\BC{C}} := \{\tilde{\B{e}}\in\tilde{\BC{E}} : \tilde P(\tilde{\B{e}}) = 1 \}$ contains a measurable set of $\Prb_{\tilde{\BC{E}}}$-measure 1, in other words, $\tilde P$ holds $\Prb_{\tilde{\BC{E}}}$-a.s..

  The converse is easier to prove. Suppose $\tilde{\BC{C}}=\{\tilde{\B{e}}\in\tilde{\BC{E}} : \tilde P(\tilde{\B{e}}) = 1 \}$ contains a measurable set $\tilde{\BC{C}}^*$ with $\Prb_{\tilde{\BC{E}}}$-measure 1, that is, $\tilde{\BC{C}}^* \subseteq \tilde{\BC{C}}$ and $\Prb_{\tilde{\BC{E}}}(\tilde{\BC{C}}^*) = 1$. Because $\B{\phi}$ is measurable, the set $\B{\phi}^{-1}(\tilde{\BC{C}}^*)$ is measurable and $\Prb_{\BC{E}}(\B{\phi}^{-1}(\tilde{\BC{C}}^*)) = 1$, and furthermore, $\B{\phi}^{-1}(\tilde{\BC{C}}^*) \subseteq \B{\phi}^{-1}(\tilde{\BC{C}})=\BC{C}$.
\end{proof}

\begin{lemma}[Some properties for the for-almost-every quantifier]
  \label{lemm:AlmostAllQuantifierLogic}
Let $\BC{X}=\C{X}\times\tilde{\C{X}}$ and $\BC{E}=\C{E}\times\tilde{\C{E}}$ be products of nonempty standard measurable spaces and $\Prb_{\BC{E}}=\Prb_{\C{E}}\times\Prb_{\tilde{\C{E}}}$ be the product measure of probability measures $\Prb_{\C{E}}$ and $\Prb_{\tilde{\C{E}}}$ on $\C{E}$ and $\tilde{\C{E}}$, respectively. Denote by ``$\foralmostall\B{e}$'' the quantifier ``for $\Prb_{\BC{E}}$-almost every $\B{e}\in\BC{E}$'' and by ``$\forall\B{x}$'' the quantifier ``for all $\B{x}\in\BC{X}$'', and similarly for their components, for example, ``$\foralmostall e$'' for ``for $\Prb_{\C{E}}$-almost every $e\in\C{E}$'' and ``$\forall x$'' for ``for all $x\in \C{X}$''. Then we have the following properties:
  \begin{enumerate}[ref=\thelemma.(\arabic*)]
    \item \label{lemm:AlmostAllQuantifierLogic1} $\foralmostall e: P(e) \implies \exists e: P(e)$ \,(similarly to $\forall x: P(x) \implies \exists x: P(x)$);
    \item \label{lemm:AlmostAllQuantifierLogic2} $\foralmostall e: P(e) \iff \foralmostall\B{e}: P(e)$ \,(similarly to $\forall x: P(x) \iff \forall\B{x}: P(x)$);
    \item \label{lemm:AlmostAllQuantifierLogic3} $\exists x \foralmostall e: P(x,e) \implies \foralmostall e \exists x: P(x,e)$ \,(similarly to $\exists x \forall e: P(x,e) \implies \forall e \exists x: P(x,e)$);
    \item \label{lemm:AlmostAllQuantifierLogic4} $\foralmostall e \forall x: P(x,e) \implies \forall x \foralmostall e: P(x,e)$ \,(similarly to $\forall e \forall x: P(x,e) \implies \forall x \forall e: P(x,e)$);
    \item \label{lemm:AlmostAllQuantifierLogic5} $\foralmostall\B{e}: P(\B{e}) \implies \exists \tilde{e} \foralmostall e: P(\B{e})$ \,(similarly to $\forall \B{x}: P(\B{x}) \implies \exists \tilde{x} \forall x: P(\B{x})$);
    \item \label{lemm:AlmostAllQuantifierLogic6} $\foralmostall e\forall x: P(x,e) \iff \foralmostall\B{e}\forall\B{x}: P(x,e)$;
    \item \label{lemm:AlmostAllQuantifierLogic7} $\foralmostall\B{e}\forall\B{x}: P(\B{x},\B{e}) \implies \exists \tilde{e} \exists \tilde{x} \foralmostall e\forall x: P(\B{x},\B{e})$,
 \end{enumerate}
where $P$ denotes a property, that is, a measurable boolean-valued function, on the corresponding measurable spaces and we write $\B{e}$ and $\B{x}$ for $(e,\tilde{e})$ and $(x,\tilde{x})$, respectively.
\end{lemma}
\Joris{I shortened the proof, only stating the non-obvious.}
\begin{proof}
We only prove the statements that may not be immediately obvious.

Property 2. Let $pr_{\C{E}}:\BC{E}\to\C{E}$ be the projection mapping on $\C{E}$. 
Then by Lemma~\ref{lemm:MeasurableMapsBetweenStandardSpaces} we have 
$$
   \foralmostall e : P(e) 
   \iff \foralmostall {\B{e}} : P\circ\text{pr}_{\C{E}}(\B{e}) \iff \foralmostall {\B{e}} : P(e) \,.
$$

Property 4: We have
$$
\begin{aligned}
 & \foralmostall e \forall x : P(x,e) \\
  &\implies \exists {\,\text{$\Prb_{\C{E}}$-null set $N$}\,} \forall {e \in \C{E}\setminus N} \, \forall x : P(x,e) \\
  &\implies \exists {\,\text{$\Prb_{\C{E}}$-null set $N$}\,} \forall x \, \forall {e \in \C{E}\setminus N} : P(x,e) \\
  &\implies \forall x \, \exists {\,\text{$\Prb_{\C{E}}$-null set $N$}\,} \forall { e \in \C{E}\setminus N} : P(x,e) \\
&\implies \forall x \foralmostall e : P(x,e) \,.
\end{aligned}
$$
%where we used Property~3 in the third implication.

Property 5:
Let $\B{N}$ be a measurable $\Prb_{\BC{E}}$-null set such that $P(\B{e})$ holds for all $\B{e}\in\BC{E}\setminus \B{N}$. Define for $\tilde{e}\in\tilde{\C{E}}$ the set $N_{\tilde{e}}:=\{ e \in \C{E} : (e,\tilde{e}) \in \B{N}\}$. Note that the sets $N_{\tilde{e}}$ are measurable. % (see Proposition~3.3.2 in~\citet{Bog07}).
From Fubini's theorem %(e.g., Theorem~3.4.1 in~\citet{Bog07}) 
it follows that for $\Prb_{\tilde{\C{E}}}$-almost every $\tilde{e}\in\tilde{\C{E}}$ we have $\Prb_{\C{E}}(N_{\tilde{e}})=0$. That is, there exists a measurable $\Prb_{\tilde{\C{E}}}$-null set $\tilde{N}$ such that $\Prb_{\C{E}}(N_{\tilde{e}})=0$ for all $\tilde{e}\in\tilde{\C{E}}\setminus\tilde{N}$.
Hence, there exists $\tilde{e}\in\tilde{\C{E}}\setminus\tilde{N}$ such that $\Prb_{\C{E}}(N_{\tilde{e}})=0$; for
all $e \in \C{E}\setminus N_{\tilde{e}}$, $P(\B{e})$ then holds.
This means $\exists \tilde{e} \foralmostall e : P(\B{e})$.
%
%Note that the set $(\BC{E}\setminus \B{N}) \cap \C{E}\times(\tilde{\C{E}}\setminus \tilde{N})$ is non-empty and thus we have in particular
%$$
%  \begin{aligned}
%  & \forall \tilde{e} \forall e \, \big( (e,\tilde{e}) \in (\BC{E}\setminus \B{N}) \cap \C{E}\times(\tilde{\C{E}}\setminus \tilde{N}) \implies P(\B{e}) \big) \\
%    & \implies \exists \tilde{e} \forall e \, \big( (e,\tilde{e}) \in (\BC{E}\setminus \B{N}) \cap \C{E}\times(\tilde{\C{E}}\setminus \tilde{N}) \implies P(\B{e}) \big)  \\
%    & \implies \exists \tilde{e} \forall e\in\C{E}\setminus N_{\tilde{e}} \, P(\B{e}) \,,
%  \end{aligned}
%$$
%holds. Hence
%$$
%\begin{aligned}
%  & \exists \tilde{e} \exists_{\text{$\Prb_{\C{E}}$-null set $N_{\tilde{e}}$}}\, \forall e\in\C{E}\setminus N_{\tilde{e}} \, P(\B{e}) \\
%  & \implies \exists \tilde{e} \foralmostall e \, P(\B{e}) \,.
%\end{aligned}
%$$

Property 7: We have 
$$
  \begin{aligned}
    \foralmostall {\B{e}} \forall {\B{x}} : P(\B{x},\B{e}) 
    &\implies \exists {\tilde{e}} \foralmostall e \forall {\B{x}} : P(\B{x},\B{e}) 
  \implies \exists {\tilde{e}} \foralmostall e \forall {\tilde{x}} \forall x : P(\B{x},\B{e}) \\
    &\implies \exists {\tilde{e}} \forall {\tilde{x}} \foralmostall e \forall x : P(\B{x},\B{e}) 
  \implies \exists {\tilde{e}} \exists {\tilde{x}} \foralmostall e \forall x : P(\B{x},\B{e}) \,,
  \end{aligned}
$$
where in the first equivalence we used Property~5, in the third equivalence we used Property~4 and in the last equivalence we used Property~1.
\end{proof}

\bibliographystyle{imsart-number}
%\bibliographystyle{imsart-nameyear}
%\bibliography{marginal_scm}

\end{document}